\pdfoutput=1
\pdfoutput=1

\ifdefined\PaperPublication\else
\def\PaperPublication{}
\def\PaperAuthorNames{Tristan Simas}
\def\PaperZenodoURL{https://doi.org/10.5281/zenodo.23900405}

\long\def\PaperPublicationAuthor#1{\ifcsname PaperPublicationAuthor@#1\endcsname\csname PaperPublicationAuthor@#1\endcsname\else\errmessage{Unsupported publication author adapter: #1}\fi}
\expandafter\def\csname PaperPublicationAuthor@article\endcsname{\author{Tristan Simas\\McGill University\\\texttt{tristan.simas@mail.mcgill.ca}}}
\expandafter\def\csname PaperPublicationAuthor@acmart\endcsname{\author{Tristan Simas}\affiliation{\institution{McGill University}\country{Canada}}\email{tristan.simas@mail.mcgill.ca}}
\expandafter\def\csname PaperPublicationAuthor@jair\endcsname{\author{Tristan Simas}\affiliation{\institution{McGill University}\country{Canada}}\email{tristan.simas@mail.mcgill.ca}}
\expandafter\def\csname PaperPublicationAuthor@IEEEtran\endcsname{\author{Tristan Simas\thanks{McGill University; tristan.simas@mail.mcgill.ca}}}
\expandafter\def\csname PaperPublicationAuthor@siamart\endcsname{\author{Tristan Simas\thanks{McGill University; tristan.simas@mail.mcgill.ca}}}
\expandafter\def\csname PaperPublicationAuthor@elsarticle\endcsname{\author{Tristan Simas}\address{McGill University}\ead{tristan.simas@mail.mcgill.ca}}
\expandafter\def\csname PaperPublicationAuthor@sn-jnl\endcsname{\author*[1]{Tristan Simas}\email{tristan.simas@mail.mcgill.ca}\affil*[1]{\orgname{McGill University},\orgaddress{\country{Canada}}}}
\AddToHook{begindocument/before}{\ifdefined\hypersetup\hypersetup{pdfauthor={Tristan Simas}}\fi}
\AddToHook{cmd/maketitle/after}{\ifdefined\hypersetup\hypersetup{pdfauthor={Tristan Simas}}\fi}
\fi
\renewcommand{\PaperPublicationAuthor}[1]{\csname PaperPublicationAuthor@article\endcsname}

\documentclass{paper-preprint}

\setcopyright{cc}
\acmDOI{}
\JAIRAE{}
\JAIRTrack{}
\acmVolume{}
\acmArticle{}
\acmMonth{9}
\acmYear{2026}

\RequirePackage[utf8]{inputenc}
\RequirePackage[T1]{fontenc}
\RequirePackage[
  datamodel=acmdatamodel,
  style=acmauthoryear,
  backend=biber,
  giveninits=true,
  uniquename=init
]{biblatex}

\DefineBibliographyStrings{english}{andothers={et\addabbrvspace al}}
\def\RunInFootnotes{}
\usepackage{compat-preamble}
\usepackage{float}
\usepackage{placeins}
\IfFileExists{content/lean_stats.tex}{

}{}
\providecommand{\ExecutableAuditInventory}{the Boolean restricted-state and
sparse-support reductions; the observation-coverage parity reduction; the
three-action selector and its product completion; the two-action selector;
the multistep QBF reduction; the horizon-two individual-bit reduction; both
tolerance reductions; the antichain approximation census; and the frozen
learned-policy comparison}

\newif\ifappendixtolerance
\appendixtolerancefalse

\providecommand{\PaperTitleAuto}{Less Is Harder Than More: Proving Unnecessary Complexity}

\title[Proving Unnecessary Complexity]{\PaperTitleAuto}

\PaperPublicationAuthor{jair}

\renewcommand{\shortauthors}{Simas}

\begin{document}
\raggedbottom
\vfuzz=8pt

\begin{abstract}

When can a sensor, feature, or API field be removed without changing a fixed
system's decisions? We require the remaining inputs to recover its complete
optimal-action set, including ties, at every admitted state.

Admitting every input combination guarantees one smallest preserving set.
Yet checking one deletion is coNP-complete when circuits give utilities on
Boolean products, even with two actions and at most two optimal-action sets.
Missing combinations permit alternative preserving sets, so an auditor must
choose inputs within budget and check the chosen set. Requiring every
declared reading to have a positive-probability state adds a support check.
We classify choosing and checking exactly for circuit-defined Boolean models,
with and without the support requirement.

A unique winner after averaging hidden states can still change the original
decisions. Certifying such winners requires weighted score comparisons:
finding inputs within budget is
$\mathsf{NP}^{(\mathsf{NP}^{\mathsf{PP}})}$-complete even with two actions and
independent uniform inputs. Explicit tables admit logarithmic approximation
of the minimum preserving input count; succinct Boolean products admit no
polynomial-time $o(\Lambda/\log\Lambda)$ approximation in total encoded
length $\Lambda$, unless $\mathsf P=\mathsf{NP}$.

Supplied decompositions and symmetry proofs can make audits polynomial-time;
discovery and validation are additional work. For a frozen digit
classifier, eleven blocks suffice on held-out rows but fifteen are needed on
the full product; four conflicting pairs expose the gap. Building the required response
from retained inputs establishes sufficiency, not minimality. Preserve the
structure that makes the dependence checkable.

\end{abstract}

\maketitle

\section{Introduction}\label{sec:introduction}

Suppose a delivery form records an address and whether a signature is
required. Removing the signature field makes two otherwise identical orders
look the same, although one requires a handover and the other permits leaving
the parcel. No rule reading only the address can recover both instructions.
The same problem arises when removing a sensor, model feature, or API field:
the remaining inputs must still determine the response a fixed system
requires on every admitted input. We study removal of a field, rather than
changing the value entered in it.

\paragraph{An unused input is easier to remove than an unnecessary one.}
If a private parameter is never read, its deletion can be checked by
inspection.  If an input affects an existing system, removing it requires a
stronger argument: the remaining inputs must still determine the required
response on every input the interface admits.  Passing tests establishes
that claim only for the calls that ran.  Likewise, a feature ablation in a
frozen model may preserve decisions on a benchmark while
failing on another input admitted by the deployment schema.

Removing inputs merges states that the remaining inputs cannot distinguish.
If two merged states require different responses, the removal fails its
preservation contract.

\paragraph{A successful deletion and a failed one.}
Consider two Boolean hazard sensors, $x_1$ and $x_2$, and three actions:
brake, steer, and continue.  Give the four hazard configurations equal
probability.  Figure~\ref{fig:input-deletion} shows what deleting $x_2$ loses.
Six additional diagnostic sensors affect none of the controller's scores;
those six can be removed without changing any required action.

\begin{figure}[H]
\centering
\includegraphics[width=\linewidth]{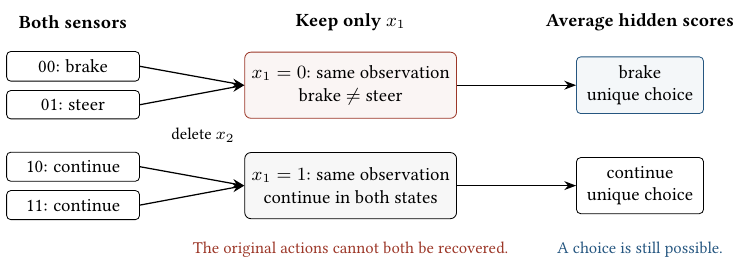}
\caption{A definite choice can still lose a required answer.  Deleting
$x_2$ merges 00 (brake) with 01 (steer).  No choice based on $x_1$ alone can
recover both original decisions.  Averaging selects brake for this group
and continue for the other: decisiveness holds, but preservation fails.}
\Description{Four equally probable states become two observations after the
second sensor is deleted. States 00 and 01 require brake and steer but both
give observation zero. Their average scores select brake. States 10 and 11
both require continue and give observation one, whose average also selects
continue.}
\label{fig:input-deletion}
\end{figure}

We call recovery of the complete original optimal-action set
\emph{preservation}, including every action tied for the highest score.
Averaging over hidden inputs asks a different question.  Each observation of
$x_1$ gives a unique conditional choice, a property we call
\emph{decisiveness}.

For hazard configurations
00, 01, 10, and 11, the utility vectors in the order brake, steer, continue are
$(4,0,0),(0,2,0),(0,0,4),(0,0,4)$.
Averaging over the hidden $x_2$ gives $(2,1,0)$ when $x_1=0$ and
$(0,0,4)$ when $x_1=1$.  Thus retiring $x_2$ changes the decision at 01,
an event of probability $1/4$.  The eight-sensor controller will remain our
running example; its full distribution is uniform.

Averaging over all four hazard states gives $(1,1/2,2)$.  Even the empty
interface is therefore decisive, while preservation still needs both hazard
sensors.

The controller's scores already explain why the six diagnostic sensors can
be omitted: every score is determined by the two hazard readings. Each hazard
sensor also has a witness for keeping it. With the second reading fixed at
zero, changing the first changes brake to continue. With the first fixed at
zero, changing the second changes brake to steer. Thus the construction
establishes that the two readings suffice, and the conflicting pairs establish
that neither can be removed. The reasons cover every diagnostic setting;
they do not depend on how many example readings we test.

\paragraph{From passing tests to checking all admitted inputs.}
The same question arises for a frozen digit classifier with sixteen image
blocks, each taking one of four values.  On its 510 distinct held-out inputs,
the smallest preserving interface keeps eleven blocks; if the API admits all
$4^{16}$ combinations, the minimum is fifteen (Figure~\ref{fig:learned-policy-scope-audit}).
The four added blocks each have a pair of inputs that exposes the failed
deletion.  The classifier is unchanged: the second claim covers more inputs.

Its decision-tree structure makes the larger audit possible without visiting
all $4^{16}$ inputs.  Its leaves cover the product with boxes of allowed input
values, each carrying one prediction.  Two boxes with different predictions refute a deletion if
their allowed values overlap on every retained input: choose agreeing values
there, and the two states become indistinguishable.  Checking the leaf boxes
therefore settles the whole-domain claim.  Section~\ref{sec:learned-policy-audit}
gives the model and complete audit.

\begin{figure}[t]
  \centering
  \includegraphics[width=0.78\linewidth]{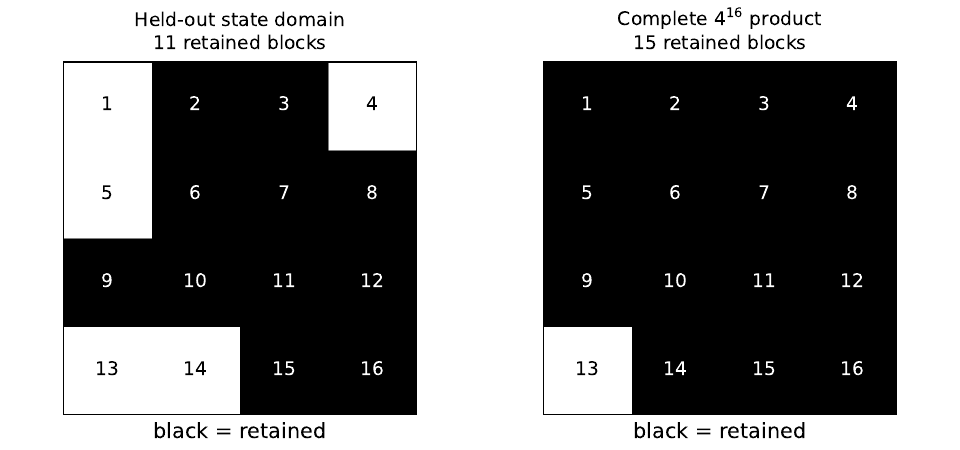}
  \par\smallskip
  \(x=\begin{pmatrix}
  0&0&0&0\\0&2&0&0\\0&2&0&0\\0&0&0&1
  \end{pmatrix}\ \mapsto\ 6
  \qquad
  x'=\begin{pmatrix}
  1&0&0&0\\0&2&0&0\\0&2&0&0\\0&0&0&1
  \end{pmatrix}\ \mapsto\ 2\)
  \Description{Two four-by-four grids show which blocks must be retained:
  eleven on the held-out table and fifteen on the complete product.
  Below them, two product inputs differ only in the top-left block, omitted
  by the listed certificate. The frozen classifier predicts six and two.}
  \caption{Same classifier, different certificate scope.  Eleven blocks
  suffice on the held-out inputs; fifteen are necessary on all admitted
  combinations.  The two arrays disagree only in the top-left block, which
  the eleven-block interface omits, but receive different predictions.
  They refute that deletion on the product.  Grid labels number blocks
  from one to sixteen; array entries are quantized input values.}
  \label{fig:learned-policy-scope-audit}
\end{figure}

Keeping all state-identifying inputs preserves the original optimal-action
set. Certifying a deletion instead requires agreement at every admitted pair
of states that the retained inputs cannot distinguish. Write $\Opt(s)$ for
the complete optimal-action set, including ties. An interface is the set of
retained inputs; a fiber is the set of states giving the same retained
observation.

The state domain changes that audit. Suppose the controller is admitted only
at hazard settings 00 and 11. Either hazard sensor distinguishes brake from
continue, so either can be kept alone. Admitting the mixed settings makes
both sensors necessary. The other six sensors still affect no score.
In general, a complete product---every combination of input values---has one
smallest sufficient interface: exactly the inputs whose change can affect
$\Opt$. Missing combinations can leave several alternatives.

Building the required response from chosen inputs establishes why those
inputs suffice. This does not select the smallest possible interface, and
the construction must still implement the response the task requires. It gives
the auditor an established relationship instead of leaving that relationship
to be recovered from an arbitrary evaluator.

The cost of recovering it depends on what the auditor receives. A complete
table exposes every state; a compact circuit can answer at one state without
listing all the others. Even checking one input's removal from such a circuit
is coNP-complete on Boolean products with two actions
(Theorem~\ref{thm:sufficiency-conp}). Missing combinations introduce a further
choice between sufficient interfaces: selecting one within a bit budget is
$\Sigma_2^P$-complete under the circuit-defined restricted-state encoding
(Theorem~\ref{thm:listed-minsuff-sigma2p}). The short reduction explains this
gap between computing answers and establishing which information they need.

Proving an input is unnecessary can be more expensive than adding it.
Passing along an already available input may be simple; safely leaving it out
requires establishing that no allowed case needs the distinction it carries.
This is not a claim that acquiring or retaining every sensor is cheap.

Useful structure supplies a constructive route. Decompositions, symmetry and
the classifier's leaf boxes let an auditor settle claims across the declared
domain without visiting every state. We present those mechanisms in
Section~\ref{sec:tractable}, before the extended classifications.
Section~\ref{sec:structural-consequences} combines them into an audit that
returns a preservation certificate, a conflicting state pair, or no verdict.
The later results ask what changes when an observation needs a supported state
behind it, hidden-state scores are averaged, or an input can guide a future
action. Table~\ref{tab:central-questions} identifies the principal audits;
Section~\ref{sec:stochastic-contracts} defines their support and conditional-score
requirements.

Explaining one supplied prediction and recovering a fixed response across
every admitted state are different tasks. Section~\ref{sec:contribution-boundary}
compares this global coordinate audit with particular local-explanation tasks,
internal-subcircuit selection, value of information, and planning complexity.

Each compact proof-handle footnote points to a finite Lean declaration;
Section~\ref{sec:verification-scope} distinguishes those checks from the
manuscript's complexity proofs.

\paragraph{Reading route.}

For preservation and practical audits, read
Sections~\ref{sec:formal-setup} and~\ref{sec:tractable}; then read
Section~\ref{sec:stochastic-contracts} for support and conditional scores.
The principal classifications are
Theorems~\ref{thm:listed-minsuff-sigma2p},
\ref{thm:stochastic-preservation-tight}, and
\ref{prop:two-action-exact-counting};
Table~\ref{tab:regime-matrix} collects the exact queries.
Section~\ref{sec:multistep-pomdp} extends the audit to future decisions.
The supplement gives the longest derivations under \emph{Fixed-block
normalization}, \emph{Within-Budget Decisiveness Under Conditional Scores},
\emph{Product completion}, and \emph{Full Proofs for the Multistep Stochastic
Classification}.
All completeness claims use polynomial-time many-one reductions. Each model
is specified where its audit is introduced. Throughout the multistep audit,
the retained input set stays fixed; selecting different sensors as events
unfold is a different problem.

\ifappendixtolerance\else
\section{What Input Removal Must Preserve}\label{sec:formal-setup}

The hazard controller gives the basic test: states that look the same after
deletion must retain the same complete set of best actions, including ties.
A sensor records one input coordinate; a state records all their values;
utility gives the score of each action. The definitions below make this test
precise for finite systems.

\subsection{Decision Problems with Coordinate Structure}

Start with all combinations of input values. A second presentation
will allow some combinations to be absent.

\begin{definition}[Decision Problem]\label{def:decision-problem}
A \emph{decision problem with product-coordinate structure} is a tuple
$\mathcal D=(A,X_1,\ldots,X_n,U)$ with finite nonempty action set $A$, finite
nonempty coordinate spaces $X_i$, state space
$S=X_1\times\cdots\times X_n$, and utility $U:A\times S\to\mathbb Q$.
\end{definition}

Rational or dyadic utilities permit exact encoded comparisons.

\begin{definition}[Finite Coordinate Presentation]\label{def:coordinate-presentation}
A \emph{finite coordinate presentation} consists of a finite nonempty state set $S$,
finite coordinate spaces $X_1,\ldots,X_n$, coordinate maps
$\pi_i:S\to X_i$, a finite nonempty action set $A$, and a utility function
$U:A\times S\to\mathbb{Q}$. A product-coordinate decision problem is the
special case $S=X_1\times\cdots\times X_n$ with the canonical projections.
\end{definition}

Only products guarantee that replacing one coordinate yields another state.
Checking whether agreeing states require the same response also works on
listed rows with missing combinations, as does the set-cover reduction in
Section~\ref{sec:approximation-across-presentations}.

\begin{definition}[Projection]\label{def:projection}
For $s\in S$ and $I\subseteq\{1,\ldots,n\}$, the projection of $s$ onto $I$
is $s_I=(\pi_i(s))_{i\in I}$.
For $I=\varnothing$, this is the unique empty tuple, so all states belong to
one projection fiber.
\end{definition}

\begin{definition}[Optimizer Map]\label{def:optimizer}
For $s\in S$, the \emph{optimal-action set} is
\[
\Opt(s)=\arg\max_{a\in A}U(a,s).
\]
\end{definition}

\subsection{Recovering Responses from Retained Inputs}

For a frozen classifier, the required response is its prediction. For a
controller, it may be a set of tied best actions. A finite function $F:S\to O$
is included by taking $A=O$ and
$U(a,s)=\mathbf1[a=F(s)]$: then $\Opt(s)=\{F(s)\}$. The encoding still
determines the cost of checking preservation.

Group states by the response they require. The resulting groups record just
enough information to recover that response; the formal name is the
\emph{optimizer quotient}.

\begin{definition}[Optimizer Quotient]\label{def:decision-equiv}\label{def:decision-quotient}
States $s,s'\in S$ are \emph{optimizer-equivalent} when
$\Opt(s)=\Opt(s')$.  Writing this kernel relation as
$s\equiv_{\Opt}s'$, the \emph{optimizer quotient} is
\[
Q_{\mathcal D}=S/{\equiv_{\Opt}},
\]
with quotient map $q_{\Opt}:S\to Q_{\mathcal D}$.
\end{definition}

Write $\pi(s)$ for what remains visible and $h$ for a rule that recovers the
required response from it. Recovery means $h(\pi(s))=\Opt(s)$ at every state.
This recovery criterion also holds for real utilities.
The kernel $\ker(\pi)$ records the pairs that look identical through $\pi$;
the theorem says none may require different responses.
Representation $\pi_1$ is finer than $\pi_2$ when $\pi_2$ can be recovered
from $\pi_1$; formally, $\pi_2$ factors through $\pi_1$.

\leanmetapending{\LH{UB1}, \LHrng{QT}{5}{7}}
\begin{theorem}[Kernel Factorization for the Optimizer Map]\label{thm:master-factorization}
For every map $\pi:S\to Q$,
\[
\bigl(\exists\,h:Q\to 2^A\setminus\{\varnothing\},
       \ \Opt=h\circ\pi\bigr)
\quad\Longleftrightarrow\quad
\ker(\pi)\subseteq {\equiv_{\Opt}}.
\]
Moreover, $Q_{\mathcal D}$ is canonically isomorphic to
$\operatorname{range}(\Opt)$.  If
$\pi$ is surjective and satisfies the equivalent conditions above, there is a
unique map $u:Q\to Q_{\mathcal D}$ with $q_{\Opt}=u\circ\pi$.
Consequently, $Q_{\mathcal D}$ is the unique coarsest exact abstraction, up to
canonical isomorphism.
\end{theorem}

\begin{proof}
Factorization through $\pi$ implies kernel containment.  Under kernel
containment, $h(\pi(s)):=\Opt(s)$ is well defined.  The map
$h$ may be extended arbitrarily off $\operatorname{range}(\pi)$.  The
canonical bijection from $Q_{\mathcal D}$ to $\operatorname{range}(\Opt)$
sends $[s]_{\Opt}$ to $\Opt(s)$.  For surjective $\pi$, the unique quotient
map is $u(\pi(s))=[s]_{\Opt}$.
\end{proof}

Thus a retained observation may merge states only when they require the same
response. The brake/steer pair fails exactly this test. The audits below check
the familiar factorization criterion without constructing all response groups.

\subsection{Sufficiency and Relevance}

Keeping both hazard sensors preserves the controller's actions. Keeping only
the first merges the brake and steer states, so it is insufficient.

\begin{definition}[Sufficient Coordinate Set]\label{def:sufficient}
A coordinate set $I\subseteq\{1,\ldots,n\}$ is \emph{sufficient} when
\[
 s_I=s'_I\quad\Longrightarrow\quad\Opt(s)=\Opt(s')
 \qquad(s,s'\in S).
\]
\end{definition}

For this static target, we use \emph{sufficient} and \emph{preserving}
interchangeably.  Later stochastic contracts qualify preservation explicitly.

The empty set suffices exactly when $\Opt$ is constant. A conflicting pair
agreeing on $I$ witnesses insufficiency.\leanmeta{\LHrng{DP}{6}{7}}

\begin{definition}[Minimal Sufficient Set]\label{def:minimal-sufficient}
A sufficient set is \emph{minimal} when none of its proper subsets is
sufficient.
\end{definition}

Can two sufficient input sets be replaced by just the inputs they share?
If this always works, intersecting all sufficient sets leaves one smallest
set. The next lemma records that shortcut; product closure will justify it.

\leanmetapending{\LHrng{IC}{1}{4}}
\begin{lemma}[Intersection Closure Gives a Least Member]\label{lem:intersection-closure-collapse}
Let $C$ be a finite coordinate set and let
$\varnothing\ne\mathcal F\subseteq\mathcal P(C)$.  If
$I,J\in\mathcal F$ implies $I\cap J\in\mathcal F$, then
\[
  I_{\mathcal F}:=\bigcap_{I\in\mathcal F} I
\]
belongs to $\mathcal F$ and is its unique least member.  Consequently, for
every budget $k$,
\[
  \exists I\in\mathcal F:\ |I|\le k
  \quad\Longleftrightarrow\quad
  |I_{\mathcal F}|\le k .
\]
\end{lemma}

\begin{proof}
Finiteness and pairwise closure place the intersection of all members in
$\mathcal F$.  This intersection is the unique least member.  Its containment
in every $I\in\mathcal F$ gives the budget equivalence.
\end{proof}

\begin{definition}[Relevant Coordinate]\label{def:relevant}
Coordinate $i$ is \emph{relevant} when
\[
\exists s,s' \in S:\;
\Big(\forall j \neq i,\; s_j=s'_j\Big)
\ \wedge\
\Opt(s)\neq\Opt(s').
\]
\end{definition}

In the opening controller, each hazard sensor has its own conflicting pair,
whereas none of the diagnostic sensors does. When all combinations are
admitted, keeping exactly such individually necessary inputs always suffices.
The next statement allows any response specified by an equivalence relation.

\leanmetapending{\LH{EF1}}
\begin{proposition}[On Products, the Necessary Inputs Suffice]\label{prop:minimal-relevant-equiv}
Let $S=X_1\times\cdots\times X_n$, and let $\equiv$ be any equivalence
relation on $S$.  Call $I$ sufficient when
$s_I=s'_I$ implies $s\equiv s'$, and call coordinate $i$ relevant when some
$s,s'$ agree outside $i$ and satisfy $s\not\equiv s'$.  The sufficient sets
are closed under finite intersection.  The relevant-coordinate set is
sufficient and is the unique inclusion-minimal sufficient set.  Equivalently,
a coordinate set is sufficient exactly when it contains every relevant
coordinate.
\end{proposition}

\begin{proof}
Let $I$ and $J$ be sufficient, and suppose $s,s'$ agree on $I\cap J$.
Let $t$ splice the $I$-coordinates of $s$ into $s'$.  Sufficiency of $I$ and $J$
gives $s\equiv t\equiv s'$.  Hence the sufficient sets are
intersection-closed and have a unique least member $K$.

Every sufficient set contains each relevant coordinate, since omitting one
retains its witness pair in a single fiber.  Conversely, if $i\in K$ were
irrelevant, replace coordinate $i$ in one state by its value in another.
Irrelevance relates the original state to the replacement, and the
replacement agrees with the second state on $K$.  Hence $K\setminus\{i\}$
would be sufficient, contradicting leastness.  Therefore $K$ is exactly the
relevant-coordinate set.
\end{proof}

For equality of $\Opt$, this gives Definitions~\ref{def:sufficient}
and~\ref{def:relevant}. Uniqueness does not make the relevance audit cheap:
succinct inputs can describe exponentially many candidate witnesses.

\leanmetapending{\LH{EF1}}
\begin{corollary}[Full Interface and Smallest Sufficient Interface]
\label{cor:free-top-unique-bottom}
For a product-coordinate decision problem, let $R$ be the set of relevant
coordinates.  Then
\[
  \{I:I\text{ is sufficient}\}
  =\{I:R\subseteq I\subseteq[n]\}.
\]
Thus the full interface $[n]$ is sufficient without search, while $R$ is the
unique least sufficient interface.
\end{corollary}

\begin{proof}
Apply Proposition~\ref{prop:minimal-relevant-equiv}: sufficiency is
equivalent to containing $R$.
\end{proof}

\leanmetapending{\LH{DQ86}}
\begin{remark}[Listed States and Product Closure]\label{rem:listed-product-closure}
Suppose the only allowed states are $(0,0)$ and $(1,1)$, with different
optimizer sets.  Either coordinate alone distinguishes them; the empty set
does not.  Thus $\{1\}$ and $\{2\}$ are incomparable minimal sufficient sets.
The missing mixed states prevent the splice in
Proposition~\ref{prop:minimal-relevant-equiv}.  This models a log of two
co-occurring settings.  It is also the set-cover instance with one universe
element and two duplicate sets.
\end{remark}

\begin{figure}[htbp]
\centering
\includegraphics[width=.94\linewidth]{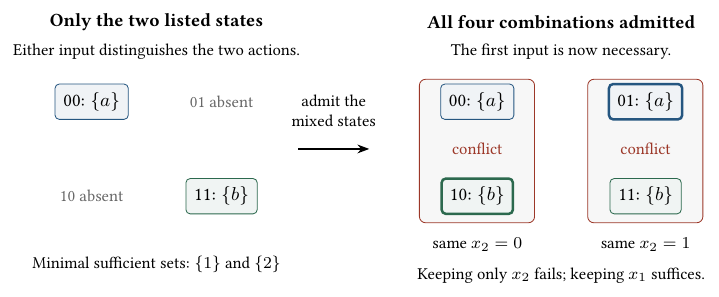}
\caption{Missing combinations can hide a coordinate's relevance.  On the
left, either input distinguishes the two listed states.  On the right,
the newly admitted states have thick borders.  Each outlined column groups
states agreeing on $x_2$ but requiring different actions.  The unique smallest
sufficient set is now $\{1\}$.}
\Description{The listed domain contains only 00 labelled a and 11 labelled b.
Either coordinate is sufficient. Extending the response to the full square
by labelling 01 with a and 10 with b makes the first coordinate necessary.
Keeping the second coordinate alone groups 00 with 10 and 01 with 11;
both groups contain conflicting labels.}
\label{fig:missing-combinations}
\end{figure}

Equivalently, $I$ is sufficient iff $\Opt=h_I\circ\pi_I$ for some $h_I$,
which assigns each realized observation its common optimizer
(Theorem~\ref{thm:master-factorization}).\leanmeta{\LH{UB2}}

\paragraph{Changing a setting is not removing a channel.}
Replacing one setting need not make an input removable. If
$F(1)=F(3)=\mathsf A$ and $F(2)=\mathsf B\ne\mathsf A$, replacing three
by one preserves the response, but discarding the count does not. Encoding
the count by one mandatory step and two optional Boolean steps gives responses
$\mathsf A,\mathsf B,\mathsf B,\mathsf A$ at $00,10,01,11$:
each Boolean input is relevant even though $00$ and $11$ agree.
Budgets count the inputs in the supplied encoding.

On a product, fixing omitted inputs to defaults preserves $F$ exactly when
the retained inputs determine it: states agreeing on those inputs have the
same reset tuple, and every state agrees with its own reset. Supplement~S2
gives the expanded illustration and criteria.
\leanmeta{\LHrng{IS}{1}{8}}

\subsection{Admissible Outputs}

For acceptable-output sets $\mathrm{Adm}(s)$, write
$s\sim_{\mathrm{Adm}}s'$ iff $\mathrm{Adm}(s)=\mathrm{Adm}(s')$.
An empty acceptable set may specify failure.

\leanmetapending{\LHrng{QT}{8}{9}}
\begin{proposition}[Admissible Outputs as an Optimizer Quotient]\label{prop:exact-output-realizes-optimizer-quotient}
Let the output set $O$ be finite.  For every admissible-output specification
$\mathrm{Adm}:S\to\mathcal{P}(O)$, there is a decision problem whose optimizer
quotient satisfies
\[
s\sim_{\Opt}s'
\quad\Longleftrightarrow\quad
\mathrm{Adm}(s)=\mathrm{Adm}(s').
\]
\end{proposition}

\begin{proof}
Use action space $\mathcal P(O)$ and make $\mathrm{Adm}(s)$ uniquely optimal
at $s$.  Then $\Opt(s)=\{\mathrm{Adm}(s)\}$.
\end{proof}

This construction can require $2^{|O|}$ actions and need not preserve encoded
size. Unlike comparison across families of decision problems
~\cite{blackwell1953equivalent,savage1954foundations,raiffa1961applied,howard1966information},
we fix one response. For additive tolerance~\cite{abel2016near}, it is
\[
\mathrm{Adm}_{\varepsilon}(s)
:=\left\{a\in A:U(a,s)\ge
\max_{b\in A}U(b,s)-\varepsilon\right\}.
\]
Write $\gamma(U)$ for the least positive regret (shortfall from the best utility):
\[
 \gamma(U):=\min\bigl\{\max_{b\in A}U(b,s)-U(a,s)>0:
                    a\in A,\ s\in S\bigr\},
\]
with $\gamma(U)=+\infty$ when the set is empty.

\leanmetapending{\LH{EF11}}
\begin{corollary}[Positive-Gap Tolerance Band]
\label{cor:positive-gap-tolerance}
For every $0\le\varepsilon<\gamma(U)$ and every state $s$,
\[
 \mathrm{Adm}_\varepsilon(s)=\Opt(s).
\]
Consequently the preservation verdicts are identical throughout this interval
whether we check a proposed input set, seek a set within a budget, or ask
whether one realized observation has a constant response.
\end{corollary}

\begin{proof}
Optimal actions have regret zero; all others have regret at least $\gamma(U)$.
Below that threshold the two correspondences agree at every state, hence so
do all three preservation queries.
\end{proof}

Supplement~S3 gives the breakpoint, stability-radius,
and tolerance-existence results beyond this fixed positive-gap interval.

Decisiveness tests uniqueness \emph{after} scores are aggregated on a fiber.
Requiring a conditional $\varepsilon$-admissible set to be a
singleton needs a separate margin for every induced fiber and becomes
weakly more demanding as the tolerance grows.

\fi
\ifappendixtolerance
\section{The Absolute-Tolerance Family}\label{sec:tolerance-family}

For the absolute-tolerance target in
Corollary~\ref{cor:positive-gap-tolerance}, write
\[
 r(a,s):=\max_{b\in A}U(b,s)-U(a,s),\qquad
 Q_\varepsilon:=S/{\ker(\mathrm{Adm}_\varepsilon)} .
\]
Thus $a\in\mathrm{Adm}_\varepsilon(s)$ exactly when
$r(a,s)\le\varepsilon$, and $Q_0$ is the optimizer quotient.

\leanmetapending{\LHrng{EF}{2}{3}, \LH{EF10}}
\begin{proposition}[Breakpoint and Tail Structure]
\label{prop:tolerance-breakpoints}
Let
\[
 R(s):=\{r(a,s):a\in A\},\qquad R:=\bigcup_{s\in S}R(s).
\]
The set $R$ contains $0$ and has at most
$1+|S|(|A|-1)$ elements.
For $|A|=1$, this gives the exact bound $|R|=1$.
The correspondence $\mathrm{Adm}_\varepsilon$ is right-continuous and changes
only at values in $R$.  If
$0=v_0<\cdots<v_\ell$ are its distinct values, then $Q_\varepsilon$ is
constant on each $[v_i,v_{i+1})$ and on $[v_\ell,\infty)$.
Consequently, the family has at most $1+|S|(|A|-1)$ distinct members.  For
\[
\varepsilon\ge \max_{a,s}r(a,s),
\]
every action is admissible at every state and $Q_\varepsilon$ has one class.
\end{proposition}

\begin{proof}
For fixed $(a,s)$, membership of $a$ in
$\mathrm{Adm}_\varepsilon(s)$ is the threshold predicate
$r(a,s)\le\varepsilon$.  It can change only when $\varepsilon$ reaches that
regret, and it includes the action at the breakpoint.  Each state contributes
zero and at most $|A|-1$ positive regrets, which gives the cardinality bound;
the interval claims follow by taking all action--state pairs.  Above the
largest regret every
threshold predicate is true.
\end{proof}

On an explicit table, the values in $R$ can be enumerated and sorted in
polynomial time.  A succinct presentation can encode exponentially many
states.  Each breakpoint remains computable from one action--state pair in
polynomial time, while enumerating all breakpoints need not be polynomial in
the succinct input length.

\leanmetapending{\LHrng{EF}{4}{8}}
\begin{example}[The Kernels Are Non-monotone]
\label{ex:tolerance-kernel-nonmonotonicity}
Let two states have utility vectors
\[
 U(\cdot,s_1)=(0,-1),\qquad U(\cdot,s_2)=(0,-2).
\]
For $0\le\varepsilon<1$, both states admit $\{a\}$ and lie in one class.  For
$1\le\varepsilon<2$, only $s_1$ admits $\{a,b\}$, so the class splits.  For
$\varepsilon\ge2$, both states admit $\{a,b\}$ and merge again.
The action sets grow monotonically with tolerance, while their equality
kernel follows a merge--split--merge pattern.  The reverse separation also
occurs: vectors $(0,-1)$ and $(-1,0)$ have disjoint optimizer sets at
$\varepsilon=0$ and both admit $\{a,b\}$ at $\varepsilon=1$.
\end{example}

\leanmetapending{\LH{EF9}}
\begin{theorem}[Meet of the Tolerance Family]
\label{thm:tolerance-family-meet}
For any states $s,s'$, equality of
$\mathrm{Adm}_\varepsilon(s)$ and
$\mathrm{Adm}_\varepsilon(s')$ at every $\varepsilon\ge0$ is equivalent to
equality of their complete regret vectors and to equality of every pairwise
action gap.  Equivalently,
\[
 \bigcap_{\varepsilon\ge0}\ker(\mathrm{Adm}_\varepsilon)
 =
 \ker\!\bigl(s\mapsto r(\cdot,s)\bigr)
 =
 \ker\!\bigl(s\mapsto
       (U(b,s)-U(a,s))_{a,b\in A}\bigr),
\]
the kernel of the complete action-gap vector.
\end{theorem}

\begin{proof}
Equal regret vectors give equal threshold sets.  Conversely, if
$r(a,s)<r(a,s')$, the threshold $\varepsilon=r(a,s)$ admits $a$ at $s$ and
excludes it at $s'$.  Hence agreement at every threshold forces coordinatewise
equality of regret.  Finally,
\[
 r(a,s)-r(b,s)=U(b,s)-U(a,s).
\]
The regret vector therefore determines every pairwise gap.  Every regret is
the maximum of those pairwise gaps, giving the converse.
\end{proof}

The optimizer quotient is the member $Q_0$, while the common refinement of
every member retains the complete action-gap data.

\leanmetapending{\LHrng{EF}{11}{15}}
\begin{proposition}[Stability Radius]\label{prop:tolerance-stability-radius}
Define the extended stability radius
$\varepsilon^\star(\mathcal D)$ as the supremum of the
$\delta\ge0$ for which $Q_\varepsilon=Q_0$ throughout
$0\le\varepsilon<\delta$.  If a positive regret occurs, then
\[
\varepsilon^\star(\mathcal D)\ge
\min\{r(a,s):r(a,s)>0\}.
\]
The inequality can be strict.
\end{proposition}

\begin{proof}
Below the least positive regret, every admissible set equals its value at
zero, so its kernel equals $\ker(\Opt)$.  For strictness, take utility vectors
$(0,-1,-5)$ and $(0,-1,-6)$.
The minimum positive regret is $1$.  Both states admit $\{a\}$ below $1$ and
$\{a,b\}$ throughout $[1,5)$, so their common optimizer class survives.
At $\varepsilon=5$, only $s_1$ admits $c$, and the class splits.  Thus
$\varepsilon^\star(\mathcal D)=5$.
\end{proof}

\begin{remark}[Absolute Tolerance Is Not Affine Invariant]\label{rem:tolerance-descent}
Under $U'(a,s)=\alpha(s)+\beta(s)U(a,s)$ with $\beta(s)>0$, regret transforms
as $r'(a,s)=\beta(s)r(a,s)$.  Absolute tolerance therefore does not descend
through statewise positive affine reparameterization.  On nonconstant states,
normalizing each regret by the state's utility range does; constant states use
the zero regret vector.
\end{remark}

\begin{remark}[Remaining Tolerance Questions]\label{rem:tolerance-deferred}
Section~\ref{sec:static-tolerance-certification} classifies three static
questions: sufficiency throughout $[0,B]$, whether the stability radius is at
least $B$, and whether some tolerance in $(0,B]$ makes an interface
sufficient.  Stochastic tolerance can widen statewise admissibility
before averaging or widen the conditional score after averaging; these
predicates differ.  Sequential tolerance also depends on backward induction,
and counting family members ranges over exponentially many implicit
breakpoints.  These directions are deferred.
\end{remark}
\fi

\section{Encodings and Access Models}\label{sec:encoding}\label{sec:dichotomy}

The eight-sensor controller has $2^8$ possible readings. A table can list
every reading and its action scores; a compact circuit can compute those
scores when given a reading. With $n$ Boolean sensors the full table has
$2^n$ rows, while the circuit need not list them. Scanning a table and
auditing a compact evaluator therefore measure work against different input
lengths. A table of logged readings can also omit combinations.

\Needspace{10\baselineskip}
\subsection{Input Models}

The encodings specify what the auditor receives and how its size is measured.
With query access, the auditor pays for calls to an evaluator.

\paragraph{Succinct product encoding.}
An instance gives a polynomial-length list of actions, the coordinate-domain
sizes in binary, and a Boolean circuit $C_U$ evaluating $U(a,s)$.
Arithmetic displays abbreviate polynomial-size bit-circuit implementations;
their encoded size includes that implementation.
Each coordinate domain is the canonical indexed set
$X_i=[|X_i|]=\{0,\ldots,|X_i|-1\}$; its values are not listed separately.
For stochastic instances it also gives a circuit $C_P$ evaluating $P(s)$.
Both circuits output dyadic rationals of polynomial bitlength (equivalently,
signed integer numerators over a common power-of-two denominator); $C_P$ is
promised nonnegative and normalized.  Hence the predicates $P(s)>0$ and
pointwise utility comparisons are decidable in polynomial time.  Comparing
scores summed over an exponentially large fiber instead requires differences
of weighted counts.  After clearing the common denominator, these differences
are GapP functions: differences between two counts of accepting computation
paths.  Writing
\(\operatorname{enc}(A)\) for the bitlength of the explicit action list, the
input length is
\[
L=\operatorname{enc}(A)+\sum_i\log|X_i|+|C_U|+|C_P|,
\]
with the $C_P$ term omitted outside the stochastic setting.

\paragraph{Circuit-defined restricted states.}
Circuits specify the allowed states, coordinate values and utilities without listing them. Each allowed state has one canonical polynomial-length identifier, and projections are total and within their declared domains. The domain is nonempty but may omit combinations. These are semantic promises, not conditions checked by the preservation audit. Sequential encodings also fix a unary horizon.

\begin{definition}[Promise-Problem Convention]
\label{def:validity-promise-convention}
A promise problem is a pair $(Y,V)$, where $V$ is the set of valid
encodings and $Y\subseteq V$ is the set of yes-instances.  An algorithm need
only be correct on $V$.  A many-one reduction from $(Y,V)$ to $(Y',V')$ maps
every input in $V$ into $V'$ and preserves membership in $Y$.
\end{definition}

Some succinct encodings below are valid only under semantic conditions, such
as normalization of $C_P$ or canonical state identifiers.  Completeness for
such an input model is relative to
Definition~\ref{def:validity-promise-convention}, and every hardness reduction
constructs a valid target instance.  A model that specifies how malformed
inputs are rejected is an ordinary decision language instead.  The multistep
encoding in Section~\ref{sec:multistep-pomdp} uses this second convention.

\paragraph{Explicit-state encoding.}
The full utility table over $A\times S$ is listed, with input length
$L_{\mathrm{exp}}=\Theta(|A||S|)$ up to entry bitlength.  Polynomial time in
this model may be exponential in the number of Boolean coordinates.

\paragraph{Query access.}
The solver receives oracle access to $\Opt(s)$ or to the utility values needed
to reconstruct it.  Complexity is measured by query count together with any
declared per-query evaluation cost. The auditor can evaluate selected states
without receiving either the full table or the circuit that computes it.

\begin{remark}[Product-Closure Scope]\label{rem:product-closure-scope}
Proposed-input and some-observation queries allow arbitrary finite coordinate
presentations.  Any within-budget upper bound that invokes the unique relevant-coordinate set
assumes a product state space; Theorem~\ref{thm:listed-minsuff-sigma2p} separately
classifies within-budget sufficiency on circuit-defined restricted states.
The conditional-winner lower bound covers every input combination: its proof
adds the missing states while preserving the required score comparisons.
The set-cover approximation theorem instead applies to listed rows.
\end{remark}

\section{Algorithms and Structural Certificates}\label{sec:tractable}

The opening controller's six diagnostic sensors affect no score. A verified
score decomposition can certify all six deletions without scanning every
reading. Other structure can help too: small interacting groups support
local calculations, and symmetry lets many states share one check.
An evaluator answers at one state; these certificates cover the whole
declared domain. The bounded-width method searches for a conflicting pair
with a constraint satisfaction problem (CSP). Table~\ref{tab:tractability-map}
states what each method needs and the work it saves.

The runtime bounds count the supplied factor values, decompositions, and
representatives.  Each method requires valid certificate data; the bounds
do not include discovering that structure in an arbitrary utility circuit.

These methods certify input sufficiency. Conditional choices and support obligations require their own checks in Section~\ref{sec:stochastic-contracts}.

\begin{table}[htbp]
\centering
\small
\setlength{\tabcolsep}{4pt}
\renewcommand{\arraystretch}{1.12}
\renewcommand{\tabularxcolumn}[1]{m{#1}}
\begin{tabularx}{\linewidth}{@{}>{\raggedright\arraybackslash}m{0.22\linewidth}>{\raggedright\arraybackslash}X>{\raggedright\arraybackslash}m{0.23\linewidth}@{}}
\toprule
\textbf{Setting} & \textbf{Algorithmic mechanism} & \textbf{Complexity} \\
\midrule
Explicit-state baseline & compute optimizer sets once, then scan state pairs & $O(|S||A|+|S|^2(n+|A|))$ \\
\midrule
State-independent utility & cancel the action-independent state term & $O(1)$ once the displayed decomposition is trusted or independently verified \\
\midrule
Rank decomposition with invariant omitted factors & compare each omitted coordinate factor with one reference value & $O(Rn(b-1))$ comparisons, where $b=\max_i|X_i|$ \\
\midrule
Factored CSP for a violating pair, of width $w$ & dynamic programming on a supplied width-$w$ tree decomposition & $O(N_{\mathrm{CSP}}b^{w+1})$ plus constraint evaluation \\
\midrule
Coordinate symmetry with a histogram certificate & compare one state per histogram (orbit representatives), checking which groups can agree on retained inputs & $O(|\mathcal R|^2)$ optimizer comparisons \\
\bottomrule
\end{tabularx}
\caption{Algorithms and certified structural reductions. The first row is
an explicit-input baseline; the remaining rows state sufficient structure or
certificate data.}
\label{tab:tractability-map}
\end{table}

For fixed alphabet size, symmetry gives a polynomial-size representative
family, as proved below. The succinct coNP-hard family in
Theorem~\ref{thm:sufficiency-conp} already uses only two actions.
\leanmeta{\LH{DQ51}, \LH{DQ72}, \LH{ACT1}}

\subsection{Two Direct Decomposition Checks}
\label{sec:tract-separable}\label{sec:tract-tensor}

If $U(a,s)=f(a)+g(s)$, then
\[
  \Opt(s)=\arg\max_{a\in A}f(a)
\]
for every state, so the empty coordinate set is sufficient.

More generally, suppose a supplied rank-$R$ decomposition has the form
\[
  U(a,s)=\sum_{r=1}^{R}w_r f_r(a)\prod_{i=1}^{n}g_{ri}(s_i).
\]
For a candidate set $I$, assume every factor $g_{ri}$ with $i\notin I$ is
constant on the declared value domain of coordinate $i$.  Put
$b:=\max_i|X_i|$.  Comparing every listed value with one reference value
decides this condition in at most $Rn(b-1)$ factor comparisons.  If the
check accepts, states agreeing on $I$ give every product term the same value,
so their action utilities and optimal-action sets are equal.
\leanmeta{\LH{DQ57}, \LHrng{DQ}{74}{77}}

If this check rejects, the displayed rank decomposition does not certify the
proposed interface.  For example, let
$S=\{0,1\}^2$, $A=\{a,b\}$, $U(a,s)=s_2$, and
$U(b,s)=1-s_2$.  For the proposed set $I=\{1\}$, the omitted factor
$s_2$ takes both values zero and one, violating constancy.  Independently,
the states $(0,0)$ and $(0,1)$ give an explicit violating pair:
they agree on
$I$ and have optimizer sets $\{b\}$ and $\{a\}$.

\subsection{Bounded Treewidth}\label{sec:tract-treewidth}

To disprove sufficiency, find two states that agree on the retained inputs but
have different optimal-action sets.  A CSP can search for those two states
directly. If its constraints can be arranged in a tree of small overlapping
variable groups, each group can be checked locally and its result passed to
its neighbors. A tree decomposition specifies those groups and their links;
its width is one less than the largest group size. The algorithm below is
exponential in that size, rather than in the total number of inputs.

\leanmetapending{\LHrng{DQ}{79}{81}, \LH{DQ87}}
\begin{proposition}[Bounded-Width CSP for Violating State Pairs]\label{thm:treewidth}
For a candidate set $I$, form a CSP with two state-copy variables per
coordinate, together with any auxiliary variables introduced by a local
factor compilation, and constraints accepting exactly the state pairs that
agree on $I$ and have different optimal-action sets.  The CSP is satisfiable
exactly when $I$ is insufficient.  If an equisatisfiable factorization of size
$N_{\mathrm{CSP}}$ has maximum domain size $b$ and a supplied width-$w$ tree
decomposition, standard dynamic programming decides sufficiency in
$O(N_{\mathrm{CSP}}b^{w+1})$ time plus constraint evaluation.
\end{proposition}

\begin{proof}
The two variable copies encode the candidate state pair.  The accepting
constraint is precisely a witness of insufficiency.  Equisatisfiability
preserves that verdict, and dynamic programming on the supplied decomposition
has the stated bound
~\cite{bodlaender1993tourist,freuder1990complexity,bessiere2013detecting}.
For fixed $w$, a decomposition need not be supplied: Bodlaender's algorithm
tests whether the CSP primal graph has treewidth at most $w$ and returns a
decomposition in $f(w)N_{\mathrm{CSP}}$ time
~\cite{bodlaender1996linear}.  The statement retains the supplied form because
the decomposition is also an independently checkable certificate.
The input contract includes correctness of the compilation of the violating-pair predicate---for
example, a canonical construction or a separately checked derivation trace.
A tree decomposition by itself certifies only the width of the supplied CSP,
not its equivalence to insufficiency.
\end{proof}

\begin{example}[A tree-factored diagnostic decision]\label{ex:treewidth-factorization}
Let $T$ be a rooted tree of maximum out-degree $\Delta$, with one Boolean
state coordinate $s_v$ at every node.  Each nonroot node carries a local
compatibility predicate $c_v(s_{\mathrm{par}(v)},s_v)$, and the root carries
a unary predicate $c_r(s_r)$.  Put
\[
 H(s):=c_r(s_r)\wedge
 \bigwedge_{v\in V(T)\setminus\{r\}}
 c_v(s_{\mathrm{par}(v)},s_v).
\]
A diagnostic controller has actions \emph{accept} and \emph{reject}, with
\[
 U(\mathit{accept},s)=\mathbf 1[H(s)],
 \qquad
 U(\mathit{reject},s)=\tfrac12.
\]
Thus the optimizer is $\{\mathit{accept}\}$ exactly on globally compatible
states and $\{\mathit{reject}\}$ otherwise.

For a proposed coordinate set $I$, introduce two state copies $s,t$ and
Boolean summary variables $h_v,h'_v$.  The local constraints recursively set
$h_v$ to the conjunction of $c_v$ and the summaries of the children, and do
the same for $h'_v$ on $t$.  Add $s_v=t_v$ for $v\in I$ and require
$h_r\mathbin{\mathsf{xor}}h'_r$.  The resulting CSP is satisfiable exactly
when two states agree on $I$ and have different optimizer sets.

The decomposition is obtained directly from $T$.  For each nonroot $v$, use
the bag
\[
 \{s_{\mathrm{par}(v)},t_{\mathrm{par}(v)},s_v,t_v,h_v,h'_v\}
 \,\cup\,
 \{h_c,h'_c:c\text{ is a child of }v\},
\]
with the smaller root analogue.  Every constraint is contained in one bag,
and the bags containing each variable form a connected subtree.  Hence the
width is at most $2\Delta+5$, the factorized CSP has linear size, and its tree
decomposition is constructed in linear time.  Proposition~\ref{thm:treewidth}
therefore decides sufficiency in
$O(|V(T)|\,2^{2\Delta+6})$ time.  This construction applies to the displayed
tree-factored utility.  An arbitrary factored utility may yield an unbounded-width
CSP for violating state pairs because the two-copy compilation adds agreement
and optimizer-comparison constraints.  The applicable parameter is therefore
the width of the compiled CSP, which the auditor must construct and check,
rather than the width of the original Bayesian network or factor graph.
\end{example}

The extra variables in Figure~\ref{fig:diagnostic-summaries} do useful work:
they summarize the existing inputs so that small checks can settle the
whole-domain claim.

\begin{figure}[!htbp]
\centering
\includegraphics[width=0.84\linewidth]{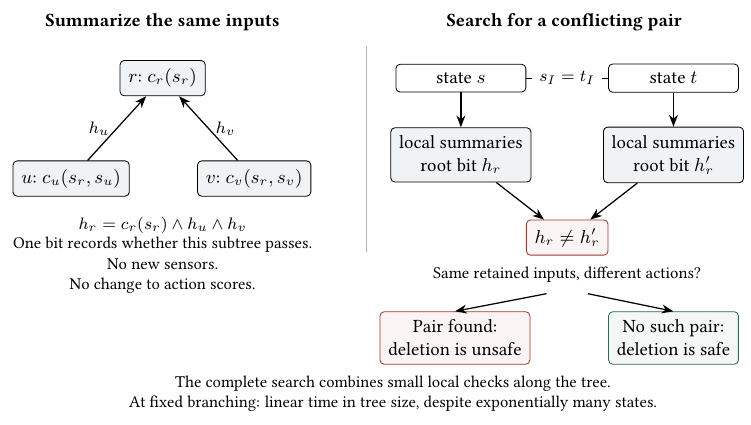}
\caption{More structure, cheaper audit.  Each summary bit records whether
all checks in its subtree pass.  Two copies search for states with the same
retained inputs but different actions.  Finding a pair refutes deletion;
exhausting this search without one certifies preservation.  For the
tree-factored utility in Example~\ref{ex:treewidth-factorization}, dynamic
programming combines bounded-size bags in linear time at fixed branching,
without enumerating all state pairs.}
\Description{A three-node tree passes two child summary bits to its root,
which conjoins them with its own local check. Two copies of that construction
search for states agreeing on retained coordinates and giving unequal root
summaries. The search has two outcomes: a conflicting pair refutes deletion,
and the absence of such a pair certifies preservation. The compiled
constraints have a tree decomposition with bounded-size bags for fixed
branching.}
\label{fig:diagnostic-summaries}
\end{figure}

\subsection{Coordinate Symmetry and Histogram Certificates}\label{sec:tract-symmetry}

If permuting coordinates leaves every utility unchanged, states with the same
histogram have the same optimal actions.  We can therefore compare histograms
instead of all states when checking a proposed input set.

For an arbitrary succinct utility this subsection is a promise-class result.
A directly checkable sufficient contract is a coefficient or factor
presentation unchanged by coordinate permutations---for example, all unary coordinate coefficients in each
action score are equal and all pairwise coefficients are equal.  Alternatively
the model may supply an independently checked symmetry proof.  Once invariance
is established, the histogram certificate below is polynomial in the supplied
representative family.

Let $S$ be a finite family of $d$-tuples over
$\{0,\ldots,k-1\}$.  Here permutation invariance means that for every
coordinate permutation $\sigma$ and every $s\in S$ with $\sigma s\in S$,
$U(a,\sigma s)=U(a,s)$ for every action $a$.  On the full product this is the
usual action of the complete symmetric group; on a listed state family it is
only required for pairs that both belong to the declaration.  The
\emph{orbit type} of a state is its coordinate histogram
$h_s$.  Equal histograms are equivalent to membership in the same permutation
orbit.  For two states define their common multiplicity by
\[
\mathrm{cm}(s,s'):=\sum_{v=0}^{k-1}\min\{h_s(v),h_{s'}(v)\}.
\]
Symmetry makes $\Opt$ constant on each orbit. \leanmeta{\LH{DQ82}}
For a complete representative family, the certificate is polynomial in the
number of inputs $d$ when the alphabet size $k$ is fixed.

On a full product, symmetry also carries any single-coordinate relevance
witness to every coordinate.  By Proposition~\ref{prop:minimal-relevant-equiv},
the least sufficient interface is therefore empty when $\Opt$ is constant and
full otherwise.  Testing the empty interface with the histogram test below
distinguishes these cases.

\leanmetapending{\LHrng{DQ}{83}{85}, \LHrng{SYM}{1}{7}}
\begin{proposition}[Orbit Representatives and the Histogram Certificate]\label{thm:symmetric-complexity}
Let $\mathcal R$ contain one representative of every orbit type occurring in the finite
state family. Two representatives are \emph{compatible on $I$} when their
orbits contain states in the family that agree on $I$. Under symmetric utility,
$I$ is sufficient if and only if every compatible representative pair has
equal optimizer sets.

On an explicitly listed state family, exact compatibility is computed by
hashing each state by its $I$-projection and recording the orbit types in each
bucket.  Constructing the buckets takes $O(|S|d)$ time, and recording all
compatible type pairs takes at most $O(|S|^2)$ additional time.

The following histogram test is a directly checkable sufficient condition on
an arbitrary listed family:
for every $r,r'\in \mathcal R$,
\[
|I|\le \mathrm{cm}(r,r')
\quad\Longrightarrow\quad
\Opt(r)=\Opt(r').
\]
Then $I$ is sufficient. The checker performs one optimizer comparison per
ordered representative pair.  If $\mathcal R$ contains one representative for
each occurring orbit type, then
\[
|\mathcal R|\le \binom{d+k-1}{k-1},
\qquad
|\mathcal R|^2\le \binom{d+k-1}{k-1}^{\!2}.
\]
On the full product $S=\{0,\ldots,k-1\}^d$, the histogram implication is also
necessary: two orbit types have representatives agreeing on $I$ if and only
if $|I|\le\mathrm{cm}(r,r')$.  The histogram test is therefore exact on the
full product and merely sufficient on an arbitrary listed subfamily.
\end{proposition}

\begin{proof}
States in one orbit have equal optimizer sets.  A pair agreeing on $I$ maps to
a compatible representative pair; compatibility supplies the converse
witness.  This proves the exact representative test.

For the histogram certificate, if states $s,s'$ agree on $I$, the multiset of their values on $I$ is a
submultiset of both state histograms. Therefore
$|I|\le\mathrm{cm}(s,s')$. Choose representatives $r,r'$ of their orbit
types. Symmetry transfers $\Opt(s)$ and $\Opt(s')$ to the representative
optimizer sets, while orbit invariance transfers common multiplicity. The
certificate hypothesis gives $\Opt(r)=\Opt(r')$, hence
$\Opt(s)=\Opt(s')$.

On the full product, the converse compatibility statement follows by choosing
any common submultiset of size $|I|$ from the two histograms, placing it on the
coordinates of $I$ in both states, and arranging the remaining symbols on the
other coordinates.  Thus $|I|\le\mathrm{cm}(r,r')$ is exactly the
compatibility condition there.

Orbit types are multisets of $d$ values from a $k$-element alphabet.  The
stars-and-bars count bounds a family with one representative per orbit type by
$\binom{d+k-1}{k-1}$. Its square gives the comparison count, which
is polynomial in $d$ for fixed $k$. \leanmeta{\LHrng{DQ}{82}{85}, \LHrng{SYM}{1}{7}}
\end{proof}

Single-action, bounded-state, strict-dominance, and constant-optimizer
instances give further elementary boundary cases.

\subsection{An Audit Procedure}\label{sec:structural-consequences}
Figure~\ref{fig:audit-procedure} turns these certificates into a preservation
audit with three possible outcomes.

\begin{figure}[t]
\centering
\includegraphics[width=\linewidth]{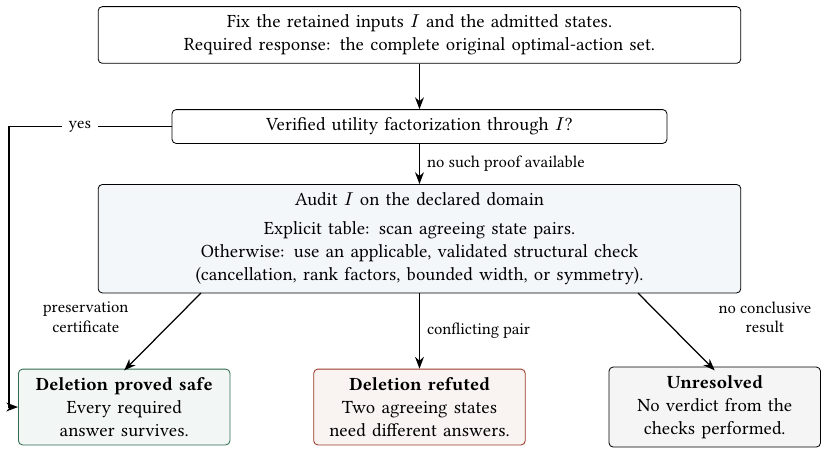}
\Description{A fixed preservation audit first asks for a verified
factorization through the retained inputs. Such a proof certifies safety.
Without it, an explicit pair scan or an applicable validated structural
check can return a preservation certificate, a conflicting pair, or no
conclusive result. These mean safe, refuted, and unresolved respectively.}
\caption{A preservation audit has three outcomes: a certificate, a
conflicting pair, or no verdict.  Minimality, support coverage, and
decisiveness require their own checks.}
\label{fig:audit-procedure}
\end{figure}

\subsubsection{Contract and Checks}\label{sec:preflight}
Record the response, declared states and observations, supplied encoding,
and candidate $I$. A listed audit covers its rows; a product audit covers
all combinations. Absence from a sample does not establish zero probability.
For stochastic or sequential audits also record normalized probabilities,
support obligations, horizon, and any alternation promise.

A verified factorization $U(a,s)=\widehat U(a,s_I)$ certifies sufficiency
by Theorem~\ref{thm:master-factorization}. Without one, scan agreeing pairs
if the table is explicit; otherwise try the applicable decomposition,
bounded-width, or symmetry checks above. Validate their prerequisites:
the compiled CSP must encode the conflicting pairs, and symmetry needs a
checked presentation, proof, or trusted contract. A failed sufficient check
is not a refutation; unequal histogram labels refute deletion only if their
orbits contain agreeing states.

The preceding proofs establish correctness and Table~\ref{tab:tractability-map}
gives the costs, excluding certificate discovery and validation. Explicit
value lists are required where values are compared; a binary domain-size
declaration does not supply them. To certify inclusion-minimality, test every
single-input deletion $I\setminus\{i\}$ and require each to fail preservation.
Coverage and decisiveness require their separate contracts (Section~\ref{sec:stochastic-contracts}).

\subsubsection{Auditing the Same Classifier on Two State Domains}
\label{sec:learned-policy-audit}

We audited one frozen decision tree on two declared domains, using the UCI
digits data distributed by scikit-learn
\cite{alpaydin1998digits,pedregosa2011scikit}. Each $8\times8$ image was
pooled into sixteen $2\times2$ block sums and quantized to four values at
cut points $8,24,40$: intervals $[0,8)$, $[8,24)$, $[24,40)$, and $[40,64]$.
A stratified 70/30 split with seed 1729 gave 1,257 training and 540 held-out
images. The tree used Gini splitting, maximum depth ten, and a three-sample
leaf floor. The fitted tree had depth ten, 121 leaves, and accuracy
$414/540=76.7\%$. The preservation target is its frozen prediction,
not the data label.

The listed domain contains 510 distinct held-out inputs. Each differently
predicted pair contributes the blocks distinguishing it. The 21,203 distinct
sets reduce to nineteen constraints after removing those
implied by smaller sets. Exhausting all $2^{16}$ interfaces gives minimum
eleven, with five smallest and nine inclusion-minimal interfaces. Greedy
selection, which repeatedly adds the block covering the most remaining
constraints, also returns eleven with deterministic tie-breaking.

The product domain admits all $4^{16}$ quantized inputs. Comparing the
tree's leaf boxes---the regions of block values reaching each leaf---finds
fifteen relevant blocks, which form the unique
smallest sufficient interface by Proposition~\ref{prop:minimal-relevant-equiv}.
It retains every block except block~13. Blocks are numbered from one to
sixteen in row-major order, as in
Figure~\ref{fig:learned-policy-scope-audit}.
Four blocks absent from the selected eleven-block minimum become necessary:
blocks~1, 4, 5, and~14. For each, the artifact supplies a pair differing
only in that block with different predictions. Figure~\ref{fig:learned-policy-scope-audit}
shows one such pair; the four pairs refute the listed certificate on the
product.

The symbolic audit checks overlaps between differently labelled leaf boxes,
covering the product without enumerating it. Some product inputs are
implausible handwriting; restricting the declaration to independently
justified plausible inputs requires another audit. The comparison establishes
the exact retention minimum for this frozen model on each declared domain.
The archive records
the canonical interfaces, tree configuration and hash, rows, witnesses, and replay
command. Greedy's general guarantee comes from the hitting-set result,
not its success on this instance.
\nameref{sec:data-code-availability} locates the archive.

\subsubsection{Running the Audit on the Sensor Controller}

\begin{example}[End-to-end pre-flight audit of the sensor controller]
For the opening controller, $I=\{1\}$ is decisive but not preserving:
$(0,0,0^6)$ and $(0,1,0^6)$ agree on $I$ and require brake and steer.
The additive decomposition check cannot certify $I$: a state term shared by
all actions cannot change which wins. Write the scores as $4(1-x_1)(1-x_2)$
for brake, $2(1-x_1)x_2$ for steer, and $4x_1$ for continue. The six unused
sensors have constant factors in this expression, whereas $x_2$ varies.
The invariant-factor check certifies $J=\{1,2\}$ but fails for $I$.

The two-copy CSP finds the same conflicting pair for $I$. After removing
the isolated diagnostic pairs, one bag contains
$x_1,x_2,x'_1,x'_2$, giving width at most three. For $J$ the CSP is
unsatisfiable because every agreeing state pair has the same utility.
Full coordinate-permutation symmetry is absent: $x_1$ and $x_2$ play
different roles, so the histogram check is inapplicable, not inconclusive.
Recovering the controller's original actions therefore needs both hazard
sensors, $J$. If the contract only asks for a unique winner after averaging
hidden states, retaining $x_1$ alone, $I$, suffices.

\end{example}

\section{What Must Be Proved Before Removing an Input}\label{sec:results-overview}

\subsection{Choosing Inputs and Checking States}
\label{sec:complexity-labels}

The agreement-only controller leaves a choice: keep either hazard sensor.
With missing combinations, an auditor may likewise have to choose among
different sufficient input sets. Once a set is chosen, every admitted pair that looks
the same through it must require the same complete optimal-action set.
Formally, when circuits specify the states and utility and $k$ is the budget,
\[
 \exists I\ (|I|\le k)\quad
 \forall s,t\in S:\quad s_I=t_I\Longrightarrow \Opt(s)=\Opt(t).
\]
Choose inputs, then check every pair: that existential--universal pattern
is what places the audit in $\Sigma_2^P$.
A completeness result establishes both an upper bound and matching hardness:
the audit can be solved within the named class, and a solver for it would
solve every problem in that class through a polynomial-time reduction.
On a complete product, the smallest sufficient set is unique and contains
exactly the inputs whose change can affect the response. Each omitted input
must still be proved irrelevant. The classification below shows that this
audit is coNP-complete, while missing combinations retain the harder choice
between sufficient sets.

Table~\ref{tab:central-questions} separates the four central questions.
For deletion, one conflicting pair witnesses failure ($\mathsf{coNP}$).
Support coverage asks whether every reading admitted by the interface can
actually arise: each needs a positive-probability state behind it
($\Pi_2^P$). A unique conditional winner instead requires adding the scores
of hidden states and comparing the totals ($\mathsf{PP}$). The budget query
must choose inputs that give every observation one winner.
The nested class in the last row records those operations in order: choose the
inputs, search for an observation that fails the requirement, and compare its
weighted scores.

\begin{table}[htbp]
\centering
\small
\setlength{\tabcolsep}{4pt}
\renewcommand{\arraystretch}{1.15}
\begin{tabularx}{\linewidth}{@{}>{\raggedright\arraybackslash}p{0.30\linewidth}>{\raggedright\arraybackslash}X>{\raggedright\arraybackslash}p{0.23\linewidth}@{}}
\toprule
\textbf{Question} & \textbf{Model supplied to the auditor} & \textbf{Complexity} \\
\midrule
Can this input be deleted? & Boolean product; circuit utility; two actions and at most two optimal-action sets & coNP-complete \\
\midrule
Can preserving inputs fit a budget? & Circuits specify allowed Boolean states and utility; combinations may be missing & $\Sigma_2^P$-complete \\
\midrule
Does the interface preserve, with support for every declared observation? & Circuit presentation; probabilities sum to one and have power-of-two denominators (dyadic) & $\Pi_2^P$-complete \\
\midrule
Can inputs within budget make every observation decisive? & Full product of fields of size at most four; circuit presentation; two actions; independent uniform inputs & $\mathsf{NP}^{(\mathsf{NP}^{\mathsf{PP}})}$-complete \\
\bottomrule
\end{tabularx}
\caption{Four input-removal questions.  The first is the baseline; the
other three expose coordinate choice, support coverage, and weighted
comparison.  These are decision problems under the stated succinct models,
not runtime predictions for individual instances.  The full classifications
appear in Table~\ref{tab:regime-matrix}.}
\label{tab:central-questions}
\end{table}

\subsection{Relation to Existing Complexity Results}
\label{sec:contribution-boundary}

We keep the decision rule fixed and ask which inputs recover all its best
actions. How the states are supplied changes the difficulty: choosing
preserving columns within a budget on listed rows is NP-complete, while a full product has one
least sufficient set
\cite{wong1985optimal,bassan2024local}. Theorem~\ref{thm:listed-minsuff-sigma2p}
classifies the missing-combinations case with ordinary Boolean bit budgets.
Bassan et al. study a local reason for one MLP prediction. Selecting an
internal neural subcircuit that remains faithful globally is the task of
Adolfi et al., while Umans minimizes the size of an equivalent Boolean
formula. These problems differ in what the auditor is allowed to select
\cite{bassan2024local,adolfi2025circuit,umans2001minimum}.

Preserving columns are reducts in an optimizer-labelled table; a backdoor to
constant-optimizer problems asks for the same constancy. The supplement's
\emph{Input Preservation in Existing AI Models} gives the exact identifications
and separates this target from selecting one action or preserving numerical
values.

W\"aldchen et al. choose variables that determine one supplied binary
prediction with a prescribed probability; their within-budget query is
$\mathsf{NP}^{\mathsf{PP}}$-complete \cite{waeldchen2021complexity}. Probabilistic decision-tree explanations
also concern one supplied input \cite{arenas2025probabilistic}.
Here the retained inputs must give every observation a unique conditional
winner. Choosing those inputs therefore leaves a universal observation
audit, giving the nested
$\mathsf{NP}^{(\mathsf{NP}^{\mathsf{PP}})}$ classification.

Compact finite-horizon POMDP planning is already PSPACE-complete
\cite{mundhenk2000mdp}. Theorem~\ref{thm:pomdp-exact-pspace} instead compares
the full and projected complete optimizer sets at roots or histories and
selects one fixed coordinate interface within a budget.

Krause and Guestrin optimize scalar value of information and prove
$\mathsf{NP}^{\mathsf{PP}}$-hardness even for polytrees
\cite{krause2009optimal}. Same-decision probability asks whether hidden
information would change a threshold decision and supports stopping and
feature selection \cite{choi2012same,chen2014same,choi2017optimal}.
Decisiveness here requires a unique conditional winner at every observation;
preservation instead recovers the complete, possibly tied optimal-action set.

\section{Certification Complexity}\label{sec:certification-complexity}

Checking a deletion means ruling out conflicting states that agree on the
retained inputs.  Three requirements add work to that baseline: choosing
among sufficient input sets when combinations are missing, finding a
supported state behind every declared observation, and comparing weighted
scores to require one conditional winner.  Table~\ref{tab:central-questions}
summarizes the resulting classifications.

The approximation results compare how closely an auditor can approach the
smallest sufficient set on explicit and succinct inputs.  The two-action
results distinguish testing a score tie from testing a score threshold;
the sequential results include future decisions.
The supplement's \emph{Fixed-block normalization}, \emph{Within-Budget
Decisiveness Under Conditional Scores}, and \emph{Product completion} sections give the
complete transcript, selector and completion derivations.

\ifappendixtolerance\else
\subsection{Input Sufficiency: Product Closure Versus Missing Combinations}\label{sec:hardness}

Deleting the second hazard sensor merges a brake state with a steer state.
One such pair settles a failed deletion; a successful deletion must exclude
every such pair. On all input combinations, the budget query asks how many
inputs are individually necessary. Missing combinations can instead permit
different sufficient choices. The formal queries use the encodings in
Remark~\ref{rem:product-closure-scope}.

\subsubsection{Product Closure Removes the Coordinate Choice}

\begin{problem}[SUFFICIENCY-CHECK]
Given a finite coordinate presentation $\mathcal D$ and $I$, is $I$ sufficient?
\end{problem}

\begin{problem}[WITHIN-BUDGET-SUFFICIENCY]
Given a product-coordinate problem $\mathcal D$ and budget $k$, is some
$|I|\le k$ sufficient?
\end{problem}

\begin{problem}[RESTRICTED-WITHIN-BUDGET-SUFFICIENCY]
The same budget query on a finite coordinate presentation, without product closure.
\end{problem}

\begin{problem}[SOME-OBSERVATION-SUFFICIENCY]
Given a finite coordinate presentation and $I$, is $\Opt$ constant on some
realized fiber $\{s:s_I=(s_0)_I\}$?
\end{problem}

Product closure reduces the budget query to $|R|\le k$ for the relevant set
$R$ (Proposition~\ref{prop:minimal-relevant-equiv}). For Boolean responses,
$R$ is the set of variables the response depends on, studied in junta
theory~\cite{fischer2004juntas}. These coNP classifications
~\cite{bassan2024local} establish the cost of checking a compactly supplied
response~\cite{papadimitriou1986succinct,balcazar1992succinct}. Restricted
domains, coverage, and conditional scores introduce the additional quantifiers.

\leanmetapending{\LH{DQ40}}
\begin{theorem}[One-Coordinate Deletion Is coNP-Complete]\label{thm:sufficiency-conp}
Under the succinct Boolean-product encoding of Section~\ref{sec:encoding},
\textsc{Sufficiency-Check} is coNP-complete.  Hardness already holds with two
actions, at most two optimizer-set values, and a proposed interface obtained
from the full interface by deleting one coordinate.
\end{theorem}

\begin{proof}
For membership, negating Definition~\ref{def:sufficient} gives a polynomial
witness $(s,s')$ with $s_I=s'_I$ and
$\Opt(s)\ne\Opt(s')$ whenever $I$ is not sufficient, so the complement is in
NP.

For hardness, reduce \textsc{Tautology} on $n$ variables to the Boolean product
$\{0,1\}^{n+1}$.  Write a state as $(x_0,x)$, propose the interface containing
all $x$-coordinates but not $x_0$, and use the actions
\emph{accept} and \emph{reject}.  Set
\[
 U(\mathit{accept},(0,x))=1,
 \qquad
 U(\mathit{accept},(1,x))=\mathbf 1[\varphi(x)],
 \qquad
 U(\mathit{reject},(x_0,x))=0.
\]
Every $x_0=0$ state has optimizer $\{\mathit{accept}\}$.  An $x_0=1$ state
has the same optimizer when $\varphi(x)$ holds and the tied optimizer
$\{\mathit{accept},\mathit{reject}\}$ otherwise.  Each proposed-interface
fiber consists of $(0,x)$ and $(1,x)$, so the interface is sufficient exactly
when $\varphi$ is a tautology.  The validity, projection, and utility circuits
have size polynomial in $|\varphi|$.

\end{proof}

\paragraph{Quotient cardinality.}
The hardness persists when the optimizer quotient has at most two classes.
Parameterization by either action count or
optimizer-quotient cardinality alone is therefore para-coNP-hard.

\leanmetapending{\LH{DQ36}, \LHrng{SK}{4}{8}}
\begin{theorem}[Finding a Sufficient Set Within Budget Is coNP-Complete]\label{thm:minsuff-conp}
Under the succinct Boolean-product encoding of Section~\ref{sec:encoding},
the apparent existential layer in \textsc{Within-Budget-Sufficiency} collapses
to relevance containment.  Consequently,
\textsc{Within-Budget-Sufficiency} is coNP-complete.
\end{theorem}

\begin{proof}
Proposition~\ref{prop:minimal-relevant-equiv} identifies sufficient sets on a
product space with the supersets of the relevant-coordinate set $R$.  The
within-budget query therefore asks whether $|R|\le k$.  Its complement has a
polynomial certificate: list
$k+1$ distinct coordinates and, for each one, a state pair witnessing
relevance. This certificate places the problem in coNP.

For hardness, reuse the Boolean-product construction in the proof of
Theorem~\ref{thm:sufficiency-conp} and set $k=0$.  A sufficient set of size
zero exists exactly when its optimizer map is constant, which holds exactly
when $\varphi$ is a tautology.  Thus the $k=0$ slice is coNP-hard, but
its optimum may be zero.  The Boolean paired family in
Proposition~\ref{prop:paired-family-approximation-gap} replaces that
degenerate gap by optimum one versus all $p$ coordinates.
\end{proof}

\subsubsection{Missing Combinations Restore the Coordinate Choice}

On the agreement-only hazard states 00 and 11, either sensor suffices.
Neither sensor is mandatory by itself, but retaining no sensor fails.
In larger circuit-defined domains, choosing a sufficient set within budget
must be followed by an audit of every state pair that choice merges.
The next theorem proves that this choose-then-check structure has
the full $\Sigma_2^P$ worst-case complexity, even when each selectable input
is one bit.

The hardness construction makes a bit budget enforce whole-block choices.
For each choice it provides two blocks of input bits. Omitting a bit from
each block hides a conflicting state pair that differs only at those two
bits. Keeping one whole block prevents this failure; after retaining the
mandatory identifier bits, the budget leaves room for exactly one block per
choice. Those choices encode a candidate assignment. The remaining state
pairs test whether that assignment works for every setting of the other
variables.

\leanmetapending{\LHrng{LMS}{1}{2}}
\begin{theorem}[Missing Combinations Keep the Choice Alive]
\label{thm:listed-minsuff-sigma2p}
\textsc{Restricted-Within-Budget-Sufficiency} is $\Sigma_2^P$-complete under the
circuit-defined restricted-state encoding.
Hardness already holds when every coordinate is Boolean and the budget counts
individual bits.
\end{theorem}

\begin{proof}
For membership, first replace $k$ by $\min\{k,n\}$, which does not change the
set of feasible coordinate subsets.  Then guess $I$ and use a
pair of agreeing states with different optimizer sets as the NP witness for
failure.  The within-budget query is therefore in $\Sigma_2^P$.

For hardness, reduce $\exists z\,\forall x\,\varphi(z,x)$ directly to
individually selectable bits. Pad so $d=|z|\ge1$ and $u=|x|\ge1$.
The main states are $m(z',x,r)$ for $r\in\{0,1\}$, with two actions and
\[
 U(0,m)=1,\qquad
 U(1,m(z',x,r))=
 \begin{cases}0,&\varphi(z',x),\\2r,&\neg\varphi(z',x),\end{cases}
\]
Put $T=d+u+1$ and $L=T+1$. Each existential bit has two $L$-bit
candidate blocks $Q_{j,0},Q_{j,1}$. Set
\[
 Q_{j,b}(m)=
 \begin{cases}
   0\,x\,0^{d+1},&z'_j=b,\\
   1\,z'\,x\,r,&z'_j\ne b.
 \end{cases}
\]
Selecting block $Q_{j,z_j}$ for every $j$ will group the two $r$-states
for each matching $(z,x)$ and expose every off-assignment state.

Let $F=dL^2$, choose
$h=\max\{3,\lceil\log_2(F+1)\rceil\}$, and prepend two marker bits and
$h$ public identifier bits.  Main states have public word $00\,0^h$.
For every $(j,i,\ell)\in[d]\times[L]\times[L]$, choose a distinct nonzero
$c_{j,i,\ell}\in\{0,1\}^h$ and add a forcing pair with public word
$00\,c_{j,i,\ell}$.  Its first state has all candidate bits zero; its second
has ones exactly at $Q_{j,0}[i]$ and $Q_{j,1}[\ell]$.  Give the two endpoints
opposite unique optimizers.

The public bits are forced by vertex-disjoint calibration edges whose
candidate bits are all zero.  For every payload position $k$, use the edge
\[
 10\,e_{k+1}\ \longleftrightarrow\
 10\,(e_{k+1}\oplus e_k),
\]
with subscripts modulo $h$.  Add
$10\,1^h\leftrightarrow11\,1^h$ and
$01\,0^h\leftrightarrow11\,0^h$ for the two marker positions, again with
opposite unique optimizers at the endpoints.  Because $h\ge3$, these
calibration endpoints are distinct from one another; their nonzero markers
separate them from main and forcing states.  The nonzero forcing identifiers
separate the forcing pairs from one another and from the main states.

Set the bit budget to
\[
 q+dL,\qquad q=h+2.
\]
Omitting a public bit merges its calibration pair, so all $q$ public bits are
mandatory.  For each $j$, the retained candidate bits must cover every edge
of the complete bipartite graph between $Q_{j,0}$ and $Q_{j,1}$.  Every vertex
cover of $K_{L,L}$ has size at least $L$; a cover of size $L$ is one whole
side, since omitting a vertex on each side leaves their edge uncovered.  The
remaining budget therefore selects exactly one complete block per $j$ and no
other candidate bits.  Decode those sides as $z$.

For this canonical mask, every forcing and calibration state is isolated.
The main states with $z'=z$ form exactly the two-state fibers indexed by $x$;
if $z'\ne z$, one retained block contains the exposed identifier
$1z'xr$, so that main state is a singleton. The two matched states both
prefer action zero when $\varphi(z,x)$ holds; otherwise $r=0$ prefers zero
and $r=1$ prefers one. Thus this mask is sufficient exactly when
$\forall x\,\varphi(z,x)$. A validity
circuit recognizes the three disjoint state families by their public words,
decodes the exposed main identifier or forcing index, and checks the displayed
bit patterns.  The number of coordinates, state-name length, utilities, and
all circuits are polynomial in the source size.  Thus the construction is a
polynomial many-one reduction with Boolean coordinates and an ordinary bit
budget.  The executable audit exhausts all source predicates and all
budget-feasible masks for one existential and one universal bit.
Supplement~S1.1 also records the preliminary token-field construction.
\end{proof}

Missing combinations keep the coordinate choice alive: choose which inputs
to retain, then check all allowed state pairs that agree on them.  The
neighboring logic-minimization precedent is Umans's $\Sigma_2^P$-completeness
result for minimum equivalent DNF and shortest
implicants~\cite{umans2001minimum}.

\paragraph{The same hardness with a positive tolerance.}
The utility gaps in the preceding two hardness constructions are integral:
every positive pointwise regret is at least one. Accepting actions within a
tolerance $0\le\varepsilon<1$ of the best score therefore changes neither
construction's required responses. Corollary~\ref{cor:positive-gap-tolerance}
transfers the coNP proposed-input lower bound and the $\Sigma_2^P$
restricted-domain classification to $\mathrm{Adm}_\varepsilon$.
The later approximation construction separating a one-input optimum from an
all-input optimum also has integral gaps
(Definition~\ref{def:paired-tautology-family}), so its lower bound survives
the same tolerance.

\begin{remark}[Rough-Set Search versus Product-Closed Certification]
\label{rem:rough-set-search-collapse}
On listed rows, the auditor may still have to choose between different
preserving input sets. A full product fixes the smallest set to the relevant
coordinates. The next query leaves the inputs fixed and asks the auditor to
choose one reading instead.
\end{remark}

\subsubsection{One Observation and an Existential--Universal Check}

For a fixed retained input set, choose one observation realized by an allowed
state.  Check that all states giving that observation have the same complete
optimal-action set.  The resulting
existential--universal pattern yields $\Sigma_2^P$.  In the reduction,
existential variables choose the observation and universal variables range
over the states that produce it.

The algorithm chooses the observation to certify.  Anchor and abductive
explanations make a related local choice
~\cite{ribeiro2018anchors,ignatiev2019abduction,marques2022delivering,darwiche2023reasons}.
The predecessor comparisons are in Section~\ref{sec:contribution-boundary}.

\leanmetapending{\LH{DQ98}, \LH{CC4}}
\begin{theorem}[Finding One Observation That Preserves the Answer]\label{thm:anchor-sigma2p}
Under the succinct Boolean-product encoding of Section~\ref{sec:encoding},
\textsc{Some-Observation-Sufficiency} is $\Sigma_2^P$-complete.
\end{theorem}

\begin{proof}
For membership, guess a representative state $s_0$ and verify
\[
\forall s\in S:\quad
 (s_I=(s_0)_I)\Rightarrow \Opt(s)=\Opt(s_0).
\]
The action list is polynomial, so both optimizer sets can be computed and
compared in polynomial time.  The existential guess followed by the universal
state check places the problem in $\Sigma_2^P$.

For hardness, reduce from an instance
$\exists z\,\forall x\,\psi(z,x)$ of
$\exists\forall$-SAT.  Pad the universal block so that $|x|\ge1$, use the
Boolean product state space on $(z,x)$, and retain exactly the $z$
coordinates.  Take actions \emph{yes} and \emph{no}, and define
\[
\begin{aligned}
 U(\mathit{yes},(z,x))&=
 \begin{cases}
 2,&\psi(z,x)\text{ is true},\\
 0,&\psi(z,x)\text{ is false},
 \end{cases}\\
 U(\mathit{no},(z,x))&=
 \begin{cases}
 1,&x=0^{|x|},\\
 0,&x\ne0^{|x|}.
 \end{cases}
\end{aligned}
\]
If $\psi(z,x)$ is true for every $x$, action \emph{yes} is uniquely optimal
throughout the $z$-fiber.  Conversely, suppose $\psi(z,x^\star)$ is false.
When $x^\star\ne0^{|x|}$, the optimizer at $(z,x^\star)$ is
$\{\mathit{yes},\mathit{no}\}$, while the optimizer at $(z,0^{|x|})$ is a
singleton.  When $x^\star=0^{|x|}$, the optimizer there is
$\{\mathit{no}\}$.  Choose any nonzero $\widehat x$, which exists after
padding.  At $(z,\widehat x)$ the optimizer is either
$\{\mathit{yes}\}$ or $\{\mathit{yes},\mathit{no}\}$, so it again differs.
Thus a $z$-fiber is optimizer-constant exactly when
$\forall x\,\psi(z,x)$, and some realized observation has a constant fiber
exactly when the
source formula is true.

The retained projection and both utilities are computed by circuits of size
polynomial in $|\psi|$.  The construction is therefore a polynomial-time
many-one reduction.
\end{proof}

\fi
\ifappendixtolerance
\subsection{Tolerance-Indexed Certification}
\label{sec:static-tolerance-certification}

The target can vary across a declared tolerance interval while the retained
coordinate set remains fixed.

\paragraph{\textsc{Tolerant-Sufficiency}.}
Given succinct $\mathcal D$, coordinate set $I$, and nonnegative dyadic $B$,
decide whether $I$ is sufficient for $\mathrm{Adm}_\varepsilon$ at every
$\varepsilon\in[0,B]$.

\leanmetapending{\LH{EF2}, \LH{DQ40}}
\begin{theorem}[\textsc{Tolerant-Sufficiency} is coNP-complete]
\label{thm:tolerant-sufficiency-conp}
\textsc{Tolerant-Sufficiency} is coNP-complete under succinct encoding.
\end{theorem}

\begin{proof}
Suppose sufficiency fails at some $\varepsilon\in[0,B]$.  There are states
$s,s'$ agreeing on $I$ and an action $c$ whose membership differs between
$\mathrm{Adm}_\varepsilon(s)$ and
$\mathrm{Adm}_\varepsilon(s')$.  After exchanging the states if necessary,
\[
 r(c,s)\le\varepsilon<r(c,s').
\]
At the breakpoint $v=r(c,s)$ the same action still separates their
admissible sets.  A complement certificate therefore consists of
$(c,s,s')$ and the circuit evaluations establishing coordinate agreement,
$v\le B$, and this membership difference.  The action list is polynomial and
each utility evaluation is polynomial, so the complement lies in NP.

For hardness set $B=0$.  Since
$\mathrm{Adm}_0=\Opt$, the resulting instance is
\textsc{Sufficiency-Check}, which is coNP-hard by
Theorem~\ref{thm:sufficiency-conp}.
\end{proof}

\paragraph{\textsc{Stability-Radius}.}
Given a succinct Boolean-product decision problem $\mathcal D$ and a
nonnegative dyadic $B$, decide whether
$\varepsilon^\star(\mathcal D)\ge B$.

\leanmetapending{\LH{EF17}}
\begin{theorem}[The Stability Radius Is coNP-Complete]
\label{thm:stability-radius-conp}
\textsc{Stability-Radius} is coNP-complete, already for two actions on a
Boolean product.
\end{theorem}

\begin{proof}
A no-instance has a polynomial witness: an attained breakpoint $v<B$ and two
states on which equality under $\mathrm{Adm}_v$ differs from equality under
$\Opt$.  Such a breakpoint exists whenever the kernels first disagree before
$B$, because the admissible correspondence is constant between attained
regrets.  The utilities, breakpoint, admissible sets, and both state
comparisons are computable in polynomial time.  The problem is therefore in
coNP.

For hardness, reduce from \textsc{Tautology}.  On states
$(t,x)\in\{0,1\}^{1+n}$ use actions $a,c$ and set
\[
 U(a,(t,x))=0,
 \qquad
 U(c,(t,x))=
 \begin{cases}
  -1,&t=1\text{ or }\varphi(x),\\
  -2,&t=0\text{ and }\neg\varphi(x).
 \end{cases}
\]
The optimizer is $\{a\}$ at every state, so $Q_0$ has one class.  If
$\varphi$ is a tautology, the regret of $c$ is one everywhere and every
$Q_\varepsilon$ still has one class.  Otherwise a state with $t=1$ and a
falsifying state with $t=0$ split throughout $[1,2)$, so the stability radius
is exactly one.  Thus, for $B=3/2$, the answer is yes exactly when
$\varphi$ is a tautology.  The utility circuit has polynomial size.
\end{proof}

The minimum positive regret $\gamma(U)$ in this reduction is always one.
Consequently, deciding whether the stability radius strictly exceeds the
general lower bound from Proposition~\ref{prop:tolerance-stability-radius} is
already coNP-hard.

\paragraph{\textsc{Robust-Within-Budget-Sufficiency}.}
Given succinct product-coordinate $\mathcal D$, integer $k$, and nonnegative
dyadic $B$, decide whether some $I$ with $|I|\le k$ is sufficient for every
$\mathrm{Adm}_\varepsilon$ with $\varepsilon\in[0,B]$.

\leanmetapending{\LHrng{EF}{1}{3}, \LH{DQ36}}
\begin{theorem}[\textsc{Robust-Within-Budget-Sufficiency} is coNP-complete]
\label{thm:robust-minimum-conp}
\textsc{Robust-Within-Budget-Sufficiency} is coNP-complete on product state
spaces under succinct encoding.
\end{theorem}

\begin{proof}
For each tolerance $\varepsilon$, apply
Proposition~\ref{prop:minimal-relevant-equiv} to
$\ker(\mathrm{Adm}_\varepsilon)$.  Its sufficient sets are exactly the
supersets of the corresponding relevant-coordinate set $R_\varepsilon$.
The sets sufficient at every tolerance in $[0,B]$ are therefore the
supersets of
\[
 R_{\le B}:=\bigcup_{0\le\varepsilon\le B}R_\varepsilon .
\]
Proposition~\ref{prop:tolerance-breakpoints} makes this a finite union over
attained breakpoints.  It is the unique least robust set, so
Lemma~\ref{lem:intersection-closure-collapse} reduces the budgeted question to
$|R_{\le B}|\le k$.

A complement certificate lists $k+1$ distinct coordinates.  For each
coordinate it supplies a breakpoint in $[0,B]$ and two states that differ
only there and have different admissible sets.  A relevance witness at an
arbitrary tolerance can be moved to an attained breakpoint by choosing an
action in the symmetric difference, exactly as in the preceding proof.
This polynomial certificate places the problem in coNP.  At $B=0$ the problem
is \textsc{Within-Budget-Sufficiency}, which is coNP-hard by
Theorem~\ref{thm:minsuff-conp}.
\end{proof}

Varying tolerance does not destroy the unique least sufficient set on a
product space.

\paragraph{\textsc{Tolerance-Existence}.}
Given succinct $\mathcal D$, coordinate set $I$, and positive dyadic $B$,
decide whether some $\varepsilon\in(0,B]$ makes $I$ sufficient for
$\mathrm{Adm}_\varepsilon$.

\leanmetapending{\LH{EF18}}
\begin{theorem}[\textsc{Tolerance-Existence} is $\Sigma_2^P$-complete]
\label{thm:tolerance-existence-sigma2p}
\textsc{Tolerance-Existence} is $\Sigma_2^P$-complete already for succinct
Boolean-product decision problems.
\end{theorem}

\begin{proof}
Membership guesses a polynomial-length dyadic representative and applies the
coNP state-pair test.  The representative may be $B$, an attained positive
breakpoint encoded by a state--action pair, or
$\min\{B,\beta\}/2$ for a guessed attained positive breakpoint $\beta$; the
last form represents the initial open interval.  Every guessed representative
lies in $(0,B]$, so soundness follows directly from the coNP test.  For
completeness, use an attained breakpoint from the satisfying interval, use
the midpoint below the least positive breakpoint when the initial interval
works, or use $B$ if no positive breakpoint exists.

For hardness, reduce from $\exists z\,\forall x\,\psi(z,x)$.  The full state
space has Boolean coordinates $(g,z,x,r)$, and the proposed interface retains
$I=\{g,z_1,\ldots,z_m\}$, where $m=|z|$.  Let
$L_z=\langle z\rangle+1$.  Use two actions
$a,c$, give $a$ regret zero, and give $c$ regret
\[
 r \quad\text{when }g=1,
 \qquad
 L_z+\mathbf 1[r=1\land\neg\psi(z,x)]
 \quad\text{when }g=0.
\]
Set $B=2^m+\tfrac12$.  Each $g=1$ fiber contains regrets zero and one, so no
tolerance in $(0,1)$ is sufficient.  At every tolerance at least one, those
guard fibers are constant.

Every $\varepsilon\in[1,B]$ belongs to a unique interval
$[L_z,L_z+1)$.  In every fiber with a smaller selector code, both branches
admit $c$; in every fiber with a larger code, neither branch admits $c$.
Within the selected fiber, branch zero admits $c$, while branch one admits
$c$ exactly when $\psi(z,x)$ holds.  Thus $I$ is sufficient at some
$\varepsilon\in(0,B]$ exactly when
$\exists z\,\forall x\,\psi(z,x)$.  The utilities have polynomial-size
circuits, and the carrier is the complete Boolean product.
For the maximal code $\langle z\rangle=2^m-1$, the final interval is
$[2^m,2^m+1)$ and contains
$B=2^m+\tfrac12$; no right-endpoint convention is used at $B$.
\end{proof}

\begin{table}[H]
\centering
\small
\renewcommand{\arraystretch}{1.13}
\renewcommand{\tabularxcolumn}[1]{m{#1}}
\begin{tabularx}{\linewidth}{@{}>{\raggedright\arraybackslash}m{0.25\linewidth}>{\raggedright\arraybackslash}X>{\raggedright\arraybackslash}m{0.24\linewidth}@{}}
\toprule
\textbf{Presentation} & \textbf{Tolerance-indexed query} & \textbf{Complexity} \\
\midrule
Explicit table & proposed-input interval sufficiency and robust within-budget sufficiency on products & polynomial time \\
\midrule
Succinct presentation & \textsc{Tolerant-Sufficiency} & coNP-complete \\
\midrule
Succinct Boolean product & \textsc{Stability-Radius} & coNP-complete \\
\midrule
Succinct product space & \textsc{Robust-Within-Budget-Sufficiency} & coNP-complete \\
\midrule
Succinct Boolean product & \textsc{Tolerance-Existence} & $\Sigma_2^P$-complete \\
\bottomrule
\end{tabularx}
\caption{Static queries whose correctness target is reindexed by an
absolute tolerance.  The table does not introduce stochastic or sequential
tolerance predicates.}
\label{tab:tolerance-static-queries}
\end{table}

For an explicit table, enumerate the regret breakpoints up to $B$ and run the
state-pair test at each one.  On a full product table, taking the union of the
relevant-coordinate sets gives the robust minimum within the same polynomial
input bound.

\paragraph{Design consequence}
A declared approximation tolerance creates a third design choice.  Sufficiency
at one threshold and sufficiency throughout a tolerance interval are different
requirements: Example~\ref{ex:tolerance-kernel-nonmonotonicity} shows that
equality of admissible sets can be lost and later recovered as tolerance
grows.  An advisory layer promising an $\varepsilon$-band must certify the
entire band; endpoint certification does not establish that contract.
Even on complete products, product closure removes the input-choice layer
from robust within-budget sufficiency but leaves the tolerance choice at full
$\Sigma_2^P$ strength.  The structure settles which inputs matter at a fixed
tolerance; it does not choose the tolerance for us.
\fi

\subsection{What to Require at Unreachable Observations}
\label{sec:stochastic-contracts}

Suppose two inputs report a commanded mode and its confirmed mode. An
interlock keeps them in agreement, although the interface schema admits
disagreement too. Does deletion have to preserve the response only during
interlocked operation, or also at the extra settings admitted by the schema?
And must every retained reading be able to occur? Example~\ref{ex:sensor-support-contract}
will give this situation explicit scores.

The first question chooses which states must keep their responses: only
those with positive probability (\emph{supported} states), or every state
the schema admits (\emph{declared} states). The second asks for
\emph{observation coverage}: every retained reading must have at least one
positive-probability state behind it. Removing inputs groups states into
fibers, each with one retained reading. A fiber can mix supported and
unsupported states or contain no supported state at all. The definitions
distinguish these two requirements.

\begin{definition}[Stochastic Decision Problem]
\label{def:stochastic-decision-problem}
A stochastic decision problem is a tuple
\[
  \mathcal D_S=(A,X_1,\ldots,X_n,U,P),
\]
where $P\in\Delta(S)$ is a distribution on
$S=X_1\times\cdots\times X_n$.
\end{definition}

For a retained coordinate set $I$, define the unnormalized score and its
complete optimizer by
\[
 F_I(s,a):=\sum_{t:t_I=s_I}P(t)U(a,t),
 \qquad
 \Opt^{\mathrm{stoch}}_I(s):=\arg\max_{a\in A}F_I(s,a).
\]
On a positive-mass fiber, normalization scales every action score by the same
positive number, so this is the ordinary conditional optimizer.  On a
zero-mass fiber, every action has score zero and
$\Opt^{\mathrm{stoch}}_I(s)=A$.

\begin{definition}[Preservation on Supported and Declared States]
\label{def:stochastic-preservation}
A coordinate set $I$ \emph{preserves on supported states} if
\[
  \forall s\in\operatorname{supp}(P),\qquad
  \Opt^{\mathrm{stoch}}_I(s)=\Opt(s).
\]
It \emph{preserves on every declared state} if
\[
  \forall s\in S,\qquad
  \Opt^{\mathrm{stoch}}_I(s)=\Opt(s).
\]
The first requirement ignores unreachable states.  The second tests every state,
including unreachable states in positive-mass fibers.  On zero-mass fibers,
it compares the pointwise optimizer with the full action set $A$.

For any finite observation map $\phi:S\to Y$, the same definitions use the
fiber $\phi^{-1}(\phi(s))$ in place of agreement on $I$.  The complexity
classifications below use coordinate projections; they do not transfer
automatically to arbitrary succinctly represented maps.
\end{definition}

\begin{definition}[Input Sufficiency with Observation Coverage]
\label{def:deployable-preservation}
Write $\mathcal F_I(s)=\{t:t_I=s_I\}$.  A coordinate set $I$ is
\emph{sufficient with observation coverage} if
\begin{enumerate}
\item $I$ is sufficient, so $\Opt$ is constant on every
      $I$-fiber; and
\item every declared $I$-fiber has a positive-mass realization:
\[
  \forall s\in S\;\exists t\in \mathcal F_I(s): P(t)>0.
\]
\end{enumerate}
For the some-observation query, an $I$-fiber qualifies if it contains a
positive-mass state and $\Opt$ is constant on that fiber.
This requirement never evaluates a conditional expectation on a null event.  It
requires every observation declared by the interface to occur and to recover
the complete pointwise optimizer set.
\end{definition}

\begin{definition}[One Conditional Winner]
\label{def:stochastic-singleton-sufficient}
A coordinate set $I$ is \emph{decisive} if every combinatorial
fiber has one conditionally optimal action:
\[
  \forall s\in S\;\exists a\in A,\qquad
  \Opt^{\mathrm{stoch}}_I(s)=\{a\}.
\]
For the some-observation version, this singleton condition need hold on only
one combinatorial $I$-fiber.
\label{def:stochastic-anchor-sufficient}
\end{definition}

\begin{table}[htbp]
\centering
\small
\setlength{\tabcolsep}{3.5pt}
\renewcommand{\arraystretch}{1.12}
\renewcommand{\tabularxcolumn}[1]{m{#1}}
\begin{tabularx}{\linewidth}{@{}>{\raggedright\arraybackslash}m{0.25\linewidth}>{\raggedright\arraybackslash}X@{}}
\toprule
\textbf{Requirement} & \textbf{Question for the retained inputs} \\
\midrule
Input sufficiency & Do the retained inputs determine the complete optimal-action set at every declared state? \\
\midrule
Preservation on supported states & Do conditional and pointwise optimal-action sets agree at every state of positive probability? \\
\midrule
Preservation on every declared state & Do those sets agree at every declared state, including zero-probability states under the stated score rule? \\
\midrule
Input sufficiency with observation coverage & Are the inputs sufficient, and does every declared observation have positive probability? \\
\midrule
One conditional winner & Does every declared observation have exactly one optimal action under the stated score rule? \\
\bottomrule
\end{tabularx}
\caption{What each requirement asks.  Under full support the four preservation
rows coincide; one conditional winner still asks a different question.  On a zero-mass
fiber the score rule gives every action score zero.}
\label{tab:support-contract-guide}
\end{table}

\begin{example}[Supported and Declared-State Sensor Retention]
\label{ex:sensor-support-contract}
Let $x_1,x_2\in\{0,1\}$ record a commanded mode and its confirmed mode, and
let $A=\{H,C\}$ be two control actions.  An interlock excludes disagreement
during observation, although the interface schema declares all four states.
Give the two matched states equal probability:
\[
 P(0,0)=P(1,1)=\tfrac12,
 \qquad P(0,1)=P(1,0)=0.
\]
The utility table is
\[
\begin{array}{c@{\qquad}cc}
(x_1,x_2) & U(H,x) & U(C,x)\\
\hline
(0,0)&2&0\\
(0,1)&0&0\\
(1,0)&0&0\\
(1,1)&0&2
\end{array}
\]
so the unmatched states represent a tie under the system's fallback rule.
Retain $I=\{1\}$.  The two fiber-score vectors are $(1,0)$ and $(0,1)$.
They reproduce the pointwise optimizer at both supported states.

Figure~\ref{fig:support-contracts} compares $I$ with the full interface
$J=\{1,2\}$.  Each $I$-fiber mixes a supported singleton optimizer with an
unmatched tie, so $I$ is not sufficient and does not preserve on every
declared state.  With $J$, each state stands alone.  It preserves on every
declared state because the zero-mass states have both pointwise and fiber
optimizer equal to $A$; it lacks observation coverage because those
observations cannot occur.  The empty interface and $\{2\}$ also mix
pointwise labels.  Thus no interface is both sufficient and observation-covering
on this product.
\end{example}

\begin{figure}[htbp]
\centering
\includegraphics[width=\linewidth]{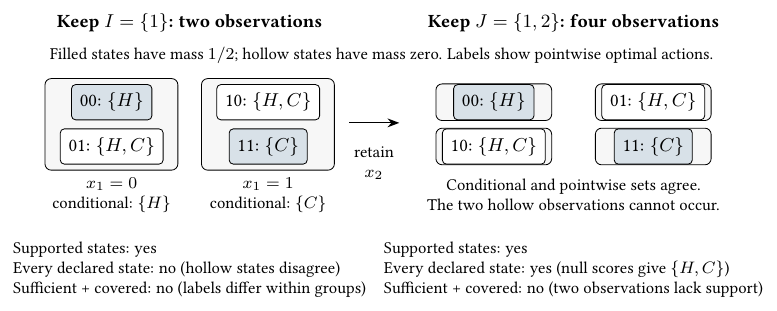}
\caption{The same interlock gives different preservation verdicts.  Boxes
group states that look the same through the retained inputs.  Keeping one
sensor recovers the actions on the support; keeping both recovers them at
every declared state under the zero-score convention.  Neither interface is
both sufficient and observation-covering: the first mixes pointwise labels,
and the second declares unsupported observations.}
\Description{Positive-mass states 00 and 11 have singleton optimizers H and
C. Zero-mass states 01 and 10 have both actions optimal. Keeping the first
coordinate groups one supported state with one zero-mass state in each
observation. Keeping both coordinates separates all states but exposes two
zero-mass observations. Both interfaces preserve on supported states; only the
full interface preserves on every declared state; neither is both sufficient
and observation-covering.}
\label{fig:support-contracts}
\end{figure}

If an application declares its observation space to be exactly the image of
$\operatorname{supp}(P)$, coverage is automatic.  If the interface schema
instead declares the full coordinate product, observation coverage audits the
extra observations as well. When each state has one best action and every
declared observation is reachable, preservation on every declared state
coincides with input sufficiency plus observation coverage;
the exact general relationship is proved in
Theorem~\ref{thm:stochastic-preservation-zero-mass-characterization}.

For a unique conditional winner, restricting the declaration to reachable observations means
testing only positive-mass fibers. The conditional-winner lower bounds use
uniform or strictly positive dyadic distributions. Their upper-bound tests
can also check that a state producing the observation has positive
probability, so these classifications do not depend on the zero-mass convention.

\paragraph{Applying the structural certificates.}
The earlier input-sufficiency checks also prove decisiveness when every observation has positive probability and all states behind it have the same unique best action (Proposition~\ref{prop:static-to-stochastic-product}). Otherwise conditional uniqueness requires its own fiber-score check.

\subsection{Markov Blankets and Preservation on Supported States}

Set $Y(s)=\Opt(s)$ and let $S_{-I}$ denote the omitted coordinates. A
Markov blanket requires that they give no further information about $Y$.

\leanmetapending{\LH{UB15}}
\begin{theorem}[Optimizer Markov-Blanket Equivalence]
\label{thm:optimizer-markov-blanket}
Assume a finite nonempty action set and nonnegative state probabilities.  For
every coordinate set $I$, the following are equivalent:
\begin{enumerate}
\item $Y$ is conditionally independent of $S_{-I}$ given $S_I$ under $P$;
\item $\Opt$ is constant on $F\cap\operatorname{supp}(P)$ for every
$I$-fiber $F$; and
\item $I$ is preserving on supported states in the sense of
Definition~\ref{def:stochastic-preservation}.
\end{enumerate}
\end{theorem}

\begin{proof}
Given a positive-mass $(v,w)$, the law of $Y$ is the point mass at $Y(v,w)$.
Conditional independence holds iff this depends only on $v$, proving
1$\Leftrightarrow$2. Averaging a common optimizer set with positive weights
preserves exactly that set, giving 2$\Rightarrow$3. Conversely, a single
fiber optimizer cannot equal two different supported pointwise sets.
\end{proof}

\leanmetapending{\LH{UB16}}
\begin{corollary}[When a Markov Blanket Matches Input Sufficiency]
\label{cor:markov-blanket-support-boundary}
Under full support, the conditions of
Theorem~\ref{thm:optimizer-markov-blanket} coincide with input sufficiency
and its coNP row in Table~\ref{tab:regime-matrix}. Without that promise,
preservation on every declared state can be stronger and is $\Pi_2^P$-complete.
The coNP classification does not include recognizing full support.
\end{corollary}

\begin{proof}
Full support turns item~2 into sufficiency. Otherwise all-state preservation
also tests unsupported states within positive fibers and imposes the null-fiber
condition of Theorem~\ref{thm:stochastic-preservation-zero-mass-characterization};
Proposition~\ref{prop:static-not-stochastic} supplies a separation.
\end{proof}

\leanmetapending{\LHrng{UB}{17}{19}}
\begin{corollary}[Full-Support Blanket Uniqueness and Sparse-Support Failure]
\label{cor:markov-blanket-uniqueness}
Under full support on a product space, the unique minimal Markov blanket for
$Y=\Opt$ is the relevant-coordinate set.  Without support closure,
incomparable minimal optimizer Markov blankets can occur.
\end{corollary}

\begin{proof}
Full support reduces the claim to
Proposition~\ref{prop:minimal-relevant-equiv}.  For the separation, take
$S=\{0,1\}^2$, put mass $1/2$ on $00$ and $11$, and set
$U(a,(x_1,x_2))=\mathbf 1[a=x_1]$.  On the support, either coordinate alone
determines $\Opt=\{x_1\}$, while the empty set does not.  Thus
$\{1\}$ and $\{2\}$ are incomparable minimal blankets.
\end{proof}

Koller and Sahami use Markov blankets for sequential feature elimination
~\cite{koller1996toward}.

\Needspace{24\baselineskip}
\subsection{Intersection Closure Governs Within-Budget Certification}
\label{sec:regime-hierarchy}

If two sensor selections each preserve the original actions, can we keep only
the sensors they share? On a full product, yes: both must contain every
individually necessary sensor (Proposition~\ref{prop:minimal-relevant-equiv}).
The two-row domain in Figure~\ref{fig:missing-combinations} permits alternatives.
Requiring a definite conditional choice permits alternatives even on a full
product: either sensor can break a tie that returns when both are hidden.
The proposition records when keeping only the shared sensors preserves
the guarantee.

\leanmetapending{\LHrng{DC}{140}{143}}
\begin{proposition}[Intersection Closure Separates Preservation from Decisiveness]
\label{prop:preservation-decisiveness-intersection}
On Boolean product spaces with nonnegative state weights, both the family of
sufficient, observation-covering sets and the family that preserves every declared state are
closed under intersection.
Decisive coordinate sets need not be intersection-closed, even
with two actions, two Boolean coordinates, and a uniform full-support
distribution.
\end{proposition}

\begin{proof}
Every sufficient and observation-covering set, and every set that preserves on
all declared states, is sufficient.  Product closure identifies sufficient
sets with the supersets of the relevant-coordinate set, so
the intersection of two preserving sets remains sufficient.
Let $J$ denote that intersection and let $I$ be either preserving set.  On a
positive-mass $J$-fiber, the static-sufficiency argument makes the pointwise
optimizer constant, so averaging preserves it.  On a zero-mass $J$-fiber,
the $I$-fiber of any state is a subset of its $J$-fiber and therefore also has
zero mass.  Preservation on every declared state by $I$ then forces $\Opt(s)=A$, exactly the
condition required for preservation on every declared state by $J$.  Thus $J$
preserves on both kinds of fiber.  For observation coverage, every $J$-fiber is a union of
$I$-fibers, and each $I$-fiber has a positive-mass representative.  Hence
every $J$-fiber is positive as well.

For failure of decisiveness closure, let the action gap
$g=U(a,\cdot)-U(b,\cdot)$ on the four states be
\[
\begin{array}{c|cc}
 & x_2=0 & x_2=1\\
\hline
x_1=0 & 3 & -1\\
x_1=1 & -1 & -1
\end{array}
\]
under the uniform distribution.  Retaining either coordinate gives weighted
fiber gaps $1/2$ and $-1/2$, so both $\{1\}$ and $\{2\}$ are decisive.  Their
intersection is empty, and its sole fiber has total gap zero, so the two
actions tie.  Thus the intersection is not decisive.
\end{proof}

Thus the product-space shortcut for within-budget input sufficiency does not
settle within-budget decisiveness.  The latter still requires choosing a set and
comparing its conditional action scores.

\subsection{Stochastic Certification}\label{sec:stochastic}
Hiding a sensor leaves several possible states behind one reading. The
controller then adds their probability-weighted action scores. In the opening
example this produces a definite brake choice, even though one hidden state
requires steer. Section~\ref{sec:stochastic-contracts} specifies which states
must still have their original actions recovered.

\subsubsection{Preservation and Decisiveness}

The two-action examples below isolate both possibilities: preserving a tie,
and obtaining one winner while losing an original response.

\begin{example}[Preservation without Decisiveness]
Let $S=\{s_1,s_2\}$, $A=\{a,b\}$, and $U\equiv0$, with positive mass on
both states and no retained coordinates.  Both optimizer sets are $\{a,b\}$:
preservation holds, but decisiveness fails. \leanmeta{\LHrng{EX}{4}{5}}
\end{example}

\begin{example}[Decisiveness without Preservation]
Let $S=\{s_1,s_2\}$ and $A=\{a,b\}$ under the uniform distribution, with
\[
U(a,s_1)=2,\quad U(b,s_1)=1,\quad U(a,s_2)=0,\quad U(b,s_2)=3.
\]
With no retained coordinates, the expected utilities are $(1,2)$, uniquely
favoring $b$.  Decisiveness holds, but preservation fails at $s_1$, whose
optimizer is $\{a\}$. \leanmeta{\LHrng{EX}{6}{7}}
\end{example}

Checking preservation starts by comparing the best-action sets within each
group of hidden states. To determine whether a conditional winner is unique,
the auditor must also add their probability-weighted scores for each action.

\subsubsection{Queries over Reachability and Conditional Scores}

\begin{problem}[Preservation and Observation-Coverage Queries]
\label{prob:stochastic-sufficiency}
For each contract in Definitions~\ref{def:stochastic-preservation}
and~\ref{def:deployable-preservation}, the \emph{check} query supplies $I$,
the \emph{within-budget} query asks for some $|I|\le k$, and the
\emph{some-observation} query asks for one qualifying realized $I$-fiber.
Their names are:
\begin{itemize}
\item \textsc{On-Support-Preservation-Check}\label{prob:support-relative-check},
\textsc{Within-Budget-Support-Preservation}\label{prob:minimum-support-relative},
and \textsc{Some-Observation-Support-Preservation}\label{prob:anchor-support-relative}.
A qualifying fiber has positive mass and a constant optimizer on its support.
\item \textsc{All-State-Preservation-Check}\label{prob:joint-preservation-check},
\textsc{Within-Budget-All-State-Preservation}\label{prob:stochastic-minimum-preservation},
and \textsc{Some-Observation-All-State-Preservation}\label{prob:stochastic-anchor-preservation}.
A qualifying fiber has $\Opt_I^{\mathrm{stoch}}(s)=\Opt(s)$ at every declared state in it.
\item \textsc{Sufficiency-With-Coverage-Check}\label{prob:deployability-check},
\textsc{Within-Budget-Sufficiency-With-Coverage}\label{prob:minimum-deployable},
and \textsc{Some-Observation-Sufficiency-With-Coverage}\label{prob:anchor-deployability}.
A qualifying fiber has positive mass and a constant optimizer throughout.
\end{itemize}
\end{problem}

\subsubsection{Sparse Support Keeps the Blanket Choice Alive}

Return to the commanded and confirmed modes. If only matched settings can
occur and their required responses differ, either input can determine the
response on the support. Admitting all four settings in the schema does not
remove that choice when the audit ignores zero-probability states.
By contrast, positive probability on every combination leaves one smallest
preserving set. The classification also describes a Markov blanket of the
action label: the retained inputs determine that label on the support.

\leanmetapending{\LH{UB15}, \LHrng{LMS}{1}{2}}
\begin{theorem}[Sparse Support Keeps the Blanket Choice Alive]
\label{thm:support-relative-blanket-classification}
Under the dyadic succinct Boolean-product encoding of
Section~\ref{sec:encoding}, with nonnegative normalized probabilities:
\begin{enumerate}
\item \textsc{On-Support-Preservation-Check} is coNP-complete;
\item \textsc{Within-Budget-Support-Preservation} is $\Sigma_2^P$-complete;
      and
\item \textsc{Some-Observation-Support-Preservation} is
      $\Sigma_2^P$-complete.
\end{enumerate}
Equivalently, deciding whether an optimizer Markov blanket of size at most the
given budget exists without a full-support promise is $\Sigma_2^P$-complete.
The lower bound already holds with two actions, Boolean coordinates, and
dyadic probabilities.
\end{theorem}

\begin{proof}
Theorem~\ref{thm:optimizer-markov-blanket} identifies preservation on supported
states with constancy of $\Opt$ on the positive-mass part of every
fiber.  A failed proposed input set therefore has a polynomial witness: two
positive-mass states that agree on $I$ and have different optimizer sets.
This puts the proposed-input query in coNP.  Guessing $I$ before that check puts the
within-budget query in $\Sigma_2^P$.  For the some-observation query, guess a positive-mass
$s_0$ and universally check every positive-mass state agreeing with it on
$I$; this also lies in $\Sigma_2^P$.

Proposed-input and some-observation hardness already hold in the full-support
subfamily.
Give the Boolean products in Theorems~\ref{thm:sufficiency-conp}
and~\ref{thm:anchor-sigma2p} their uniform distributions and apply
Theorem~\ref{thm:optimizer-markov-blanket}.

For within-budget hardness, take the Boolean restricted-state instance
$\mathcal L_\varphi\subseteq\{0,1\}^N$ constructed in
Theorem~\ref{thm:listed-minsuff-sigma2p}.  Write $d=|z|$, $u=|x|$,
$L=d+u+2$, $F=dL^2$, $h=\max\{3,\lceil\log_2(F+1)\rceil\}$, and
$q=h+2$.  The construction has
\[
 M=2^{d+u+1}\quad\text{main states},\qquad
 K=2F+2q\quad\text{forcing and calibration states}.
\]
Choose $e=\lceil\log_2(2K)\rceil$, put $\delta=2^{-e}$, and set
\[
 \mu=\frac{1-K\delta}{M}.
\]
On the full Boolean cube, assign mass $\mu$ to each main codeword, mass
$\delta$ to each forcing or calibration codeword, and zero elsewhere.
Then $K\delta\le1/2$, so every allowed codeword has positive dyadic mass, and
$M\mu+K\delta=1$.  The probability circuit recognizes the same three
disjoint codeword families as the validity circuit and returns one of these
two polynomial-bit constants.  No model count is computed.  Extend both
utilities arbitrarily, say by zero, off the support.

Positive-mass pairs agreeing on an interface are now exactly the allowed
pairs of $\mathcal L_\varphi$ agreeing on that interface.  By
Theorem~\ref{thm:optimizer-markov-blanket}, an interface is
preserving on supported states exactly when it is sufficient for
$\mathcal L_\varphi$.  The budget is unchanged, so
Theorem~\ref{thm:listed-minsuff-sigma2p} supplies the $\Sigma_2^P$ lower
bound.  The state, utility, support, and dyadic-mass circuits all have
polynomial encoded size.  The executable audit checks the normalized masses,
the Boolean forcing geometry, and the on-support equivalence for every
one-bit source predicate and every budget-feasible interface.
\end{proof}

\subsubsection{Observation Coverage Adds a Support Witness}

Keeping both mode inputs distinguishes all four schema readings, but the
interlock still prevents the two unmatched readings. Distinguishing states
does not establish that a reading can occur. The next results separate these
jobs: preserve the common response within each group, then check its support.
This avoids summing scores when constancy already fixes their winner.
\leanmeta{\LHrng{DC}{91}{93}, \LH{OU12}}
When every state has positive probability, the support check is automatic
and all preservation contracts give the same verdict and response groups.
\leanmeta{\LHrng{DC}{95}{99}}

\leanmetapending{\LHrng{DQ}{94}{95}}
\begin{theorem}[Full Support Reduces Preservation to Input Sufficiency]\label{thm:stochastic-preservation-full-support-variants}
Assume the distribution $P$ has full support and the action set is nonempty.
Then:
\begin{enumerate}
\item both \textsc{Within-Budget-Sufficiency-With-Coverage} and
\textsc{Within-Budget-All-State-Preservation} are equivalent to the
underlying \textsc{Within-Budget-Sufficiency} query;
\item for each fixed $I$, both some-observation variants are equivalent to the
underlying \textsc{Some-Observation-Sufficiency} query on $I$.
\end{enumerate}
Consequently, under full support, input sufficiency, preservation on supported
states, preservation on every declared state, and sufficiency with observation
coverage all have the same coNP-complete within-budget query; their
some-observation variants all have the same $\Sigma_2^P$-complete query.
\end{theorem}

The supplement's \emph{Full-support consequences} section proves the two underlying
equivalences and the full classification.

Without full support, the null-score convention adds the following condition.
\leanmeta{\LHrng{DC}{108}{110}, \LHrng{OU}{14}{18}}

\leanmetapending{\LH{DC113}, \LH{DC162}}
\begin{theorem}[Preservation Can Hold Without Observation Coverage]
\label{thm:stochastic-preservation-zero-mass-characterization}
Assume that $A$ is nonempty and $P$ is nonnegative.  Write
$\mathcal F_I(s)=\{t:t_I=s_I\}$.  Then
\begin{equation}
\label{eq:stochastic-preservation-normal-form}
 I\text{ is preserving on every declared state}
 \iff
 \left\{
 \begin{array}{l}
 \Opt\text{ is constant on every }I\text{-fiber},\\
 \forall s\;\bigl[(\nexists t\in \mathcal F_I(s):P(t)>0)\Rightarrow\Opt(s)=A\bigr].
 \end{array}
 \right.
\end{equation}
By contrast, $I$ is sufficient with observation coverage exactly when the
first conjunct holds and
\[
 \forall s\;\exists t\in \mathcal F_I(s):P(t)>0.
\]
Sufficiency with observation coverage implies preservation on every declared
state.  If $\Opt(s)\ne A$ for every
state, then the two predicates are equivalent.  In particular, the equivalence
holds when every pointwise optimizer is a singleton and $|A|\ge2$; the action
cardinality assumption is necessary because, for $|A|=1$, the unique singleton
optimizer is $A$ itself.
\end{theorem}

\begin{proof}
Preservation identifies the common fiber optimizer with every pointwise
optimizer in that fiber, so it implies input sufficiency.  On a zero-mass
fiber every score is zero, making the fiber optimizer $A$; preservation there
is therefore equivalent to the second conjunct of
\eqref{eq:stochastic-preservation-normal-form}.

Conversely, suppose both conjuncts hold.  On a positive-mass fiber, let $O$
be the common pointwise optimizer.  Every action in $O$ ties pointwise, and
every action outside $O$ is strictly worse at each state.  Nonnegative
averaging with at least one positive coefficient preserves exactly that
argmax set.  On a zero-mass fiber, all aggregate scores tie and the second
conjunct gives the same pointwise optimizer $A$.  This proves
\eqref{eq:stochastic-preservation-normal-form}.

The sufficiency-and-coverage formula is
Definition~\ref{def:deployable-preservation}, written with
$\mathcal F_I(s)$.  Its coverage conjunct removes the zero-mass case, so the
positive-fiber argument proves preservation on every declared state.
Conversely, if preservation on every declared state holds and some fiber has no
positive-mass representative, the normal form forces $\Opt(s)=A$ on that
fiber.  The stated no-universal-optimizer hypothesis rules this out.
\end{proof}

\leanmetapending{\LH{DC108}}
\begin{proposition}[Input Sufficiency Need Not Preserve Every Declared State]\label{prop:static-not-stochastic}
There exist instances where the retained inputs determine $\Opt$ pointwise but fail
preservation on every declared state after a distribution $P$ is fixed.
\end{proposition}

\leanmetapending{\LH{DQ91}, \LH{DC138}}
\begin{proposition}[Singleton Sufficiency Gives One Conditional Winner]\label{prop:static-to-stochastic-product}
Suppose that every $I$-fiber has positive probability and that its pointwise
optimal-action set is a singleton constant on that fiber. Then $I$ preserves
on every declared state and is decisive.
\end{proposition}

\leanmetapending{\LHrng{DQ}{90}{91}, \LHrng{DQ}{96}{97}}
\begin{corollary}[Relevance and Positive Support]
\label{cor:stochastic-preservation-relevance}
Every coordinate set that preserves on every declared state contains every
statically relevant coordinate.  Hence, on the Boolean cube, a yes-instance of
\textsc{Within-Budget-All-State-Preservation} with budget $k$ has at most $k$
statically relevant coordinates.  Conversely, input sufficiency on $I$
implies preservation on every declared state whenever every $I$-fiber contains a positive-mass
state.  On one named fiber, optimizer constancy together with positive mass
implies preservation on that observation.
\end{corollary}

\begin{proof}[Proof of Propositions~\ref{prop:static-not-stochastic}--\ref{prop:static-to-stochastic-product}
and Corollary~\ref{cor:stochastic-preservation-relevance}]
For the separation, retain the full interface on two states with
$P(s_0)=1$, $P(s_1)=0$, and opposite singleton optimizers. Static sufficiency
holds, but the null state's aggregate optimizer is $A$.
The normal form proves the positive-fiber implications, including uniqueness
for a common singleton. Preservation implies sufficiency, hence relevance
containment on products and the budget bound; the same averaging argument
works on a named positive fiber.
\end{proof}

\leanmetapending{\LH{DC126}, \LH{DC128}, \LH{DC150}}
\begin{proposition}[When No Interface Preserves Every Declared State]
\label{prop:joint-preservation-empty}
On a Boolean product space, let $R$ be the statically relevant-coordinate
set.  No coordinate set preserves on every declared state if and only if there is a
state $s$ such that
\[
 \Opt(s)\ne A
 \quad\text{and}\quad
 \nexists t\in F_R(s):P(t)>0.
\]
In particular, under the all-state convention, the full coordinate set fails
exactly when some zero-probability state $s$ has $\Opt(s)\ne A$.  By contrast,
the full coordinate set has observation coverage exactly when every declared
state has positive probability.
\end{proposition}

\begin{proof}
A preserving $I$ contains $R$. Every $R$-fiber is a union of $I$-fibers
with a common pointwise optimizer. If it is positive, averaging preserves
that optimizer; if null, all its $I$-fibers are null and preservation forces
$\Opt=A$. Thus any preserving interface makes $R$ preserve too.
The normal form applied to $R$ gives the criterion. Full-interface fibers
are singletons, giving the stated special cases.
\end{proof}

A smaller interface can restore preservation on every declared state by grouping an unsupported
state with a supported state having the same optimizer.

Coverage asks a different question: for every reading the interface admits,
can we find a state of positive probability that produces it? That is a
for-every-reading, find-a-state test. The theorem proves that this extra
support obligation reaches $\Pi_2^P$ even when each state has one best action.

\leanmetapending{\LHrng{DC}{126}{130}}
\begin{theorem}[Every Declared Observation Needs a Witness]\label{thm:stochastic-preservation-tight}
Under the dyadic succinct product-state encoding of
Section~\ref{sec:encoding}, with nonnegative normalized probabilities:
\begin{enumerate}
\item \textsc{Sufficiency-With-Coverage-Check} is $\Pi_2^P$-complete;
\item \textsc{Within-Budget-Sufficiency-With-Coverage} is $\Pi_2^P$-complete; and
\item \textsc{Some-Observation-Sufficiency-With-Coverage} is $\Sigma_2^P$-complete.
\end{enumerate}
The lower bounds already hold on Boolean product spaces with two actions and
dyadic probabilities.  The proposed-input and within-budget lower-bound instances have a
unique pointwise optimizer at every state.
\end{theorem}

\begin{proof}
For a proposed set, input sufficiency is a coNP condition and coverage has the
form
\[
 \forall s\;\exists t:
 \quad t_I=s_I\ \land\ C_P(t)>0.
\]
The matrix is polynomial-time computable, so their conjunction lies in
$\Pi_2^P$.  The some-observation query guesses an observation and a positive
representative, then universally checks optimizer constancy on that fiber,
placing it in $\Sigma_2^P$.

For the within-budget problem, let $R$ be the unique statically relevant-coordinate
set.  If $I$ is sufficient with observation coverage and $R\subseteq I$, every $R$-fiber is a union of
$I$-fibers and therefore contains positive mass; sufficiency of $R$
then gives $R$ observation coverage.  Thus $R$ is the least sufficient and
observation-covering set whenever any one exists, and
\begin{equation}
\label{eq:preservation-minimum-collapse}
 \begin{aligned}
 &\exists I\ \bigl(|I|\le k\land
     I\text{ is sufficient and observation-covering}\bigr)\\
 &\qquad\iff |R|\le k\land
     R\text{ is sufficient and observation-covering}.
 \end{aligned}
\end{equation}
The complement of \eqref{eq:preservation-minimum-collapse} is in
$\Sigma_2^P$.  If $|R|>k$, a certificate lists
$k+1$ distinct relevant coordinates and a local relevance witness for each.
If $R$ fails coverage, a certificate lists a state $s$ and a set
$W\subseteq R$ of relevant coordinates, with local relevance witnesses, such
that every positive-mass state differs from $s$ on some coordinate in $W$.
The local witnesses certify $W\subseteq R$, while the disagreement condition
excludes every positive-mass realization of $s_R$.  Thus the whole $R$-fiber
contains no positive-mass state.  Conversely, when an $R$-fiber has no
positive-mass state, $W=R$ is a complete
certificate.  The remaining test is universal with a polynomial-time matrix,
so the within-budget problem lies in $\Pi_2^P$.

The matching within-budget reduction below has two separate jobs.  Parity
and local relevance first pin the only budget-feasible interface; the support
circuit then makes coverage of that interface express
$\forall x\,\exists y\,\varphi(x,y)$.

For matching hardness of the first two queries, reduce the standard
$\forall x\exists y\,\varphi(x,y)$ problem~\cite{stockmeyer1976polynomial}.
Use states $(x,c,y,b)$, observe $(x,b)$, and put
\[
 \psi(x,c,y)=c\land\varphi(x,y).
\]
For each $(x,c,y)$, give the unique state with $b=\psi(x,c,y)$ the dyadic atom
$2^{-(|x|+|y|+1)}$ and give its companion mass zero.  The distribution is
normalized without computing a model count because exactly one state per
triple receives an atom.  The $(x,0)$ fiber is always positive: choose
$c=0$.  The $(x,1)$ fiber is positive exactly when
$\exists y\,\varphi(x,y)$.

For actions $a_0,a_1$, set
\[
 U(a_1,x,c,y,b)=\mathbf 1[b\oplus\operatorname{parity}(x)=1],
 \qquad
 U(a_0,x,c,y,b)=1-U(a_1,x,c,y,b).
\]
The unique pointwise label is therefore $a_{b\oplus\operatorname{parity}(x)}$.
This optimizer depends only on the observed
coordinates, and toggling any $x_i$ or $b$ changes it, so the relevant set is
exactly $(x,b)$.  The observed interface is therefore sufficient with
observation coverage exactly when
$\forall x\exists y\,\varphi(x,y)$, and the budget $|x|+1$ forces the same
interface in the within-budget query.  All circuits and dyadic numbers have
polynomial encoded size.  Some-observation hardness already holds in the
full-support subfamily by the static $\Sigma_2^P$ reduction and
Theorem~\ref{thm:stochastic-preservation-full-support-variants}.
\end{proof}

\leanmetapending{\LHrng{DC}{126}{130}}
\begin{corollary}[Complexity of Preserving Every Declared State]
\label{cor:joint-preservation-complexity}
Under the same encoding,
\textsc{All-State-Preservation-Check} and
\textsc{Within-Budget-All-State-Preservation} are $\Pi_2^P$-complete, and
\textsc{Some-Observation-All-State-Preservation} is $\Sigma_2^P$-complete.
\end{corollary}

\begin{proof}
The normal form in
Theorem~\ref{thm:stochastic-preservation-zero-mass-characterization} gives the
proposed-input upper bound.  Passing to the relevant-coordinate set gives the
within-budget upper bound, and the selected observation is either a positive
constant fiber or a null fiber with universal optimizer, giving the
some-observation upper bound.  For proposed-input and within-budget hardness,
the construction above has a unique pointwise optimizer at every state;
sufficiency with observation coverage and preservation on every declared state
therefore coincide.  Some-observation hardness uses the same full-support
static family.
\end{proof}

The lower-bound instances have unique pointwise winners; support coverage
supplies the $\forall\exists$ test.

\begin{problem}[One-Conditional-Winner Queries]
\label{prob:stochastic-singleton-sufficiency}
\textsc{Decisiveness-Check} asks whether a supplied $I$ is decisive;
\textsc{Some-Observation-Decisiveness-Check}\label{prob:stochastic-anchor-sufficiency}
asks whether one combinatorial $I$-fiber has a unique winner; and
\textsc{Within-Budget-Decisive-Set}\label{prob:stochastic-minimum-sufficient}
asks for a decisive $|I|\le k$.
\end{problem}

\subsubsection{One Fiber: Equality Versus a Threshold}
For two actions, such as brake and steer, subtract the probability-weighted
steer score from the brake score in each hidden state and add the differences.
A positive total selects brake, a negative total selects steer, and zero is a
tie. Asking for one winner therefore tests whether this total is nonzero;
it does not ask which action wins. With three actions, two holds can instead
tie at a threshold that accept must exceed.

The class $\mathsf C_{=}\mathsf P$ tests whether a signed count is zero;
its complement, $\mathsf{coC}_{=}\mathsf P$, tests whether it is nonzero.
For a two-action fiber, zero means a tie and nonzero means one winner.
These are different counting tests. A PP-hard test of equality on
one two-action fiber would imply
$\mathsf{PP}=\mathsf C_{=}\mathsf P$. This boundary concerns one observation:
the construction assigns each forbidden smaller count its own observation,
where equality produces a tie. Checking that every observation avoids a tie
therefore expresses a threshold with two actions, giving the proposed-input and within-budget classes
(Theorem~\ref{prop:two-action-exact-counting}). The individual-bit horizon-two
reduction will reuse the three-action score mechanism on its own restricted
state set.

Throughout the stochastic counting results, $\mathrm{MAJ}$ means weak
majority.  For a Boolean formula $\psi(y)$,
\begin{equation}
\label{eq:weak-majority-convention}
 \mathrm{MAJ}_y\,\psi
 \quad\Longleftrightarrow\quad
 \bigl|\{y:\psi(y)\}\bigr|\ge 2^{|y|-1}.
\end{equation}

For a language class $\mathcal C$, write $\exists\!\cdot\!\mathcal C$ for
polynomially bounded existential projection: an input is accepted exactly
when some polynomial-length string makes the paired input belong to a language
in $\mathcal C$~\cite{wagner1986complexity,toran1991counting}.

The next lemma connects these guesses to solver calls. An oracle
\emph{transcript} records the questions asked and the replies received.
Calls are adaptive when a later question depends on an earlier reply.
Guessing a complete transcript fixes the questions, but leaves every reply
to be checked. For an oracle that asks whether a witness exists, a yes-reply
needs one valid witness; a no-reply must rule out every candidate witness.

The supplement, \emph{Exact-Counting Projection Identities}, gives the adaptive-oracle identities and their derivation.

Decisiveness implies the some-observation version: choose any state $s_0$.
The optimizer of its $I$-fiber is a singleton, although the fiber itself may
contain many states.
\leanmeta{\LH{DC19}}

\leanmetapending{\LH{DC134}, \LH{OU1}, \LHrng{EH}{4}{6}}
\begin{theorem}[Finding One Observation with a Unique Winner]\label{thm:stochastic-anchor-hard}
Under the dyadic succinct encoding,
\textsc{Some-Observation-Decisiveness-Check} is
$\mathsf{NP}^{\mathsf{PP}}$-complete even when the action set has exactly
three elements.  The hardness already holds for the uniform distribution.
\end{theorem}

\begin{proof}
Membership follows by guessing a state $s_0$ and an action $a$, then using PP
comparisons to verify that $a$ strictly beats every competitor on the
$s_0$-fiber.  For hardness, reduce E-MAJSAT, the canonical
$\mathsf{NP}^{\mathsf{PP}}$-complete problem~\cite{littman1998probplanning}.
Given $\varphi(x,y)$, take the uniform state space on $(x,y)$, retain the $x$
coordinates, and set
\[
U(\mathrm{accept},x,y)=\mathbf{1}[\varphi(x,y)],\qquad
U(\mathrm{hold}_L,x,y)=U(\mathrm{hold}_R,x,y)
=\frac12-2^{-(|y|+1)}.
\]
The fiber indexed by $x$ has a unique optimum exactly when
$\mathrm{MAJ}_y\,\varphi(x,y)$ holds, so some observation is decisive exactly
when $\exists x\,\mathrm{MAJ}_y\,\varphi(x,y)$.
\end{proof}

\subsubsection{A Common Counting Source for Two Constructions}

The budget query combines three jobs: choose sensors, check every retained
reading, and add the scores of the hidden states behind that reading.
The reductions encode those jobs by existential, universal, and counted
variables of fixed lengths. This adds a representation requirement to the counting-quantifier and SSAT
settings~\cite{wagner1986complexity,toran1991counting,littman2001ssat,lee2018existrandom}:
retained coordinates must encode the existential assignment. The majority-test
observations use counted blocks of equal length so that they share one
threshold. In the supplementary three-action construction, exposed singleton
observations have separate tie-breaking scores; they are not additional
majority tests. We use right-associated oracle
notation~\cite{kwisthout2011mpe}; complementing an oracle leaves its adaptive
power unchanged. The supplement's \emph{Fixed-block normalization} section
also normalizes strict to weak majority.

\leanmetapending{\LH{OU13}}
\begin{lemma}[One Threshold for Every Observation]
\label{lem:nested-pp-oracle-normal-form}
The language of Boolean circuits satisfying
\begin{equation}
\label{eq:nested-fixed-block-source}
  \exists z\;\forall x\;\mathrm{MAJ}_y\,\varphi(z,x,y).
\end{equation}
is $\mathsf{NP}^{(\mathsf{NP}^{\mathsf{PP}})}$-complete under
polynomial-time many-one reductions.  Here $\varphi$ is evaluated as a circuit,
so no auxiliary gate variables enter the majority block.  Moreover, the
reduction from an input $w$ produces a majority block of length
$n_y=p(|w|)$ for one polynomial $p$;
the length is independent of the values assigned to $(z,x)$.
\end{lemma}

\paragraph{Why this holds.}
Choosing the outer assignment leaves the universal/counting test. Padding makes the majority block length independent of the assignments.
The supplement's section \emph{Supporting Proofs for the Main Results} gives the complete proof of this result.

\subsubsection{Exact Two-Action Classification}
\label{sec:two-action-exact}

Two actions are enough for the full choose-then-check difficulty. Although
one reading only tests for a tie, a check over every reading can rule out
every count below a threshold. For example, a count is at least four exactly
when it differs from each of zero, one, two, and three. The threshold lemma
below formalizes this step; the selector construction makes the input budget
encode the initial choice.

The supplement's \emph{Exact-Counting Projection Identities} identifies the
resulting projection forms with the standard oracle classes used below
~\cite{wagner1986complexity,toran1991counting}.

The zero set of a GapP function defines $\mathsf C_{=}\mathsf P$.  After
padding two nondeterministic machines to the same path length, every GapP
function has the form
\begin{equation}
\label{eq:gapp-signed-path-normal-form}
 g(w,u)=\sum_{p\in\{0,1\}^{m}}h(w,u,p),
 \qquad h(w,u,p)\in\{-1,0,1\},
\end{equation}
where $h$ is computed by a polynomial-size circuit.  The relevant complete
predicates are
\[
 \forall x\;g(w,x)\ne0,
 \qquad
 \exists x\;g(w,x)\ne0,
 \qquad
 \exists z\;\forall x\;g(w,z,x)\ne0.
\]
The next observation lets the universal fiber layer recover a threshold
without adding a third action.

\leanmetapending{\LH{DC163}}
\begin{lemma}[A Threshold Is a Family of Nonzero Tests]
\label{lem:majority-as-universal-nonzero}
Let $m\ge1$, let $c=|\{y\in\{0,1\}^m:\varphi(y)\}|$, and encode
$\ell\in\{0,\ldots,2^{m-1}-1\}$ with $m-1$ bits.  Then
\[
 c\ge2^{m-1}
 \quad\Longleftrightarrow\quad
 \forall \ell<2^{m-1}:\ c-\ell\ne0.
\]
Moreover, with a signed path $(b,v)\in\{0,1\}^{m+1}$,
\[
 h(\ell,b,v)=
 \begin{cases}
  \mathbf 1[\varphi(v)],&b=0,\\
  -\mathbf 1[v<\ell],&b=1
 \end{cases}
 \quad\text{satisfies}\quad
 \sum_{b,v}h(\ell,b,v)=c-\ell.
\]
\end{lemma}

\begin{proof}
The forbidden values $0,\ldots,2^{m-1}-1$ are exactly the counts below weak
majority.  The signed-path identity follows because precisely $\ell$ strings
$v\in\{0,1\}^m$ represent integers below $\ell$.
\end{proof}

\leanmetapending{\LH{DC154}, \LHrng{DC}{158}{159}}
\begin{theorem}[Exact Two-Action Decisiveness Classification]
\label{prop:two-action-exact-counting}
Suppose $A=\{a,b\}$ and define the signed score gap on an $I$-fiber by
\[
 G_I(s):=\sum_{t:t_I=s_I}P(t)\bigl(U(a,t)-U(b,t)\bigr).
\]
Then $I$ is decisive if and only if $G_I(s)\ne0$ for every state $s$.
Under the succinct dyadic product encoding with exactly two actions:
\begin{enumerate}
\item\label{item:two-action-check} \textsc{Decisiveness-Check} is
$\operatorname{co}(\mathsf{NP}^{\mathsf{PP}})$-complete;
\item\label{item:two-action-some} \textsc{Some-Observation-Decisiveness-Check} is
$\mathsf{coC}_{=}\mathsf P$-complete; and
\item\label{item:two-action-budget} \textsc{Within-Budget-Decisive-Set} is
$\mathsf{NP}^{(\mathsf{NP}^{\mathsf{PP}})}$-complete.
\end{enumerate}
The lower bounds use an independent uniform distribution on a full product.
The first two use Boolean coordinates.  The minimum reduction uses coordinate
domains of size at most four.  For $I=\varnothing$, the proposed-input problem
specializes to $\mathsf{coC}_{=}\mathsf P$-completeness.
\end{theorem}

In query language, item~\ref{item:two-action-check} asks whether every fiber
avoids a tie, item~\ref{item:two-action-some} asks whether some fiber avoids
a tie, and item~\ref{item:two-action-budget} first chooses an interface and
then asks whether every resulting fiber avoids a tie.

\paragraph{How the input budget chooses an assignment.}
For the within-budget lower bound, turn each majority threshold into the
nonzero tests of Lemma~\ref{lem:majority-as-universal-nonzero}. Write
$u=(x,\ell)$ for the universal bits together with the threshold index, and
pad $u$ and the existential assignment $z$ to a common length $N\ge1$.
Each candidate input $q_{j,b}$ has four values: a sign bit $r_{j,b}$ and a
bit $\xi_{j,b}$ of the observation. Selecting one input from each pair
$(q_{j,0},q_{j,1})$ records $z_j=b$; its observed $\xi$ bits specify $u$.

Hidden coordinates hold another assignment $z'$ and a signed path $p$.
Set the two utilities to $1+d$ and $1-d$, where $d$ is the signed path value
multiplied by the signs selected by $z'$. If an interface misses both inputs
of any pair, the hidden sign takes opposite values equally often, cancelling
the score gap. That observation ties. Selecting hidden coordinates cannot
repair the missing sign. A decisive interface with budget $N$ must therefore
select exactly one input from each of the $N$ pairs.

For such an interface, every term with $z'\ne z$ cancels at a hidden sign
from its first mismatch. The terms with $z'=z$ survive. Their gap is a
nonzero sign and a positive common factor times $g(w,z,u)$, the signed count
from Equation~\eqref{eq:gapp-signed-path-normal-form}. With the threshold
construction, this count is $c_{z,x}-\ell$, where $c_{z,x}$ counts satisfying
assignments of the original $n_y$-bit block. Every observation avoids a tie
exactly when every original count reaches its majority threshold.
Figure~\ref{fig:paired-selector-gadget} shows these cancellations; the
supplement gives the exact factor, circuit sizes, and complete proof.

\begin{figure}[htbp]
\centering
\resizebox{\linewidth}{!}{%
\begin{tikzpicture}[
  x=1cm,y=1cm,
  every node/.style={font=\small},
  box/.style={draw,rounded corners=2pt,fill=black!3,inner sep=4pt,align=center},
  chosen/.style={draw,rounded corners=2pt,fill=black!12,inner sep=4pt,align=center},
  heading/.style={font=\footnotesize\bfseries,align=center,text width=4.7cm},
  note/.style={align=center,text width=4.2cm}
]
  \path[use as bounding box] (-0.2,0) rectangle (15.2,5.9);
  \draw[black!35] (5.0,0.25) -- (5.0,5.75);
  \draw[black!35] (10.2,0.25) -- (10.2,5.75);
  \node[heading] at (2.3,5.25) {Choose one input per bit};
  \node[box,text width=1.8cm] at (1.1,4.0) {$q_{j,0}$\\$(r_{j,0},\xi_{j,0})$};
  \node[chosen,text width=1.8cm] at (3.5,4.0) {$q_{j,1}$\\chosen: $z_j=1$};
  \node[box,text width=4.0cm] at (2.3,2.4) {omit both inputs:\\hidden signs cancel $G_I$};
  \node[note] at (2.3,0.9) {budget $N$ forces one of each pair};
  \node[heading] at (7.6,5.25) {Which terms survive?};
  \node[chosen,text width=1.6cm] (match) at (6.1,4.0) {$z'=z$};
  \node[chosen,text width=2.3cm] (count) at (8.8,4.0) {signed count\\$g(w,z,u)$};
  \draw[->] (match.east) -- (count.west);
  \node[box,text width=1.6cm] (mismatch) at (6.1,2.4) {$z'\ne z$};
  \node[box,text width=2.3cm] (cancel) at (8.8,2.4) {hidden sign\\cancels};
  \draw[->] (mismatch.east) -- (cancel.west);
  \node[heading] at (12.7,5.25) {Two actions: tie or no tie};
  \node[box,text width=4.2cm] at (12.7,3.55) {
    $U(a)=1+d$, $U(b)=1-d$\\[2pt]
    $G_I=C\chi(c_{z,x}-\ell)$\\
    $C>0$, $\chi\in\{-1,1\}$};
  \node[note] at (12.7,0.95) {all $\ell<2^{n_y-1}$ avoid zero\\iff $c_{z,x}\ge2^{n_y-1}$};
\end{tikzpicture}}
\caption{The two-action selector for
Theorem~\ref{prop:two-action-exact-counting},
item~\ref{item:two-action-budget}. Left: each four-valued input carries a
sign and an observation bit; omitting an entire pair leaves a cancelling
hidden sign, so the budget chooses one input per pair. Middle: hidden
assignments inconsistent with that choice cancel; the chosen assignment
contributes its signed count. Right: the two action scores tie exactly when
that count is zero. Testing every threshold index excludes every count below
weak majority. All coordinates are independent and uniform.}
\Description{Three panels explain the two-action uniform-product reduction.
A budget selects one four-valued input per pair. Hidden sign bits cancel
terms from unselected assignments. The surviving gap is a positive factor
times a sign and the difference between the accepting count and a threshold
index. Nonzero gaps for every index below half the paths establish weak
majority.}
\label{fig:paired-selector-gadget}
\end{figure}

\begin{remark}[Why the nested class is not simplified]
\label{rem:nested-class-not-simplified}
No collapse of
$\mathsf{NP}^{(\mathsf{NP}^{\mathsf{PP}})}$ to
$\mathsf{NP}^{\mathsf{PP}}$ is known.  The closure facts used above---PP
closure under complement and polynomial truth-table reductions---combine the
counting tests inside one fixed oracle transcript.  They do not absorb an
adaptive outer NP oracle around an $\mathsf{NP}^{\mathsf{PP}}$ predicate.
The theorem therefore states the exact complete class without assuming that
the displayed nesting is strict~\cite{wagner1986complexity,toran1991counting}.
\end{remark}

Since GapP is closed under products, $g\ne0$ is equivalent to $g^2>0$, so
$\mathsf{coC}_{=}\mathsf{P}\subseteq\mathsf{PP}$; PP complement closure also
places $\mathsf{C}_{=}\mathsf{P}$ in PP.
The action-count boundary therefore remains only for the some-observation
query: two actions give $\mathsf{coC}_{=}\mathsf P$, whereas three actions give
$\mathsf{NP}^{\mathsf{PP}}$.  Equality of those two classes would imply
$\mathsf{NP}^{\mathsf{PP}}=\mathsf{PP}$.

\paragraph{Which distribution restrictions already preserve hardness.}
The proposed-input, some-observation, and two-action within-budget reductions
use independent uniform coordinates.  These distributions have Bayesian
networks of moral treewidth zero, so the hardness survives a stronger restriction than
the polytree value-of-information lower bound of Krause and Guestrin
~\cite{krause2009optimal}.  The separate three-action construction supplies a
polynomial-size circuit for its positive dyadic distribution and product
completion; it is retained with the horizon-two proof in the supplement's
\emph{Three-Action Carrier for Horizon-Two Inheritance}. The two-action
construction already matches the one-step upper bounds. The
$\Pi_2^P$ observation-coverage reduction still uses an arbitrary succinct
support circuit.  Whether that result survives a polytree or fixed-treewidth
bound on the supplied distribution remains open.  This restriction concerns the
representation of $P$, independently of how $U$ is checked.\leanmeta{\LH{DC57}, \LHrng{OU}{8}{11}, \LHrng{OU}{36}{37}, \LHrng{EH}{2}{10}}

\subsection{Approximation Across Presentations}\label{sec:approximation-across-presentations}

Suppose we accept keeping more sensors than necessary. Every original
response must still survive: the approximation concerns how many inputs we
keep, not the decisions they recover. On succinct Boolean products, even
distinguishing a one-input optimum from an all-input optimum is hard.
An explicit table exposes the conflicting state pairs. We can repeatedly
keep the input that distinguishes the most remaining conflicts; this is
greedy hitting set, with a tight logarithmic guarantee.

Formally, a factor-$\rho$ approximation is a polynomial-time
algorithm for the stated promised family that returns an actual interface
satisfying the preservation or decisiveness contract and whose cardinality is
at most $\rho$ times the optimum.  The promise includes feasibility and a
positive optimum where a multiplicative ratio is used.

Probabilistic sufficient reasons provide a neighboring local theory with
sharp algorithmic and inapproximability results
~\cite{izza2022precise,arenas2025probabilistic,kozachinskiy2023inapproximability}.
Here one projection must preserve the complete optimizer globally.

The sufficiency indicator need not have diminishing returns: on two-bit parity,
adding the second bit helps only after the first has been selected.
Submodular guarantees therefore need an additional hypothesis
~\cite{krause2009optimal,golovin2011adaptive}. On listed tables, the number
of conflicting row pairs distinguished by the retained inputs has diminishing
returns, supplying that hypothesis.

\subsubsection{Succinct Boolean Products: A One-to-$p$ Gap}
The next family always needs at least one input. This gives a meaningful
approximation ratio even when the earlier construction in
Theorem~\ref{thm:sufficiency-conp} can have optimum zero.

\begin{definition}[Paired Tautology Family]
\label{def:paired-tautology-family}
Let $\varphi$ be a Boolean formula on $n\ge1$ variables.  A state is
$(g,x,y)\in\{0,1\}\times\{0,1\}^n\times\{0,1\}^n$, so the presentation has
$p=2n+1$ Boolean coordinates.  There are two actions, \emph{accept} and
\emph{reject}.  Define
\[
 F_\varphi(g,x,y)
   :=\mathbf 1[g=1]\,
      \mathbf 1\!\left[\varphi(x)=1\ \vee\ x=y\right],
\]
and put
\[
 U(\mathit{accept},g,x,y)=F_\varphi(g,x,y),
 \qquad
 U(\mathit{reject},g,x,y)=0.
\]
The utility circuit has size polynomial in $|\varphi|+n$ and takes only the
values zero and one.  In particular, the gate indicator multiplies the whole
disjunction: $F_\varphi(0,x,y)=0$ for every $x,y$.
\end{definition}

\leanmetapending{\LHrng{PBA}{1}{7}}
\begin{proposition}[Boolean Product One-to-$p$ Gap]
\label{prop:paired-family-approximation-gap}
On the paired tautology family,
\[
 \operatorname{OPT}(\varphi)=
 \begin{cases}
  1,&\varphi\text{ is a tautology},\\
  p,&\varphi\text{ is not a tautology},
 \end{cases}
 \qquad p=2n+1,
\]
where $\operatorname{OPT}$ denotes the minimum cardinality of a sufficient
coordinate set.  In particular, every instance in the family has
$\operatorname{OPT}\ge1$.
\end{proposition}

\begin{proof}
If $\varphi$ is a tautology, $F_\varphi=g$, so the gate alone is necessary
and sufficient. Otherwise fix $a$ with $\varphi(a)=0$.
The states $(0,a,a),(1,a,a)$ witness gate relevance. For each $i$,
$(1,a,a^{\oplus i})$ ties, while $(1,a^{\oplus i},a^{\oplus i})$ and
$(1,a,a)$ strictly prefer accept. These pairs witness relevance of $x_i$
and $y_i$, respectively. All $p$ coordinates are therefore necessary.
\end{proof}

\leanmetapending{\LHrng{PBA}{1}{7}}
\begin{theorem}[One Input or All Inputs: A Hard Approximation Gap]
\label{thm:succinct-product-approximation-threshold}
On the paired tautology family, returning all $p$ coordinates is a
polynomial-time factor-$p$ approximation.  Unless
$\mathsf P=\mathsf{NP}$, no polynomial-time algorithm has a guarantee
$\rho(p)<p$ on all sufficiently large odd $p$.
Thus this nonzero-optimum family has a gap from $1$ to $p$, and no guarantee
strictly below $p$ holds uniformly on the family unless
$\mathsf P=\mathsf{NP}$.
\end{theorem}

\begin{proof}
Returning the full set has ratio at most $p$ because $\operatorname{OPT}\ge1$.
A ratio $\rho(p)<p$ returns fewer than $p$ inputs on tautologies and must
return all $p$ otherwise, deciding \textsc{Tautology}. Padding with unused
formula variables reaches every sufficiently large odd $p$ in polynomial time.
\end{proof}

\leanmetapending{\LHrng{PBA}{1}{7}}
\begin{corollary}[Input-Length Inapproximability]
\label{cor:succinct-input-length-inapproximation}
Let $\Lambda$ be the total encoded length of a succinct minimum-sufficiency
instance.  Unless $\mathsf P=\mathsf{NP}$, no polynomial-time algorithm has
approximation ratio $O(\Lambda^{1-\varepsilon})$ for any constant
$\varepsilon>0$.  More generally, no such algorithm has ratio
$o(\Lambda/\log\Lambda)$.
\end{corollary}

\paragraph{Why this holds.}
Bounded-occurrence formulas keep the gate-and-wire encoding nearly linear in the coordinate count. The one-to-all gap then gives the stated encoded-length obstruction.
The supplement's section \emph{Approximation Transfers} gives the complete proof of this result.

\subsubsection{Within-Budget Decisiveness Has the Same Worst-Case Gap}

The construction uses cancellation rather than relevance.  Let $\varphi$ be
a Boolean formula on $n\ge1$ variables, put $p=n+1$, and use the uniform
product on states $(g,x)\in\{0,1\}\times\{0,1\}^n$.  Define
\[
 M=2^{n+1},\qquad
 \chi(g,x)=(-1)^{g+\sum_i x_i},\qquad
 d_\varphi(g,x)=
 \begin{cases}
   (-1)^gM,&\varphi(x)=1,\\
   \chi(g,x),&\varphi(x)=0.
 \end{cases}
\]
For two actions $a,b$, set $U(a)=M+d_\varphi$ and
$U(b)=M-d_\varphi$.  These utilities are nonnegative integers with
polynomial bit length.

\Needspace{11\baselineskip}
\leanmetapending{\LHrng{DC}{160}{161}}
\begin{theorem}[A One-to-All Gap for Within-Budget Decisiveness]
\label{thm:minimum-decisiveness-approximation-threshold}
On the family above,
\[
 \operatorname{OPT}_{\mathrm{dec}}(\varphi)=
 \begin{cases}
  1,&\varphi\text{ is satisfiable},\\
  p,&\varphi\text{ is unsatisfiable}.
 \end{cases}
\]
Consequently returning all coordinates is a factor-$p$ approximation, while,
unless $\mathsf P=\mathsf{NP}$, no polynomial-time algorithm has guarantee
$\rho(p)<p$ on all sufficiently large $p$.  If $\Lambda$ is the total succinct
input length, the same assumption rules out every ratio
$o(\Lambda/\log\Lambda)$ and hence every
$O(\Lambda^{1-\varepsilon})$ for fixed $\varepsilon>0$.
The lower bound already holds with two actions and an independent uniform
distribution on Boolean coordinates.
\end{theorem}

\paragraph{Why this holds.}
A satisfying assignment makes the gate singleton decisive. Otherwise parity cancellation makes every proper interface tie; bounded-occurrence formulas transfer this gap to encoded length.
The supplement's section \emph{Approximation Transfers} gives the complete proof of this result.

\subsubsection{Explicit Listed Tables: A Logarithmic Guarantee}

For a listed presentation and two states with different optimizer labels,
define the discernibility set
\[
 D(s,t):=\{i:\pi_i(s)\ne\pi_i(t)\},
 \qquad
 \mathcal D_{\ne}:=
 \{D(s,t):\Opt(s)\ne\Opt(t)\}.
\]
This is the classical rough-set discernibility construction
~\cite{skowron1992discernibility}.

\leanmetapending{\LH{DSC1}}
\begin{proposition}[Discernibility Hitting-Set Formulation]
\label{prop:listed-discernibility-hitting-set}
A coordinate set $I$ is sufficient for an explicit listed table if and only
if it hits every member of $\mathcal D_{\ne}$.  An empty member certifies that
no sufficient coordinate set exists because the coordinate maps fail to
separate two different optimizer labels; otherwise finding the smallest sufficient set is
exactly the minimum hitting-set problem for $\mathcal D_{\ne}$.

If $\varnothing\notin\mathcal D_{\ne}$ and the table has $N$ states and
$M=|\mathcal D_{\ne}|\ge1$ distinct discernibility constraints, the greedy
covering algorithm returns a
sufficient set of size at most
\[
 H_M\operatorname{OPT}
 \le (1+\ln M)\operatorname{OPT}
 \le \left(1+\ln\binom N2\right)\operatorname{OPT},
\]
in polynomial time~\cite{chvatal1979greedy}.  When $M=0$, the empty set is
optimal.
\end{proposition}

\begin{proof}
A conflicting pair agrees on $I$ iff $I$ misses its discernibility set.
An empty constraint is therefore infeasible. Otherwise apply greedy
set cover to the $M\le\binom N2$ constraints, one covering set per coordinate.
\end{proof}

An exhaustive finite census covers all $7{,}773$ antichains of nonempty subsets on at
most five coordinates under the declared greedy tie-break.  Greedy is exact
through four coordinates; its first and largest gap is $3/2$ at five
coordinates, where it returns three inputs against an optimum of two.
Supplement~S4 records the protocol and checks.

\leanmetapending{\LH{DQ11}, \LH{DQ34}, \LH{DQ86}}
\begin{theorem}[Cardinality-Preserving Equivalence with Set Cover on the Gadget Family]
\label{thm:set-cover-cardinality-equivalence}
The set-cover gadget uses the listed state set
$S=\{\mathsf{false},\mathsf{true}\}\times\mathcal U$ with one coordinate map
for each set $S_i$ in the instance:
\[
q_i(b,u)=b\wedge\mathbf 1[u\in S_i].
\]
The optimizer is $\{\mathit{accept}\}$ on every false-tagged state and
$\{\mathit{accept},\mathit{reject}\}$ on every true-tagged state.  On this
family, a coordinate set is sufficient if and only if the corresponding set
family is a cover.  Feasible solutions are in bijection, and optimum
cardinalities coincide.
\end{theorem}

\begin{proof}
An uncovered $u$ leaves $(\mathsf{false},u)$ and $(\mathsf{true},u)$
indistinguishable despite their different labels. Conversely, a cover
separates every opposite-tag pair: a selected set containing its true-tagged
element returns one there and zero on the false-tagged state.
Equal-tag states have equal labels. This proves both directions without
changing cardinality.

The presentation need not be product-closed: duplicate singleton sets
already omit mixed coordinate vectors, so Proposition~\ref{prop:minimal-relevant-equiv}
does not apply. \leanmeta{\LH{DQ86}}
\end{proof}

\subsection{Sequential Certification}\label{sec:sequential}

An action can have the same immediate payoff in two sensor states yet lead
to different later opportunities. To decide whether a sensor is unnecessary,
the controller must compare those continuations too. Here the controller
chooses moves and the environment may respond adversarially; the worked
two-move example below shows how this makes an input necessary at the root.
Formally, replace immediate utility by a state-indexed game and work backward:
\[
 s\ \longmapsto\ (G_{s,a})_{a\in A}
 \ \xrightarrow{\text{backward induction}}\
 (Q(s,a))_{a\in A}
 \ \longmapsto\ \Opt_{\mathrm{seq}}(s).
\]
The audit preserves the complete induced root optimizer.  Imperfect-recall
abstraction instead controls game-value loss~\cite{cermak2020imperfect}.

\begin{example}[A coordinate revealed after two moves]
\label{ex:two-move-sequential-separation}
Let the root state be one Boolean coordinate $h$, and let the root actions be
\textsc{Carry} and \textsc{Leave}.  The immediate utility vector is $(2,1)$ at
both $h=0$ and $h=1$, so the immediate optimizer is always
$\{\textsc{Carry}\}$ and the empty coordinate set is sufficient.

Replace the value of \textsc{Carry} by a horizon-two arena.  The controller
chooses a row, the environment chooses a column, and the terminal payoff
matrix is
\[
 G_0=\begin{pmatrix}2&2\\0&3\end{pmatrix}
 \quad\text{when }h=0,
 \qquad
 G_1=\begin{pmatrix}0&2\\-1&3\end{pmatrix}
 \quad\text{when }h=1.
\]
The induced value is the maximum row minimum.  It equals $2$ for $G_0$ and
$0$ for $G_1$.  Let \textsc{Leave} terminate with payoff $1$ at both states.
Backward induction therefore gives
\[
 \Opt_{\mathrm{seq}}(0)=\{\textsc{Carry}\},
 \qquad
 \Opt_{\mathrm{seq}}(1)=\{\textsc{Leave}\}.
\]
The coordinate has no effect on the immediate optimizer and becomes relevant
at the root after the controller and environment moves are evaluated.
\end{example}

\leanmetapending{\LH{DQ62}}
\begin{proposition}[Irrelevant Now, Necessary Later]\label{prop:stochastic-not-sequential}
There exist instances where $I$ is sufficient for the immediate optimizer and
fails to be sufficient for the optimizer induced by a sequential arena.
\end{proposition}

\begin{proof}
Example~\ref{ex:two-move-sequential-separation} has a constant immediate
optimizer and different induced root optimizers at its two states.  Thus the
empty coordinate set is immediately sufficient and fails sequential
sufficiency.
\end{proof}

Each switch between choosing a move and checking every environment response
adds another layer to the game evaluation. Certification must then compare
its root actions across states that look the same. The number of such move
blocks controls the complexity ladder in Table~\ref{tab:sequential-matrix}.

\subsubsection{Sequential Decision Problems}

The formal model records whose move comes next, which successors are
available, and when the payoff is evaluated. Its horizon is written in
unary: allowing $h$ moves contributes $h$ symbols to the input length.
The complexity bound therefore counts the available number of moves,
rather than treating an exponentially long game as a short binary number.

\begin{definition}[Sequential Decision Problem]\label{def:sequential-decision-problem}
A succinct instance gives finite root actions $A$, a product state space,
one unary horizon $h$, and circuits describing each arena $G_{s,a}$.
Positions have polynomial-length $\ell$-bit names. Total circuits on these
names compute ownership (controller or environment), two labelled successors,
a terminal predicate, a rational payoff $r$, and root code $v_{s,a}$.
The payoff applies at terminals or when the horizon expires; successors
are valid $\ell$-bit positions.

For a position $v$ with $t$ steps remaining, define
\[
V_t(v)=
\begin{cases}
r(v), & v\text{ terminal or }t=0,\\
\max_{b\in\{0,1\}}V_{t-1}(\operatorname{succ}(v,b)),
  & v\text{ controlled by the decision maker},\\
\min_{b\in\{0,1\}}V_{t-1}(\operatorname{succ}(v,b)),
  & v\text{ controlled by the environment}.
\end{cases}
\]
The induced root value is $Q(s,a)=V_h(v_{s,a})$, and
\[
\Opt_{\mathrm{seq}}(s)=\operatorname*{arg\,max}_{a\in A}Q(s,a).
\]
A \emph{quantifier block} is a maximal consecutive run of nonterminal
positions with the same owner in the depth-$h$ unfolding from a root.  The
arena block bound is the maximum number of such runs along every unfolded
history of length at most $h$, where a declared terminal or a depth-$h$ cutoff
ends the history.  Thus cyclic arenas and paths that do not reach a declared
terminal are still covered by the promise.  A terminal root has zero blocks.
\end{definition}

\begin{definition}[Sequential Sufficiency]\label{def:sequential-sufficient}
For a sequential instance $\mathcal D_{\mathrm{seq}}$, a coordinate set $I$ is
\emph{sequentially sufficient} if
\[
\forall s,s'\in S:\quad s_I=s'_I
\implies \Opt_{\mathrm{seq}}(s)=\Opt_{\mathrm{seq}}(s').
\]
The equality defines certification at the declared root states.
\end{definition}

The proposed-input query supplies $I$ and asks whether this equality holds.  The
within-budget query asks whether it holds for some $I$ of size at most $k$.  The
some-observation query asks whether one $I$-fiber has a constant induced root optimizer.
These are the three static queries with $\Opt$ replaced by
$\Opt_{\mathrm{seq}}$.

\leanmetapending{\LHrng{DC}{49}{50}}
\begin{proposition}[The Least Sequentially Sufficient Set Is the Relevant Set]\label{prop:sequential-minimal-relevant}
Assume the state space is product-closed, so single-coordinate replacements
remain valid states.  The sequentially sufficient sets have a unique least
member, and it is exactly the set of coordinates relevant for the induced root
optimizer.
\end{proposition}

\begin{proof}
Sequential sufficiency is ordinary sufficiency of the induced root optimizer
map on a product-closed state space.  The static minimal-set/relevance
equivalence therefore applies.
\end{proof}

\subsubsection{Bounded Control Alternation}

Fix a constant $d\ge0$, with $\Sigma_0^P=\Pi_0^P=\mathsf P$.  The pathwise
block bound is a promise on the succinct
arena descriptions.  Recognizing that promise is not part of the decision
problems.

Completeness uses the promise-problem convention of
Definition~\ref{def:validity-promise-convention}; every generated arena
satisfies the declared block bound.

\leanmetapending{\LHrng{DC}{155}{157}}
\begin{corollary}[Bounded-Alternation Certification Ladder]
\label{thm:sequential-bounded-ladder}
Under the promise that every succinct arena on a Boolean product state space
has at most $d$ quantifier blocks:
\begin{enumerate}
\item \textsc{Sequential-Sufficiency-Check} is
      $\Pi_{d+1}^P$-complete;
\item the sequential within-budget query is $\Pi_{d+1}^P$-complete; and
\item the sequential some-observation query is $\Sigma_{d+2}^P$-complete.
\end{enumerate}
The lower bounds use two root actions and arenas with a Boolean terminal
payoff.
\end{corollary}

\paragraph{Why this holds.}
Threshold evaluation follows the promised arena blocks; conflicting pairs or relevance witnesses add the certification quantifiers. The wrapper adds no internal arena block.
The supplement's section \emph{Supporting Proofs for the Main Results} gives the complete proof of this result.

At $d=0$ this is the static coNP/coNP/$\Sigma_2^P$ row. Each extra
owner block raises the row by one level.

\leanmetapending{\LH{DC46}, \LHrng{DC}{114}{125}}
\begin{corollary}[Unbounded-Alternation Limit]
\label{thm:sequential-pspace}
\label{thm:sequential-anchor-hard}
\label{thm:sequential-minimum-hard}
When the number of arena quantifier blocks may grow polynomially with the
input, the proposed-input, product within-budget, and some-observation
sequential queries are all
PSPACE-complete.
\end{corollary}

\paragraph{Why this holds.}
Depth-first evaluation uses polynomial space. The same constant-size base/test wrapper gives the empty-interface, zero-budget and sole-observation lower bounds.
The supplement's section \emph{Supporting Proofs for the Main Results} gives the complete proof of this result.

\section{Input Removal in Multistep POMDPs}
\label{sec:multistep-pomdp}

Both sensor states may call for continuing now, while requiring different
actions later. Hiding the sensor can then make continuation less valuable
and change even today's choice. The next example isolates this effect.
The audit keeps one fixed input set at every stage and checks the resulting
beliefs and choices; the horizon is encoded in unary.

\begin{example}[One-step constancy does not compose through a belief update]
\label{ex:pomdp-two-step-noncomposition}
Let a hidden bit $h$ be uniform and fully observed in the original model.
At the first of two decision epochs, action \textsc{Exit} terminates with
reward $3/4$, while \textsc{Continue} moves to a second epoch with no immediate
reward.  At that epoch, actions $a_0,a_1$ receive reward one exactly when their
index equals $h$, and zero otherwise.  Figure~\ref{fig:future-information}
compares the two information sets.  With the bit, \textsc{Continue} wins at
both states; without it, \textsc{Exit} wins.  Agreement of the original root
optimizers does not settle deletion when the deleted input also guides a
later action.
\end{example}

\begin{figure}[htbp]
\centering
\includegraphics[width=\linewidth]{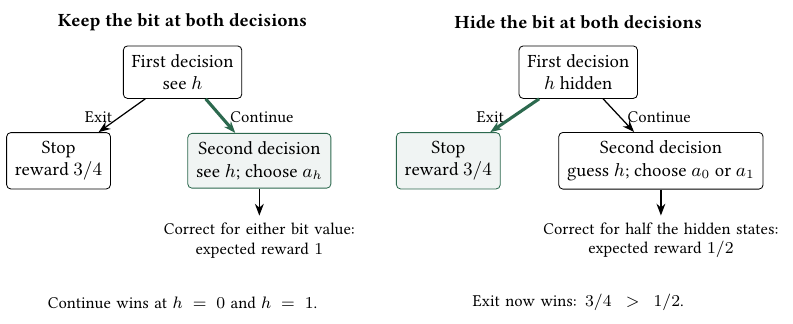}
\caption{A sensor can matter later even when today's action is the same in
both states.  The bit $h$ is equally likely to be zero or one.  Guessing it
correctly pays one; a wrong guess pays zero.  Keeping it makes continuation
worth one.  Hiding it at both decisions reduces that value to $1/2$, below
the exit reward $3/4$.  Thick green branches mark the best initial actions.}
\Description{Two matching decision trees compare full and reduced
information. The fully observed agent continues, sees h at the second stage,
and chooses the correct action for reward one. The reduced agent cannot see
h at either stage. Any later guess has expected reward one half, so it
instead exits for three quarters.}
\label{fig:future-information}
\end{figure}

There are two choices about the audit. Do we compare the first decision only,
or every later decision as well? And do we include only histories that can
occur, or every history admitted by the model? Table~\ref{tab:multistep-contracts}
keeps the timing and support requirements separate.

\begin{table}[htbp]
\centering
\small
\setlength{\tabcolsep}{4pt}
\renewcommand{\arraystretch}{1.1}
\begin{tabularx}{\linewidth}{@{}>{\raggedright\arraybackslash}m{0.3\linewidth}>{\raggedright\arraybackslash}X@{}}
\toprule
\textbf{Audit choice} & \textbf{Requirement} \\
\midrule
Temporal scope: root & stage-zero decision histories \\
\midrule
Temporal scope: historywise & decision histories at every stage before the horizon \\
\midrule
Supported-state preservation & compare full and projected optimizers at supported full histories \\
\midrule
Sufficiency with coverage & compare at all declared full histories; every declared projection needs a supported lift \\
\midrule
All-declared-state preservation & compare at all declared full histories; assign optimizer $A$ to a null projection \\
\midrule
Decisiveness & a singleton optimizer at each supported projected history \\
\bottomrule
\end{tabularx}
\caption{Multistep scopes and contracts. Scope specifies when to test;
the contract specifies which histories and support obligations to include.}
\label{tab:multistep-contracts}
\end{table}

\begin{definition}[Succinct Unary-Horizon POMDP Encoding]
\label{def:pomdp-succinct-interface}
A succinct interface instance has polynomial-length Boolean state names, a
listed action set, a unary horizon $H$, and polynomial-size circuits for
the dyadic initial distribution, stage-indexed dyadic transition
probabilities, and stage-indexed state--action--successor rewards.  Rewards
are accumulated without discount through horizon $H$.  Stages
$t=0,\ldots,H-1$ are decision stages, and stage $H$ is terminal.  The input is
a valid-POMDP encoding: masses are nonnegative, the initial masses sum to one,
and, for every declared state and every listed action at a decision stage, the
outgoing masses over declared stage-$(t+1)$ successors sum to one.  Separate
Boolean circuits declare the initial-state domain and the permitted
transitions; the probability circuit is zero outside those declarations.
Probability zero is therefore distinct from absence from the model.  This is a
deterministic-observation POMDP: the full-information model observes the
Boolean state name, and interface $I$ applies the same fixed coordinate
projection to that name at every stage.  Arbitrary observation kernels and
adaptive sensor-selection policies are outside the model.
\end{definition}

After deletion, the agent remembers its retained readings and past actions.
Several full histories may fit that record; their probabilities determine
the posterior used for the next decision.

A decision history records the states and actions so far. Its Bellman
continuation compares the best expected reward available from each next
action through the remaining stages, using the observations still available.

\begin{definition}[Decision Histories and Optimizers]
\label{def:pomdp-history-semantics}
A full \emph{decision history} at stage $t<H$ has the form
\[
 h=(s_0,a_0,s_1,\ldots,a_{t-1},s_t).
\]
Actions are recorded interventions, including suboptimal actions.
Stage $H$ is terminal and is not audited for a next action. A root history
is the observed stage-zero state, not the pre-observation empty history.
Write $\bar r(h,a)=\sum_{s'}P_t(s'\mid s,a)r_t(s,a,s')$ for the immediate
expected reward when $h$ ends at $s$ in stage $t$.

A full history is \emph{declared} if its initial state and transitions satisfy
the declaration circuits, and \emph{supported} if their masses have positive
product. Its projection retains the actions and only the $I$-coordinates.
A full history is \emph{compatible} with a projected history if it projects
to it; a projected history is declared or supported if it has a full lift
with that property. Support is dynamic feasibility under the recorded actions,
not on-policy reachability.

Define $\Opt_{\rm full}(h)$ as the complete argmax of the stage-$t$ Bellman
continuation from $s$ with full future observation. This also applies to a
null prefix: it does not condition on that prefix or include past reward.
$\Opt_I(\pi_Ih)$ instead uses the posterior and future $I$-observations.
Policies may use their entire available history and maximize total reward to $H$.

For $B\in\{\mathrm{root},\mathrm{hist}\}$, let $\mathcal H_B^D$ and
$\mathcal H_B^+$ be the declared and supported full decision histories at the
chosen scope.  The root scope contains only stage-zero histories; the
historywise scope contains all decision histories at stages $t<H$.  Let
\[
 \mathcal G_{I,B}^D=\{\pi_Ih:h\in\mathcal H_B^D\},
 \qquad
 \mathcal G_{I,B}^+=\{\pi_Ih:h\in\mathcal H_B^+\}
\]
be the corresponding declared and supported projected histories.
\end{definition}

\begin{definition}[Multistep Certification Contracts]
\label{def:pomdp-certification-contracts}
Fix $B\in\{\mathrm{root},\mathrm{hist}\}$ and an interface $I$.
Preservation on supported states requires
\[
  \forall h\in\mathcal H_B^+,\qquad
  \Opt_{\rm full}(h)=\Opt_I(\pi_I h).
\]
Decisiveness requires
\[
  \forall g\in\mathcal G_{I,B}^+,\qquad |\Opt_I(g)|=1.
\]
Input sufficiency with observation coverage requires both
\[
 \mathcal G_{I,B}^D\subseteq\mathcal G_{I,B}^+
 \quad\text{and}\quad
 \forall h\in\mathcal H_B^D:
 \Opt_{\rm full}(h)=\Opt_I(\pi_Ih).
\]
Thus every declared projected history has a compatible full history of
positive probability.  Every compatible declared full history is audited,
including one of probability zero when its projection can also arise from a
supported full history.

For preservation on every declared state, extend the projected optimizer by
\[
 \Opt_I^{\rm joint}(g)=
 \begin{cases}
   \Opt_I(g), & g\in\mathcal G_{I,B}^+,\\
   A, & g\in\mathcal G_{I,B}^D\setminus\mathcal G_{I,B}^+,
 \end{cases}
\]
and require
$\Opt_{\rm full}(h)=\Opt_I^{\rm joint}(\pi_Ih)$ for every
$h\in\mathcal H_B^D$.  Proposed-input problems supply $I$;
within-budget problems ask whether some $|I|\le k$ satisfies the target.
\end{definition}

When declarations equal positive initial and transition support, coverage
is automatic. Otherwise declarations can include null histories. Malformed
encodings are rejected. Validity is decidable in polynomial space by
enumerating state names and transition triples and summing the dyadic masses
with polynomial-bit accumulators.
\leanmetapending{\LHrng{PM}{1}{3}}
\begin{lemma}[Polynomial-Space Exact Continuation]
\label{lem:pomdp-polynomial-space}
For Definitions~\ref{def:pomdp-succinct-interface}--\ref{def:pomdp-certification-contracts}, exact next-action values,
complete optimizer sets, and tests for compatible declared or supported full histories are
computable in polynomial space.
\end{lemma}

\paragraph{Why this holds.}
Unnormalized values avoid posterior division; the unary horizon bounds recursion and exact accumulator sizes polynomially.
The supplement's section \emph{Supporting Proofs for the Main Results} gives the complete proof of this result.

\leanmetapending{\LHrng{PM}{4}{5}}
\begin{lemma}[Sign-Separated QBF Planning Core]
\label{lem:pomdp-sign-core}
From a prenex QBF $\Phi$ one can construct in polynomial time a succinct,
fully observed, two-action, unary-horizon MDP with dyadic data and value
$V_\Phi$ such that, for a polynomial-bit $\delta>0$,
\[
 \Phi\text{ true}\Longrightarrow V_\Phi\ge\delta,
 \qquad
 \Phi\text{ false}\Longrightarrow V_\Phi\le-\delta.
\]
Every declared history is supported, and every decision history
has a unique optimal action.
\end{lemma}

\paragraph{Why this holds.}
Actions choose existential bits and fair transitions choose universal bits; terminal rewards separate truth from falsity, while smaller bonuses break ties.
The supplement's section \emph{Supporting Proofs for the Main Results} gives the complete proof of this result.

\leanmetapending{\LHrng{PM}{6}{9}, \LH{DC10}}
\begin{theorem}[Unary-Horizon Interface Certification Is PSPACE-Complete]
\label{thm:pomdp-exact-pspace}
Under Definitions~\ref{def:pomdp-succinct-interface}--\ref{def:pomdp-certification-contracts}, at both the root and
historywise audit scopes, each of the following is PSPACE-complete:
\begin{enumerate}
\item proposed-input and within-budget preservation;
\item proposed-input and within-budget input sufficiency with observation coverage; and
\item proposed-input and within-budget decisiveness.
\end{enumerate}
Hardness already holds with Boolean observation coordinates, two actions,
unary horizon, polynomial-bit dyadic data, and every declared history
supported.  Consequently the same classification holds for preservation on every declared state
under the null-history convention in
Definition~\ref{def:pomdp-certification-contracts}.
Completeness is under polynomial-time many-one reductions whose outputs
satisfy every validity and normalization condition in
Definition~\ref{def:pomdp-succinct-interface}; invalid encodings are never
used as promise instances.
\end{theorem}

\paragraph{Why this holds.}
Compare a baseline root with a formula root for preservation; calibration pairs force the retained public bits. A zero-value Quit action supplies the decisiveness obstruction.
The supplement's section \emph{Supporting Proofs for the Main Results} gives the complete proof of this result.

Even two stages retain the one-step support and counting obstructions.
In the lift below, the first action is forced by its uniquely optimal reward;
the other actions remain available.

\leanmetapending{\LHrng{PM}{10}{15}, \LH{EX1}}
\begin{proposition}[Horizon-Two Inheritance]
\label{prop:pomdp-horizon-two-inheritance}
Every succinct one-step stochastic decision problem on Boolean coordinates has a polynomial-size
horizon-two POMDP lift with one forced first action and one last-stage decision
such that, for every retained coordinate set $I$:
\begin{enumerate}
\item the last-decision observation interface is sufficient with observation coverage if
      and only if $I$ is sufficient with observation coverage in the one-step instance; and
\item every supported last-decision observation has a unique optimal action if
      and only if $|\Opt_I^{\mathrm{stoch}}(s)|=1$ for every $s$ whose
      one-step $I$-fiber has positive mass.  In particular, when the supplied
      distribution is strictly positive, this is equivalent to decisiveness
      of $I$ in the one-step instance.
\end{enumerate}
In the same individual-bit interface model, horizon-two proposed-input and
within-budget \emph{historywise} sufficiency-with-coverage certification are
$\Pi_2^P$-hard.  Horizon-two within-budget historywise
decisiveness is
$\mathsf{NP}^{(\mathsf{NP}^{\mathsf{PP}})}$-hard with three actions and
positive dyadic masses on every declared initial state and permitted transition.
\end{proposition}

\paragraph{Why this holds.}
A forced first action samples the original distribution, leaving the original fiber scores at the last stage. The bit-budget transfer needs its separate block-forcing construction.
The supplement's section \emph{Supporting Proofs for the Main Results} gives the complete proof of this result.

Dynamic sensor selection lets a POMDP policy choose sensors from past
observations, while sensor synthesis and optimal observability select or
evaluate observation structures under reachability or reward objectives
~\cite{spaan2009dynamicsensor,chatterjee2018sensorsynthesis,konsta2024optimalobservability}.
If retained coordinates may depend on history, the certificate is such a
sensor-selection policy rather than the fixed projection studied here; that
adaptive-interface classification remains open.

The closest complexity premise is PSPACE-completeness of history-dependent
policy existence in compressed finite-horizon POMDPs
~\cite[Theorem~4.27]{mundhenk2000mdp}.  Compact-POMDP results also show why the
representation and horizon conventions matter
~\cite{littman1998probplanning}.  That premise evaluates one supplied
information structure.  Theorem~\ref{thm:pomdp-exact-pspace} additionally
compares the full and projected optimal-action sets, chooses one
fixed retained-coordinate set, and audits declared histories without a
compatible full history of positive probability.

\section{The Complete Classifications}

\paragraph{Explicit-state baseline.}
Explicit tables permit direct pair scans, fiber sums, and arena evaluation.
An explicit Boolean product also permits testing all $2^n=|S|$ input subsets
in polynomial time. Missing combinations in an arbitrary log instead leave
the classical NP-complete minimum-reduct problem
(Theorem~\ref{thm:set-cover-cardinality-equivalence}).
\leanmeta{\LHrng{DC}{62}{77}, \LHrng{DC}{80}{82}, \LHrng{DC}{91}{93}, \LH{DQ12}, \LH{DQ14}, \LH{DQ45}, \LHrng{DQ}{86}{89}, \LHrng{DQ}{92}{93}}

Tables~\ref{tab:regime-matrix} and~\ref{tab:sequential-matrix} separate the
one-step classification from its sequential extension.
Bold cells mark the principal boundaries.  The unbolded cells include both
adjacent new classifications and the static-product and explicit-reduct
baselines~\cite{bassan2024local,wong1985optimal}.

For multistep models, the audit selects one fixed input set and compares the
full and reduced optimal-action sets at the declared root or histories,
after belief and policy updates.

With two actions, subtract their weighted scores within a fiber. A zero
total means a tie; a nonzero total means one winner, whichever action wins.
Testing for a winner therefore differs from testing whether a particular
action's score exceeds a threshold. The last row records that distinction.

The superscript in $\mathsf{NP}^{\mathcal C}$ allows a polynomial-time
nondeterministic computation to call a solver for a problem in $\mathcal C$;
later calls may depend on earlier replies. The nested label below combines
the choice of retained inputs with a test for an observation whose weighted
scores tie. It records the work of these nested tests, not just the number
of actions.

\begin{table}[!htb]
\centering
\begingroup
\footnotesize
\setlength{\tabcolsep}{4pt}
\renewcommand{\arraystretch}{1.1}
\renewcommand{\tabularxcolumn}[1]{m{#1}}
\begin{tabularx}{\linewidth}{@{}>{\raggedright\arraybackslash}m{0.27\linewidth}>{\raggedright\arraybackslash}m{0.19\linewidth}>{\raggedright\arraybackslash}m{0.27\linewidth}>{\raggedright\arraybackslash}X@{}}
\toprule
\textbf{One-step requirement} & \textbf{Proposed inputs} & \textbf{Within budget} & \textbf{Some observation} \\
\midrule
\multicolumn{4}{@{}l}{\textbf{Baselines}}\\[-1pt]
Input sufficiency, product coordinates & coNP-complete & coNP-complete & $\Sigma_2^P$-complete \\
\midrule
Preservation on supported states, full support & coNP-complete & coNP-complete & $\Sigma_2^P$-complete \\
\midrule
\multicolumn{4}{@{}l}{\textbf{Three central boundaries}}\\[-1pt]
Input sufficiency, circuit-defined restricted states & coNP-complete & \textbf{$\boldsymbol{\Sigma_2^P}$-complete} & $\Sigma_2^P$-complete \\
\midrule
Preservation on supported states, circuit-defined support & coNP-complete & \textbf{$\boldsymbol{\Sigma_2^P}$-complete} & $\Sigma_2^P$-complete \\
\midrule
Input sufficiency with observation coverage & \textbf{$\boldsymbol{\Pi_2^P}$-complete} & \textbf{$\boldsymbol{\Pi_2^P}$-complete} & $\Sigma_2^P$-complete \\
\midrule
Decisiveness under conditional scores, two actions & $\operatorname{co}(\mathsf{NP}^{\mathsf{PP}})$-complete & \textbf{$\boldsymbol{\mathsf{NP}^{(\mathsf{NP}^{\mathsf{PP}})}}$-complete} & $\mathsf{coC}_{=}\mathsf P$-complete \\
\bottomrule
\end{tabularx}
\caption{Static and one-step stochastic classifications for the displayed
succinct encodings.  Product closure supplies one least sufficient interface;
missing support can restore incomparable blanket choices; observation coverage
adds a support witness; decisiveness retains weighted score comparison.  For a
fixed retained input set, a some-observation query asks whether at
least one realized observation satisfies the row predicate.  Semantic validity promises use
Definition~\ref{def:validity-promise-convention}.}
\label{tab:regime-matrix}
\endgroup
\end{table}

\begin{table}[!htb]
\centering
\begingroup
\footnotesize
\setlength{\tabcolsep}{5pt}
\renewcommand{\arraystretch}{1.1}
\renewcommand{\tabularxcolumn}[1]{m{#1}}
\begin{tabularx}{\linewidth}{@{}>{\raggedright\arraybackslash}m{0.31\linewidth}>{\raggedright\arraybackslash}m{0.2\linewidth}>{\raggedright\arraybackslash}m{0.24\linewidth}>{\raggedright\arraybackslash}X@{}}
\toprule
\textbf{Sequential model} & \textbf{Proposed inputs} & \textbf{Within budget} & \textbf{Some observation} \\
\midrule
Unbounded deterministic arena blocks & PSPACE-complete & PSPACE-complete & PSPACE-complete \\
\midrule
At most $d$ arena blocks, promised & \textbf{$\boldsymbol{\Pi_{d+1}^P}$-complete} & \textbf{$\boldsymbol{\Pi_{d+1}^P}$-complete} & \textbf{$\boldsymbol{\Sigma_{d+2}^P}$-complete} \\
\midrule
Multistep stochastic, unary horizon & \textbf{PSPACE-complete} & \textbf{PSPACE-complete} & --- \\
\bottomrule
\end{tabularx}
\caption{Sequential extension. A fixed bound on controller--environment move
blocks determines the polynomial-hierarchy level. When that bound grows with
the input, or a unary-horizon POMDP updates beliefs after observations, the
corresponding certification problems are PSPACE-complete. Here $d$ is
fixed, the block bound is promised, and $\Sigma_0^P=\Pi_0^P=\mathsf P$.  The
multistep row covers root-only and historywise preservation, observation coverage,
and decisiveness; no some-observation query is posed.}
\label{tab:sequential-matrix}
\endgroup
\end{table}

Budgets count fields in the supplied presentation.  The static restricted-state
and sparse-support lower bounds already use Boolean fields and ordinary bit
budgets.  The exact two-action minimum construction needs domains of size at
most four. The three-action lower bound lets each selected field carry a
polynomial-length token. Its two-stage POMDP version uses a separate
construction in which each selectable input is one bit. The unit
of selection stays fixed within each theorem.

Oracle superscripts above denote literal adaptive access.  With two actions,
one fiber tests exact nonvanishing of a signed score gap. The construction
assigns a separate reading to each count below the required threshold.
Every reading must avoid a tie, ruling out all those smaller counts.
Proposed-input and within-budget decisiveness therefore retain the displayed
oracle classes. Only the
some-observation query falls to $\mathsf{coC}_{=}\mathsf P$.

The three principal proofs are Theorem~\ref{thm:listed-minsuff-sigma2p}
for restricted-state within-budget sufficiency,
Theorem~\ref{thm:stochastic-preservation-tight} for observation coverage, and
Theorem~\ref{prop:two-action-exact-counting} for within-budget decisiveness.
\leanmeta{\LHrng{DC}{126}{136}, \LHrng{DC}{138}{149}, \LHrng{DQ}{18}{19}, \LH{DQ91}, \LH{OU1}, \LHrng{OU}{9}{11}, \LH{OU38}}

\section{Verification Scope}
\label{sec:verification-scope}

Lean checks the cited finite identities, gadget correspondences, and
certificate equivalences. Exact programs audit \ExecutableAuditInventory.
The complexity transfers, encoded-size arguments, and imported asymptotic
results are manuscript proofs. The supplement maps every cited finite
statement as Lean-complete or conditional on a named premise; those statuses
do not mechanize the surrounding complexity theorem.

The accompanying archive pins the environments. From its
\texttt{paper4\_jair/} directory, replay with
\texttt{python3.12 experiments/verify\_jair\_artifact.py -{}-bootstrap}.
\nameref{sec:data-code-availability} gives its locator.

\section{Open Problems}\label{sec:open-problems}

Five questions remain about structured models and broader decision settings.
First, does a simple probability model make the audit easier? The two-action
decisiveness reductions already use independent uniform inputs for all three
queries, so their distributions have moral treewidth zero. Whether the
$\Pi_2^P$ observation-coverage hardness survives a polytree or fixed-treewidth
restriction on the supplied distribution remains open.

Second, having few actions or few distinct optimal-action sets does not
by itself make deletion easy: the coNP lower bound already uses two of each.
These parameters alone therefore cannot give fixed-parameter tractability
unless P = NP. Tractable
parameterizations need additional restrictions on how the model is represented.
Circuit depth together with distribution width, or compiled representation
class together with the retained-coordinate budget, are natural candidates.

Third, average-case or smoothed analysis needs a model for generating or
perturbing the supplied tables or circuits.  It could then ask when greedy
usually finds near-smallest sufficient sets, or when constraint systems
for conflicting state pairs usually have small width.

Fourth, multistep input-removal audits are PSPACE-complete for the
unary-horizon model (Theorem~\ref{thm:pomdp-exact-pspace}).
Bounded-memory, bounded-width,
binary-horizon, infinite-horizon, and adaptive sensor-selection variants
remain separate questions because they change the policy class, the history
encoding, or the object being selected.

Fifth, how much model is needed for a specified task and checking guarantee?
Fix the allowed cases, an accuracy or regret target, the representation and
size measure, and the resources available for extracting and checking
structure. Models meeting that task target may still produce different
responses; audit each frozen response rather than treating them as
interchangeable. The current results separate evaluation from preservation
auditing. They give no parameter-count cutoff.

\section{Conclusion}

Build the required behavior from the inputs you intend to keep, and preserve
the reasoning that proves their sufficiency. This fixes neither the
smallest interface nor the correctness of the required behavior itself. It
makes explicit the dependence that an input-removal audit would otherwise
have to recover.

An input is unnecessary relative to a response and a declared state domain.
Missing combinations can permit alternative sufficient interfaces; all
combinations give one least sufficient set. Under the stated circuit encodings,
choosing an interface within a bit budget on restricted Boolean states is
$\Sigma_2^P$-complete, and checking one deletion on Boolean products is
coNP-complete even with two actions and at most two optimal-action sets.
Efficient evaluation does not supply an efficient general audit.

Structure changes what can be checked. The diagnostic-tree construction uses
derived summaries so local checks settle a whole-domain claim in linear time
at fixed branching. The classifier's leaf boxes support a symbolic product
audit: eleven blocks suffice on the held-out rows, while fifteen are necessary
across every allowed combination. With the supplied structure, these audits
run in polynomial time; fast evaluation alone does not make a general audit
tractable.

Uncertainty and future decisions add different obligations. Coverage needs a
supported state behind every declared observation; decisiveness compares
weighted scores; sequential preservation accounts for information used later.
For preservation, approximation permits keeping more inputs than the optimum
while still recovering every required response. The decisiveness
approximation instead requires a unique conditional winner, not recovery of the
original response. The extended classifications in
Table~\ref{tab:regime-matrix} state the cost of those particular questions.

\paragraph{What retained reasoning can contribute.}
Intermediate work can expose calculations, assumptions and proposed
dependencies that an answer alone does not show. Some steps may be checked or
corrected even when the whole trace is not a certificate. Training can also
use this material: Hsieh et al. improve task scores for the same 220M T5 model
by adding rationale generation to label prediction~\cite{hsieh2023distilling}.
DeepSeek-R1 reports distillation into Qwen and Llama students using about
800k mixed examples, including about 600k reasoning
examples~\cite{deepseekai2025r1}. The latter comparison does not isolate the
contribution of traces from the other training data.

These results do not establish that a natural-language trace faithfully
reports a model's internal computation, or that a learned response is correct
on every allowed case. They show why intermediate work should not be dismissed
merely because it is incomplete. Our question is whether the retained inputs
determine a fixed required response throughout the declared domain. A learned
component may operate inside a construction that restricts which inputs it
can use; training selects its behavior within that restriction. Evidence that
it meets the task's requirements is still needed.

For an existing system, try applicable structural checks first; retaining an
input can remain the cheaper correct choice under the stated cost comparison.
During design, preserve the constructions and evidence that make the dependence
checkable. Their value includes what they let a later auditor verify, even
when producing an answer no longer needs them.

\begin{acks}
The author thanks Tobias Fritz for pointing out the coimage formulation of the
optimizer quotient in \textbf{Set}.
\end{acks}

\section*{Declaration of AI Use and Author Contribution}

\emph{Research direction and responsibility.}
The author originated and determined the research questions, semantic
contracts, classification plan, theorem acceptance criteria, novelty claims,
and final exposition, and is responsible for every claim, citation, and
interpretation.

\emph{Candidate generation and drafting.}
Generative AI systems proposed candidate definitions, proof and reduction
drafts, examples, code, and prose.  These candidates were not treated as
evidence: the author checked and edited every retained output and discarded
any candidate that failed the predeclared criteria.

\emph{Independent checks.}
Section~\ref{sec:verification-scope} distinguishes the Lean-checked finite
cores, executable audits, and manuscript proofs; the supplement records the
boundary claim by claim.

\section*{Data and Code Availability}
\phantomsection
\label{sec:data-code-availability}

The accompanying supplementary package contains the proofs, Lean sources,
exact audit programs, raw audit records, licenses, and a file-hash manifest.
Its \texttt{publication\_manifest.json} identifies this package's sources;
the submission manifest identifies the PDFs and archives supplied for review.

The supplementary package for this manuscript is available at
\url{\PaperZenodoURL}~\cite{paperArtifact}. It includes this version's claim
ledger and replay instructions.

\printbibliography

\clearpage
\section{Reproducibility Checklist for JAIR}
\label{app:reproducibility-checklist}

The answers below complete the checklist supplied in the current JAIR Author
Kit.  One answer is selected for every applicable item.

\begingroup
\footnotesize
\setlength{\parskip}{0pt}
\setlength{\parsep}{0pt}
\setlength{\itemsep}{0pt}
\setlength{\topsep}{0.35\baselineskip}

\subsection*{All Articles:}
\begin{enumerate}
    \setlength{\itemsep}{0pt}
    \setlength{\parsep}{0pt}
    \item All claims investigated in this work are clearly stated.
    \textbf{Answer: yes.}
    \item Clear explanations are given how the work reported substantiates the claims.
    \textbf{Answer: yes.}
    \item Limitations or technical assumptions are stated clearly and explicitly.
    \textbf{Answer: yes.}
    \item Conceptual outlines and/or pseudo-code descriptions of the AI methods introduced in this work are provided, and important implementation details are discussed.
    \textbf{Answer: yes.} Section~\ref{sec:tractable} gives the structural
    checks and their audit procedure; the executable artifact
    implements the exact finite audits used in the manuscript.
    \item
    Motivation is provided for all design choices, including algorithms, implementation choices, parameters, data sets and experimental protocols beyond metrics.
    \textbf{Answer: yes.}
\end{enumerate}

\subsection*{Articles Containing Theoretical Contributions:}
Does this paper make theoretical contributions?
\textbf{Answer: yes.}

\begin{enumerate}
    \setlength{\itemsep}{0pt}
    \setlength{\parsep}{0pt}
    \item All assumptions and restrictions are stated clearly and formally.
    \textbf{Answer: yes.}
    \item All novel claims are stated formally (e.g., in theorem statements).
    \textbf{Answer: yes.}
    \item Proofs of all non-trivial claims are provided in sufficient detail to permit verification by readers with a reasonable degree of expertise (e.g., that expected from a PhD candidate in the same area of AI).
    \textbf{Answer: yes.}
    \item
    Complex formalism, such as definitions or proofs, is motivated and explained clearly.
    \textbf{Answer: yes.}
    \item
    The use of mathematical notation and formalism serves the purpose of enhancing clarity and precision; gratuitous use of mathematical formalism (i.e., use that does not enhance clarity or precision) is avoided.
    \textbf{Answer: yes.}
    \item
    Appropriate citations are given for all non-trivial theoretical tools and techniques.
    \textbf{Answer: yes.}
\end{enumerate}

\subsection*{Articles Reporting on Computational Experiments:}
Does this paper include computational experiments?
\textbf{Answer: yes.}

\begin{enumerate}
    \setlength{\itemsep}{0pt}
    \setlength{\parsep}{0pt}
    \item All source code required for conducting experiments is included in an online appendix or will be made publicly available upon publication of the paper. The online appendix follows best practices for source code readability and documentation as well as for long-term accessibility.
    \textbf{Answer: yes.} \nameref{sec:data-code-availability} gives the
    archive locator.
    \item The source code comes with a license that allows free usage for reproducibility purposes.
    \textbf{Answer: yes.}  The artifact's source code is released under the MIT License.
    \item The source code comes with a license that allows free usage for research purposes in general.
    \textbf{Answer: yes.}  The artifact's source code is released under the MIT License.
    \item Raw, unaggregated data from all experiments is included in an online appendix or will be made publicly available upon publication of the paper. The online appendix follows best practices for long-term accessibility.
    \textbf{Answer: yes.}  The artifact includes one CSV row for every
    exhaustively enumerated antichain family and the reduction-audit
    summaries.  It also includes the held-out transformed rows, predictions,
    targets, and split indices for the learned-policy audit.
    \nameref{sec:data-code-availability} gives the archive locator.
    \item The unaggregated data comes with a license that allows free usage for reproducibility purposes.
    \textbf{Answer: yes.}  Generated records are released under CC BY 4.0; third-party inputs retain their cited licenses.
    \item The unaggregated data comes with a license that allows free usage for research purposes in general.
    \textbf{Answer: yes.}  Generated records are released under CC BY 4.0; third-party inputs retain their cited licenses.
    \item If an algorithm depends on randomness, then the method used for generating random numbers and for setting seeds is described in a way sufficient to allow replication of results.
    \textbf{Answer: yes.}  The reduction and census audits are exhaustive.  The
    learned-policy audit uses one stratified split and one decision tree, both
    with seed 1729; no result was selected across seeds.
    \item The execution environment for experiments, i.e., the computing infrastructure (hardware and software) used for running them, is described, including GPU/CPU makes and models; amount of memory (cache and RAM); make and version of operating system; names and versions of relevant software libraries and frameworks.
    \textbf{Answer: yes.}  The learned-policy result JSON records the operating
    system, Python and library versions, CPU model and count, and memory.
    The multistep reduction JSON also records elapsed time.  No GPU is used.
    \item The evaluation metrics used in experiments are clearly explained and their choice is explicitly motivated.
    \textbf{Answer: yes.}  The approximation census compares greedy solution
    cardinality with the exact optimum on the same discernibility family.  The
    reduction audits check the source/target verdict, selector decoding, fiber
    comparisons, support normalization, and dominance inequalities using exact
    arithmetic.
    The learned-policy audit compares the exact smallest input set on its listed
    state domain with the exact semantic relevance set on the declared product; its
    accuracy is reported only to describe the frozen classifier.
    \item The number of algorithm runs used to compute each result is reported.
    \textbf{Answer: yes.} Supplement~S4 reports every audit's execution
    counts and exhaustive boundary for \ExecutableAuditInventory.
    The block-size-two selector audit checks $65{,}536$ semantic outcome
    matrices, not literally all $2^{64}$ predicates; its factorization covers
    every gadget-relevant equivalence class. The learned-policy audit uses
    one predeclared split and model configuration, exhausts all $2^{16}$
    listed-domain interfaces, and symbolically covers the $4^{16}$ product.
    \item Reported results have not been ``cherry-picked'' by silently ignoring unsuccessful or unsatisfactory experiments.
    \textbf{Answer: yes.}  The protocol retains every family in the declared
    antichain census and every case in each
    declared reduction-audit factorization.
    \item Analysis of results goes beyond single-dimensional summaries of performance (e.g., average, median) to include measures of variation, confidence, or other distributional information.
    \textbf{Answer: not applicable.} The reduction audits enumerate their
    finite test families exhaustively; the classifier audit also covers its
    full product symbolically. Supplement~S4 reports the first
    failing dimension, exact maximum ratio, and complete cardinality counts.
    \item All (hyper-) parameter settings for the algorithms/methods used in experiments have been reported, along with the rationale or method for determining them.
    \textbf{Answer: yes.}  The boundary $p\le5$, coordinate-index
    tie-breaker, antichain normalization, and exact cardinality-layer search
    order are stated in Supplement~S4, the code, and the artifact README.  The
    learned-policy split, pooling, quantization, tree depth, leaf-size floor,
    criterion, and seed are likewise recorded, and its Python environment is
    pinned in a lock file.
    \item The number and range of (hyper-) parameter settings explored prior to conducting final experiments have been indicated, along with the effort spent on (hyper-) parameter optimisation.
    \textbf{Answer: yes.}  There is no parameter optimization in the
    theorem-level audits.  The learned policy uses one didactic configuration
    fixed before the exact state-domain comparison; no alternatives were searched
    or filtered by the reported interface gap.
    \item Appropriately chosen statistical hypothesis tests are used to establish statistical significance in the presence of noise effects.
    \textbf{Answer: not applicable.}  The reduction and census calculations
    are exhaustive finite audits.  The learned-policy calculation is one exact
    model/state-domain comparison and makes no population claim.
\end{enumerate}

\subsection*{Articles Using Data Sets:}
Does this work rely on one or more data sets (possibly obtained from a benchmark generator or similar software artifact)?
\textbf{Answer: yes.}

\begin{enumerate}
    \setlength{\itemsep}{0pt}
    \setlength{\parsep}{0pt}
    \item All newly introduced data sets are included in an online appendix or will be made publicly available upon publication of the paper. The online appendix follows best practices for long-term accessibility with a license that allows free usage for research purposes.
    \textbf{Answer: yes.}  All $7{,}773$ census records are included and
    released under CC BY 4.0.
    \nameref{sec:data-code-availability} gives the archive locator.
    \item The newly introduced data set comes with a license that allows free usage for reproducibility purposes.
    \textbf{Answer: yes.}  The census records are released under CC BY 4.0.
    \item The newly introduced data set comes with a license that allows free usage for research purposes in general.
    \textbf{Answer: yes.}  The census records are released under CC BY 4.0.
    \item All data sets drawn from the literature or other public sources (potentially including authors' own previously published work) are accompanied by appropriate citations.
    \textbf{Answer: yes.}  The learned-policy audit cites the UCI Optical
    Recognition of Handwritten Digits data and the scikit-learn distribution
    used to load it~\cite{alpaydin1998digits,pedregosa2011scikit}.
    \item All data sets drawn from the existing literature (potentially including authors' own previously published work) are publicly available.
    \textbf{Answer: yes.}  The digits data are publicly available from the UCI
    Machine Learning Repository under CC BY 4.0.
    \item All new data sets and data sets that are not publicly available are described in detail, including relevant statistics, the data collection process and annotation process if relevant.
    \textbf{Answer: yes.}  The antichain universe and exhaustive boundary are
    stated in Supplement~S4; there is no
    human annotation.
    \item All methods used for preprocessing, augmenting, batching or splitting data sets (e.g., in the context of hold-out or cross-validation) are described in detail.
    \textbf{Answer: yes.}  The antichain census has no train/test split,
    augmentation, sampling, or filtering.  For the learned policy, the paper
    and artifact specify $2\times2$ sum pooling, four-bin quantization, the
    stratified $70/30$ split, seed, and classifier configuration.  The JSON
    retains the held-out indices, transformed rows, targets, and predictions
    and hashes the split and fitted tree.
\end{enumerate}

\endgroup

\end{document}


\title{Supplementary Material: \PaperTitleAuto}

\PaperPublicationAuthor{article}

\date{\today}
\maketitle

\begin{abstract}
This supplement contains the complete deferred proofs and construction
expansions, the absolute-tolerance extension, encoding details, the
verification protocol, and the claim-to-Lean finite-core map. The main article
states the headline results and explains their mechanisms; the derivations
below retain their exact assumptions and finite-core evidence.
\end{abstract}

\tableofcontents
\clearpage
\appendix
\setcounter{section}{0}
\renewcommand{\thesection}{S\arabic{section}}
\renewcommand{\theHsection}{supp.S\arabic{section}}

\section{Auxiliary Construction Proofs}

This section gives the token-coordinate core and exact two-action selector completion. The following supplementary sections retain the full headline-classification proofs. Supplement~S4 records the Lean and executable audits.

\subsection{Within-Budget Sufficiency on Circuit-Defined Restricted States}
\label{app:proof-listed-minimum}

\begin{proof}[Token-coordinate core for Theorem~\ref{thm:listed-minsuff-sigma2p}]
Membership follows the displayed quantifier pattern.  Guess a coordinate set
\(I\) with \(|I|\le k\).  Failure of sufficiency has a polynomial witness:
two encoded states that agree on every coordinate in \(I\) but have different
optimizer sets.  A nondeterministic machine can guess and verify this pair,
so the universal certification after the guess lies in coNP and the minimum
query lies in \(\Sigma_2^P\).

For hardness, reduce from an instance
\[
 \exists z\in\{0,1\}^{n_z}\;\forall x\in\{0,1\}^{n_x}
 \;\varphi(z,x)
\]
of \(\exists\forall\)-SAT, padding so that \(n_z\ge1\).  For every
existential bit \(z_j\), introduce two coordinates
\(q_{j,0}\) and \(q_{j,1}\), and set the budget to \(n_z\).
The restricted code contains main states \(m(z',x,r)\), where
\(z'\in\{0,1\}^{n_z}\), \(x\in\{0,1\}^{n_x}\), and
\(r\in\{0,1\}\), together with forcing states \(f(j,b)\) for every
\(j\) and \(b\in\{0,1\}\).

Coordinate \(q_{j,b}\) returns four disjoint token types.  On a main state,
it returns a common token carrying only \(x\) when \(z'_j=b\), and otherwise
returns an exposed token containing the complete state identifier.  On a
forcing state \(f(\ell,c)\), it returns a side token carrying
\((\ell,c)\) when \(j=\ell\), and otherwise a neutral token carrying only
\(\ell\).  These tags prevent collisions between main and forcing states.

Use two actions, denoted \(0\) and \(1\).  At a main state put
\[
 U(0,m(z',x,r))=1,
 \qquad
 U(1,m(z',x,r))=
 \begin{cases}
 0,&\varphi(z',x)\text{ is true},\\
 2r,&\varphi(z',x)\text{ is false}.
 \end{cases}
\]
Thus a satisfying main pair has optimizer \(\{0\}\) at both values of
\(r\); an unsatisfying pair has optimizer \(\{0\}\) at \(r=0\) and
\(\{1\}\) at \(r=1\).  At \(f(j,b)\), action \(b\) is uniquely optimal.

Fix an assignment \(z\) and select
\(I_z=\{q_{j,z_j}:1\le j\le n_z\}\).  Two main states with first component
\(z\) agree on \(I_z\) exactly when their \(x\)-components agree; their
auxiliary bits are hidden.  A main state with \(z'\ne z\) is exposed by any
coordinate on which \(z'_j\ne z_j\), and hence lies in a singleton fiber.
Every forcing state is likewise isolated by the selected coordinate from its
own pair, and the outer token tags separate forcing states from main states.
It follows that
\[
 I_z\text{ is sufficient}
 \quad\Longleftrightarrow\quad
 \forall x\;\varphi(z,x).
\]

Conversely, let \(I\) be sufficient with \(|I|\le n_z\).  If \(I\) misses
both coordinates of pair \(j\), then \(f(j,0)\) and \(f(j,1)\) agree on
every selected coordinate: coordinates from other pairs return the same
neutral token.  Their unique optimizers differ, contradicting sufficiency.
Hence \(I\) hits all \(n_z\) pairs.  The budget forces it to contain exactly
one coordinate from each pair, so \(I=I_z\) for the decoded assignment
\(z\).  The preceding fiber analysis then yields
\(\forall x\,\varphi(z,x)\).

State identifiers, coordinate tokens, and utilities are computable by
circuits of size polynomial in the source formula.  The exponentially many
allowed codewords are therefore represented by the declared succinct
finite-coordinate presentation.  This proves both directions of the
many-one reduction.
The main article then replaces every token coordinate by an individually
selectable Boolean block, forces each public bit with a disjoint calibration
edge, and reuses the same optimizer argument under an ordinary bit budget.
\end{proof}

\subsection{Two-action minimum selector}
\label{app:proof-two-action-minimum}

\begin{proof}[Expanded proof of Theorem~\ref{prop:two-action-exact-counting}, item 3]
Write the source predicate as
$\exists z\,\forall x\;g(w,z,x)\ne0$, using the signed-path normal form in
Equation~\ref{eq:gapp-signed-path-normal-form}.  Pad the existential and
universal blocks to one length $N\ge1$.  For the headline lower bound,
Lemma~\ref{lem:majority-as-universal-nonzero} takes the universal block to be
$(x,\ell)$ and sets $g(w,z,x,\ell)=\sum_{b,v}h(z,x,\ell,b,v)$, so
nonvanishing for every $\ell$ is exactly weak majority over $v$.  For each
$j\in[N]$ introduce two
candidate coordinates $q_{j,0},q_{j,1}$, each with domain $\{0,1\}^2$; write
its value as $(r_{j,b},\xi_{j,b})$.  Add hidden Boolean coordinates
$z'_1,\ldots,z'_N$ and the path block $p$.  The supplied distribution is
uniform on this full product, and the budget is $N$.  Put
\[
\begin{aligned}
 x(z',q)_j&=\xi_{j,z'_j},\\
 \chi(z',q)&=\prod_{j=1}^{N}(-1)^{r_{j,z'_j}},\\
 d(q,z',p)&=\chi(z',q)h\bigl(w,z',x(z',q),p\bigr),
\end{aligned}
\]
and use utilities $U(a)=1+d$ and $U(b)=1-d$.

Any coordinate set of size at most $N$ that misses both candidates of some
pair leaves the corresponding sign bit hidden.  Summing over that bit cancels
the fiber gap exactly, even if $z'_j$ or path bits were retained.  Such an
interface is therefore not decisive.  Every remaining budget-feasible set
chooses exactly one coordinate $q_{j,z_j}$ from each pair and hence encodes
one string $z$.

Fix such a selector and one of its fibers, with observed values $(r_j,x_j)$.
For every $z'\ne z$, the first mismatching pair contributes an unobserved sign
whose two values cancel.  Only $z'=z$ survives, and direct summation gives
\begin{equation}
\label{eq:two-action-selector-gap}
 G_I(r,x)=2^{1-3N-m}\left(\prod_{j=1}^{N}(-1)^{r_j}\right)g(w,z,x).
\end{equation}
The prefactor and sign never vanish.  The selected interface is therefore
decisive exactly when $\forall x\,g(w,z,x)\ne0$, and a budget-feasible decisive
interface exists exactly when the source predicate holds.  The domains,
uniform dyadic mass, and utility circuit all have polynomial-size
descriptions.

Lean checks the singleton--nonzero equivalence and the two one-pair parity
identities iterated in this construction.  The executable audit exhausts every
one-bit signed payload table and checks the two-bit factorization with exact
integer fiber sums.  A composed audit also checks all $65{,}536$ predicates
with one existential bit, one universal bit, and two majority bits against all
$46$ interfaces of size at most two.  Formula-independent cancellation
discharges the $42$ nonselectors; the four canonical selectors are compared
with the source majority verdict predicate by predicate.
\leanmeta{\LH{DC154}, \LHrng{DC}{158}{159}}
\end{proof}

\section{Further Results on Preserving Responses}
\label{app:auxiliary-ai-corollaries}

The main article gives the coordinate identification and the resulting
presentation comparison.  This section records the exact selector and rough-set
corollaries used to obtain those entries, then proves the distinction between
replacing a setting and removing an input channel.

\leanmetapending{\LHrng{KP}{11}{13}}
\begin{remark}[A Declared Selector Has Its Own Canonical
Quotient]\label{cor:declared-selector-quotient}
Fix an optimal selector $\sigma:S\to A$ as part of the recovery target.  Then
$\sigma$ factors through $\pi$ if and only if
$\ker(\pi)\subseteq\ker(\sigma)$, and $S/\ker(\sigma)$ is the unique coarsest
exact abstraction for that declared selector.  Any representation recovering
$\sigma$ supports an optimal selector.
\end{remark}

\begin{proof}
Repeat the kernel argument of Theorem~\ref{thm:master-factorization} with
$\sigma$ in place of $\Opt$.  If
$\sigma=g\circ\pi$, then $g(\pi(s))=\sigma(s)\in\Opt(s)$.
\end{proof}

\leanmetapending{\LH{UB13}}
\begin{lemma}[Budgeted Sufficiency and Reduct Existence]
\label{lem:budgeted-reduct-sufficiency}
In an optimizer-labelled finite decision table, a sufficient coordinate set
of cardinality at most $k$ exists if and only if a reduct of cardinality at
most $k$ exists.
\end{lemma}

\begin{proof}
Every reduct is sufficient.  Conversely, let $I$ be sufficient with
$|I|\le k$ and choose an inclusion-minimal sufficient subset
$I'\subseteq I$.  Proposition~\ref{prop:coordinate-projection-calibrations}(1)
identifies $I'$ with a reduct, and $|I'|\le |I|\le k$.
\end{proof}

\leanmetapending{\LH{UB13}, \LHrng{LMS}{1}{2}}
\begin{corollary}[Within-Budget Reduct Complexity on Restricted States]
\label{cor:succinct-minimum-reduct}
Under the circuit-defined restricted-state encoding, deciding whether
an optimizer-labelled decision table has a reduct of cardinality at most $k$
is $\Sigma_2^P$-complete.
\end{corollary}

\begin{proof}
Lemma~\ref{lem:budgeted-reduct-sufficiency} identifies the query with
\textsc{Restricted-Within-Budget-Sufficiency}, which is $\Sigma_2^P$-complete by
Theorem~\ref{thm:listed-minsuff-sigma2p}.
\end{proof}

\leanmetapending{\LH{UB3}, \LH{DQ36}, \LHrng{SK}{4}{8}}
\begin{corollary}[Within-Budget Reduct Complexity on Products]
\label{cor:product-minimum-reduct}
On product-closed decision tables, the reduct is unique and equals the
relevant-attribute set.  The within-budget reduct decision problem is
coNP-complete under the succinct encoding.
\end{corollary}

\begin{proof}
Proposition~\ref{prop:minimal-relevant-equiv} gives uniqueness and identifies
the reduct with the relevant-coordinate set.  Theorem~\ref{thm:minsuff-conp}
gives the coNP classification.
\end{proof}

\subsection{Replacing a Setting and Removing an Input}

The count example in the main article distinguishes a particular replacement
from removing an input channel.  Both have simple exact criteria, for any
required response $F$.

\begin{figure}[htbp]
\centering
\includegraphics[width=.96\linewidth]{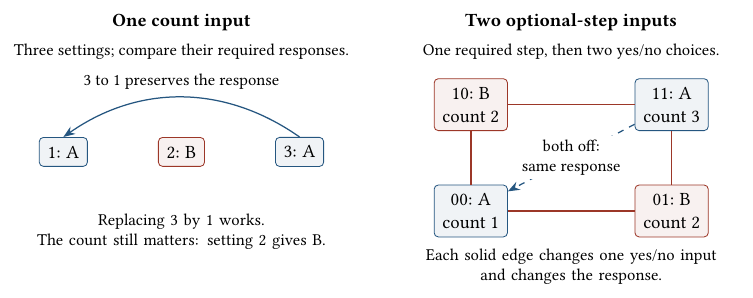}
\caption{A valid replacement need not make an input irrelevant.  Counts one
and three agree, but two differs.  The right panel represents the same counts
by one mandatory step and two optional steps: removing both optional steps
preserves the response at $11$, while removing either alone changes it.
A channel budget counts different inputs in the two encodings.}
\ifdefined\Description
\Description{One count input has settings 1 and 3 labelled A and setting 2
labelled B. An arrow replaces 3 by 1. In a square of two optional Boolean
steps, 00 and 11 have response A, while 10 and 01 have response B. Every
one-bit edge changes the response, although the diagonal from 11 to 00 does
not.}
\fi
\label{fig:input-settings}
\end{figure}

\leanmetapending{\LH{IS1}}
\begin{proposition}[Replacing One Setting]
\label{prop:setting-replacement}
Let $F:C\times V\to O$, and fix $v,w\in V$.  Define $T_{v\to w}$ to
replace the setting $v$ by $w$ and leave every other setting unchanged,
without changing the context $c$.  Then
\[
 F=F\circ T_{v\to w}
 \quad\Longleftrightarrow\quad
 \forall c\in C:\ F(c,v)=F(c,w).
\]
No equality involving another setting is required.
\end{proposition}

\begin{proof}
Apply preservation to $(c,v)$ for necessity.  For sufficiency, a state with
setting $v$ is covered by the stated equality; every other state is unchanged.
\end{proof}

\leanmetapending{\LHrng{IS}{2}{4}}
\begin{proposition}[Fixing Defaults and Removing Channels]
\label{prop:default-reset-sufficiency}
Let $F:\prod_{i=1}^n X_i\to O$, fix a baseline tuple $b$, and let
$R_{I,b}(s)$ retain $s_i$ for $i\in I$ and use $b_i$ otherwise.  Then
\[
 \forall s:\ F(s)=F(R_{I,b}(s))
 \quad\Longleftrightarrow\quad
 \forall s,t:\ s_I=t_I\Longrightarrow F(s)=F(t).
\]
If preservation holds for one baseline, it holds for every baseline.
Taking $F=\Opt$ identifies this condition with coordinate sufficiency and
with recovery of $\Opt$ from the retained projection.
\end{proposition}

\begin{proof}
If $s_I=t_I$, their resets are identical.  Preservation under reset gives
$F(s)=F(R_{I,b}(s))=F(R_{I,b}(t))=F(t)$.  Conversely, $s$ and its reset
agree on $I$, so fiber constancy gives reset preservation.  The right-hand
condition does not depend on $b$.  The optimizer specialization is
Definition~\ref{def:sufficient} and Theorem~\ref{thm:master-factorization}.
\end{proof}

The product assumption ensures that resets stay in the declared state domain.
The same proof works on a domain closed under the particular reset; an
arbitrary listed domain need not have that property.

For the figure's finite separation, set $F(1)=F(3)=\mathsf A$ and
$F(2)=\mathsf B\ne\mathsf A$.  Replacing three by one preserves $F$;
discarding the count does not.  With one mandatory step and two optional
Boolean steps, the count map sends $00,10,01,11$ to $1,2,2,3$.
It is surjective, so a count-dependent response is constant on this Boolean
square exactly when it is constant on the three count settings.  Nevertheless
the coordinate presentations differ: one has one three-valued channel, the
other has two Boolean channels.  Each Boolean channel can change the displayed
response while the other is fixed.  Agreement between $11$ and $00$ does not
make either channel irrelevant.\leanmeta{\LHrng{IS}{5}{8}}

\appendixtolerancetrue
\ifappendixtolerance\else
\section{What Input Removal Must Preserve}\label{sec:formal-setup}

The hazard controller gives the basic test: states that look the same after
deletion must retain the same complete set of best actions, including ties.
A sensor records one input coordinate; a state records all their values;
utility gives the score of each action. The definitions below make this test
precise for finite systems.

\subsection{Decision Problems with Coordinate Structure}

Start with all combinations of input values. A second presentation
will allow some combinations to be absent.

\begin{definition}[Decision Problem]\label{def:decision-problem}
A \emph{decision problem with product-coordinate structure} is a tuple
$\mathcal D=(A,X_1,\ldots,X_n,U)$ with finite nonempty action set $A$, finite
nonempty coordinate spaces $X_i$, state space
$S=X_1\times\cdots\times X_n$, and utility $U:A\times S\to\mathbb Q$.
\end{definition}

Rational or dyadic utilities permit exact encoded comparisons.

\begin{definition}[Finite Coordinate Presentation]\label{def:coordinate-presentation}
A \emph{finite coordinate presentation} consists of a finite nonempty state set $S$,
finite coordinate spaces $X_1,\ldots,X_n$, coordinate maps
$\pi_i:S\to X_i$, a finite nonempty action set $A$, and a utility function
$U:A\times S\to\mathbb{Q}$. A product-coordinate decision problem is the
special case $S=X_1\times\cdots\times X_n$ with the canonical projections.
\end{definition}

Only products guarantee that replacing one coordinate yields another state.
Checking whether agreeing states require the same response also works on
listed rows with missing combinations, as does the set-cover reduction in
Section~\ref{sec:approximation-across-presentations}.

\begin{definition}[Projection]\label{def:projection}
For $s\in S$ and $I\subseteq\{1,\ldots,n\}$, the projection of $s$ onto $I$
is $s_I=(\pi_i(s))_{i\in I}$.
For $I=\varnothing$, this is the unique empty tuple, so all states belong to
one projection fiber.
\end{definition}

\begin{definition}[Optimizer Map]\label{def:optimizer}
For $s\in S$, the \emph{optimal-action set} is
\[
\Opt(s)=\arg\max_{a\in A}U(a,s).
\]
\end{definition}

\subsection{Recovering Responses from Retained Inputs}

For a frozen classifier, the required response is its prediction. For a
controller, it may be a set of tied best actions. A finite function $F:S\to O$
is included by taking $A=O$ and
$U(a,s)=\mathbf1[a=F(s)]$: then $\Opt(s)=\{F(s)\}$. The encoding still
determines the cost of checking preservation.

Group states by the response they require. The resulting groups record just
enough information to recover that response; the formal name is the
\emph{optimizer quotient}.

\begin{definition}[Optimizer Quotient]\label{def:decision-equiv}\label{def:decision-quotient}
States $s,s'\in S$ are \emph{optimizer-equivalent} when
$\Opt(s)=\Opt(s')$.  Writing this kernel relation as
$s\equiv_{\Opt}s'$, the \emph{optimizer quotient} is
\[
Q_{\mathcal D}=S/{\equiv_{\Opt}},
\]
with quotient map $q_{\Opt}:S\to Q_{\mathcal D}$.
\end{definition}

Write $\pi(s)$ for what remains visible and $h$ for a rule that recovers the
required response from it. Recovery means $h(\pi(s))=\Opt(s)$ at every state.
This recovery criterion also holds for real utilities.
The kernel $\ker(\pi)$ records the pairs that look identical through $\pi$;
the theorem says none may require different responses.
Representation $\pi_1$ is finer than $\pi_2$ when $\pi_2$ can be recovered
from $\pi_1$; formally, $\pi_2$ factors through $\pi_1$.

\leanmetapending{\LH{UB1}, \LHrng{QT}{5}{7}}
\begin{theorem}[Kernel Factorization for the Optimizer Map]\label{thm:master-factorization}
For every map $\pi:S\to Q$,
\[
\bigl(\exists\,h:Q\to 2^A\setminus\{\varnothing\},
       \ \Opt=h\circ\pi\bigr)
\quad\Longleftrightarrow\quad
\ker(\pi)\subseteq {\equiv_{\Opt}}.
\]
Moreover, $Q_{\mathcal D}$ is canonically isomorphic to
$\operatorname{range}(\Opt)$.  If
$\pi$ is surjective and satisfies the equivalent conditions above, there is a
unique map $u:Q\to Q_{\mathcal D}$ with $q_{\Opt}=u\circ\pi$.
Consequently, $Q_{\mathcal D}$ is the unique coarsest exact abstraction, up to
canonical isomorphism.
\end{theorem}

\begin{proof}
Factorization through $\pi$ implies kernel containment.  Under kernel
containment, $h(\pi(s)):=\Opt(s)$ is well defined.  The map
$h$ may be extended arbitrarily off $\operatorname{range}(\pi)$.  The
canonical bijection from $Q_{\mathcal D}$ to $\operatorname{range}(\Opt)$
sends $[s]_{\Opt}$ to $\Opt(s)$.  For surjective $\pi$, the unique quotient
map is $u(\pi(s))=[s]_{\Opt}$.
\end{proof}

Thus a retained observation may merge states only when they require the same
response. The brake/steer pair fails exactly this test. The audits below check
the familiar factorization criterion without constructing all response groups.

\subsection{Sufficiency and Relevance}

Keeping both hazard sensors preserves the controller's actions. Keeping only
the first merges the brake and steer states, so it is insufficient.

\begin{definition}[Sufficient Coordinate Set]\label{def:sufficient}
A coordinate set $I\subseteq\{1,\ldots,n\}$ is \emph{sufficient} when
\[
 s_I=s'_I\quad\Longrightarrow\quad\Opt(s)=\Opt(s')
 \qquad(s,s'\in S).
\]
\end{definition}

For this static target, we use \emph{sufficient} and \emph{preserving}
interchangeably.  Later stochastic contracts qualify preservation explicitly.

The empty set suffices exactly when $\Opt$ is constant. A conflicting pair
agreeing on $I$ witnesses insufficiency.\leanmeta{\LHrng{DP}{6}{7}}

\begin{definition}[Minimal Sufficient Set]\label{def:minimal-sufficient}
A sufficient set is \emph{minimal} when none of its proper subsets is
sufficient.
\end{definition}

Can two sufficient input sets be replaced by just the inputs they share?
If this always works, intersecting all sufficient sets leaves one smallest
set. The next lemma records that shortcut; product closure will justify it.

\leanmetapending{\LHrng{IC}{1}{4}}
\begin{lemma}[Intersection Closure Gives a Least Member]\label{lem:intersection-closure-collapse}
Let $C$ be a finite coordinate set and let
$\varnothing\ne\mathcal F\subseteq\mathcal P(C)$.  If
$I,J\in\mathcal F$ implies $I\cap J\in\mathcal F$, then
\[
  I_{\mathcal F}:=\bigcap_{I\in\mathcal F} I
\]
belongs to $\mathcal F$ and is its unique least member.  Consequently, for
every budget $k$,
\[
  \exists I\in\mathcal F:\ |I|\le k
  \quad\Longleftrightarrow\quad
  |I_{\mathcal F}|\le k .
\]
\end{lemma}

\begin{proof}
Finiteness and pairwise closure place the intersection of all members in
$\mathcal F$.  This intersection is the unique least member.  Its containment
in every $I\in\mathcal F$ gives the budget equivalence.
\end{proof}

\begin{definition}[Relevant Coordinate]\label{def:relevant}
Coordinate $i$ is \emph{relevant} when
\[
\exists s,s' \in S:\;
\Big(\forall j \neq i,\; s_j=s'_j\Big)
\ \wedge\
\Opt(s)\neq\Opt(s').
\]
\end{definition}

In the opening controller, each hazard sensor has its own conflicting pair,
whereas none of the diagnostic sensors does. When all combinations are
admitted, keeping exactly such individually necessary inputs always suffices.
The next statement allows any response specified by an equivalence relation.

\leanmetapending{\LH{EF1}}
\begin{proposition}[On Products, the Necessary Inputs Suffice]\label{prop:minimal-relevant-equiv}
Let $S=X_1\times\cdots\times X_n$, and let $\equiv$ be any equivalence
relation on $S$.  Call $I$ sufficient when
$s_I=s'_I$ implies $s\equiv s'$, and call coordinate $i$ relevant when some
$s,s'$ agree outside $i$ and satisfy $s\not\equiv s'$.  The sufficient sets
are closed under finite intersection.  The relevant-coordinate set is
sufficient and is the unique inclusion-minimal sufficient set.  Equivalently,
a coordinate set is sufficient exactly when it contains every relevant
coordinate.
\end{proposition}

\begin{proof}
Let $I$ and $J$ be sufficient, and suppose $s,s'$ agree on $I\cap J$.
Let $t$ splice the $I$-coordinates of $s$ into $s'$.  Sufficiency of $I$ and $J$
gives $s\equiv t\equiv s'$.  Hence the sufficient sets are
intersection-closed and have a unique least member $K$.

Every sufficient set contains each relevant coordinate, since omitting one
retains its witness pair in a single fiber.  Conversely, if $i\in K$ were
irrelevant, replace coordinate $i$ in one state by its value in another.
Irrelevance relates the original state to the replacement, and the
replacement agrees with the second state on $K$.  Hence $K\setminus\{i\}$
would be sufficient, contradicting leastness.  Therefore $K$ is exactly the
relevant-coordinate set.
\end{proof}

For equality of $\Opt$, this gives Definitions~\ref{def:sufficient}
and~\ref{def:relevant}. Uniqueness does not make the relevance audit cheap:
succinct inputs can describe exponentially many candidate witnesses.

\leanmetapending{\LH{EF1}}
\begin{corollary}[Full Interface and Smallest Sufficient Interface]
\label{cor:free-top-unique-bottom}
For a product-coordinate decision problem, let $R$ be the set of relevant
coordinates.  Then
\[
  \{I:I\text{ is sufficient}\}
  =\{I:R\subseteq I\subseteq[n]\}.
\]
Thus the full interface $[n]$ is sufficient without search, while $R$ is the
unique least sufficient interface.
\end{corollary}

\begin{proof}
Apply Proposition~\ref{prop:minimal-relevant-equiv}: sufficiency is
equivalent to containing $R$.
\end{proof}

\leanmetapending{\LH{DQ86}}
\begin{remark}[Listed States and Product Closure]\label{rem:listed-product-closure}
Suppose the only allowed states are $(0,0)$ and $(1,1)$, with different
optimizer sets.  Either coordinate alone distinguishes them; the empty set
does not.  Thus $\{1\}$ and $\{2\}$ are incomparable minimal sufficient sets.
The missing mixed states prevent the splice in
Proposition~\ref{prop:minimal-relevant-equiv}.  This models a log of two
co-occurring settings.  It is also the set-cover instance with one universe
element and two duplicate sets.
\end{remark}

\begin{figure}[htbp]
\centering
\includegraphics[width=.94\linewidth]{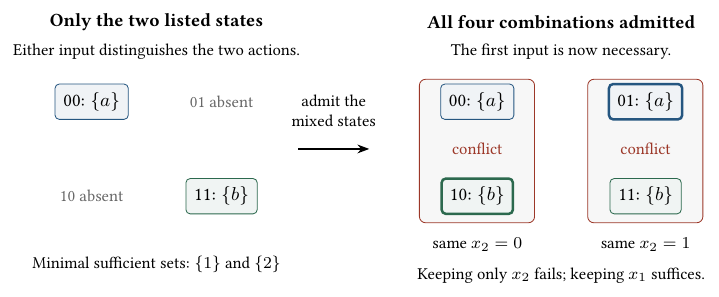}
\caption{Missing combinations can hide a coordinate's relevance.  On the
left, either input distinguishes the two listed states.  On the right,
the newly admitted states have thick borders.  Each outlined column groups
states agreeing on $x_2$ but requiring different actions.  The unique smallest
sufficient set is now $\{1\}$.}
\Description{The listed domain contains only 00 labelled a and 11 labelled b.
Either coordinate is sufficient. Extending the response to the full square
by labelling 01 with a and 10 with b makes the first coordinate necessary.
Keeping the second coordinate alone groups 00 with 10 and 01 with 11;
both groups contain conflicting labels.}
\label{fig:missing-combinations}
\end{figure}

Equivalently, $I$ is sufficient iff $\Opt=h_I\circ\pi_I$ for some $h_I$,
which assigns each realized observation its common optimizer
(Theorem~\ref{thm:master-factorization}).\leanmeta{\LH{UB2}}

\paragraph{Changing a setting is not removing a channel.}
Replacing one setting need not make an input removable. If
$F(1)=F(3)=\mathsf A$ and $F(2)=\mathsf B\ne\mathsf A$, replacing three
by one preserves the response, but discarding the count does not. Encoding
the count by one mandatory step and two optional Boolean steps gives responses
$\mathsf A,\mathsf B,\mathsf B,\mathsf A$ at $00,10,01,11$:
each Boolean input is relevant even though $00$ and $11$ agree.
Budgets count the inputs in the supplied encoding.

On a product, fixing omitted inputs to defaults preserves $F$ exactly when
the retained inputs determine it: states agreeing on those inputs have the
same reset tuple, and every state agrees with its own reset. Supplement~S2
gives the expanded illustration and criteria.
\leanmeta{\LHrng{IS}{1}{8}}

\subsection{Admissible Outputs}

For acceptable-output sets $\mathrm{Adm}(s)$, write
$s\sim_{\mathrm{Adm}}s'$ iff $\mathrm{Adm}(s)=\mathrm{Adm}(s')$.
An empty acceptable set may specify failure.

\leanmetapending{\LHrng{QT}{8}{9}}
\begin{proposition}[Admissible Outputs as an Optimizer Quotient]\label{prop:exact-output-realizes-optimizer-quotient}
Let the output set $O$ be finite.  For every admissible-output specification
$\mathrm{Adm}:S\to\mathcal{P}(O)$, there is a decision problem whose optimizer
quotient satisfies
\[
s\sim_{\Opt}s'
\quad\Longleftrightarrow\quad
\mathrm{Adm}(s)=\mathrm{Adm}(s').
\]
\end{proposition}

\begin{proof}
Use action space $\mathcal P(O)$ and make $\mathrm{Adm}(s)$ uniquely optimal
at $s$.  Then $\Opt(s)=\{\mathrm{Adm}(s)\}$.
\end{proof}

This construction can require $2^{|O|}$ actions and need not preserve encoded
size. Unlike comparison across families of decision problems
~\cite{blackwell1953equivalent,savage1954foundations,raiffa1961applied,howard1966information},
we fix one response. For additive tolerance~\cite{abel2016near}, it is
\[
\mathrm{Adm}_{\varepsilon}(s)
:=\left\{a\in A:U(a,s)\ge
\max_{b\in A}U(b,s)-\varepsilon\right\}.
\]
Write $\gamma(U)$ for the least positive regret (shortfall from the best utility):
\[
 \gamma(U):=\min\bigl\{\max_{b\in A}U(b,s)-U(a,s)>0:
                    a\in A,\ s\in S\bigr\},
\]
with $\gamma(U)=+\infty$ when the set is empty.

\leanmetapending{\LH{EF11}}
\begin{corollary}[Positive-Gap Tolerance Band]
\label{cor:positive-gap-tolerance}
For every $0\le\varepsilon<\gamma(U)$ and every state $s$,
\[
 \mathrm{Adm}_\varepsilon(s)=\Opt(s).
\]
Consequently the preservation verdicts are identical throughout this interval
whether we check a proposed input set, seek a set within a budget, or ask
whether one realized observation has a constant response.
\end{corollary}

\begin{proof}
Optimal actions have regret zero; all others have regret at least $\gamma(U)$.
Below that threshold the two correspondences agree at every state, hence so
do all three preservation queries.
\end{proof}

Supplement~S3 gives the breakpoint, stability-radius,
and tolerance-existence results beyond this fixed positive-gap interval.

Decisiveness tests uniqueness \emph{after} scores are aggregated on a fiber.
Requiring a conditional $\varepsilon$-admissible set to be a
singleton needs a separate margin for every induced fiber and becomes
weakly more demanding as the tolerance grows.

\fi
\ifappendixtolerance
\section{The Absolute-Tolerance Family}\label{sec:tolerance-family}

For the absolute-tolerance target in
Corollary~\ref{cor:positive-gap-tolerance}, write
\[
 r(a,s):=\max_{b\in A}U(b,s)-U(a,s),\qquad
 Q_\varepsilon:=S/{\ker(\mathrm{Adm}_\varepsilon)} .
\]
Thus $a\in\mathrm{Adm}_\varepsilon(s)$ exactly when
$r(a,s)\le\varepsilon$, and $Q_0$ is the optimizer quotient.

\leanmetapending{\LHrng{EF}{2}{3}, \LH{EF10}}
\begin{proposition}[Breakpoint and Tail Structure]
\label{prop:tolerance-breakpoints}
Let
\[
 R(s):=\{r(a,s):a\in A\},\qquad R:=\bigcup_{s\in S}R(s).
\]
The set $R$ contains $0$ and has at most
$1+|S|(|A|-1)$ elements.
For $|A|=1$, this gives the exact bound $|R|=1$.
The correspondence $\mathrm{Adm}_\varepsilon$ is right-continuous and changes
only at values in $R$.  If
$0=v_0<\cdots<v_\ell$ are its distinct values, then $Q_\varepsilon$ is
constant on each $[v_i,v_{i+1})$ and on $[v_\ell,\infty)$.
Consequently, the family has at most $1+|S|(|A|-1)$ distinct members.  For
\[
\varepsilon\ge \max_{a,s}r(a,s),
\]
every action is admissible at every state and $Q_\varepsilon$ has one class.
\end{proposition}

\begin{proof}
For fixed $(a,s)$, membership of $a$ in
$\mathrm{Adm}_\varepsilon(s)$ is the threshold predicate
$r(a,s)\le\varepsilon$.  It can change only when $\varepsilon$ reaches that
regret, and it includes the action at the breakpoint.  Each state contributes
zero and at most $|A|-1$ positive regrets, which gives the cardinality bound;
the interval claims follow by taking all action--state pairs.  Above the
largest regret every
threshold predicate is true.
\end{proof}

On an explicit table, the values in $R$ can be enumerated and sorted in
polynomial time.  A succinct presentation can encode exponentially many
states.  Each breakpoint remains computable from one action--state pair in
polynomial time, while enumerating all breakpoints need not be polynomial in
the succinct input length.

\leanmetapending{\LHrng{EF}{4}{8}}
\begin{example}[The Kernels Are Non-monotone]
\label{ex:tolerance-kernel-nonmonotonicity}
Let two states have utility vectors
\[
 U(\cdot,s_1)=(0,-1),\qquad U(\cdot,s_2)=(0,-2).
\]
For $0\le\varepsilon<1$, both states admit $\{a\}$ and lie in one class.  For
$1\le\varepsilon<2$, only $s_1$ admits $\{a,b\}$, so the class splits.  For
$\varepsilon\ge2$, both states admit $\{a,b\}$ and merge again.
The action sets grow monotonically with tolerance, while their equality
kernel follows a merge--split--merge pattern.  The reverse separation also
occurs: vectors $(0,-1)$ and $(-1,0)$ have disjoint optimizer sets at
$\varepsilon=0$ and both admit $\{a,b\}$ at $\varepsilon=1$.
\end{example}

\leanmetapending{\LH{EF9}}
\begin{theorem}[Meet of the Tolerance Family]
\label{thm:tolerance-family-meet}
For any states $s,s'$, equality of
$\mathrm{Adm}_\varepsilon(s)$ and
$\mathrm{Adm}_\varepsilon(s')$ at every $\varepsilon\ge0$ is equivalent to
equality of their complete regret vectors and to equality of every pairwise
action gap.  Equivalently,
\[
 \bigcap_{\varepsilon\ge0}\ker(\mathrm{Adm}_\varepsilon)
 =
 \ker\!\bigl(s\mapsto r(\cdot,s)\bigr)
 =
 \ker\!\bigl(s\mapsto
       (U(b,s)-U(a,s))_{a,b\in A}\bigr),
\]
the kernel of the complete action-gap vector.
\end{theorem}

\begin{proof}
Equal regret vectors give equal threshold sets.  Conversely, if
$r(a,s)<r(a,s')$, the threshold $\varepsilon=r(a,s)$ admits $a$ at $s$ and
excludes it at $s'$.  Hence agreement at every threshold forces coordinatewise
equality of regret.  Finally,
\[
 r(a,s)-r(b,s)=U(b,s)-U(a,s).
\]
The regret vector therefore determines every pairwise gap.  Every regret is
the maximum of those pairwise gaps, giving the converse.
\end{proof}

The optimizer quotient is the member $Q_0$, while the common refinement of
every member retains the complete action-gap data.

\leanmetapending{\LHrng{EF}{11}{15}}
\begin{proposition}[Stability Radius]\label{prop:tolerance-stability-radius}
Define the extended stability radius
$\varepsilon^\star(\mathcal D)$ as the supremum of the
$\delta\ge0$ for which $Q_\varepsilon=Q_0$ throughout
$0\le\varepsilon<\delta$.  If a positive regret occurs, then
\[
\varepsilon^\star(\mathcal D)\ge
\min\{r(a,s):r(a,s)>0\}.
\]
The inequality can be strict.
\end{proposition}

\begin{proof}
Below the least positive regret, every admissible set equals its value at
zero, so its kernel equals $\ker(\Opt)$.  For strictness, take utility vectors
$(0,-1,-5)$ and $(0,-1,-6)$.
The minimum positive regret is $1$.  Both states admit $\{a\}$ below $1$ and
$\{a,b\}$ throughout $[1,5)$, so their common optimizer class survives.
At $\varepsilon=5$, only $s_1$ admits $c$, and the class splits.  Thus
$\varepsilon^\star(\mathcal D)=5$.
\end{proof}

\begin{remark}[Absolute Tolerance Is Not Affine Invariant]\label{rem:tolerance-descent}
Under $U'(a,s)=\alpha(s)+\beta(s)U(a,s)$ with $\beta(s)>0$, regret transforms
as $r'(a,s)=\beta(s)r(a,s)$.  Absolute tolerance therefore does not descend
through statewise positive affine reparameterization.  On nonconstant states,
normalizing each regret by the state's utility range does; constant states use
the zero regret vector.
\end{remark}

\begin{remark}[Remaining Tolerance Questions]\label{rem:tolerance-deferred}
Section~\ref{sec:static-tolerance-certification} classifies three static
questions: sufficiency throughout $[0,B]$, whether the stability radius is at
least $B$, and whether some tolerance in $(0,B]$ makes an interface
sufficient.  Stochastic tolerance can widen statewise admissibility
before averaging or widen the conditional score after averaging; these
predicates differ.  Sequential tolerance also depends on backward induction,
and counting family members ranges over exponentially many implicit
breakpoints.  These directions are deferred.
\end{remark}
\fi

\ifappendixtolerance\else
\subsection{Input Sufficiency: Product Closure Versus Missing Combinations}\label{sec:hardness}

Deleting the second hazard sensor merges a brake state with a steer state.
One such pair settles a failed deletion; a successful deletion must exclude
every such pair. On all input combinations, the budget query asks how many
inputs are individually necessary. Missing combinations can instead permit
different sufficient choices. The formal queries use the encodings in
Remark~\ref{rem:product-closure-scope}.

\subsubsection{Product Closure Removes the Coordinate Choice}

\begin{problem}[SUFFICIENCY-CHECK]
Given a finite coordinate presentation $\mathcal D$ and $I$, is $I$ sufficient?
\end{problem}

\begin{problem}[WITHIN-BUDGET-SUFFICIENCY]
Given a product-coordinate problem $\mathcal D$ and budget $k$, is some
$|I|\le k$ sufficient?
\end{problem}

\begin{problem}[RESTRICTED-WITHIN-BUDGET-SUFFICIENCY]
The same budget query on a finite coordinate presentation, without product closure.
\end{problem}

\begin{problem}[SOME-OBSERVATION-SUFFICIENCY]
Given a finite coordinate presentation and $I$, is $\Opt$ constant on some
realized fiber $\{s:s_I=(s_0)_I\}$?
\end{problem}

Product closure reduces the budget query to $|R|\le k$ for the relevant set
$R$ (Proposition~\ref{prop:minimal-relevant-equiv}). For Boolean responses,
$R$ is the set of variables the response depends on, studied in junta
theory~\cite{fischer2004juntas}. These coNP classifications
~\cite{bassan2024local} establish the cost of checking a compactly supplied
response~\cite{papadimitriou1986succinct,balcazar1992succinct}. Restricted
domains, coverage, and conditional scores introduce the additional quantifiers.

\leanmetapending{\LH{DQ40}}
\begin{theorem}[One-Coordinate Deletion Is coNP-Complete]\label{thm:sufficiency-conp}
Under the succinct Boolean-product encoding of Section~\ref{sec:encoding},
\textsc{Sufficiency-Check} is coNP-complete.  Hardness already holds with two
actions, at most two optimizer-set values, and a proposed interface obtained
from the full interface by deleting one coordinate.
\end{theorem}

\begin{proof}
For membership, negating Definition~\ref{def:sufficient} gives a polynomial
witness $(s,s')$ with $s_I=s'_I$ and
$\Opt(s)\ne\Opt(s')$ whenever $I$ is not sufficient, so the complement is in
NP.

For hardness, reduce \textsc{Tautology} on $n$ variables to the Boolean product
$\{0,1\}^{n+1}$.  Write a state as $(x_0,x)$, propose the interface containing
all $x$-coordinates but not $x_0$, and use the actions
\emph{accept} and \emph{reject}.  Set
\[
 U(\mathit{accept},(0,x))=1,
 \qquad
 U(\mathit{accept},(1,x))=\mathbf 1[\varphi(x)],
 \qquad
 U(\mathit{reject},(x_0,x))=0.
\]
Every $x_0=0$ state has optimizer $\{\mathit{accept}\}$.  An $x_0=1$ state
has the same optimizer when $\varphi(x)$ holds and the tied optimizer
$\{\mathit{accept},\mathit{reject}\}$ otherwise.  Each proposed-interface
fiber consists of $(0,x)$ and $(1,x)$, so the interface is sufficient exactly
when $\varphi$ is a tautology.  The validity, projection, and utility circuits
have size polynomial in $|\varphi|$.

\end{proof}

\paragraph{Quotient cardinality.}
The hardness persists when the optimizer quotient has at most two classes.
Parameterization by either action count or
optimizer-quotient cardinality alone is therefore para-coNP-hard.

\leanmetapending{\LH{DQ36}, \LHrng{SK}{4}{8}}
\begin{theorem}[Finding a Sufficient Set Within Budget Is coNP-Complete]\label{thm:minsuff-conp}
Under the succinct Boolean-product encoding of Section~\ref{sec:encoding},
the apparent existential layer in \textsc{Within-Budget-Sufficiency} collapses
to relevance containment.  Consequently,
\textsc{Within-Budget-Sufficiency} is coNP-complete.
\end{theorem}

\begin{proof}
Proposition~\ref{prop:minimal-relevant-equiv} identifies sufficient sets on a
product space with the supersets of the relevant-coordinate set $R$.  The
within-budget query therefore asks whether $|R|\le k$.  Its complement has a
polynomial certificate: list
$k+1$ distinct coordinates and, for each one, a state pair witnessing
relevance. This certificate places the problem in coNP.

For hardness, reuse the Boolean-product construction in the proof of
Theorem~\ref{thm:sufficiency-conp} and set $k=0$.  A sufficient set of size
zero exists exactly when its optimizer map is constant, which holds exactly
when $\varphi$ is a tautology.  Thus the $k=0$ slice is coNP-hard, but
its optimum may be zero.  The Boolean paired family in
Proposition~\ref{prop:paired-family-approximation-gap} replaces that
degenerate gap by optimum one versus all $p$ coordinates.
\end{proof}

\subsubsection{Missing Combinations Restore the Coordinate Choice}

On the agreement-only hazard states 00 and 11, either sensor suffices.
Neither sensor is mandatory by itself, but retaining no sensor fails.
In larger circuit-defined domains, choosing a sufficient set within budget
must be followed by an audit of every state pair that choice merges.
The next theorem proves that this choose-then-check structure has
the full $\Sigma_2^P$ worst-case complexity, even when each selectable input
is one bit.

The hardness construction makes a bit budget enforce whole-block choices.
For each choice it provides two blocks of input bits. Omitting a bit from
each block hides a conflicting state pair that differs only at those two
bits. Keeping one whole block prevents this failure; after retaining the
mandatory identifier bits, the budget leaves room for exactly one block per
choice. Those choices encode a candidate assignment. The remaining state
pairs test whether that assignment works for every setting of the other
variables.

\leanmetapending{\LHrng{LMS}{1}{2}}
\begin{theorem}[Missing Combinations Keep the Choice Alive]
\label{thm:listed-minsuff-sigma2p}
\textsc{Restricted-Within-Budget-Sufficiency} is $\Sigma_2^P$-complete under the
circuit-defined restricted-state encoding.
Hardness already holds when every coordinate is Boolean and the budget counts
individual bits.
\end{theorem}

\begin{proof}
For membership, first replace $k$ by $\min\{k,n\}$, which does not change the
set of feasible coordinate subsets.  Then guess $I$ and use a
pair of agreeing states with different optimizer sets as the NP witness for
failure.  The within-budget query is therefore in $\Sigma_2^P$.

For hardness, reduce $\exists z\,\forall x\,\varphi(z,x)$ directly to
individually selectable bits. Pad so $d=|z|\ge1$ and $u=|x|\ge1$.
The main states are $m(z',x,r)$ for $r\in\{0,1\}$, with two actions and
\[
 U(0,m)=1,\qquad
 U(1,m(z',x,r))=
 \begin{cases}0,&\varphi(z',x),\\2r,&\neg\varphi(z',x),\end{cases}
\]
Put $T=d+u+1$ and $L=T+1$. Each existential bit has two $L$-bit
candidate blocks $Q_{j,0},Q_{j,1}$. Set
\[
 Q_{j,b}(m)=
 \begin{cases}
   0\,x\,0^{d+1},&z'_j=b,\\
   1\,z'\,x\,r,&z'_j\ne b.
 \end{cases}
\]
Selecting block $Q_{j,z_j}$ for every $j$ will group the two $r$-states
for each matching $(z,x)$ and expose every off-assignment state.

Let $F=dL^2$, choose
$h=\max\{3,\lceil\log_2(F+1)\rceil\}$, and prepend two marker bits and
$h$ public identifier bits.  Main states have public word $00\,0^h$.
For every $(j,i,\ell)\in[d]\times[L]\times[L]$, choose a distinct nonzero
$c_{j,i,\ell}\in\{0,1\}^h$ and add a forcing pair with public word
$00\,c_{j,i,\ell}$.  Its first state has all candidate bits zero; its second
has ones exactly at $Q_{j,0}[i]$ and $Q_{j,1}[\ell]$.  Give the two endpoints
opposite unique optimizers.

The public bits are forced by vertex-disjoint calibration edges whose
candidate bits are all zero.  For every payload position $k$, use the edge
\[
 10\,e_{k+1}\ \longleftrightarrow\
 10\,(e_{k+1}\oplus e_k),
\]
with subscripts modulo $h$.  Add
$10\,1^h\leftrightarrow11\,1^h$ and
$01\,0^h\leftrightarrow11\,0^h$ for the two marker positions, again with
opposite unique optimizers at the endpoints.  Because $h\ge3$, these
calibration endpoints are distinct from one another; their nonzero markers
separate them from main and forcing states.  The nonzero forcing identifiers
separate the forcing pairs from one another and from the main states.

Set the bit budget to
\[
 q+dL,\qquad q=h+2.
\]
Omitting a public bit merges its calibration pair, so all $q$ public bits are
mandatory.  For each $j$, the retained candidate bits must cover every edge
of the complete bipartite graph between $Q_{j,0}$ and $Q_{j,1}$.  Every vertex
cover of $K_{L,L}$ has size at least $L$; a cover of size $L$ is one whole
side, since omitting a vertex on each side leaves their edge uncovered.  The
remaining budget therefore selects exactly one complete block per $j$ and no
other candidate bits.  Decode those sides as $z$.

For this canonical mask, every forcing and calibration state is isolated.
The main states with $z'=z$ form exactly the two-state fibers indexed by $x$;
if $z'\ne z$, one retained block contains the exposed identifier
$1z'xr$, so that main state is a singleton. The two matched states both
prefer action zero when $\varphi(z,x)$ holds; otherwise $r=0$ prefers zero
and $r=1$ prefers one. Thus this mask is sufficient exactly when
$\forall x\,\varphi(z,x)$. A validity
circuit recognizes the three disjoint state families by their public words,
decodes the exposed main identifier or forcing index, and checks the displayed
bit patterns.  The number of coordinates, state-name length, utilities, and
all circuits are polynomial in the source size.  Thus the construction is a
polynomial many-one reduction with Boolean coordinates and an ordinary bit
budget.  The executable audit exhausts all source predicates and all
budget-feasible masks for one existential and one universal bit.
Supplement~S1.1 also records the preliminary token-field construction.
\end{proof}

Missing combinations keep the coordinate choice alive: choose which inputs
to retain, then check all allowed state pairs that agree on them.  The
neighboring logic-minimization precedent is Umans's $\Sigma_2^P$-completeness
result for minimum equivalent DNF and shortest
implicants~\cite{umans2001minimum}.

\paragraph{The same hardness with a positive tolerance.}
The utility gaps in the preceding two hardness constructions are integral:
every positive pointwise regret is at least one. Accepting actions within a
tolerance $0\le\varepsilon<1$ of the best score therefore changes neither
construction's required responses. Corollary~\ref{cor:positive-gap-tolerance}
transfers the coNP proposed-input lower bound and the $\Sigma_2^P$
restricted-domain classification to $\mathrm{Adm}_\varepsilon$.
The later approximation construction separating a one-input optimum from an
all-input optimum also has integral gaps
(Definition~\ref{def:paired-tautology-family}), so its lower bound survives
the same tolerance.

\begin{remark}[Rough-Set Search versus Product-Closed Certification]
\label{rem:rough-set-search-collapse}
On listed rows, the auditor may still have to choose between different
preserving input sets. A full product fixes the smallest set to the relevant
coordinates. The next query leaves the inputs fixed and asks the auditor to
choose one reading instead.
\end{remark}

\subsubsection{One Observation and an Existential--Universal Check}

For a fixed retained input set, choose one observation realized by an allowed
state.  Check that all states giving that observation have the same complete
optimal-action set.  The resulting
existential--universal pattern yields $\Sigma_2^P$.  In the reduction,
existential variables choose the observation and universal variables range
over the states that produce it.

The algorithm chooses the observation to certify.  Anchor and abductive
explanations make a related local choice
~\cite{ribeiro2018anchors,ignatiev2019abduction,marques2022delivering,darwiche2023reasons}.
The predecessor comparisons are in Section~\ref{sec:contribution-boundary}.

\leanmetapending{\LH{DQ98}, \LH{CC4}}
\begin{theorem}[Finding One Observation That Preserves the Answer]\label{thm:anchor-sigma2p}
Under the succinct Boolean-product encoding of Section~\ref{sec:encoding},
\textsc{Some-Observation-Sufficiency} is $\Sigma_2^P$-complete.
\end{theorem}

\begin{proof}
For membership, guess a representative state $s_0$ and verify
\[
\forall s\in S:\quad
 (s_I=(s_0)_I)\Rightarrow \Opt(s)=\Opt(s_0).
\]
The action list is polynomial, so both optimizer sets can be computed and
compared in polynomial time.  The existential guess followed by the universal
state check places the problem in $\Sigma_2^P$.

For hardness, reduce from an instance
$\exists z\,\forall x\,\psi(z,x)$ of
$\exists\forall$-SAT.  Pad the universal block so that $|x|\ge1$, use the
Boolean product state space on $(z,x)$, and retain exactly the $z$
coordinates.  Take actions \emph{yes} and \emph{no}, and define
\[
\begin{aligned}
 U(\mathit{yes},(z,x))&=
 \begin{cases}
 2,&\psi(z,x)\text{ is true},\\
 0,&\psi(z,x)\text{ is false},
 \end{cases}\\
 U(\mathit{no},(z,x))&=
 \begin{cases}
 1,&x=0^{|x|},\\
 0,&x\ne0^{|x|}.
 \end{cases}
\end{aligned}
\]
If $\psi(z,x)$ is true for every $x$, action \emph{yes} is uniquely optimal
throughout the $z$-fiber.  Conversely, suppose $\psi(z,x^\star)$ is false.
When $x^\star\ne0^{|x|}$, the optimizer at $(z,x^\star)$ is
$\{\mathit{yes},\mathit{no}\}$, while the optimizer at $(z,0^{|x|})$ is a
singleton.  When $x^\star=0^{|x|}$, the optimizer there is
$\{\mathit{no}\}$.  Choose any nonzero $\widehat x$, which exists after
padding.  At $(z,\widehat x)$ the optimizer is either
$\{\mathit{yes}\}$ or $\{\mathit{yes},\mathit{no}\}$, so it again differs.
Thus a $z$-fiber is optimizer-constant exactly when
$\forall x\,\psi(z,x)$, and some realized observation has a constant fiber
exactly when the
source formula is true.

The retained projection and both utilities are computed by circuits of size
polynomial in $|\psi|$.  The construction is therefore a polynomial-time
many-one reduction.
\end{proof}

\fi
\ifappendixtolerance
\subsection{Tolerance-Indexed Certification}
\label{sec:static-tolerance-certification}

The target can vary across a declared tolerance interval while the retained
coordinate set remains fixed.

\paragraph{\textsc{Tolerant-Sufficiency}.}
Given succinct $\mathcal D$, coordinate set $I$, and nonnegative dyadic $B$,
decide whether $I$ is sufficient for $\mathrm{Adm}_\varepsilon$ at every
$\varepsilon\in[0,B]$.

\leanmetapending{\LH{EF2}, \LH{DQ40}}
\begin{theorem}[\textsc{Tolerant-Sufficiency} is coNP-complete]
\label{thm:tolerant-sufficiency-conp}
\textsc{Tolerant-Sufficiency} is coNP-complete under succinct encoding.
\end{theorem}

\begin{proof}
Suppose sufficiency fails at some $\varepsilon\in[0,B]$.  There are states
$s,s'$ agreeing on $I$ and an action $c$ whose membership differs between
$\mathrm{Adm}_\varepsilon(s)$ and
$\mathrm{Adm}_\varepsilon(s')$.  After exchanging the states if necessary,
\[
 r(c,s)\le\varepsilon<r(c,s').
\]
At the breakpoint $v=r(c,s)$ the same action still separates their
admissible sets.  A complement certificate therefore consists of
$(c,s,s')$ and the circuit evaluations establishing coordinate agreement,
$v\le B$, and this membership difference.  The action list is polynomial and
each utility evaluation is polynomial, so the complement lies in NP.

For hardness set $B=0$.  Since
$\mathrm{Adm}_0=\Opt$, the resulting instance is
\textsc{Sufficiency-Check}, which is coNP-hard by
Theorem~\ref{thm:sufficiency-conp}.
\end{proof}

\paragraph{\textsc{Stability-Radius}.}
Given a succinct Boolean-product decision problem $\mathcal D$ and a
nonnegative dyadic $B$, decide whether
$\varepsilon^\star(\mathcal D)\ge B$.

\leanmetapending{\LH{EF17}}
\begin{theorem}[The Stability Radius Is coNP-Complete]
\label{thm:stability-radius-conp}
\textsc{Stability-Radius} is coNP-complete, already for two actions on a
Boolean product.
\end{theorem}

\begin{proof}
A no-instance has a polynomial witness: an attained breakpoint $v<B$ and two
states on which equality under $\mathrm{Adm}_v$ differs from equality under
$\Opt$.  Such a breakpoint exists whenever the kernels first disagree before
$B$, because the admissible correspondence is constant between attained
regrets.  The utilities, breakpoint, admissible sets, and both state
comparisons are computable in polynomial time.  The problem is therefore in
coNP.

For hardness, reduce from \textsc{Tautology}.  On states
$(t,x)\in\{0,1\}^{1+n}$ use actions $a,c$ and set
\[
 U(a,(t,x))=0,
 \qquad
 U(c,(t,x))=
 \begin{cases}
  -1,&t=1\text{ or }\varphi(x),\\
  -2,&t=0\text{ and }\neg\varphi(x).
 \end{cases}
\]
The optimizer is $\{a\}$ at every state, so $Q_0$ has one class.  If
$\varphi$ is a tautology, the regret of $c$ is one everywhere and every
$Q_\varepsilon$ still has one class.  Otherwise a state with $t=1$ and a
falsifying state with $t=0$ split throughout $[1,2)$, so the stability radius
is exactly one.  Thus, for $B=3/2$, the answer is yes exactly when
$\varphi$ is a tautology.  The utility circuit has polynomial size.
\end{proof}

The minimum positive regret $\gamma(U)$ in this reduction is always one.
Consequently, deciding whether the stability radius strictly exceeds the
general lower bound from Proposition~\ref{prop:tolerance-stability-radius} is
already coNP-hard.

\paragraph{\textsc{Robust-Within-Budget-Sufficiency}.}
Given succinct product-coordinate $\mathcal D$, integer $k$, and nonnegative
dyadic $B$, decide whether some $I$ with $|I|\le k$ is sufficient for every
$\mathrm{Adm}_\varepsilon$ with $\varepsilon\in[0,B]$.

\leanmetapending{\LHrng{EF}{1}{3}, \LH{DQ36}}
\begin{theorem}[\textsc{Robust-Within-Budget-Sufficiency} is coNP-complete]
\label{thm:robust-minimum-conp}
\textsc{Robust-Within-Budget-Sufficiency} is coNP-complete on product state
spaces under succinct encoding.
\end{theorem}

\begin{proof}
For each tolerance $\varepsilon$, apply
Proposition~\ref{prop:minimal-relevant-equiv} to
$\ker(\mathrm{Adm}_\varepsilon)$.  Its sufficient sets are exactly the
supersets of the corresponding relevant-coordinate set $R_\varepsilon$.
The sets sufficient at every tolerance in $[0,B]$ are therefore the
supersets of
\[
 R_{\le B}:=\bigcup_{0\le\varepsilon\le B}R_\varepsilon .
\]
Proposition~\ref{prop:tolerance-breakpoints} makes this a finite union over
attained breakpoints.  It is the unique least robust set, so
Lemma~\ref{lem:intersection-closure-collapse} reduces the budgeted question to
$|R_{\le B}|\le k$.

A complement certificate lists $k+1$ distinct coordinates.  For each
coordinate it supplies a breakpoint in $[0,B]$ and two states that differ
only there and have different admissible sets.  A relevance witness at an
arbitrary tolerance can be moved to an attained breakpoint by choosing an
action in the symmetric difference, exactly as in the preceding proof.
This polynomial certificate places the problem in coNP.  At $B=0$ the problem
is \textsc{Within-Budget-Sufficiency}, which is coNP-hard by
Theorem~\ref{thm:minsuff-conp}.
\end{proof}

Varying tolerance does not destroy the unique least sufficient set on a
product space.

\paragraph{\textsc{Tolerance-Existence}.}
Given succinct $\mathcal D$, coordinate set $I$, and positive dyadic $B$,
decide whether some $\varepsilon\in(0,B]$ makes $I$ sufficient for
$\mathrm{Adm}_\varepsilon$.

\leanmetapending{\LH{EF18}}
\begin{theorem}[\textsc{Tolerance-Existence} is $\Sigma_2^P$-complete]
\label{thm:tolerance-existence-sigma2p}
\textsc{Tolerance-Existence} is $\Sigma_2^P$-complete already for succinct
Boolean-product decision problems.
\end{theorem}

\begin{proof}
Membership guesses a polynomial-length dyadic representative and applies the
coNP state-pair test.  The representative may be $B$, an attained positive
breakpoint encoded by a state--action pair, or
$\min\{B,\beta\}/2$ for a guessed attained positive breakpoint $\beta$; the
last form represents the initial open interval.  Every guessed representative
lies in $(0,B]$, so soundness follows directly from the coNP test.  For
completeness, use an attained breakpoint from the satisfying interval, use
the midpoint below the least positive breakpoint when the initial interval
works, or use $B$ if no positive breakpoint exists.

For hardness, reduce from $\exists z\,\forall x\,\psi(z,x)$.  The full state
space has Boolean coordinates $(g,z,x,r)$, and the proposed interface retains
$I=\{g,z_1,\ldots,z_m\}$, where $m=|z|$.  Let
$L_z=\langle z\rangle+1$.  Use two actions
$a,c$, give $a$ regret zero, and give $c$ regret
\[
 r \quad\text{when }g=1,
 \qquad
 L_z+\mathbf 1[r=1\land\neg\psi(z,x)]
 \quad\text{when }g=0.
\]
Set $B=2^m+\tfrac12$.  Each $g=1$ fiber contains regrets zero and one, so no
tolerance in $(0,1)$ is sufficient.  At every tolerance at least one, those
guard fibers are constant.

Every $\varepsilon\in[1,B]$ belongs to a unique interval
$[L_z,L_z+1)$.  In every fiber with a smaller selector code, both branches
admit $c$; in every fiber with a larger code, neither branch admits $c$.
Within the selected fiber, branch zero admits $c$, while branch one admits
$c$ exactly when $\psi(z,x)$ holds.  Thus $I$ is sufficient at some
$\varepsilon\in(0,B]$ exactly when
$\exists z\,\forall x\,\psi(z,x)$.  The utilities have polynomial-size
circuits, and the carrier is the complete Boolean product.
For the maximal code $\langle z\rangle=2^m-1$, the final interval is
$[2^m,2^m+1)$ and contains
$B=2^m+\tfrac12$; no right-endpoint convention is used at $B$.
\end{proof}

\begin{table}[H]
\centering
\small
\renewcommand{\arraystretch}{1.13}
\renewcommand{\tabularxcolumn}[1]{m{#1}}
\begin{tabularx}{\linewidth}{@{}>{\raggedright\arraybackslash}m{0.25\linewidth}>{\raggedright\arraybackslash}X>{\raggedright\arraybackslash}m{0.24\linewidth}@{}}
\toprule
\textbf{Presentation} & \textbf{Tolerance-indexed query} & \textbf{Complexity} \\
\midrule
Explicit table & proposed-input interval sufficiency and robust within-budget sufficiency on products & polynomial time \\
\midrule
Succinct presentation & \textsc{Tolerant-Sufficiency} & coNP-complete \\
\midrule
Succinct Boolean product & \textsc{Stability-Radius} & coNP-complete \\
\midrule
Succinct product space & \textsc{Robust-Within-Budget-Sufficiency} & coNP-complete \\
\midrule
Succinct Boolean product & \textsc{Tolerance-Existence} & $\Sigma_2^P$-complete \\
\bottomrule
\end{tabularx}
\caption{Static queries whose correctness target is reindexed by an
absolute tolerance.  The table does not introduce stochastic or sequential
tolerance predicates.}
\label{tab:tolerance-static-queries}
\end{table}

For an explicit table, enumerate the regret breakpoints up to $B$ and run the
state-pair test at each one.  On a full product table, taking the union of the
relevant-coordinate sets gives the robust minimum within the same polynomial
input bound.

\paragraph{Design consequence}
A declared approximation tolerance creates a third design choice.  Sufficiency
at one threshold and sufficiency throughout a tolerance interval are different
requirements: Example~\ref{ex:tolerance-kernel-nonmonotonicity} shows that
equality of admissible sets can be lost and later recovered as tolerance
grows.  An advisory layer promising an $\varepsilon$-band must certify the
entire band; endpoint certification does not establish that contract.
Even on complete products, product closure removes the input-choice layer
from robust within-budget sufficiency but leaves the tolerance choice at full
$\Sigma_2^P$ strength.  The structure settles which inputs matter at a fixed
tolerance; it does not choose the tolerance for us.
\fi

\appendixtolerancefalse
\section{Verification Protocol for AI-Assisted Work}
\label{app:verification-protocol}

Generated drafts were admitted only after the semantic target, input encoding,
and acceptance criterion had been fixed independently.  A retained reduction
had to preserve the source truth value in both directions, cover every state
or fiber case, have polynomial encoded size, satisfy its arithmetic bounds,
and use only the external complexity premises stated in the proof.  These
checks distinguish three forms of support:

\begin{enumerate}
\item \emph{Lean-checked finite cores}: factorization identities, finite
      separating examples, optimizer correspondences, selector decoding,
      support gadgets, and finite certificate equivalences;
\item \emph{executable audits}: deterministic programs using exact arithmetic
      for \ExecutableAuditInventory; and
\item \emph{pen-and-paper proofs in this article}: oracle-machine simulation,
      asymptotic completeness, polynomial encoded-size bounds, and imported
      complexity results.
\end{enumerate}

The supplementary archive accompanying this manuscript contains the sources for all three
layers: a complete theorem-to-declaration ledger, the pinned
\texttt{lean-toolchain} and \texttt{lake-manifest.json}, the executable
regression suite, the frozen learned-policy inputs and hashes, and a fully
pinned Python lock file. After extracting it, enter the enclosed
\texttt{paper4\_jair/} directory. From that directory,
\begin{center}
\texttt{python3.12 experiments/verify\_jair\_artifact.py -{}-bootstrap}
\end{center}
creates the Python environment, runs this manuscript's regression suite, and
builds the pinned Lean project.  It deliberately excludes the companion
task-arrival experiment. The accompanying archive's \texttt{SMindex.txt}
gives this review path, and \texttt{publication\_manifest.json} records its
file hashes. These identify the exact package used to verify this manuscript.
The Python audits are independent implementations checked against the same
finite specifications.  They are regression evidence, not code extracted
from Lean and not a formally verified bridge from Python to Lean.

\noindent\textbf{Public artifact for this manuscript:} \url{\PaperZenodoURL}.
It contains the current manuscript, claim ledger, and replay package.

\paragraph{Within-budget decisiveness under conditional scores.}
For Theorem~\ref{thm:stochastic-minimum-hard}, acceptance required every
decisive budget-feasible coordinate set to decode exactly one existential assignment
and every nonselector to create a nondecisive forcing fiber.  The audit covered
the selected main fibers, exposed codewords, perturbation cancellation,
positive dyadic mass normalization, product-background fibers, and the
accept--hold and hold--hold dominance inequalities.  The exact program checks
all $256$ predicates at block size one.  At block size two, gadget behavior
depends only on the $16$ weak-majority outcomes of a source table.  The
program exhausts all $2^{16}=65{,}536$ such outcome matrices together with
every local exposed-singleton slice, thereby covering every gadget-relevant
equivalence class without literally enumerating $2^{64}$ source tables.  The
proof that these summaries
determine every gadget branch is what makes the reduced enumeration
exhaustive.  This finite audit exercises the
constants in Lemma~\ref{lem:paired-selector-product-completion}; it does not
    prove the asymptotic class lift.  The pen-and-paper proof protects the direct
accept--hold margin $\gamma\mu(\theta-\varepsilon)$; background utility is
identical for the two hold actions and therefore cannot change either their
intended tie or their exposed-state separation.

\paragraph{Polynomial-hierarchy and sequential rows.}
For within-budget sufficiency on circuit-defined restricted states and input sufficiency with observation coverage, the author
separately checked the selector/support witness correspondence and the
polynomial-size circuit realization before applying the standard
$\exists\forall$ or $\forall\exists$ machine characterization.  The
Boolean-coordinate refinement of that restricted-state reduction has a separate
exact audit.  On its complete one-existential-bit, one-universal-bit source
family, the audit checks all $32{,}192$ masks within the bit budget and leaves
exactly the two canonical assignment masks.  It then checks those masks on all
$16$ source predicates.  The same $54$ codewords are assigned positive dyadic
mass with common denominator $1{,}024$; exact fiber sums verify the
on-support/Markov-blanket reduction and normalization.  This exercises
the public-bit calibration edges, every $K_{4,4}$ forcing edge, the collision
partition, both optimizer directions, and the sparse-support lift.
The
observation-coverage construction has a unique pointwise optimizer, so the
all-state corollary uses the same checked finite instance rather than a
second null-fiber gadget.  For the
bounded-alternation ladder, each source instance was checked to satisfy the
declared arena-block promise, and the certification quantifiers were counted
outside the arena evaluation.  The unbounded row uses the standard TQBF and
polynomial-space premises stated in Appendix~\ref{app:proof-sequential-ladder}.
For the horizon-two POMDP lift, Lean checks at the last decision that the
lifted fiber scores and optimizer sets equal their one-step
counterparts.  Its
coverage predicate combines positive-fiber coverage with optimizer
equality at every declared last-decision state, so zero-mass declared states are
audited exactly as they are in the one-step input-sufficiency clause.  The
executable regression for Example~\ref{ex:sensor-support-contract} checks the
semantic distinction directly: retaining the first coordinate passes
preservation on supported states and coverage but fails both one-step and
horizon-two sufficiency with observation coverage.  The separate individual-bit within-budget-decisiveness
reduction is proved in main-article
Appendix~\ref{app:pomdp-horizon-two-bits}.
Its executable audit checks every one of the $30{,}827$ budget-feasible bit
masks at block size one.  Formula-independent parity and forcing-pair failures
leave exactly two canonical masks; both are then checked against all $256$
source predicates using exact integer scores after clearing positive common
denominators.  This validates the finite bit-budget construction, not its
asymptotic complexity transfer or a Lean formalization of that transfer.
For the unary-horizon PSPACE
classification, Lean checks the finite depth-first Bellman evaluator and stack
bound, dyadic-denominator growth, geometric sign preservation, both root
optimizer wrappers, and the coordinate-forcing bank.  The direct
QBF-to-compressed-MDP circuit construction, Turing-machine space accounting,
polynomial encoding bounds, and asymptotic completeness transfer remain
pen-and-paper mathematics in this article; the classical citation supplies
provenance rather than a hidden wrapper lemma.  The executable component test exhausts every
quantifier prefix and Boolean terminal predicate through three variables:
$2{,}120$ source QBFs, $18{,}776$ sign-core checks, $6{,}360$ preservation
checks, $133{,}560$ budget-feasible coordinate-forcing checks, and $18{,}776$
decisiveness checks.  Every value and comparison uses exact rational
arithmetic.  It does not execute the historywise belief recursion or enumerate
all projected policies used by the value-of-information argument.  It
independently detects finite component or constant errors but does not prove
the Turing-machine space bound or asymptotic reduction.

\paragraph{Action-count boundary and approximation.}
\label{sec:explicit-approximation-calibration}
The two-action classification is supported at three distinct levels.  The
manuscript proves the GapP normal form and the polynomial selector reduction;
Lean checks the singleton--nonzero equivalence and both local parity
identities.  The local selector audit exhausts all $3^8=6{,}561$ one-bit
signed payload tables, then checks the two-bit cancellation and isolation
factorizations.  A separate composed audit joins weak majority, universal
exact-count enumeration, and the selector on all $65{,}536$ source predicates
and all $46$ budget-feasible interfaces at the smallest nontrivial block size.
Formula-independent
cancellation discharges the $42$ nonselectors; the four canonical selectors
are checked predicate by predicate.  The same audit exhausts all decisiveness-gap truth tables
through three formula variables and checks the full-parity obstruction at
four.  The approximation census fixed
the antichain universe and coordinate-index tie-breaker before enumeration and
retains all $7{,}773$ records.  After redundant discernibility constraints are
removed, the remaining families are antichains of nonempty subsets of $[p]$;
the program exhausts $2,5,19,167,$ and $7{,}580$ families for $p=1,\ldots,5$.
Greedy is exact through $p=4$; at $p=5$ its first and largest gap is $3/2$,
with optimum two and greedy size three.  This is a finite statement about the
declared tie-break, not an estimate of typical behavior.  The total-input-length
inapproximability corollary is a pen-and-paper argument using bounded-occurrence
$3$-SAT and the standard binary circuit encoding; it is not inferred from the
finite census.

\paragraph{Tolerance reductions.}
The stability-radius audit exhausts all $278$ Boolean truth tables through
three variables, checks that the reduction's minimum positive regret is one,
and compares the optimizer and admissibility kernels at $B=3/2$.  The
full-product tolerance selector is checked on all $528$ predicates at block
sizes $(|z|,|x|)\in\{(1,1),(1,2),(2,1)\}$.  For every predicate the program
enumerates the complete Boolean carrier and the breakpoint intervals up to
$B$, then compares tolerance existence with $\exists z\,\forall x\,\psi$.

For Proposition~\ref{prop:listed-discernibility-hitting-set}, Lean checks the finite
equivalence between sufficiency and hitting every discernibility set.  The
polynomial-time hash construction, greedy analysis, parameterized transfer,
and asymptotic lower bound remain manuscript or cited results; the label does
not mark an unresolved semantic step.

\paragraph{Learned-policy carrier audit.}
The case study in Section~\ref{sec:learned-policy-audit} is separate from the
formal claim map: it exercises the paper's carrier distinction on one fitted
classifier.  The frozen tree has SHA-256
\nolinkurl{893b2e77e9bbd266231e4e7ea6f42d8410e32b34d6c947eb307e603533e4f990};
the retained JSON also hashes the transformed data, held-out indices, rows,
and predictions. Its pinned configuration reconstructs the classifier,
whose hash and minimum-input outputs are checked on replay. Exhaustive search
checks all $2^{16}$ interfaces on the
listed carrier.  On the product carrier, a different-label pair of
axis-aligned tree leaves contributes the coordinates on which their boxes are
disjoint.  A coordinate is semantically relevant exactly when some such pair
has it as its only separator.  The program records the resulting one-coordinate
prediction-changing witnesses, checks that the relevant set hits every
different-label leaf pair, and is regression-tested against brute-force
four-valued products on small random trees.  From the supplementary package
root, the audit and its brute-force regression test run as part of
\texttt{python3.12 experiments/verify\_jair\_artifact.py -{}-bootstrap}; the
pinned lockfile fixes the CPython 3.12 scientific stack.  These checks certify
the frozen policy on the two declared carriers; they make no population claim.

The claim map names the Lean declarations supporting each finite statement.
Claims outside those declarations are classified as conditional in the
artifact report; no claim is marked partial or unverified.  This protocol records how candidate
AI-generated material was evaluated; it does not transfer authorship or
responsibility to the tools.

\section{Formal Contracts for Finite Reduction Schemas}
\label{app:formalization-contracts}

The Lean development separates semantic preservation from encoded
polynomial-time computability.  This separation prevents a finite semantic map
with a small output from being reported as a Karp reduction before its
constructor and encodings have been supplied.

\paragraph{Excluded-witness presentation.}
For a predicate $p:X\to\mathsf{Prop}$, a finite excluded-witness presentation
consists of a finite witness set $W$ and a decidable relation
$R\subseteq X\times W$ such that
\[
 p(x)\quad\Longleftrightarrow\quad
 \neg\exists w\in W\;R(x,w).
\]

\paragraph{Existential excluded-witness presentation.}
A finite existential presentation consists of finite sets $W,V$ and decidable
relations $G\subseteq X\times W$ and
$B\subseteq X\times W\times V$ such that
\[
 p(x)\quad\Longleftrightarrow\quad\exists w\in W\;G(x,w),
 \qquad
 G(x,w)\quad\Longleftrightarrow\quad
 \neg\exists v\in V\;B(x,w,v).
\]

\paragraph{Semantic map with polynomial output size.}
For sized instance spaces $X,Y$ and predicates
$L_X\subseteq X$, $L_Y\subseteq Y$, this contract is a map $f:X\to Y$
satisfying
\[
 x\in L_X\quad\Longleftrightarrow\quad f(x)\in L_Y
\]
and constants $c,d\in\mathbb N$ such that
\[
 |f(x)|\le c|x|^d+c.
\]
The contract imposes no evaluation-time bound.

\paragraph{Encoded Karp reduction.}
An encoded Karp reduction additionally fixes finite encodings of $X$ and $Y$
and requires $f$ to be computable in polynomial time between those encodings.
Membership preservation then gives the standard polynomial-time many-one
reduction.  The formal bridge between the two map structures is indexed by
Lean handles \texttt{OU36--OU37}.

\section{Encoding Details}

\paragraph{Circuit-defined restricted states.}
Polynomial-size circuits recognize valid
state identifiers and compute each coordinate projection.  The instance fixes
an identifier length $m$ in unary, and every candidate state is a string in
$\{0,1\}^m$; $m$ is polynomial in the encoded instance size.  State
identifiers are canonical: $C_S$ accepts exactly one
identifier for each allowed state, the represented state domain is promised
nonempty, and $C_{\pi}$ is total on accepted identifiers.  Actions are
represented by indices in the supplied action list, while coordinate values
are binary indices in the canonical domains $[|X_i|]$; $C_{\pi}$ must output
an index below the declared domain size.  The representation may therefore describe
exponentially many states and need not be closed under coordinate replacement.
The circuits specify which states are allowed; no table of rows is supplied.
Sequential instances additionally provide circuits
for arena ownership, two successors, terminal status, and terminal payoff; the
polynomial horizon is encoded in unary.

\begin{table}[htbp]
\centering
\small
\setlength{\tabcolsep}{4pt}
\renewcommand{\arraystretch}{1.12}
\renewcommand{\tabularxcolumn}[1]{m{#1}}
\begin{tabularx}{\linewidth}{@{}>{\raggedright\arraybackslash}m{0.16\linewidth}>{\raggedright\arraybackslash}m{0.18\linewidth}>{\raggedright\arraybackslash}m{0.29\linewidth}>{\raggedright\arraybackslash}X@{}}
\toprule
\textbf{Access} & \textbf{State domain} & \textbf{Input length, up to entry bits} & \textbf{What the input exposes} \\
\midrule
Explicit & listed & $\Theta(m(n+|A|))$ for $m$ rows & an arbitrary logged state set; unlisted coordinate combinations stay absent \\
\midrule
Explicit & product-closed & $\Theta(N(n+|A|))$, $N=\prod_i|X_i|$ & every product state; for Boolean domains, $N=2^n$ and scanning all coordinate subsets is polynomial in the table length \\
\midrule
Succinct & restricted & $|C_S|+|C_{\pi}|+|C_U|+\operatorname{poly}(n)$ & exponentially many valid state identifiers recognized by $C_S$, with projections computed by $C_{\pi}$; no closure promise \\
\midrule
Succinct & product-closed & $\operatorname{enc}(A)+\sum_i\log|X_i|+|C_U|+|C_P|$ & the entire product is implicit; $C_P$ is present only in the stochastic setting \\
\bottomrule
\end{tabularx}
\caption{Explicit tables can list arbitrary rows or a complete product.
Only the complete product guarantees all combinations and has at least
$2^n$ states when every coordinate has at least two values.}
\label{tab:encoding-two-by-two}
\end{table}

\begin{remark}[Why Explicit and Succinct Inputs Differ]\label{rem:explicit-succinct-boundary}
An explicit utility table exposes all $N=\prod_i|X_i|$ states, so pair
enumeration is polynomial in its encoded length even when it is exponential in
the coordinate count.  Singleton coordinate domains may be discarded first;
after this harmless preprocessing, $2^n\le N$ and scanning all coordinate
subsets is polynomial in an explicit product table.  A succinct circuit
may describe the same state family without listing it.  Darwiche and Marquis
organize the analogous
knowledge-compilation tradeoff by representation succinctness and supported
queries~\cite{darwiche2002knowledge}; here the certification query is fixed and
only the presentation changes.

\end{remark}

\section{Supporting Proofs for the Main Results}

The article states each result and explains its mechanism. The following proofs retain the complete derivations and point to the detailed constructions where needed.

\subsection{Checking Every Observation in Theorem~\ref{prop:two-action-exact-counting}}
\label{app:two-action-check-proof}

\begin{proof}
For membership, failure is witnessed by a state $s$ and two distinct
actions tied for maximum fiber score.  The witness is checked using
polynomially many PP comparisons.  PP is closed under complement, intersection,
and polynomial-time truth-table reductions, hence under the Boolean
combinations needed to verify simultaneous
maximality~\cite{beigel1991pp,fortnow1996pp}.  The complement
is therefore in $\mathsf{NP}^{\mathsf{PP}}$, and decisiveness is in
$\operatorname{co}(\mathsf{NP}^{\mathsf{PP}})$.

For hardness, the first identity in
Lemma~\ref{lem:exact-counting-projection-identities} gives the complete normal
form $\forall x\;g(w,x)\ne0$, where $g$ is GapP.  Write
$g(w,x)=\sum_p h(w,x,p)$ with $h\in\{-1,0,1\}$.  Use the uniform state space
on $(x,p)$, retain $x$, and take two actions with
\[
 U(a,x,p)=1+h(w,x,p),\qquad U(b,x,p)=1-h(w,x,p).
\]
The signed score gap on the $x$-fiber is a common positive factor times
$g(w,x)$.  Every fiber is decisive exactly when
$\forall x\;g(w,x)\ne0$, proving the lower bound with two actions.
\end{proof}

\subsection{Proof of Lemma~\ref{lem:nested-pp-oracle-normal-form}}

\begin{proof}[Proof outline]
Guess $z$ and test $\forall x$ using a $\operatorname{co}(\mathsf{NP}^{\mathsf{PP}})$
oracle. For hardness, guess an outer oracle transcript with no-answer
counterwitnesses; universally verify the yes-answers. PP complement and
truth-table closure combine these into one predicate, whose padded random
tape gives $|y|=p(|w|)$. Appendix~\ref{app:proof-nested-normalization} proves
both directions and every encoding bound.
\end{proof}

\subsection{Proof of Theorem~\ref{prop:two-action-exact-counting}}

\begin{proof}
With two actions, a fiber is decisive exactly when its signed score gap
$G_I(s)$ is nonzero. Clearing the common dyadic denominator makes that gap
a GapP function. Its zero set is in $\mathsf C_{=}\mathsf P$
\cite{fenner1994gap}; Lemma~\ref{lem:exact-counting-projection-identities}
gives the three upper bounds from their projection forms.

Section~\ref{app:two-action-check-proof} proves
item~\ref{item:two-action-check}. For
item~\ref{item:two-action-some}, existential
nonvanishing is absorbed into one nonzero GapP test by the sum-of-squares
identity in Lemma~\ref{lem:exact-counting-projection-identities}.
For matching hardness, represent an arbitrary GapP function as a difference
of two accepting-path counts, pad their paths to one length, and use their
acceptance indicators as two action utilities under a uniform distribution.
Retaining no coordinates gives one fiber, whose gap is nonzero exactly when
the source GapP function is nonzero. This proves $\mathsf{coC}_{=}\mathsf P$
hardness with one observation and the stated empty-interface specialization.

For item~\ref{item:two-action-budget}, start with
$\exists z\,\forall x\,\mathrm{MAJ}_y\,\varphi(z,x,y)$ from
Lemma~\ref{lem:nested-pp-oracle-normal-form}.  By
Lemma~\ref{lem:majority-as-universal-nonzero}, this is equivalent to
\[
 \exists z\;\forall(x,\ell):
 \sum_{b,v}h(z,x,\ell,b,v)\ne0,
\]
where $h$ is the signed circuit in that lemma.  Thus the threshold has become
the exact-counting source used by the selector below.

The budget itself encodes the outer existential string.  Each bit
receives a pair of candidate coordinates $q_{j,0},q_{j,1}$, each with domain
$\{0,1\}^2$ and observed value $(r_{j,b},\xi_{j,b})$.  Hidden coordinates
store $z'$ and the signed-path block $p$.  Treat $(x,\ell)$ as one universal
block, pad it and $z$ to one length $N\ge1$, and write $x$ below for that
combined block.  Let $m$ be the path-block length in
Equation~\eqref{eq:gapp-signed-path-normal-form}.  Put
\[
 x(z',q)_j=\xi_{j,z'_j},\qquad
 \chi(z',q)=\prod_j(-1)^{r_{j,z'_j}},\qquad
 d=\chi\,h(w,z',x(z',q),p),
\]
and use $U(a)=1+d$, $U(b)=1-d$ under the uniform distribution.  Missing a
pair leaves one sign bit hidden, so its two values cancel the fiber gap and
force the selector to hit every pair.  The hidden $z'$ and path coordinates
are eligible choices, but retaining them cannot repair an omitted pair: its
sign bit remains hidden and still cancels.  The budget is therefore exhausted
by exactly one coordinate from each selector pair.  For the decoded assignment
$z$, all terms with $z'\ne z$ cancel at their first mismatch, while the
surviving gap is
\[
 G_I(r,x)=2^{1-3N-m}
 \left(\prod_{j=1}^{N}(-1)^{r_j}\right)g(w,z,x).
\]
Its prefactor never vanishes, so the interface is decisive exactly when
$\forall(x,\ell)\,g(w,z,x,\ell)\ne0$, hence exactly when every $x$ satisfies the
weak-majority test.  The candidate fields have four values and all
hidden fields are Boolean.  The exact audit composes the threshold
enumeration with this selector on all $65{,}536$ source predicates and all
$46$ budget-feasible interfaces at the smallest nontrivial block size.
\end{proof}

\subsection{Proof of Corollary~\ref{thm:sequential-bounded-ladder}}

\begin{proof}[Proof outline]
Threshold evaluation uses $d$ owner blocks; polynomially many threshold
queries compare rational root values in $\Delta_{d+1}^P$. Conflicting pairs
or $k+1$ relevance witnesses refute the first two queries; choosing a fiber
adds an existential block for the third. Bounded QBF supplies matching
hardness through a base/test root wrapper with no added arena block.
Appendix~\ref{app:proof-sequential-ladder} proves the rational comparison
bound, encodings, and all quantifier directions.
\end{proof}

\subsection{Proof of Corollary~\ref{thm:sequential-pspace}}

\begin{proof}
Depth-first evaluation and the outer certification witnesses use polynomial
space. The base/test wrapper for a prenex TQBF has a constant optimizer
exactly when the formula is true, proving hardness for the empty interface,
budget zero, and its sole observation. Appendix~\ref{app:proof-sequential-ladder}
gives the full max--min encoding.
\end{proof}

\subsection{Proof of Lemma~\ref{lem:pomdp-polynomial-space}}

\begin{proof}[Proof outline]
Evaluate continuation depth first with unnormalized values $W(g)=m(g)V_I(g)$.
This removes the posterior denominator. For $q$-bit state names and $b$-bit
dyadics, summing at most $2^{O(Hq)}$ paths needs $O(H(b+q))$ accumulator
bits and $O(H)$ frames. Unary $H$ makes this polynomial. The same traversal
tests compatible lifts. Appendix~\ref{app:multistep-pomdp-proofs} proves the
recurrence, null-history convention, and policy consistency.
\end{proof}

\subsection{Proof of Lemma~\ref{lem:pomdp-sign-core}}

\begin{proof}[Proof outline]
Actions choose existential bits and fair transitions choose universal bits.
With $u$ universal variables, terminal rewards $1$ and $-2^u$ separate
truth from falsity. Bonuses $2^{-u}4^{-(i+1)}$ at stage $i$ break ties
without reversing a nonzero comparison. Appendix~\ref{app:multistep-pomdp-proofs}
proves this for every prefix and bounds the encoding.
\end{proof}

\subsection{Proof of Theorem~\ref{thm:pomdp-exact-pspace}}

\begin{proof}[Proof outline]
Lemma~\ref{lem:pomdp-polynomial-space} and NPSPACE$=$PSPACE give membership.
For preservation hardness, merge a baseline root and a QBF-test root:
their full optimizers agree exactly on true formulas. Calibration pairs force
the public coordinates within budget, and resetting the root tag preserves
continuations. For decisiveness, a zero-value Quit action makes a false
formula tie at the root under every projection; the full interface witnesses
truth. Appendix~\ref{app:multistep-pomdp-proofs} proves both wrappers, including
normalized masses, absorbing padding, and all historywise obligations.
\end{proof}

\subsection{The Three-Action Carrier for Horizon-Two Inheritance}
\label{app:three-action-carrier}

The principal one-step classification uses two actions, independent uniform
inputs, and fields of size at most four
(Theorem~\ref{prop:two-action-exact-counting}). The construction below
instead retains explicit three-action scores: the individual-bit
horizon-two reduction reuses them in
Proposition~\ref{prop:pomdp-horizon-two-inheritance}. Its theorem and product
completion are kept here with that proof, rather than presented as sharper
one-step classifications.

Because the counted block has fixed length, every assignment uses the same
threshold $\theta$.  Small opposite adjustments to the hold scores then
cancel where the threshold test runs and break the tie on exposed singletons.
Figure~\ref{fig:three-action-selector-gadget} shows how these parts fit together.

\begin{figure}[htbp]
\centering
\resizebox{\linewidth}{!}{%
\begin{tikzpicture}[
  x=1.35cm,y=1.2cm,
  every node/.style={font=\small},
  box/.style={draw,rounded corners=2pt,fill=black!3,inner sep=4pt,align=center},
  chosen/.style={draw,rounded corners=2pt,fill=black!12,inner sep=4pt,align=center},
  note/.style={align=center,text width=3.6cm},
  heading/.style={font=\small\bfseries,align=center,text width=3.3cm}
]
  \draw[black!35] (3.65,-0.15) -- (3.65,3.3);
  \draw[black!35] (7.55,-0.15) -- (7.55,3.3);

  \node[heading] at (1.75,3.05) {Force one choice per bit};
  \node[box] (q0) at (0.75,1.85) {$q_{j,0}$};
  \node[chosen] (q1) at (2.55,1.85) {$q_{j,1}$\\chosen};
  \node[box] (f0) at (0.75,0.85) {$f(j,0)$\\large $\mathit{hold}_L$};
  \node[box] (f1) at (2.55,0.85) {$f(j,1)$\\large $\mathit{hold}_R$};
  \draw[->] (q0) -- (f0);
  \draw[->] (q1) -- (f1);
  \node[note] at (1.65,-0.05) {omitting both leaves a tie};

  \node[heading] at (5.6,3.05) {Which states look the same?};
  \node[chosen] (match) at (4.55,1.85) {$z'=z$\\$m(z,x,y,r)$};
  \node[chosen] (fiber) at (6.75,1.85) {$F_x$\\hide $(y,r)$};
  \draw[->] (match) -- (fiber);
  \node[box] (off) at (4.55,0.75) {$z'\ne z$\\$m(z',x,y,r)$};
  \node[box] (single) at (6.75,0.75) {exposed\\singleton};
  \draw[->] (off) -- (single);

  \node[heading] at (9.35,3.05) {Decision test inside $F_x$};
  \node[box,text width=3.3cm] at (9.35,1.72) {
    $\sigma_0+\sigma_1=0$\\[2pt]
    $\mathit{accept}:c_x$\\
    $\mathit{hold}_{L,R}:2^{n_y-1}-\tfrac12$};
  \node[note] at (9.35,0.42) {unique iff $c_x\ge 2^{n_y-1}$};
\end{tikzpicture}}
\caption{Choosing inputs, checking every observation, comparing scores.
Left: the budget forces one coordinate per pair, encoding an assignment.
Middle: each matching group is one observation; nonmatching states remain
singletons.  Right: both hold actions tie at a fixed score, so the group has
a unique winner exactly when accept clears that score.  The opposite hold
adjustments cancel in the groups but untie the exposed singletons.}
\par\smallskip
{\small\raggedright\noindent\emph{Diagram description.} Three-stage schematic of the coordinate-pair reduction.  A
coordinate budget chooses exactly one coordinate from each pair and
thereby decodes an existential assignment.  Matching codewords are grouped
into a fiber indexed by the universal assignment, while nonmatching codewords
remain exposed singletons.  Inside each grouped fiber, cancelling auxiliary
perturbations leave a weak-majority threshold that determines decisiveness.}
\label{fig:three-action-selector-gadget}
\end{figure}

The following construction makes the three-action score mechanism explicit.
Its product completion proves the one-step result; the individual-bit
horizon-two reduction reuses the scores on a separately constructed
restricted state set (the supplement's \emph{Horizon-two inheritance with
individual selectable bits}).

\leanmetapending{\LHrng{DC}{131}{133}, \LHrng{DC}{144}{149}, \LHrng{DC}{151}{153}, \LH{OU38}}
\begin{theorem}[A Three-Action Carrier for Decisive Observations]\label{thm:stochastic-minimum-hard}
Under the dyadic succinct product-coordinate encoding,
\textsc{Within-Budget-Decisive-Set} is
$\mathsf{NP}^{(\mathsf{NP}^{\mathsf{PP}})}$-complete even when the action set
has exactly three elements and the supplied distribution is strictly positive
and dyadic.  Equivalently, since an
oracle and its complement give the same Turing-oracle power, the class may be
written $\mathsf{NP}^{(\operatorname{co}(\mathsf{NP}^{\mathsf{PP}}))}$
~\cite{kwisthout2011mpe}.
\end{theorem}

\paragraph{Why this holds.}
Paired coordinates encode the initial choice; matching groups run the majority test, and forcing states rule out missed pairs. Figure~\ref{fig:three-action-selector-gadget} shows the selector.
The supplement's section \emph{Supporting Proofs for the Main Results} gives the complete proof of this result.

\Needspace{9\baselineskip}
\leanmetapending{\LHrng{DC}{144}{149}, \LHrng{DC}{151}{153}}
\begin{lemma}[Product Completion of the Coordinate-Pair Construction]
\label{lem:paired-selector-product-completion}
Let $\mathcal C_\varphi$ be the restricted selector code constructed in
the proof of Theorem~\ref{thm:stochastic-minimum-hard}.  A polynomial-time
transformation produces a dyadic succinct instance
$\widehat{\mathcal C}_\varphi$ on a full product with the same $2n_z$
coordinates and the same three actions such that, for every $I$ with
$|I|\le n_z$,
\[
 I\text{ is decisive for }\mathcal C_\varphi
 \quad\Longleftrightarrow\quad
 I\text{ is decisive for }\widehat{\mathcal C}_\varphi.
\]
\end{lemma}

\paragraph{Why this holds.}
Add missing states with small positive background mass. Equal hold scores and bounded mass preserve the required ties and strict comparisons.
The supplement's section \emph{Supporting Proofs for the Main Results} gives the complete proof of this result.

Exact-rational audits exhaust the block-size-one source predicates. At block
size two, they cover every distinct weak-majority outcome pattern and local
singleton case on which the construction depends, across the budget-feasible
interfaces and completed-product fibers (Section~\ref{sec:verification-scope}).
\leanmeta{\LH{PA3}, \LH{OU1}, \LHrng{OU}{9}{11}, \LH{OU38}}

\subsection{Proof of Theorem~\ref{thm:stochastic-minimum-hard}}

\begin{proof}[Proof outline]
Guess the interface and apply Theorem~\ref{prop:two-action-exact-counting}, item~\ref{item:two-action-check}. For hardness,
the selector in Figure~\ref{fig:three-action-selector-gadget} encodes $z$ by one
coordinate per pair. Its selected $x$-fiber has scores proportional to
$c_x,2^{n_y-1}-1/2,2^{n_y-1}-1/2$, where $c_x$ counts accepting paths.
It is decisive iff $\mathrm{MAJ}_y\varphi(z,x,y)$ holds. Opposite hold
perturbations cancel there but untie exposed singletons; forcing states
exclude missed pairs. Appendix~\ref{app:proof-stochastic-minimum} proves the
normalized construction, including the empty interface, and
Lemma~\ref{lem:paired-selector-product-completion} completes its product.
\end{proof}

\subsection{Proof of Lemma~\ref{lem:paired-selector-product-completion}}

\begin{proof}[Proof outline]
Add the missing coordinate tuples with positive dyadic background mass.
Giving both holds equal background utility preserves their deliberate ties
and separations; bounding total background mass below the rescaled
accept--hold margins preserves those comparisons. The forcing payoff still
excludes missed pairs. Appendix~\ref{app:proof-product-completion} proves the
polynomial-bit bounds and every budget-feasible case.
\end{proof}

\subsection{Within-Budget Decisiveness Under Conditional Scores}
\label{app:proof-stochastic-minimum}

The construction below keeps each forcing pair distinguishable from the
others. On $f(\ell,c)$, a coordinate belonging to pair $\ell$ returns a token
carrying $(\ell,c)$; every other coordinate returns a neutral token carrying
$\ell$. Side, neutral, and main-state tokens have disjoint outer tags.

\begin{proof}[Full proof of Theorem~\ref{thm:stochastic-minimum-hard}]
For membership, guess $I$ with $|I|\le k$ and invoke the
$\operatorname{co}(\mathsf{NP}^{\mathsf{PP}})$ procedure of
Theorem~\ref{prop:two-action-exact-counting}, item~\ref{item:two-action-check} for fixed-$I$ decisiveness.

For hardness, take an instance of
Lemma~\ref{lem:nested-pp-oracle-normal-form} with block lengths
$n_z,n_x,n_y\ge 1$.  For each existential bit $z_j$, create two coordinates
$q_{j,0}$ and $q_{j,1}$ and set the budget to $n_z$. Begin with a restricted
selector code $\mathcal C_\varphi$ whose codewords comprise the main states
\[
 m(z',x,y,r),\qquad
 z'\in\{0,1\}^{n_z},\ x\in\{0,1\}^{n_x},\
 y\in\{0,1\}^{n_y},\ r\in\{0,1\},
\]
and forcing states $f(j,0),f(j,1)$ for every $j$.
Figure~\ref{fig:three-action-selector-gadget} summarizes the resulting fiber
geometry.

Coordinate $q_{j,b}$ maps a main state to a common value carrying only $x$
when $z'_j=b$.  When $z'_j\ne b$, it returns a value that identifies the entire
main state.  On forcing states, $q_{j,b}$ records the side of $f(j,\cdot)$; a
coordinate from another pair returns a neutral value.  Selecting
$q_{j,z_j}$ for every $j$ therefore has two effects.  The main states with
$z'=z$ form fibers indexed by $x$ and hide $(y,r)$.  Every main state with
$z'\ne z$, and every forcing state, lies in a singleton fiber.
The coordinate-value type uses disjoint outer tags for main-state values and
forcing-state values, so no forcing token can collide with a main-state token.

The three actions are \emph{accept}, $\mathit{hold}_L$, and
$\mathit{hold}_R$.  On a main state,
\[
\begin{aligned}
 U(\mathit{accept},m(z',x,y,r))
   &=\mathbf 1[\varphi(z',x,y)],\\
 U(\mathit{hold}_L,m(z',x,y,r))
   &=\theta+\sigma_r,\\
 U(\mathit{hold}_R,m(z',x,y,r))
   &=\theta-\sigma_r,
\end{aligned}
\]
where
\[
 \theta=\frac12-2^{-(n_y+1)},\qquad
 \varepsilon=2^{-(n_y+3)},\qquad
 \sigma_0=-\varepsilon,\quad \sigma_1=\varepsilon.
\]
The two perturbations cancel when $r$ is hidden.  On an exposed singleton,
they make exactly one hold action
optimal whenever \emph{accept} has utility zero; when \emph{accept} has utility
one, it is uniquely optimal because $\theta+\varepsilon<1$.
In short, they cancel where the threshold test belongs and untie the holds
where it does not.

It remains to force every budget-feasible selector to hit every pair.  Give
each of the $2n_z$ forcing states mass
\[
 \delta=2^{-(2n_z+2)},\qquad
 M=\frac{2}{\delta}=2^{2n_z+3}.
\]
There are $N=2^{n_z+n_x+n_y+1}$ main states.  Give each main state mass
\[
 \mu=\frac{1-2n_z\delta}{N}.
\]
The elementary inequality $2n_z<2^{2n_z+2}$ gives $\mu>0$, and
$2n_z\delta+N\mu=1$.  The distribution is therefore positive, normalized, and
dyadic.

On forcing states, \emph{accept} has utility zero,
$\mathit{hold}_L$ has utility $M$ only at $f(j,0)$, and
$\mathit{hold}_R$ has utility $M$ only at $f(j,1)$.
Each designated hold contribution is
\[
 \delta M=2>1.
\]
If a nonempty coordinate set $I$ misses both $q_{j,0}$ and $q_{j,1}$, the
neutral tags isolate $f(j,0)$ and $f(j,1)$ in one two-state fiber.  The two
hold actions tie there at score $\delta M$, strictly above \emph{accept}.  If
$I=\varnothing$, all states share one fiber; symmetry makes the total scores
of the two hold actions equal, and the forcing contribution alone puts each
above one, while the accept score is at most one.  Thus the empty-set edge
case is nondecisive as well.  Consequently every decisive $I$ must hit all
$n_z$ pairs.  The budget
$|I|\le n_z$ then forces exactly one coordinate from each pair and decodes one
assignment $z$.

For such a valid selector, forcing fibers and off-assignment main fibers are
singletons and are decisive by the utility comparisons above.  On the selected
fiber indexed by $x$, let $c_x$ be the number of $y$ satisfying
$\varphi(z,x,y)$.  The unnormalized fiber scores have the common positive scale
$2\mu$ and reduce to
\[
 \mathit{accept}:c_x,\qquad
 \mathit{hold}_L=\mathit{hold}_R:
 2^{n_y}\theta=2^{n_y-1}-\frac12.
\]
Because $c_x$ is an integer, \emph{accept} is uniquely optimal exactly when
$c_x\ge 2^{n_y-1}$, equivalently
$\mathrm{MAJ}_y\,\varphi(z,x,y)$.  If the majority condition fails, the two
hold actions tie for maximum.  Hence all fibers are decisive exactly when
\[
 \forall x\;\mathrm{MAJ}_y\,\varphi(z,x,y).
\]
A decisive set of size at most $n_z$ exists exactly when
\eqref{eq:nested-fixed-block-source} holds.
Coordinate projections, utility evaluation, and dyadic weights have
polynomial-size circuits; the exponent and the forcing payoff require only
polynomially many bits.  The construction uses the pinned-threshold
three-action carrier from the some-observation reduction; the sharper
proposed-input and within-budget classifications use the two-action
exact-counting construction of
Theorem~\ref{prop:two-action-exact-counting}.

Lemma~\ref{lem:paired-selector-product-completion} converts
$\mathcal C_\varphi$ to a full product instance without changing the
decisiveness of any coordinate set within the budget.  The product instance is
therefore a yes-instance exactly when
\eqref{eq:nested-fixed-block-source} holds.
\leanmeta{\LHrng{DC}{131}{133}, \LHrng{DC}{144}{149}, \LHrng{DC}{151}{153}}
\end{proof}

\subsection{Product completion}
\label{app:proof-product-completion}

\begin{table}[htbp]
\centering
\begingroup
\small
\PaperTableSetup
\begin{tabularx}{\linewidth}{@{}>{\raggedright\arraybackslash}m{0.14\linewidth}>{\raggedright\arraybackslash}m{0.30\linewidth}>{\raggedright\arraybackslash}X@{}}
\toprule
\textbf{Constant} & \textbf{Role} & \textbf{Required dominance condition}\\
\midrule
$\delta$ & mass of each forcing state & $2n_z\delta<1$, leaving positive mass for the main states\\
\midrule
$\mu$ & common mass of each codeword main state & $2n_z\delta+N\mu=1$ and $\mu>2^{-(T+2)}$\\
\midrule
$M$ & forcing-state payoff & $\delta M=2>1$, so a missed pair creates a maximal hold--hold tie\\
\midrule
$\theta,\varepsilon$ & hold level and exposed-state perturbation & $\theta-\varepsilon\ge3/16$; background utility is identical for both holds\\
\midrule
$B,\lambda$ & background precision and per-state background mass, with $\lambda=2^{-B}$ & the displayed $B$ gives $\eta<2^{-(T+6)}$\\
\midrule
$\eta$ & total background mass, $\eta=Q\lambda$ & $\eta<\min\{1/16,\gamma\mu(\theta-\varepsilon)\}$\\
\midrule
$\gamma$ & mass retained for the selector code, $\gamma=1-\eta$ & $\gamma>15/16$ and every codeword mass exceeds $\eta$\\
\bottomrule
\end{tabularx}
\caption{Dominance audit for the paired-selector construction and its product
completion.  Violating the first three conditions breaks selector forcing;
violating the last four can reverse an accept--hold comparison after
completion.}
\label{tab:paired-selector-constant-audit}
\endgroup
\end{table}

\begin{proof}[Full proof of Lemma~\ref{lem:paired-selector-product-completion}]
Let $d=2n_z$, let $V$ contain the coordinate-token types used by the selector
code, and encode every codeword state $s$ by its complete tuple
\[
 c(s)=\bigl(q(s)\bigr)_{q\in[d]}\in V^d.
\]
The exposed-state and force-side tags make $c$ injective.  They also give a
polynomial-time decoder: read the state named by any exposed or force-side
token and verify its $d$ projections. Passing from the codeword states to the
full product $V^d$ preserves agreement between codewords on every coordinate
set.

Write $Q=|V|^d$, $T=n_z+n_x+n_y$, and
$h=\max\{2n_z+2,T+2\}$.  The token count satisfies
$|V|\le 2^{n_x}+2^{T+1}+5n_z<2^{T+4}$, so
$Q<2^{2n_z(T+4)}$.  Choose
\[
 B=2n_z(T+4)+h+4,\qquad \lambda=2^{-B},\qquad
 \eta=Q\lambda,\qquad \gamma=1-\eta.
\]
Then $\eta<2^{-h-4}\leq 1/16$ and $\gamma>15/16$.  Give every product state background
mass $\lambda$ and give a codeword $c(s)$ the additional mass
$\gamma P(s)$.  These masses sum to
$Q\lambda+\gamma\sum_sP(s)=1$, are positive dyadic rationals, and have
polynomial bitlength.  Moreover every codeword has mass greater than $\eta$:
the original forcing mass is $2^{-(2n_z+2)}$, while the main-state mass is
greater than $2^{-(T+2)}$.

No default action is added.  Reuse \emph{accept} as the background action: on
a codeword give each of the three actions a common baseline utility $1$ in
addition to its codeword utility, while on a non-code tuple give utility $1$ to
\emph{accept} and $0$ to both hold actions.  A code-free fiber is therefore
uniquely optimized by \emph{accept}.  For a valid selector $I_z$, the
codewords in each fiber are exactly the corresponding codeword fiber. The
disjoint outer tags ensure that such a fiber contains either main-state
codewords or forcing-state codewords, never both.  On its codeword set $C$,
the original mass is uniform, say $P(s)=\mu_C$ for $s\in C$.  For two
actions $a,b$, define their codeword fiber-score gap by
\[
 \Delta_C(a,b)
 :=\sum_{s\in C}P(s)\bigl(U(a,s)-U(b,s)\bigr)
 =\mu_C\sum_{s\in C}\bigl(U(a,s)-U(b,s)\bigr).
\]
The common baseline cancels.  The corresponding codeword contribution in the
completed instance is
\[
 \widehat\Delta_C(a,b)
 =\bigl(\lambda+\gamma\mu_C\bigr)
   \sum_{s\in C}\bigl(U(a,s)-U(b,s)\bigr)
 =\left(\gamma+\frac{\lambda}{\mu_C}\right)\Delta_C(a,b).
\]
The positive multiplier preserves every codeword comparison and satisfies
$|\widehat\Delta_C(a,b)|\ge\gamma|\Delta_C(a,b)|$.
The only new bias is the non-code mass in the fiber, at most $\eta$, and it
favors \emph{accept}.  The comparisons that can be affected have a direct
lower bound.  On a selected main fiber, the accept--hold gap is $2\mu$ times
the distance between an integer and a half-integer, hence has magnitude at
least $\mu$.  On an exposed main singleton, the winning hold beats
\emph{accept} by at least $\mu(\theta-\varepsilon)$ when the predicate is
false, while \emph{accept} wins by the larger margin
$\mu(1-\theta-\varepsilon)$ when it is true.  After completion, the codeword
part of each such gap has magnitude at least
$\gamma\mu(\theta-\varepsilon)$.

The displayed value of $B$ already protects these gaps.  Since $n_y\ge1$,
\[
 \theta-\varepsilon
 =\frac12-5\cdot2^{-(n_y+3)}
 \ge\frac{3}{16}>\frac18.
\]
Together with $\mu>2^{-(T+2)}$ and $\gamma>1/2$, this gives
$\gamma\mu(\theta-\varepsilon)>2^{-(T+6)}$.  Because
$h\ge T+2$, the choice above gives
$\eta<2^{-h-4}\le2^{-(T+6)}$.  Thus background mass cannot reverse any
accept--hold comparison used by the reduction.

No analogous small-gap budget is needed for the two hold actions.  Every
non-code tuple gives them the same utility, so background mass is exactly
neutral on their difference.  It preserves both the deliberate tie on a
failed selected fiber and the deliberate $2\varepsilon$ separation on an
exposed singleton.  This is why the perturbation scale does not enter the
background dominance condition.

If a budget-feasible set misses pair $j$, the codewords $c(f(j,0))$ and
$c(f(j,1))$ still share a product fiber.  Retain the forcing payoff
$M=2/\delta$.  After subtracting the common baseline, each corresponding hold
action receives the forcing contribution
\[
  (\lambda+\gamma\delta)M>\gamma\delta M=2\gamma>1,
\]
whereas the entire nonbaseline score of \emph{accept}, including non-code
mass, is at most one.  The two hold actions tie for the maximum.  The same
symmetry and domination argument covers $I=\varnothing$, where all codewords
and non-code tuples share one fiber. Both the codeword and completed fibers are
therefore nondecisive.

Now let $|I|\le n_z$.  Either $I$ misses a coordinate pair, in which case the
previous paragraph applies, or $I$ hits all $n_z$ pairs and contains exactly
one coordinate from each.  In the second case $I$ is a valid selector, and the
comparison-preservation argument applies to every codeword fiber.  Thus
decisiveness agrees for every budget-feasible $I$.  All added tests,
arithmetic operations, and rational constants have polynomial-size circuits
and polynomial bitlength.
\leanmeta{\LHrng{DC}{144}{149}, \LHrng{DC}{151}{153}}
\end{proof}

\subsection{Proof of Proposition~\ref{prop:pomdp-horizon-two-inheritance}}

\begin{proof}
Use the original actions and a dummy initial state. Every first action samples
$s$ from $P$; distinguished $a_0$ pays zero and all others pay $-1$. At the
last stage, action $a$ pays $U(a,s)$ and terminates. Stage-indexed circuits
require no selectable phase bit. Every first action has the same continuation
$V_I$, so $a_0$ uniquely wins for every interface, regardless of the sign or
size of $U$. After any recorded first action the distribution remains $P$;
the past penalty does not affect the last optimizer.

Last-stage scores are exactly the one-step fiber scores. Full optimizers
agree statewise even at null states, and supported projected histories are
exactly the positive-mass fibers. This proves both equivalences and, using
Theorem~\ref{thm:stochastic-preservation-tight}, the Boolean coverage bounds.
The supported initial history adds no further obligation. Transition and
reward circuits have polynomial size.

For within-budget decisiveness, the source has token fields. At stage zero,
parity rewards force the public identifier bits. At stage one, isolated
forcing pairs realize every edge of $K_{L,L}$ between opposite $L$-bit
candidate blocks. Covering these edges with budget $p+n_zL$ forces one whole
block per variable. The remaining fibers implement the majority tests.
Appendix~\ref{app:pomdp-horizon-two-bits} supplies the complete individual-bit
construction, positive normalized masses, and both correctness directions.
\end{proof}

\subsection{Exact-Counting Projection Identities}

\leanmetapending{\LHrng{OU}{9}{11}, \LH{OU13}}
\begin{lemma}[Exact-Counting Projection Identities]
\label{lem:exact-counting-projection-identities}
Use polynomial-time nondeterministic oracle machines with polynomially many
adaptive queries, and polynomially bounded projection as defined above.  Then
\[
 \exists\!\cdot\!\mathsf C_{=}\mathsf P
   =\mathsf{NP}^{\mathsf{PP}},
 \qquad
 \exists\!\cdot\!\mathsf{coC}_{=}\mathsf P
   =\mathsf{coC}_{=}\mathsf P,
\]
and
\[
 \exists\!\cdot\!\operatorname{co}
   (\exists\!\cdot\!\mathsf C_{=}\mathsf P)
   =\mathsf{NP}^{(\mathsf{NP}^{\mathsf{PP}})}.
\]
Because $\mathsf C_{=}\mathsf P\subseteq\mathsf{PP}$,
$\exists\!\cdot\!\mathsf C_{=}\mathsf P\subseteq
\mathsf{NP}^{\mathsf{PP}}$.  Conversely, PP is contained in
$\exists\!\cdot\!\mathsf C_{=}\mathsf P$: for a threshold test
$\#R\ge T$, guess $j\in\{T,\ldots,2^m\}$ and check $\#R=j$.
Guessing an adaptive PP-oracle transcript and using PP closure under
complement and polynomial truth-table reductions then gives
$\mathsf{NP}^{\mathsf{PP}}=\exists\!\cdot\!\mathsf{PP}$, proving the first
identity.

For the second, if $g$ is GapP, then
\[
 \exists x\;g(w,x)\ne0
 \quad\Longleftrightarrow\quad
 \sum_x g(w,x)^2\ne0.
\]
GapP is closed under products and exponential sums, so one nonzero GapP test
absorbs the existential projection.

For the third identity, substitute the first one.  Consider an
$\mathsf{NP}^{(\mathsf{NP}^{\mathsf{PP}})}$ machine and choose an
$\exists\!\cdot\!\mathsf C_{=}\mathsf P$ language for its oracle.  Guess an
accepting outer path, the complete adaptive oracle transcript, and a witness
for every reported yes-answer.  The transcript fixes all queries.  A reported
yes-answer leaves an exact-counting test for its guessed witness; a reported
no-answer universally quantifies over every possible witness and requires the
corresponding exact-counting test to fail.  One universal string selects the
query and candidate witness.  PP closure under complement and polynomial
truth-table reductions combines the transcript-consistency and exact-counting
conditions, so correctness of the guessed transcript is in
$\operatorname{co}(\mathsf{NP}^{\mathsf{PP}})$.  The outer guess gives the
left-hand projection class.  Conversely, an NP machine guesses the projected
string and queries the $\operatorname{co}(\mathsf{NP}^{\mathsf{PP}})$
predicate; replacing an oracle by its complement does not change adaptive
oracle access.  This gives
$\mathsf{NP}^{(\mathsf{NP}^{\mathsf{PP}})}$.
\end{lemma}

\section{Approximation Transfers}

\subsection{Proof of Corollary~\ref{cor:succinct-input-length-inapproximation}}

\begin{proof}
Use $\varphi=\neg\Phi$ for a 3-CNF $\Phi$ with at most four occurrences
per variable, an NP-complete restriction~\cite{tovey1984simplified}.
There are $O(n)$ clauses, giving $\Lambda=O(n\log n)$ in binary gate-and-wire
encoding and $p=\Omega(\Lambda/\log\Lambda)$. A ratio
$o(\Lambda/\log\Lambda)$ returns a proper interface exactly when $\Phi$ is
unsatisfiable. Finally $\Lambda^{1-\varepsilon}=o(\Lambda/\log\Lambda)$.
\end{proof}

\subsection{Proof of Theorem~\ref{thm:minimum-decisiveness-approximation-threshold}}

\begin{proof}
Suppose first that $\varphi$ has a satisfying assignment.  On a fiber with
the gate value $g$ fixed, the signed utility gap has the sign $(-1)^g$ and,
after omitting its common positive factor, magnitude at least
\[
 M\,|\{x:\varphi(x)=1\}|-|\{x:\varphi(x)=0\}|
 \;\ge\;2^{n+1}-2^n>0.
\]
Thus the gate singleton is decisive.  The empty interface is not decisive,
because the contributions at $g=0$ and $g=1$ cancel state by state.  Hence the
optimum is one.

If $\varphi$ is unsatisfiable, then $d_\varphi=\chi$ is the parity character
on all $p$ coordinates.  Every proper interface leaves some parity bit hidden;
pairing the two values of that bit makes every corresponding fiber gap zero.
The full interface is decisive because $\chi$ is pointwise nonzero.  Its size
$p$ is therefore optimal.

The full interface always supplies the factor-$p$ upper bound.  An algorithm
with ratio below $p$ returns a proper interface exactly on satisfiable inputs,
deciding \textsc{SAT}.  Dummy variables reach every sufficiently large $p$.
For the input-length form, restrict to bounded-occurrence $3$-SAT as in
Corollary~\ref{cor:succinct-input-length-inapproximation}.  Then
$\Lambda=O(n\log n)$ and $p=n+1=\Omega(\Lambda/\log\Lambda)$, so any smaller
asymptotic ratio distinguishes the two cases.  Lean checks that the parity
character is nonzero pointwise and cancels when a bit is hidden.  The
executable audit exhausts every truth table through three variables using
exact integer fiber sums.
\leanmeta{\LHrng{DC}{160}{161}}
\end{proof}

\paragraph{An ETH bound in the coordinate count.}  For an $n$-variable
formula $\varphi$, use $n$ coordinates with domain
$R=\{\mathsf{ref}\}\cup\{0,1\}^n$.  Give \emph{reject} utility zero and
\emph{accept} utility one exactly when every coordinate is either
$\mathsf{ref}$ or a satisfying assignment of $\varphi$; otherwise both actions
tie.  A tautology gives a constant optimizer.  A falsifying assignment placed
in any one coordinate changes the optimizer of the all-reference state, so
every coordinate is relevant.  Thus $k^*\in\{0,n\}$, and checking the empty
interface decides tautology.  The utility circuit uses $n$ copies of
$\varphi$, hence has polynomial size.  Applying this construction to the
negation of a sparse 3-CNF gives, by the
sparsification lemma and the Exponential Time Hypothesis
~\cite{impagliazzo2001which}, no
$2^{o(n)}$ algorithm for succinct \textsc{Sufficiency-Check} on this family.
Its explicit table has $N=(1+2^n)^n$ states, so the explicit polynomial scan
and the succinct lower bound measure the same semantics at different input
scales.  \leanmeta{\LH{DC70}, \LH{DQ27}, \LH{DQ30}}

\subsection{Parameterized and Logarithmic Transfers}

\leanmetapending{\LH{DQ11}}
\begin{corollary}[Parameterized Complexity of Explicit Listed Sufficiency]
\label{cor:listed-minsuff-w2}
On explicit listed tables, \textsc{Restricted-Within-Budget-Sufficiency} is
$\mathsf{W}[2]$-complete when parameterized by the budget $k$.
\end{corollary}

\begin{proof}
Theorem~\ref{thm:set-cover-cardinality-equivalence} is a parameter-preserving
reduction from \textsc{Set Cover}, giving $\mathsf{W}[2]$-hardness
~\cite{downey1995fpt}.  Proposition~\ref{prop:listed-discernibility-hitting-set}
is a parameter-preserving reduction to \textsc{Hitting Set}, giving
membership.  The cited Lean handles check the cardinality-preserving hardness
core and the exact discernibility equivalence; the standard
$\mathsf{W}[2]$ classifications of \textsc{Set Cover} and \textsc{Hitting Set}
are external premises.
\end{proof}

\leanmetapending{\LH{DQ25}, \LH{DQ32}}
\begin{corollary}[Tight Logarithmic Transfer on the Set-Cover Family]
\label{cor:set-cover-approximation-transfer}
On the gadget family, greedy set cover gives a
$1+\ln|\mathcal U|$ approximation.  Conversely, for every
$\varepsilon>0$, a polynomial-time
$(1-\varepsilon)\ln|\mathcal U|$ approximation would imply the same
approximation for set cover and hence $\mathsf P=\mathsf{NP}$
~\cite[Corollary~1.5]{dinur2014analytical}.  Since the translated table has
$N=2|\mathcal U|$ states, the lower bound is equivalently
$(1-\varepsilon)\ln(N/2)$.
\end{corollary}

\begin{proof}
The upper bound is the usual greedy set-cover guarantee.  For the lower bound,
compose a minimum-sufficiency solver with the gadget translation and use
Theorem~\ref{thm:set-cover-cardinality-equivalence}.  The bijection preserves
every feasible cardinality, so it preserves both optimum and approximation
ratio.
\end{proof}

\section{Proofs Completing the Classification Tables}

This supplementary section proves the full-support and sequential rows summarized in
Tables~\ref{tab:regime-matrix} and~\ref{tab:sequential-matrix} through
the preservation equivalences and sequential quantifier accounting below.
The other supplementary sections give auxiliary construction details and finite-core audits.

\subsection{Full-support consequences}
\label{app:proof-full-support-consequences}

\leanmetapending{\LHrng{DC}{95}{96}}
\begin{proposition}[Full-Support Equivalence]\label{prop:stochastic-full-support-equiv}
Assume the distribution $P$ has full support and the action set is nonempty.
Then the following are equivalent: $I$ is sufficient, $I$ preserves on
supported states, $I$ preserves on every declared state, and $I$ is sufficient
with observation coverage.
\end{proposition}

\leanmetapending{\LHrng{DC}{97}{99}}
\begin{proposition}[Quotient Equivalence for a Sufficient Full-Support Interface]\label{prop:stochastic-quotient-equiv}
Assume the distribution has full support, the action set is nonempty, and $I$
is sufficient---equivalently, sufficient with observation coverage or
preserving on every declared state by
Proposition~\ref{prop:stochastic-full-support-equiv}. Then the partition by
the conditional optimal-action set induced by $I$ coincides with the original
optimizer quotient.
\end{proposition}

\begin{proof}[Proof of Propositions~\ref{prop:stochastic-full-support-equiv}--\ref{prop:stochastic-quotient-equiv}
and Theorem~\ref{thm:stochastic-preservation-full-support-variants}]
Full support makes every $I$-fiber positive.  Input sufficiency is therefore
equivalent to preservation on supported states and to sufficiency with
observation coverage, while nonnegative averaging preserves a fiberwise
constant optimizer.  Conversely, preservation on every declared state always
implies input sufficiency by
Theorem~\ref{thm:stochastic-preservation-zero-mass-characterization}.  This
proves Proposition~\ref{prop:stochastic-full-support-equiv}.

Under these equivalent conditions, the conditional and fully informed
optimal-action sets agree statewise.  They induce the same equivalence
relation and hence the same quotient, proving
Proposition~\ref{prop:stochastic-quotient-equiv}.  Applying the first
equivalence to every candidate input set leaves both within-budget queries
unchanged.  For a selected observation, full support supplies the support
witness, so optimizer constancy is equivalent to both some-observation
predicates.  The complexity claims then follow from
Theorems~\ref{thm:minsuff-conp} and~\ref{thm:anchor-sigma2p}.
\end{proof}

\subsection{Bounded and unbounded sequential alternation}
\label{app:proof-sequential-ladder}

\begin{proof}[Full proof of Corollary~\ref{thm:sequential-bounded-ladder}]
An arena with at most $d$ blocks begins either existentially or universally.
Its threshold predicate $V_h(v)\ge c$ is therefore in $\Sigma_d^P$ or
$\Pi_d^P$; a $\Sigma_d^P$ oracle handles either polarity. Let $B$ be a
polynomial bound on the numerator and denominator bit lengths output by the
cutoff-payoff circuit. Every root value equals one of those rational payoffs,
lies in $[-2^B,2^B]$, and any two distinct such rationals differ by at least
$2^{-2B}$. Binary search with dyadic thresholds narrows a root value to an
interval of width below $2^{-2B-2}$ using $O(B)$ threshold-oracle calls. Two
such intervals determine whether two root values are equal and, if not, their
order. Comparing the polynomially many root-action values, and hence two
complete optimizer sets, therefore lies in
$\Delta_{d+1}^P=\mathsf P^{\Sigma_d^P}$.

Failure of proposed-input sufficiency is witnessed by a pair of agreeing root
states with different optimizer sets. Guessing that pair above the
$\Delta_{d+1}^P$ comparison puts failure in $\Sigma_{d+1}^P$, and therefore
sufficiency in $\Pi_{d+1}^P$. On a product space, the minimum sufficient set
is the relevant-coordinate set. A relevance witness is one state pair
differing in a single coordinate whose optimizer sets differ. A no-instance
for budget $k$ supplies $k+1$ such coordinate witnesses, again a
$\Sigma_{d+1}^P$ certificate, so the within-budget query is in
$\Pi_{d+1}^P$. The some-observation query guesses $s_0$ and universally
excludes a disagreeing state in its fiber, which gives $\Sigma_{d+2}^P$.

For hardness, start with the standard bounded-QBF complete problems
~\cite{stockmeyer1976polynomial}. For the first two queries, put the leading
universal assignment $x$ in the root state and leave only the following $d$
quantifier blocks inside the arena. A baseline-tagged root always prefers
\textsc{Enter}; the formula-tagged root prefers \textsc{Enter} exactly when
the inner arena is true. The empty coordinate set, and equivalently the
zero-budget within-budget query, is sufficient exactly when
$\forall x\,q_x$ holds. For the some-observation query, retain coordinates
encoding $z$ and hide $(x,\text{tag})$. A $z$-fiber has a constant optimizer
exactly when $\forall x\,q_{z,x}$ holds, so a qualifying observation exists
exactly when $\exists z\,\forall x\,q_{z,x}$ holds. The wrappers add no arena
position and therefore do not increase its $d$ internal blocks. Each
reduction maps every source instance into the promised target domain, so
completeness holds under promise-preserving many-one reductions.

Explicitly, let \textsc{Leave} terminate with payoff $0$ at both tags. At a
baseline tag, \textsc{Enter} terminates with payoff $1$; at a formula tag it
enters the Boolean-payoff arena for $q_x$. The optimizer is therefore
$\{\textsc{Enter}\}$ at every baseline state, and at a formula state it is
$\{\textsc{Enter}\}$ when $q_x$ is true and
$\{\textsc{Enter},\textsc{Leave}\}$ otherwise. This realizes the claimed
optimizer correspondence using only Boolean terminal and cutoff payoffs.
Lean checks both optimizer correspondences and the zero-budget equivalence.
\end{proof}

\begin{proof}[Full proof of Corollary~\ref{thm:sequential-pspace}]
Depth-first backward induction stores one polynomial-size position per
recursion frame and therefore uses polynomial space. The outer state-pair,
coordinate-set, or some-observation guesses do not exceed polynomial space
because NPSPACE equals PSPACE. For hardness, the prenex-TQBF arena uses
controller maxima for existential quantifiers and environment minima for
universal quantifiers. Its constant-size two-root wrapper has a common
optimizer under the empty observation exactly when the formula is true. The
same instance gives the some-observation lower bound because the sole
empty-interface fiber is the full state space, and gives the within-budget
lower bound at budget zero.
\end{proof}

\section{Proofs for the Headline Classifications}

The article states the headline results and explains their mechanisms. These proofs expand the oracle transcripts, selector construction and product completion. Supplement~S4 records the Lean and executable audits.

\subsection{Fixed-block normalization}
\label{app:proof-nested-normalization}

Weak and strict majority are polynomially interreducible.  Weak majority is
in PP because $c\ge 2^{n-1}$ iff $2c-2^n+1>0$.  Conversely, for $n\ge1$ add
bits $r,t$ and form
\[
 (\neg r\wedge\psi(y))
 \ \vee\
 (r\wedge\neg t\wedge y\notin\{0^n,1^n\}).
\]
It has $2c+2^n-2$ satisfying assignments out of $2^{n+2}$ and meets the weak
threshold exactly when $c>2^{n-1}$; padding handles an empty source block.

\begin{proof}[Full proof of Lemma~\ref{lem:nested-pp-oracle-normal-form}]
The machine-checked artifact records the finite existential--universal witness
schema used below.  The oracle-machine completeness argument is proved here.
Membership is immediate after guessing $z$: the remaining predicate
$\forall x\,\mathrm{MAJ}_y\,\varphi(z,x,y)$ belongs to
$\operatorname{co}(\mathsf{NP}^{\mathsf{PP}})$.  A nondeterministic oracle machine has the same
power with an oracle language or its complement, so
$\mathsf{NP}^{(\operatorname{co}(\mathsf{NP}^{\mathsf{PP}}))}
=\mathsf{NP}^{(\mathsf{NP}^{\mathsf{PP}})}$.

For hardness, first normalize an arbitrary oracle language
$K\in\operatorname{co}(\mathsf{NP}^{\mathsf{PP}})$.  Let $M_K$ be a universal
nondeterministic polynomial-time machine---it accepts exactly when every
nondeterministic branch accepts---with adaptive access to a fixed $\mathsf{PP}$
language $B$.  For an input $q$, encode in a polynomial-length
string $u$ one branch of $M_K$ and one proposed answer bit for every oracle
query on that branch.  The proposed answers determine the adaptive queries:
the $i$th query is reconstructed by simulating $M_K$ with the first $i-1$
recorded answers.

For a recorded yes answer to query $v_i$, consistency is the predicate
$v_i\in B$.  For a recorded no answer, consistency is
$v_i\notin B$.  Both predicates lie in $\mathsf{PP}$ because
$\mathsf{PP}$ is closed under complement.  Their polynomially many
conjunctions and disjunctions remain in $\mathsf{PP}$ by closure under
polynomial-time truth-table reductions
~\cite{beigel1991pp,fortnow1996pp}.  Define $C_K(q,u)$ to accept when either
some recorded oracle answer is inconsistent or the simulated branch ends in
an accepting state.  A consistent transcript is the unique genuine oracle
transcript for its branch.  Therefore
\begin{equation}
\label{eq:coNP-pp-transcript}
 q\in K
 \quad\Longleftrightarrow\quad
 \forall u\;C_K(q,u),
 \qquad C_K\in\mathsf{PP}.
\end{equation}
Malformed branch or transcript encodings are assigned the accepting value and
do not affect the universal condition.

Now let $M$ be a nondeterministic polynomial-time machine with adaptive oracle
access to $K$.  The existential block $z$ records a nondeterministic branch of
$M$, the proposed answer transcript for its polynomially many $K$-queries, and
a counterwitness $u_i$ for every query declared outside $K$.  The transcript
again determines each query $q_i$ from the earlier recorded answers.  The
deterministic checks require that the simulated path is well formed and
accepting.  For every no answer they also require
$\neg C_K(q_i,u_i)$, which certifies $q_i\notin K$ by
\eqref{eq:coNP-pp-transcript}.

The universal block $x$ encodes a pair $(i,u)$.  When the $i$th transcript
answer is yes, the verifier requires $C_K(q_i,u)$.  When the index is malformed
or the answer is no, that universal test accepts.  Thus quantification over
$x$ checks every universal witness for every declared yes answer, while the
counterwitnesses for declared no answers remain in $z$.  The polynomially many
path, no-answer, and selected yes-answer checks form one
$\mathsf{PP}$ predicate $D(w,z,x)$ by complement and truth-table closure.
Consequently,
\begin{equation}
\label{eq:nested-oracle-transcript}
 w\in L(M^K)
 \quad\Longleftrightarrow\quad
 \exists z\;\forall x\;D(w,z,x).
\end{equation}

Use a uniform $\mathsf{PP}$ machine in the fixed-random-tape normal form:
on an input of length $\ell$, it receives exactly $p(\ell)$ independent fair
bits, for one polynomial $p$.  Early halting leaves the unused suffix ignored.
Unroll this machine for
$D$ into one Boolean circuit $\varphi_w(z,x,y)$.  The strings $(z,x)$ remain
ordinary free inputs, and $y$ is exactly the computation-path tape; no Tseitin
or gate variables are counted.  Thus, uniformly in $w$,
\[
 D(w,z,x)
 \quad\Longleftrightarrow\quad
 \mathrm{MAJ}_y\,\varphi_w(z,x,y).
\]
If the source machine uses strict majority, apply the polynomial normalization
above to its accepting-path predicate.  The resulting circuit uses the weak
convention required by the fiber gadget.  The lengths of $z$ and $x$ are fixed
by $|w|$, so the tape length is independent of their values.  More generally,
if a fixed-tape predicate $C_{w,\alpha}$ uses $\nu$ bits and is padded to
$N\ge\nu$, write
$\#C_{w,\alpha}:=|\{u:C_{w,\alpha}(u)=1\}|$.  One polynomial clock gives
$N$ polynomial in $|w|$.  Padding with ignored bits $v$ gives
\begin{equation}
 C'_{w,\alpha}(u,v):=C_{w,\alpha}(u),
 \qquad |u|=\nu,\quad |v|=N-\nu,
 \label{eq:uniform-majority-padding}
\end{equation}
so that
\[
 \#C'_{w,\alpha}=2^{N-\nu}\#C_{w,\alpha},
 \qquad
 \#C'_{w,\alpha}\ge 2^{N-1}
 \Longleftrightarrow
 \#C_{w,\alpha}\ge 2^{\nu-1}.
\]
Thus the compiled majority block has one polynomial length
$|y|=N$, independent of the particular values of $(z,x)$.  This fixed
length is what later turns every selected fiber into the same weak-majority
comparison.
All branch, transcript, counterwitness, and circuit encodings have polynomial
length.  Padding supplies nonempty existential, universal, and majority blocks
without changing the predicate.  Substituting this equivalence into
\eqref{eq:nested-oracle-transcript} gives the required polynomial-time
many-one reduction to \eqref{eq:nested-fixed-block-source}.
\end{proof}

\section{Full Proofs for the Multistep Stochastic Classification}
\label{app:multistep-pomdp-proofs}

This appendix supplies the complete continuation, sign-separation, and wrapper
arguments summarized in Section~\ref{sec:multistep-pomdp}.

\begin{proof}[Full proof of Lemma~\ref{lem:pomdp-polynomial-space}]
A history has length $O(H)$ and a polynomial-length code.  Traverse the
action--observation tree depth first, using unnormalized continuation values.
For a projected history $g$, let $m(g)$ be its absolute probability, summed
over its full lifts under the recorded action sequence treated as fixed
interventions, and put $W(g)=m(g)V_I(g)$, where $V_I(g)$ is its optimal
conditional continuation value.  No policy-selection probability is
included.  If $\mathcal L(g)$ is the set of full lifts and $\alpha(h)$ is the
initial-and-transition mass of lift $h$, then, with the next action fixed to
$a$,
\[
 W(g,a)=
 \sum_{h\in\mathcal L(g)}\alpha(h)\bar r(h,a)
 +\sum_o W(gao),
 \qquad W(g)=\max_{a\in A}W(g,a).
\]
The terms in the second sum partition the next full-history lifts by their
projected observation.  Dividing by the common positive $m(g)$ gives the
ordinary conditional Bellman equation, so the argmax is exactly
$\Opt_I(g)$.  If $m(g)=0$, every term is zero and the argmax is $A$, as
required by the zero-mass convention for preservation on every declared
history.
At each projected history the maximizing action is chosen independently, so
the recursive maximizers assemble into one admissible history-dependent
policy rather than a different policy for each hidden trajectory.

Let $q$ bound a state-name length and let the numerator and denominator of
every input dyadic have at most $b$ bits.  A dyadic denominator $2^e$ then has
$e\le b$, so a depth-$H$ path product has
denominator exponent $O(Hb)$.  At most $2^{O(Hq)}$ full paths contribute to
one projected history, so putting them over their common power-of-two
denominator adds only $O(Hq)$ numerator bits.  Rewards and the sum over at
most $H$ stages change this by $O(b+\log H)$.  Hence every exact accumulator
has $O(H(b+q)+b)$ bits, polynomial because $H$ is unary.  Enumerate consistent
hidden trajectories and projected observations with $O(Hq)$-bit counters and
evaluate the recurrence depth first.  Only the current history, counters,
rational accumulators, and $O(H)$ recursion frames are stored.  Repeating the
computation with the first action fixed gives every next-action value and
hence the complete argmax set.

A declared or supported full lift is guessed one state at a time and checked
by the corresponding circuits.  Thus a supported-history preservation
violation, a declared-history observation-coverage mismatch, a non-singleton
optimizer, or a declared projected history without a supported lift can all be
recognized in polynomial space.  Closure of PSPACE under complement completes
the stated tests.  This is the branchwise compressed-POMDP argument cited in
the main proof, applied here to exact optimizer sets and history lifts.
\end{proof}

\begin{proof}[Full proof of Lemma~\ref{lem:pomdp-sign-core}]
Let $\Phi=Q_1x_1\cdots Q_nx_n\,\varphi(x_1,\ldots,x_n)$ and let $u$ be the
number of universal variables.  A nonterminal state is $(i,\alpha)$, where
$i\in\{0,\ldots,n-1\}$ and $\alpha\in\{0,1\}^i$ is the assignment prefix.
Both have polynomial-length binary encodings.  If $Q_{i+1}=\exists$, action
$b\in\{0,1\}$ moves deterministically to $(i+1,\alpha b)$.  If
$Q_{i+1}=\forall$, both actions induce the same fair transition to
$(i+1,\alpha0)$ and $(i+1,\alpha1)$.  After the last bit, a satisfying
assignment pays $+1$ and a falsifying assignment pays $-2^u$.  Circuits compare
the stage, copy or append the prefix, evaluate $\varphi$, and output the
indicated dyadic probability or integer reward; their size is polynomial in
$|\Phi|$.

A true QBF has a history-dependent policy for which every terminal assignment
satisfies, hence value $1$.  If the QBF is false, every policy has a falsifying
universal branch of probability at least $2^{-u}$, so its value is at most
\[
 (1-2^{-u})-2^{-u}2^u=-2^{-u}.
\]
Write $\Delta=2^{-u}$ for this unperturbed gap.

It remains to make ties below the wrapper harmless.  At variable stage $i$,
add immediate bonus $\epsilon_i=\Delta 4^{-(i+1)}$ to action one and zero to
action zero.  The current bonus
exceeds the sum of every later bonus, while the total perturbation is below
$\Delta/3$.  Every nonzero unperturbed continuation-value difference is an
integer multiple of $2^{-r}$, where $r$ is the number of universal choices
remaining at that prefix, and hence has magnitude at least
$2^{-u}=\Delta$.  The maximum difference contributed by all later bonuses is
at most their scalar sum, strictly below the current bonus.  The perturbation
therefore preserves every nonzero sign, and an unperturbed tie is broken at
the first action difference.  This prefix argument applies to every declared
history, including one reached after an interventionally recorded suboptimal
action; it does not assume the earlier actions followed an optimal policy.
At a universal state the unperturbed action values coincide, so the current
bonus selects action one.  At an existential state, either the unperturbed
values differ and their sign is preserved, or the first current-or-later
action difference selects one branch.  Thus every decision history has a
unique optimizer.  The final value is at least $1$ in the true case and at most $-2\Delta/3$ in
the false case, so the lemma holds with $\delta=\Delta/2$.  All rewards still
have polynomial bit length.  Declaring exactly the nonzero QBF transitions
makes every declared history supported.
\end{proof}

\begin{proof}[Full proof of Theorem~\ref{thm:pomdp-exact-pspace}]
Lemma~\ref{lem:pomdp-polynomial-space} gives the proposed-input upper bounds.
The within-budget query first guesses a polynomial-length coordinate mask; NPSPACE
equals PSPACE, so all within-budget versions remain in PSPACE.  To check input
sufficiency with observation coverage, the traversal compares optimizers at
every declared full history and checks that every declared projected history
has a supported lift.  To check preservation on every declared state, the same
traversal accepts an unsupported projected history only when every declared
full lift has optimizer set $A$; this adds one polynomial-space history-lift
and optimizer test.
The lower-bound instances below have no null declared histories, so their
sufficiency-with-coverage hardness transfers unchanged.

For hardness, start from Lemma~\ref{lem:pomdp-sign-core}.  The names below are
state-dependent roles for the same two action symbols: action one serves as
\textsc{Plan} or \textsc{Continue}, and action zero as \textsc{Safe} or
\textsc{Quit}.  A branch described as terminating pays its stated reward and
enters a reserved absorbing state until the common horizon.  Both actions
self-loop there; action one pays zero and action zero pays $-1$ at every
remaining decision stage.  The optimal continuation value is therefore zero,
and its optimizer is always $\{1\}$, including after suboptimal interventions.
The stage-indexed circuits implement this padding with polynomial size.
Reserve its state code in the existing padded public payload and reset the
root tag to zero; no selectable coordinate is added.  Under the preservation
interface below the absorbing code is distinct from every live code.  Under
full observation it is also distinct, as required for the true decisiveness
instances.  Padding therefore preserves all root values and creates no
historywise failure in those witness interfaces.

For preservation, use $p$ public state bits and one root-tag
bit, retain all public bits, and omit the tag.  Give positive initial mass to
two main roots with the same reserved public code and opposite tag.  At the
baseline root, \textsc{Plan} terminates with reward one and \textsc{Safe}
terminates with reward zero.  At the test root, \textsc{Safe} terminates with
reward zero and \textsc{Plan} enters the sign-separated QBF core with zero
entry reward.  If $\Phi$ is true, both full root optimizer sets are
$\{\textsc{Plan}\}$.  If $\Phi$ is false, their optimizer sets differ inside
one projected root history.  After the root, the baseline branch terminates.
The test-root transition resets the tag to zero, and each subsequent QBF state
stores its phase, stage, and quantified prefix injectively in reserved public
codes.  Thus
the tag is constant throughout the only nonterminal branch, so full and
projected continuation problems are identical.  The same wrapper therefore
proves both root and historywise proposed-input hardness.

To force this interface in the within-budget query, pad the QBF code to $m\ge3$
payload bits and add two public marker bits.  QBF states have marker $00$.
For each payload direction $j$, add a positive-mass pair of one-step
calibration roots on the cube edge with marker $10$ from cyclic unit vector
$e_{j+1}$ to $e_{j+1}\oplus e_j$, with subscripts modulo $m$.  Add the marker
edges
\[
 (01,0^m)\longleftrightarrow(11,0^m),
 \qquad
 (10,1^m)\longleftrightarrow(11,1^m).
\]
Every calibration root has root-tag value zero; its two endpoints differ only
in the displayed public-coordinate direction.
The edges are disjoint: payload-edge endpoints have Hamming weight one or two
under marker $10$, while the marker edges use payloads $0^m$ and $1^m$; the
$m\ge3$ condition prevents a collision.  Give action zero terminal reward one
at one endpoint of each edge and action one terminal reward one at the other,
with the remaining terminal reward zero.  Omitting any public coordinate
therefore merges two full roots
with different optimizer sets, which no projected optimizer can equal at both
states.  Retaining every public coordinate separates all calibration roots.
With budget $p$ among the $p$ public coordinates and the tag, every successful
interface must retain the public coordinates and omit the tag.  Minimum
preservation is equivalent to truth of $\Phi$.

There are only polynomially many initial calibration and main roots.  If their
number is $N$, choose $L$ with $2^L>N$, assign mass $2^{-L}$ to the first
$N-1$ roots, and assign the remaining mass
$1-(N-1)2^{-L}$ to the last.  These positive polynomial-bit dyadics sum to one.
The QBF core declares exactly its positive transitions.  Hence every declared history
has a supported realization, so the same construction proves the two
sufficiency-with-coverage lower bounds and removes null histories from the all-state
variant.

For decisiveness, use one supported root.  \textsc{Safe} terminates with value
zero.  \textsc{Plan} enters a gate whose two actions are \textsc{Quit}, of
value zero, and \textsc{Continue}, of value $V_\Phi$.  Thus the plan branch
has value $\max\{0,V_\Phi\}$.  Sign separation makes the gate optimizer unique
and gives
\[
 \Phi\text{ true}\Longrightarrow
 Q(\textsc{Plan})>Q(\textsc{Safe}),
 \qquad
 \Phi\text{ false}\Longrightarrow
 Q(\textsc{Plan})=Q(\textsc{Safe})=0.
\]
Every later decision is unique by Lemma~\ref{lem:pomdp-sign-core}.  Thus the
full interface is historywise decisive exactly when $\Phi$ is true.  If
$\Phi$ is false, the root tie survives every coordinate projection.  Any
policy using projected observations can be simulated under full observation
by ignoring the extra information, so hiding coordinates cannot raise the
\textsc{Continue} value above the fully observed value $V_\Phi<0$.
Conversely, \textsc{Quit} is available under every projection and keeps the
plan branch at value at least zero.  Its projected value is therefore exactly
zero, and \textsc{Plan} and \textsc{Safe} remain tied at the root.  If $\Phi$
is true, the full coordinate
set is a within-budget witness.  This proves proposed-input and within-budget hardness at
both boundaries, taking the budget to be the number of state coordinates.

Every state, transition, reward, declaration, and initial-mass test used by
the wrappers is a Boolean or dyadic circuit over the phase tag, marker bits,
stage counter, and assignment prefix.  Their number and encoded size are
polynomial in $|\Phi|$.  Each reduction therefore emits a normalized valid
instance of Definitions~\ref{def:pomdp-succinct-interface}--\ref{def:pomdp-certification-contracts}, rather than relying
on a promise about malformed inputs.
\end{proof}

\subsection{Horizon-two inheritance with individual selectable bits}
\label{app:pomdp-horizon-two-bits}

\begin{proof}[Within-budget-decisiveness part of Proposition~\ref{prop:pomdp-horizon-two-inheritance}]
The two-stage lift in the main text preserves Boolean coordinates directly.
For within-budget decisiveness, the source construction uses larger coordinate
values.  We enforce whole-block selection using individual bits, with the
first stage forcing the identifiers that isolate the last-stage tests.

Take the fixed-block source
$\exists z\,\forall x\,\mathrm{MAJ}_y\,\varphi(z,x,y)$ from
Lemma~\ref{lem:nested-pp-oracle-normal-form}, with all three block lengths
positive.  Write $n=|z|$, $T=|z|+|x|+|y|+1$, and $L=T+1$.
For each $j\in[n]$, introduce two length-$L$ bit blocks $Q_{j,0}$ and
$Q_{j,1}$.  A main state $m(z',x,y,r)$ has the same token contents as in
Appendix~\ref{app:proof-stochastic-minimum}:
\[
 Q_{j,b}(m)=
 \begin{cases}
  (0,\operatorname{pad}_T(x)),&z'_j=b,\\
  (1,z',x,y,r),&z'_j\ne b.
 \end{cases}
\]
The leading bit separates a common token from an exposed identifier.
Every main state has an exposed block, so its complete bit tuple identifies
$(z',x,y,r)$ uniquely.  Selecting $Q_{j,z_j}$ in full for every $j$ groups
the $z'=z$ states by $x$ and isolates every other main state.

Put $F=nL^2$ and $p=\lceil\log_2(F+1)\rceil$, and prepend $p$ public
identifier bits.  Main states have public identifier zero.  For each
$(j,i,\ell)\in[n]\times[L]\times[L]$, choose a distinct nonzero public
identifier and add two forcing states.  Their candidate bits are respectively
all zero and zero except at positions $Q_{j,0}[i]$ and $Q_{j,1}[\ell]$.
In action order $(\mathit{accept},\mathit{hold}_L,\mathit{hold}_R)$, their
utility vectors are $(0,1,0)$ and $(0,0,1)$.  There are $2F$ forcing states.
Set
\[
 e=\lceil\log_2(4F)\rceil,\qquad
 \delta=2^{-e},\qquad
 \mu=\frac{1-2F\delta}{2^T}.
\]
Each forcing state has mass $\delta$; each main state has mass $\mu$.
These dyadic masses are positive and sum to one, since $2F\delta\le1/2$.
The main utilities are those of the three-action source:
\[
 U(\mathit{accept},m)=\mathbf1[\varphi(z',x,y)],\quad
 U(\mathit{hold}_L,m)=\theta+\sigma_r,\quad
 U(\mathit{hold}_R,m)=\theta-\sigma_r,
\]
where $\theta=1/2-2^{-(|y|+1)}$, $\sigma_0=-\varepsilon$,
$\sigma_1=\varepsilon$, and $\varepsilon=2^{-(|y|+3)}$.

At stage zero, the declared initial states are
$(u,0^{2nL})$ for all $u\in\{0,1\}^p$, with uniform mass $2^{-p}$.
Every initial state and every action transition to the same normalized
distribution over the main and forcing states.  The immediate reward vector
is
\[
 (-2,\chi(u),-\chi(u)),\qquad \chi(u)=(-1)^{\operatorname{parity}(u)}.
\]
At stage one, action $a$ pays the utility just defined and transitions
deterministically to one terminal state.  The stage-indexed circuits use the
same $p+2nL$ state bits at both decisions; no phase bit is selectable.

Fix any interface $I$.  Conditional on any initial observation and any first
action, the distribution at stage one is the same.  Optimal continuation
therefore contributes a common term $V_I$ to the three first-action values.
If a public bit is omitted, conditional parity is balanced: candidate bits
are constant initially, and summing over the omitted public bit cancels
$\chi$.  The two holds then tie at $V_I$, above accept at $V_I-2$.
If all public bits are retained, parity is known and one hold is uniquely
optimal.  Thus historywise decisiveness forces all $p$ public bits.
The transition remains the same after a recorded suboptimal action, and its
initial reward is past at the last decision.  These observations cover every
declared history, not just histories of optimal policies.

Set the individual-bit budget to $p+nL$.  Once the public bits are retained,
each forcing pair is isolated from every other forcing pair and all main
states.  Omitting both $Q_{j,0}[i]$ and $Q_{j,1}[\ell]$ leaves that pair in
one fiber, where equal masses give the two holds equal positive scores and
accept score zero.  A decisive mask consequently covers every edge of the
complete bipartite graph between $Q_{j,0}$ and $Q_{j,1}$.  Such a cover
contains an entire side: if each side had an omitted vertex, the edge between
those two vertices would be uncovered.  At least $L$ bits are needed per
variable.  The remaining budget $nL$ therefore forces exactly one whole side
per variable, with no extra bits.  Decode these choices as $z$.

For this canonical mask, the forcing states are singletons and have unique
optimizers.  Every off-assignment main state is also a singleton; its optimum
is unique because $0<\theta-\varepsilon<\theta+\varepsilon<1$.
The selected fiber indexed by $x$ hides $(y,r)$.  After removing the positive
common factor $2\mu$, its scores are
\[
 c_x,\quad 2^{|y|}\theta,\quad 2^{|y|}\theta,
 \qquad c_x=|\{y:\varphi(z,x,y)\}|.
\]
Since $2^{|y|}\theta=2^{|y|-1}-1/2$, integrality makes the optimum unique
exactly when $c_x\ge2^{|y|-1}$.  All last-stage observations are decisive
exactly when $\forall x\,\mathrm{MAJ}_y\varphi(z,x,y)$.
Together with the first-stage test and the budget argument, this proves both
directions of the reduction.

Finally, $L$, $F$, $p$, and the state-name length are polynomial in the source
length.  To recognize a main state, decode an exposed block and check every
other block.  To recognize a forcing state, decode its public identifier and
check either of its two specified bit patterns.  Initial validity, these
last-stage validity tests, the formula evaluation, rewards, and dyadic masses
all have polynomial-size circuits.  Their denominator exponents are
polynomial, and evaluating a mass requires no satisfying-assignment count.
All declared initial states and permitted transitions have positive mass;
each transition row is normalized.  Positivity concerns these declared
carriers, not unused Boolean names.  The construction uses the same three
actions throughout and satisfies the individual-bit model in
Definition~\ref{def:pomdp-succinct-interface}.
\end{proof}

\section{Input Preservation in Existing AI Models}\label{sec:four-instances}

Think of the sensor readings as table columns and the required action set as
the row label. Removing columns must leave that label recoverable. An
inclusion-minimal set of preserving columns is a \emph{reduct} in rough-set
theory~\cite{pawlak1982rough}. Finding its minimum size has different costs
when the auditor receives listed rows, circuits defining allowed rows, or
every combination.

\subsection{One Reduct Problem, Three Presentations}

Fixing enough sensor readings can also make every remaining decision
constant. Backdoor methods describe this by specifying a class of remaining
problems. Formally, label each row by $Y(s)=\Opt(s)$, and let
$\mathcal T_0$ contain problems with constant optimizer. A strong backdoor
places every problem remaining after its coordinates are fixed in that class.

\Needspace{10\baselineskip}
\leanmetapending{\LHrng{UB}{2}{4}}
\begin{proposition}[Coordinate-Projection Identifications]
\label{prop:coordinate-projection-calibrations}
For every product-coordinate decision problem and coordinate set $I$, the
following three conditions are equivalent to sufficiency of $I$:
\begin{enumerate}
\item projection to $I$ preserves the decision label in the
optimizer-labelled decision table, so the inclusion-minimal sufficient sets
are exactly its reducts;
\item $I$ is a strong backdoor to $\mathcal T_0$; and
\item the exact feature-selection target $Y=\Opt$ factors through projection
to $I$.
\end{enumerate}
\end{proposition}

\begin{proof}
All three conditions say that $\Opt$ is constant on every $I$-fiber.
Apply Theorem~\ref{thm:master-factorization} with $\pi(s)=s_I$.
\end{proof}

Every sufficient set contains a reduct of no larger cardinality, so the
budgeted queries coincide. The explicit listed-table problem is NP-complete
and W[2]-complete by budget
(Theorem~\ref{thm:set-cover-cardinality-equivalence} and
Corollary~\ref{cor:listed-minsuff-w2}; \cite{wong1985optimal}).
Circuit-defined restricted states give $\Sigma_2^P$-completeness
(Theorem~\ref{thm:listed-minsuff-sigma2p}); product closure leaves one least
reduct and the coNP audit (Theorem~\ref{thm:minsuff-conp}).

The following comparisons distinguish this common target from choosing one
optimal action or preserving numerical values.

\subsection{Which Responses and States Are Preserved}
Recovering the complete optimal-action set is stronger than selecting one
common optimal action: $\{a,b\}$ and $\{a\}$ permit the same choice but
require different responses. Preserving the full $Q^*$ vector---the optimal
continuation value of every action---also preserves which actions are best,
but requires their numerical values to remain unchanged.
Appendix~\ref{app:reviewed-model-comparisons} proves
these comparisons, including the selector, backdoor, and MDP separations.

For the deterministic label $Y=\Opt$, a Markov blanket requires the retained
inputs to determine $Y$ on the probability support~\cite{koller1996toward}.
This is equivalent to optimizer constancy within each supported fiber and
hence to preservation on supported states: positive averaging preserves a
common optimizer set. Under full support it is input sufficiency; sparse
support can leave incomparable minimal blankets. Theorem~\ref{thm:optimizer-markov-blanket}
and Corollaries~\ref{cor:markov-blanket-support-boundary}--\ref{cor:markov-blanket-uniqueness}
give the precise equivalences and counterexample in the main article.

\subsection{Observation Reduction in One-Step POMDPs}
\label{sec:worked-examples}

An agent sees a reading that may hide part of the state, chooses an action,
and receives a reward. With deterministic observations and no later
transition, this one-step \emph{partially observable Markov decision process}
(POMDP) is a finite Bayesian decision problem.

\begin{definition}[One-step POMDP]\label{def:one-step-pomdp}
A one-step POMDP is a tuple $(S,A,O,p,o,r)$ where $S,A,O$ are finite sets of
states, actions, and observations, $p\in\Delta(S)$ is a prior distribution,
$o:S\to O$ is the observation map, and $r:A\times S\to\mathbb{R}$ is the
immediate reward function.
\end{definition}

Write $\Opt_{\mathrm{full}}(s):=\arg\max_{a\in A}r(a,s)$ for the
optimal-action set under full information.  For an observation map
$\psi:S\to Y$ and $y\in Y$, set
\[
\Opt_{\psi}(y):=\arg\max_{a\in A}
 \sum_{s:\psi(s)=y}p(s)r(a,s).
\]
This is the conditional optimizer on positive observations and $A$ on null
ones. A coarsening $g:O\to O'$ gives $\phi=g\circ o$; the classified coordinate
instances have full original observation and $\phi(s)=s_I$.

\Needspace{9\baselineskip}
\leanmetapending{\LH{EX1}}
\begin{proposition}[Full-Information Sufficiency as All-State Preservation]
\label{prop:pomdp-to-preservation}
Let $(S,A,O,p,o,r)$ be a one-step POMDP and let $\phi:S\to O'$ be as above.
Assume that the declared reduced observation set is
$O'=\operatorname{range}(\phi)$.
Use the induced finite decision problem with $U(a,s):=r(a,s)$ and state
distribution $p$.  Then $\phi$ is preserving on every declared state if
and only if
\[
  \forall s\in S:\quad
  \Opt_{\mathrm{full}}(s)=\Opt_{\phi}(\phi(s)).
\]
\end{proposition}

\begin{proof}
The pointwise optimizer of the induced decision problem is
$\Opt_{\mathrm{full}}$.  Its optimizer after observing $\phi(s)$ maximizes the
unnormalized fiber score in the preceding display and is therefore
$\Opt_{\phi}(\phi(s))$.
\end{proof}

Coverage separately rejects declarations outside $\operatorname{range}(\phi)$.
Preserving an original partial-observation policy would instead compare
$\Opt_o(o(s))$ with $\Opt_\phi(\phi(s))$, not the full-information optimizer.
Classical succinct planning evaluates the supplied information structure
~\cite{mundhenk2000mdp,lusena2001nonapproximability};
Section~\ref{sec:multistep-pomdp} audits removal through later beliefs and choices.

\section{Comparisons with Other Preservation Targets}
\label{app:reviewed-model-comparisons}
These comparisons distinguish recovery of the complete optimal-action set
from selecting one optimal action, preserving numerical values, and recovering
a deterministic label only on its probability support.

\subsection{Recovering All Optimal Actions versus Choosing One}
\label{sec:correspondence-versus-selector}

Choosing one common optimal action need not recover every tied action.

\begin{definition}[Optimal-Selector Implementation]
\label{def:optimal-selector-implementation}
An abstraction $\pi:S\to Q$ \emph{supports an optimal selector} when some
$g:Q\to A$ satisfies
\[
  g(\pi(s))\in\Opt(s)\qquad\text{for every }s\in S.
\]
\end{definition}

\leanmetapending{\LHrng{KP}{1}{4}}
\begin{remark}[Correspondence Recovery Strictly Implies Selector
Implementation]\label{prop:recovery-strictly-implies-selector}
For finite nonempty $A$, exact recovery of $\Opt$ through $\pi$ implies that
$\pi$ supports an optimal selector.  The converse fails, even for two states
and two actions.
\end{remark}

\begin{proof}
Choose from each recovered nonempty optimizer set and extend arbitrarily
off $\operatorname{range}(\pi)$. Conversely, a constant abstraction selects
$a$ for $\Opt(x)=\{a,b\}$ and $\Opt(y)=\{a\}$ but cannot recover the tie.
\end{proof}

If a tie-breaking rule is declared, that choice becomes the target. Without
one, two state groupings can each allow a common optimal action while their
common coarsening merges states with no action in common. A deterministic
choice from a tie also cannot respect every relabelling of the actions.

\leanmetapending{\LHrng{KP}{5}{10}}
\begin{proposition}[Selector-Only Abstraction Is Generally
Noncanonical]\label{prop:selector-abstraction-noncanonical}
There is a finite decision problem with two incomparable selector-supporting
abstractions but no selector-supporting common coarsening.  Hence the
selector-supporting abstractions need not have a coarsest member.  Moreover,
on a two-action tie, no deterministic selection rule is equivariant under
swapping the action labels.
\end{proposition}

\begin{proof}
Take three states $\ell,m,r$ and two actions $a,b$ with
\[
  \Opt(\ell)=\{a\},\qquad
  \Opt(m)=\{a,b\},\qquad
  \Opt(r)=\{b\}.
\]
Merging $\ell,m$ permits $a$; merging $m,r$ permits $b$. A common
coarsening merges $\ell,r$, which share no optimal action. Swapping $a,b$
fixes their tie set but neither action, precluding an equivariant selection.
\end{proof}

A declared selector has its own kernel quotient, stated in the supplement.

\subsection{Backdoors and Value-Preserving Abstractions}

\leanmetapending{\LHrng{UB}{5}{7}}
\begin{example}[A Larger-Class Backdoor Need Not Be Sufficient]
\label{ex:backdoor-separation}
Let \(S=\{s_0,s_1,s_2\}\) and \(A=\{a,b,c\}\).  Give \(a,c\) utility one
at \(s_0\), give \(b,c\) utility one at \(s_1,s_2\), and give every other
state--action pair utility zero.  Let \(\mathcal T_2\) contain optimizer maps
with at most two attained values.  The empty coordinate set is a strong
backdoor to \(\mathcal T_2\), because the complete residual instance attains
exactly \(\{a,c\}\) and \(\{b,c\}\).  It is not sufficient because those
optimizer sets differ.  The separation is strict: on the same three-state,
three-action state domain an optimizer map may attain three distinct singleton
values and lie outside \(\mathcal T_2\).
\end{example}

If $\mathcal T_0\subseteq\mathcal T$, every sufficient set is a strong
backdoor to $\mathcal T$, but Example~\ref{ex:backdoor-separation} shows that
the converse can fail when $\mathcal T$ allows different optimal-action sets
~\cite{williams2003backdoors,bessiere2013detecting,ganian2017discovering}.

Choosing one constant-optimizer fiber has the weak-backdoor quantifiers of
\textsc{Some-Observation-Sufficiency}, without a satisfiability condition.
Unlike standard FPT backdoor settings~\cite{gaspers2012backdoors}, succinct
access is hard even at budget zero; explicit states permit polynomial scans.

Li, Walsh, and Littman's Definition~3 gives five MDP abstraction relations,
and their Theorem~2 orders them from finer to coarser
~\cite{li2006towards}.  Their $\phi_{Q^*}$ refines the optimizer quotient.
Their $\phi_{a^*}$ instead requires a shared optimal action and equality of
optimal value, whereas the optimizer quotient requires equality of the
complete argmax set and ignores numerical value.  Under unique optimal
actions, $\phi_{\pi^*}$ is the closest exact member of their hierarchy.

Writing $\Opt^*(s)=\arg\max_a Q^*(s,a)$, equality of $Q^*$ vectors gives
equality of $\Opt^*$ and the quotient factorization of
Theorem~\ref{thm:master-factorization}. The converse fails.
\leanmeta{\LHrng{UB}{8}{9}}

\leanmetapending{\LHrng{UB}{10}{12}, \LH{UB14}}
\begin{example}[Strictness of the $Q^*$ Refinement]\label{ex:mdp-separation}
Consider a two-state, two-action deterministic MDP with discount
$\gamma=1/2$.  Action $g$ moves from $s_0$ to $s_1$ and from $s_1$ to $s_0$;
action $b$ self-loops.  Set
\[
 R(s_0,g)=1,\qquad R(s_1,g)=2,\qquad
 R(s_0,b)=R(s_1,b)=0 .
\]
The Bellman solution gives $Q^*(s_0,\cdot)=(8/3,4/3)$ and
$Q^*(s_1,\cdot)=(10/3,5/3)$: $g$ strictly wins at both states, although
their values differ. They merge in the optimizer quotient but separate under
$\phi_{Q^*}$ and reward-preserving bisimulation. Thus the separation persists
with positive discount and state-changing transitions.
\end{example}

Model minimization, bisimulation, and identity-action MDP homomorphisms
preserve reward and transition data sufficient for $Q^*$, so they refine the
optimizer target~\cite{dean1997model,givan2003equivalence,ravindran2004algebraic}.
Approximate and learned variants instead control value loss or behavioral
distance~\cite{ferns2004metrics,abel2016near,zhang2021learning,castro2021mico}.
Here a tolerance is part of the declared response $\mathrm{Adm}_{\varepsilon}$
and that response is preserved exactly. Causal abstractions instead relate
models by preserving declared interventions
~\cite{rubenstein2017causal,beckers2019abstracting}. Optimizer factorization
keeps one model fixed and recovers its complete best-action set from fewer
inputs.

\paragraph{Feature selection}
For target $Y=\Opt$, a Markov blanket concerns the support; input sufficiency
concerns every declared state~\cite{koller1996toward}.
CSP value interchangeability asks which values can replace one another while
preserving satisfiability~\cite{freuder1991interchangeable,bordeaux2008structural}.
Interventional boundary discovery identifies variables needed for causal
prediction~\cite{triantafillou2021causal}.

\section{Claim-to-Lean Finite-Core Mapping}

Each row maps a manuscript claim to the finite Lean declarations cited for
it.  The mapping does not assert whole-theorem mechanization: complexity-class
membership and completeness, approximation and parameterized transfers, and
asymptotic running-time conclusions remain manuscript proofs using named
external results.  In particular, a handle beside a complexity theorem means
``Lean-checked finite semantic or reduction core,'' not ``the complete class
lift is proved in Lean.''

The accompanying supplementary package contains the complete machine-readable
declaration ledger, \texttt{lean\_handle\_ledger.csv}. The cited subset is also
typeset below so that each compact handle can be audited without leaving the
supplement.  Main-article numbers appear directly in the first column;
supplement-only labels retain their lettered appendix prefix.

The map uses three statuses.  \emph{Full} means that Lean proves the complete
finite statement.  \emph{Conditional} means that Lean proves the finite core
and the full manuscript claim also uses a named external or manuscript
premise.  A declaration explicitly named \texttt{informal}, such as the ETH
transfer handle, records an assumption boundary rather than a mechanized
theorem; it is neither full nor silently absorbed into the conditional tier.

\IfFileExists{content/claim_mapping_auto.tex}{%
\noindent Prose rows quote the local annotated passage; reference numbers come from the compiled source documents. Handles support the tagged mathematical claim, not surrounding external observations or recorded executions.\par
\begingroup
\scriptsize
\setlength{\tabcolsep}{3pt}
\renewcommand{\arraystretch}{1.0}
\ifdefined\PaperSupplementTableSetup\PaperSupplementTableSetup\fi
\setlength{\LTpre}{0pt}
\setlength{\LTpost}{0pt}
\ifdefined\PaperGenericLayout
\fontsize{10}{12}\selectfont
\setlength{\LTpre}{\baselineskip}
\setlength{\LTpost}{\baselineskip}
\fi
\begin{longtable}{@{}>{\raggedright\arraybackslash}m{0.65\linewidth}>{\raggedleft\arraybackslash}m{0.30\linewidth}@{}}
\caption{Claim-to-Lean handle mapping.}\label{tab:claim-lean-handle-mapping}\\
\toprule
\textbf{Manuscript claim} & \textbf{Lean handle} \\
\midrule
\endfirsthead
\toprule
\textbf{Manuscript claim} & \textbf{Lean handle} \\
\midrule
\endhead
\endfoot
\bottomrule
\endlastfoot
Theorem 2.6: Kernel Factorization for the Optimizer Map & \LHrng{QT}{5}{7}, \LH{UB1} \\
\midrule
Lemma 2.9: Intersection Closure Gives a Least Member & \LHrng{IC}{1}{4} \\
\midrule
Proposition 2.11: On Products, the Necessary Inputs Suffice & \LH{EF1} \\
\midrule
Corollary 2.12: Full Interface and Smallest Sufficient Interface & \LH{EF1} \\
\midrule
Proposition 2.14: Admissible Outputs as an Optimizer Quotient & \LHrng{QT}{8}{9} \\
\midrule
Corollary 2.15: Positive-Gap Tolerance Band & \LH{EF11} \\
\midrule
Proposition S3.1: Breakpoint and Tail Structure & \LHrng{EF}{2}{3}, \LH{EF10} \\
\midrule
Theorem S3.3: Meet of the Tolerance Family & \LH{EF9} \\
\midrule
Proposition S3.4: Stability Radius & \LHrng{EF}{11}{15} \\
\midrule
Proposition 4.1: Bounded-Width CSP for Violating State Pairs & \LHrng{DQ}{79}{81}, \LH{DQ87} \\
\midrule
Proposition 4.3: Orbit Representatives and the Histogram Certificate & \LHrng{DQ}{83}{85}, \LHrng{SYM}{1}{7} \\
\midrule
Theorem 6.5: One-Coordinate Deletion Is coNP-Complete & \LH{DQ40} \\
\midrule
Theorem 6.6: Finding a Sufficient Set Within Budget Is coNP-Complete & \LH{DQ36}, \LHrng{SK}{4}{8} \\
\midrule
Theorem 6.7: Missing Combinations Keep the Choice Alive & \LHrng{LMS}{1}{2} \\
\midrule
Theorem 6.9: Finding One Observation That Preserves the Answer & \LH{CC4}, \LH{DQ98} \\
\midrule
Theorem S3.7: \textsc{Tolerant-Sufficiency} is coNP-complete & \LH{DQ40}, \LH{EF2} \\
\midrule
Theorem S3.8: The Stability Radius Is coNP-Complete & \LH{EF17} \\
\midrule
Theorem S3.9: \textsc{Robust-Within-Budget-Sufficiency} is coNP-complete & \LH{DQ36}, \LHrng{EF}{1}{3} \\
\midrule
Theorem S3.10: \textsc{Tolerance-Existence} is $\Sigma\_2^P$-complete & \LH{EF18} \\
\midrule
Theorem 6.15: Optimizer Markov-Blanket Equivalence & \LH{UB15} \\
\midrule
Corollary 6.16: When a Markov Blanket Matches Input Sufficiency & \LH{UB16} \\
\midrule
Corollary 6.17: Full-Support Blanket Uniqueness and Sparse-Support Failure & \LHrng{UB}{17}{19} \\
\midrule
Proposition 6.18: Intersection Closure Separates Preservation from Decisiveness & \LHrng{DC}{140}{143} \\
\midrule
Theorem 6.22: Sparse Support Keeps the Blanket Choice Alive & \LHrng{LMS}{1}{2}, \LH{UB15} \\
\midrule
Theorem 6.23: Full Support Reduces Preservation to Input Sufficiency & \LHrng{DQ}{94}{95} \\
\midrule
Theorem 6.24: Preservation Can Hold Without Observation Coverage & \LH{DC113}, \LH{DC162} \\
\midrule
Proposition 6.25: Input Sufficiency Need Not Preserve Every Declared State & \LH{DC108} \\
\midrule
Proposition 6.26: Singleton Sufficiency Gives One Conditional Winner & \LH{DC138}, \LH{DQ91} \\
\midrule
Corollary 6.27: Relevance and Positive Support & \LHrng{DQ}{90}{91}, \LHrng{DQ}{96}{97} \\
\midrule
Proposition 6.28: When No Interface Preserves Every Declared State & \LH{DC126}, \LH{DC128}, \LH{DC150} \\
\midrule
Theorem 6.29: Every Declared Observation Needs a Witness & \LHrng{DC}{126}{130} \\
\midrule
Corollary 6.30: Complexity of Preserving Every Declared State & \LHrng{DC}{126}{130} \\
\midrule
Theorem 6.32: Finding One Observation with a Unique Winner & \LH{DC134}, \LHrng{EH}{4}{6}, \LH{OU1} \\
\midrule
Lemma 6.33: One Threshold for Every Observation & \LH{OU13} \\
\midrule
Lemma 6.34: A Threshold Is a Family of Nonzero Tests & \LH{DC163} \\
\midrule
Theorem 6.35: Exact Two-Action Decisiveness Classification & \LH{DC154}, \LHrng{DC}{158}{159} \\
\midrule
Proposition 6.38: Boolean Product One-to-$p$ Gap & \LHrng{PBA}{1}{7} \\
\midrule
Theorem 6.39: One Input or All Inputs: A Hard Approximation Gap & \LHrng{PBA}{1}{7} \\
\midrule
Corollary 6.40: Input-Length Inapproximability & \LHrng{PBA}{1}{7} \\
\midrule
Theorem 6.41: A One-to-All Gap for Within-Budget Decisiveness & \LHrng{DC}{160}{161} \\
\midrule
Proposition 6.42: Discernibility Hitting-Set Formulation & \LH{DSC1} \\
\midrule
Theorem 6.43: Cardinality-Preserving Equivalence with Set Cover on the Gadget Family & \LH{DQ11}, \LH{DQ34}, \LH{DQ86} \\
\midrule
Proposition 6.45: Irrelevant Now, Necessary Later & \LH{DQ62} \\
\midrule
Proposition 6.48: The Least Sequentially Sufficient Set Is the Relevant Set & \LHrng{DC}{49}{50} \\
\midrule
Corollary 6.49: Bounded-Alternation Certification Ladder & \LHrng{DC}{155}{157} \\
\midrule
Corollary 6.50: Unbounded-Alternation Limit & \LH{DC46}, \LHrng{DC}{114}{125} \\
\midrule
Lemma 7.5: Polynomial-Space Exact Continuation & \LHrng{PM}{1}{3} \\
\midrule
Lemma 7.6: Sign-Separated QBF Planning Core & \LHrng{PM}{4}{5} \\
\midrule
Theorem 7.7: Unary-Horizon Interface Certification Is PSPACE-Complete & \LH{DC10}, \LHrng{PM}{6}{9} \\
\midrule
Proposition 7.8: Horizon-Two Inheritance & \LH{EX1}, \LHrng{PM}{10}{15} \\
\midrule
Remark S2.1: A Declared Selector Has Its Own Canonical
Quotient & \LHrng{KP}{11}{13} \\
\midrule
Lemma S2.2: Budgeted Sufficiency and Reduct Existence & \LH{UB13} \\
\midrule
Corollary S2.3: Within-Budget Reduct Complexity on Restricted States & \LHrng{LMS}{1}{2}, \LH{UB13} \\
\midrule
Corollary S2.4: Within-Budget Reduct Complexity on Products & \LH{DQ36}, \LHrng{SK}{4}{8}, \LH{UB3} \\
\midrule
Proposition S2.5: Replacing One Setting & \LH{IS1} \\
\midrule
Proposition S2.6: Fixing Defaults and Removing Channels & \LHrng{IS}{2}{4} \\
\midrule
Theorem S7.1: A Three-Action Carrier for Decisive Observations & \LHrng{DC}{131}{133}, \LHrng{DC}{144}{149}, \LHrng{DC}{151}{153}, \LH{OU38} \\
\midrule
Lemma S7.2: Product Completion of the Coordinate-Pair Construction & \LHrng{DC}{144}{149}, \LHrng{DC}{151}{153} \\
\midrule
Lemma S7.3: Exact-Counting Projection Identities & \LHrng{OU}{9}{11}, \LH{OU13} \\
\midrule
Corollary S8.1: Parameterized Complexity of Explicit Listed Sufficiency & \LH{DQ11} \\
\midrule
Corollary S8.2: Tight Logarithmic Transfer on the Set-Cover Family & \LH{DQ25}, \LH{DQ32} \\
\midrule
Proposition S9.1: Full-Support Equivalence & \LHrng{DC}{95}{96} \\
\midrule
Proposition S9.2: Quotient Equivalence for a Sufficient Full-Support Interface & \LHrng{DC}{97}{99} \\
\midrule
Proposition S12.1: Coordinate-Projection Identifications & \LHrng{UB}{2}{4} \\
\midrule
Proposition S12.3: Full-Information Sufficiency as All-State Preservation & \LH{EX1} \\
\midrule
Remark S13.2: Correspondence Recovery Strictly Implies Selector
Implementation & \LHrng{KP}{1}{4} \\
\midrule
Proposition S13.3: Selector-Only Abstraction Is Generally
Noncanonical & \LHrng{KP}{5}{10} \\
\midrule
Tagged prose (main text, p.~6): The empty set suffices exactly when $\Opt$ is constant. A conflicting pair agreeing on $I$ witnesses insufficiency. & \LHrng{DP}{6}{7} \\
\midrule
Tagged prose (main text, p.~8): On a product, fixing omitted inputs to defaults preserves $F$ exactly when the retained inputs determine it: states agreeing on those inputs have the same reset tuple, and every state agrees with its own reset. Supplement~S2 gives the expanded illustration and criteria. & \LHrng{IS}{1}{8} \\
\midrule
Tagged prose: Let two states have utility vectors [displayed formula] For $0\le\varepsilon<1$, both states admit $\{a\}$ and lie in one class. For $1\le\varepsilon<2$, only $s_1$ admits $\{a,b\}$, so the class splits. For $\varepsilon\ge2$, both states admit $\{a,b\}$ and merge again. The action sets grow monotonically with tolerance, while their equality kernel follows a merge--split--merge pattern. The reverse separation also occurs: vectors $(0,-1)$ and $(-1,0)$ have disjoint optimizer sets at $\varepsilon=0$ and both admit $\{a,b\}$ at $\varepsilon=1$. & \LHrng{EF}{4}{8} \\
\midrule
Tagged prose (main text, p.~11): For fixed alphabet size, symmetry gives a polynomial-size representative family, as proved below. The succinct coNP-hard family in Theorem~6.5 already uses only two actions. & \LH{ACT1}, \LH{DQ51}, \LH{DQ72} \\
\midrule
Tagged prose (main text, p.~12): More generally, suppose a supplied rank-$R$ decomposition has the form [displayed formula] For a candidate set $I$, assume every factor $g_{ri}$ with $i\notin I$ is constant on the declared value domain of coordinate $i$. Put $b:=\max_i|X_i|$. Comparing every listed value with one reference value decides this condition in at most $Rn(b-1)$ factor comparisons. If the check accepts, states agreeing on $I$ give every product term the same value, so their action utilities and optimal-action sets are equal. & \LH{DQ57}, \LHrng{DQ}{74}{77} \\
\midrule
Tagged prose (main text, p.~14): Let $S$ be a finite family of $d$-tuples over $\{0,\ldots,k-1\}$. Here permutation invariance means that for every coordinate permutation $\sigma$ and every $s\in S$ with $\sigma s\in S$, $U(a,\sigma s)=U(a,s)$ for every action $a$. On the full product this is the usual action of the complete symmetric group; on a listed state family it is only required for pairs that both belong to the declaration. The \emph{orbit type} of a state is its coordinate histogram $h_s$. Equal histograms are equivalent to membership in the same permutation orbit. For two states define their common multiplicity by [displayed formula] Symmetry makes $\Opt$ constant on each orbit. For a complete representative family, the certificate is polynomial in the number of inputs $d$ when the alphabet size $k$ is fixed. & \LH{DQ82} \\
\midrule
Tagged prose (main text, p.~29): Let $S=\{s_1,s_2\}$, $A=\{a,b\}$, and $U\equiv0$, with positive mass on both states and no retained coordinates. Both optimizer sets are $\{a,b\}$: preservation holds, but decisiveness fails. & \LHrng{EX}{4}{5} \\
\midrule
Tagged prose (main text, p.~29): Let $S=\{s_1,s_2\}$ and $A=\{a,b\}$ under the uniform distribution, with [displayed formula] With no retained coordinates, the expected utilities are $(1,2)$, uniquely favoring $b$. Decisiveness holds, but preservation fails at $s_1$, whose optimizer is $\{a\}$. & \LHrng{EX}{6}{7} \\
\midrule
Tagged prose (main text, p.~31): Keeping both mode inputs distinguishes all four schema readings, but the interlock still prevents the two unmatched readings. Distinguishing states does not establish that a reading can occur. The next results separate these jobs: preserve the common response within each group, then check its support. This avoids summing scores when constancy already fixes their winner. When every state has positive probability, the support check is automatic and all preservation contracts give the same verdict and response groups. & \LHrng{DC}{91}{93}, \LH{OU12} \\
\midrule
Tagged prose (main text, p.~31): Without full support, the null-score convention adds the following condition. & \LHrng{DC}{108}{110}, \LHrng{OU}{14}{18} \\
\midrule
Tagged prose (main text, p.~36): Decisiveness implies the some-observation version: choose any state $s_0$. The optimizer of its $I$-fiber is a singleton, although the fiber itself may contain many states. & \LH{DC19} \\
\midrule
Tagged prose (main text, p.~39): The proposed-input, some-observation, and two-action within-budget reductions use independent uniform coordinates. These distributions have Bayesian networks of moral treewidth zero, so the hardness survives a stronger restriction than the polytree value-of-information lower bound of Krause and Guestrin ~. The separate three-action construction supplies a polynomial-size circuit for its positive dyadic distribution and product completion; it is retained with the horizon-two proof in the supplement's \emph{Three-Action Carrier for Horizon-Two Inheritance}. The two-action construction already matches the one-step upper bounds. The $\Pi_2^P$ observation-coverage reduction still uses an arbitrary succinct support circuit. Whether that result survives a polytree or fixed-treewidth bound on the supplied distribution remains open. This restriction concerns the representation of $P$, independently of how $U$ is checked. & \LH{DC57}, \LHrng{EH}{2}{10}, \LHrng{OU}{8}{11}, \LHrng{OU}{36}{37} \\
\midrule
Tagged prose (main text, p.~50): Explicit tables permit direct pair scans, fiber sums, and arena evaluation. An explicit Boolean product also permits testing all $2^n=|S|$ input subsets in polynomial time. Missing combinations in an arbitrary log instead leave the classical NP-complete minimum-reduct problem (Theorem~6.43). & \LHrng{DC}{62}{77}, \LHrng{DC}{80}{82}, \LHrng{DC}{91}{93}, \LH{DQ12}, \LH{DQ14}, \LH{DQ45}, \LHrng{DQ}{86}{89}, \LHrng{DQ}{92}{93} \\
\midrule
Tagged prose (main text, p.~51): The three principal proofs are Theorem~6.7 for restricted-state within-budget sufficiency, Theorem~6.29 for observation coverage, and Theorem~6.35 for within-budget decisiveness. & \LHrng{DC}{126}{136}, \LHrng{DC}{138}{149}, \LHrng{DQ}{18}{19}, \LH{DQ91}, \LH{OU1}, \LHrng{OU}{9}{11}, \LH{OU38} \\
\midrule
Tagged prose: Exact-rational audits exhaust the block-size-one source predicates. At block size two, they cover every distinct weak-majority outcome pattern and local singleton case on which the construction depends, across the budget-feasible interfaces and completed-product fibers (Section~9). & \LH{OU1}, \LHrng{OU}{9}{11}, \LH{OU38}, \LH{PA3} \\
\midrule
Tagged prose: For an $n$-variable formula $\varphi$, use $n$ coordinates with domain $R=\{\mathsf{ref}\}\cup\{0,1\}^n$. Give \emph{reject} utility zero and \emph{accept} utility one exactly when every coordinate is either $\mathsf{ref}$ or a satisfying assignment of $\varphi$; otherwise both actions tie. A tautology gives a constant optimizer. A falsifying assignment placed in any one coordinate changes the optimizer of the all-reference state, so every coordinate is relevant. Thus $k^*\in\{0,n\}$, and checking the empty interface decides tautology. The utility circuit uses $n$ copies of $\varphi$, hence has polynomial size. Applying this construction to the negation of a sparse 3-CNF gives, by the sparsification lemma and the Exponential Time Hypothesis ~, no $2^{o(n)}$ algorithm for succinct \textsc{Sufficiency-Check} on this family. Its explicit table has $N=(1+2^n)^n$ states, so the explicit polynomial scan and the succinct lower bound measure the same semantics at different input scales. & \LH{DC70}, \LH{DQ27}, \LH{DQ30} \\
\midrule
Tagged prose: Let \(S=\{s_0,s_1,s_2\}\) and \(A=\{a,b,c\}\). Give \(a,c\) utility one at \(s_0\), give \(b,c\) utility one at \(s_1,s_2\), and give every other state--action pair utility zero. Let \(\mathcal T_2\) contain optimizer maps with at most two attained values. The empty coordinate set is a strong backdoor to \(\mathcal T_2\), because the complete residual instance attains exactly \(\{a,c\}\) and \(\{b,c\}\). It is not sufficient because those optimizer sets differ. The separation is strict: on the same three-state, three-action state domain an optimizer map may attain three distinct singleton values and lie outside \(\mathcal T_2\). & \LHrng{UB}{5}{7} \\
\midrule
Tagged prose: Writing $\Opt^*(s)=\arg\max_a Q^*(s,a)$, equality of $Q^*$ vectors gives equality of $\Opt^*$ and the quotient factorization of Theorem~2.6. The converse fails. & \LHrng{UB}{8}{9} \\
\midrule
Tagged prose: Consider a two-state, two-action deterministic MDP with discount $\gamma=1/2$. Action $g$ moves from $s_0$ to $s_1$ and from $s_1$ to $s_0$; action $b$ self-loops. Set [displayed formula] The Bellman solution gives $Q^*(s_0,\cdot)=(8/3,4/3)$ and $Q^*(s_1,\cdot)=(10/3,5/3)$: $g$ strictly wins at both states, although their values differ. They merge in the optimizer quotient but separate under $\phi_{Q^*}$ and reward-preserving bisimulation. Thus the separation persists with positive discount and state-changing transitions. & \LHrng{UB}{10}{12}, \LH{UB14} \\
\end{longtable}
\endgroup

}{%
  \textbf{Error:} \texttt{content/claim\_mapping\_auto.tex} not found.
}

\section{Lean Axiom Inventory}

The release build derives this table by running \texttt{\#print axioms} on
the claim-mapped declarations.  It is evidence from the compiled development,
not a hand-maintained assertion.

\IfFileExists{content/axiom_scope_auto.tex}{%
\providecommand{\LeanAxiomScanDeclarations}{175}
\providecommand{\LeanAxiomCheckedDeclarations}{175}
\begingroup
\small
\setlength{\tabcolsep}{4pt}
\renewcommand{\arraystretch}{1.12}
\newcommand{\AxiomScopeTableHeader}{%
\toprule
\textbf{Axiom-scan category} & \textbf{Count} & \textbf{Scope reading} \\
\midrule}
\begin{longtable}{@{}p{0.43\linewidth} r p{0.43\linewidth}@{}}
\caption{Lean axiom-scope summary generated from the build-system \texttt{\#print axioms} scan. Model premises remain explicit theorem hypotheses or cited empirical inputs.}
\label{tab:lean-axiom-scope-auto}\\
\AxiomScopeTableHeader\endfirsthead
\AxiomScopeTableHeader\endhead
Claim-mapped Lean declarations & 175 & Declarations derived at build time from the manuscript's inline claim anchors and the nominal handle-alias map. \\
\midrule
Successfully checked by Lean & 175 & Declarations for which \texttt{\#print axioms} returned a parsed result. \\
\midrule
No Lean kernel axioms reported & 11 & Fully constructive relative to imported definitions and theorem hypotheses. \\
\midrule
Reported Lean axiom \texttt{Classical.choice} & 155 & Distinct axiom name reported by \texttt{\#print axioms}; the count is declarations depending on this axiom. \\
\midrule
Reported Lean axiom \texttt{Quot.sound} & 159 & Distinct axiom name reported by \texttt{\#print axioms}; the count is declarations depending on this axiom. \\
\midrule
Reported Lean axiom \texttt{propext} & 163 & Distinct axiom name reported by \texttt{\#print axioms}; the count is declarations depending on this axiom. \\
\midrule
Exact axiom signature & 5 & \{\texttt{propext}\}; declarations whose reported axiom set is exactly this set. \\
\midrule
Exact axiom signature & 1 & \{\texttt{Classical.choice}, \texttt{Quot.sound}\}; declarations whose reported axiom set is exactly this set. \\
\midrule
Exact axiom signature & 4 & \{\texttt{Quot.sound}, \texttt{propext}\}; declarations whose reported axiom set is exactly this set. \\
\midrule
Exact axiom signature & 154 & \{\texttt{Classical.choice}, \texttt{Quot.sound}, \texttt{propext}\}; declarations whose reported axiom set is exactly this set. \\
\midrule
Private, timeout, or unresolved declarations & 0 & Declarations without a parsed result; inspect the raw scan when nonzero. \\
\bottomrule
\end{longtable}
\endgroup

}{%
  \textbf{Error:} \texttt{content/axiom\_scope\_auto.tex} not found.
}

\section{Cited Lean Handle Ledger}

\IfFileExists{content/lean_handle_ids_auto.tex}{%
\begingroup
\scriptsize
\setlength{\tabcolsep}{4pt}
\renewcommand{\arraystretch}{1.06}
\ifdefined\PaperSupplementTableSetup\PaperSupplementTableSetup\fi
\setlength{\LTpre}{2pt}
\setlength{\LTpost}{2pt}
\setlength{\emergencystretch}{3em}
\sloppy
\urlstyle{tt}
\makeatletter
\newif\ifleanhandlelist
\if@twocolumn\leanhandlelisttrue\else\leanhandlelistfalse\fi
\@ifundefined{LeanHandleIndexListLayout}{}{\leanhandlelisttrue}
\ifleanhandlelist
\begin{list}{}{\leftmargin=0pt\itemindent=0pt\itemsep=2pt\parsep=0pt\topsep=4pt}
\item \textbf{\nolinkurl{ACT1}}\hypertarget{lh:ACT1}{}\enspace{\ttfamily DecisionQuotient.\allowbreak{}two\_\allowbreak{}actions\_\allowbreak{}coNP\_\allowbreak{}hard} {\tiny\ttfamily DecisionQuotient/\allowbreak Tractability/\allowbreak Tightness.lean}
\item \textbf{\nolinkurl{CC4}}\hypertarget{lh:CC4}{}\enspace{\ttfamily DecisionQuotient.\allowbreak{}ClaimClosure.\allowbreak{}anchor\_\allowbreak{}sigma2p\_\allowbreak{}complete\_\allowbreak{}conditional} {\tiny\ttfamily DecisionQuotient/\allowbreak ClaimClosure.lean}
\item \textbf{\nolinkurl{DC10}}\hypertarget{lh:DC10}{}\enspace{\ttfamily StochasticSequential.\allowbreak{}sequential\_\allowbreak{}to\_\allowbreak{}PSPACE} {\tiny\ttfamily DecisionQuotient/\allowbreak StochasticSequential/\allowbreak Hierarchy.lean}
\item \textbf{\nolinkurl{DC19}}\hypertarget{lh:DC19}{}\enspace{\ttfamily StochasticSequential.\allowbreak{}stochastic\_\allowbreak{}anchor\_\allowbreak{}sufficient\_\allowbreak{}of\_\allowbreak{}stochastic\_\allowbreak{}sufficient} {\tiny\ttfamily DecisionQuotient/\allowbreak StochasticSequential/\allowbreak Quotient.lean}
\item \textbf{\nolinkurl{DC46}}\hypertarget{lh:DC46}{}\enspace{\ttfamily StochasticSequential.\allowbreak{}prenexTQBF\_\allowbreak{}iff\_\allowbreak{}empty\_\allowbreak{}alternatingSequentialSufficient} {\tiny\ttfamily DecisionQuotient/\allowbreak StochasticSequential/\allowbreak AlternatingTQBF.lean}
\item \textbf{\nolinkurl{DC49}}\hypertarget{lh:DC49}{}\enspace{\ttfamily StochasticSequential.\allowbreak{}sequentialMinimalSufficient\_\allowbreak{}iff\_\allowbreak{}relevant} {\tiny\ttfamily DecisionQuotient/\allowbreak StochasticSequential/\allowbreak Basic.lean}
\item \textbf{\nolinkurl{DC50}}\hypertarget{lh:DC50}{}\enspace{\ttfamily StochasticSequential.\allowbreak{}sequentialRelevantSet\_\allowbreak{}is\_\allowbreak{}minimal} {\tiny\ttfamily DecisionQuotient/\allowbreak StochasticSequential/\allowbreak Basic.lean}
\item \textbf{\nolinkurl{DC57}}\hypertarget{lh:DC57}{}\enspace{\ttfamily StochasticSequential.\allowbreak{}fiberDecisionProblem\_\allowbreak{}sufficient} {\tiny\ttfamily DecisionQuotient/\allowbreak StochasticSequential/\allowbreak Basic.lean}
\item \textbf{\nolinkurl{DC62}}\hypertarget{lh:DC62}{}\enspace{\ttfamily StochasticSequential.\allowbreak{}countedStochasticSufficiencySearch\_\allowbreak{}spec} {\tiny\ttfamily DecisionQuotient/\allowbreak StochasticSequential/\allowbreak Computation.lean}
\item \textbf{\nolinkurl{DC63}}\hypertarget{lh:DC63}{}\enspace{\ttfamily StochasticSequential.\allowbreak{}countedStochasticSufficiencySearch\_\allowbreak{}steps} {\tiny\ttfamily DecisionQuotient/\allowbreak StochasticSequential/\allowbreak Computation.lean}
\item \textbf{\nolinkurl{DC64}}\hypertarget{lh:DC64}{}\enspace{\ttfamily StochasticSequential.\allowbreak{}countedStochasticAnchorSearch\_\allowbreak{}spec} {\tiny\ttfamily DecisionQuotient/\allowbreak StochasticSequential/\allowbreak Computation.lean}
\item \textbf{\nolinkurl{DC65}}\hypertarget{lh:DC65}{}\enspace{\ttfamily StochasticSequential.\allowbreak{}countedStochasticAnchorSearch\_\allowbreak{}steps} {\tiny\ttfamily DecisionQuotient/\allowbreak StochasticSequential/\allowbreak Computation.lean}
\item \textbf{\nolinkurl{DC66}}\hypertarget{lh:DC66}{}\enspace{\ttfamily StochasticSequential.\allowbreak{}countedSequentialSufficiencySearch\_\allowbreak{}spec} {\tiny\ttfamily DecisionQuotient/\allowbreak StochasticSequential/\allowbreak Computation.lean}
\item \textbf{\nolinkurl{DC67}}\hypertarget{lh:DC67}{}\enspace{\ttfamily StochasticSequential.\allowbreak{}countedSequentialSufficiencySearch\_\allowbreak{}steps} {\tiny\ttfamily DecisionQuotient/\allowbreak StochasticSequential/\allowbreak Computation.lean}
\item \textbf{\nolinkurl{DC68}}\hypertarget{lh:DC68}{}\enspace{\ttfamily StochasticSequential.\allowbreak{}countedSequentialAnchorSearch\_\allowbreak{}spec} {\tiny\ttfamily DecisionQuotient/\allowbreak StochasticSequential/\allowbreak Computation.lean}
\item \textbf{\nolinkurl{DC69}}\hypertarget{lh:DC69}{}\enspace{\ttfamily StochasticSequential.\allowbreak{}countedSequentialAnchorSearch\_\allowbreak{}steps} {\tiny\ttfamily DecisionQuotient/\allowbreak StochasticSequential/\allowbreak Computation.lean}
\item \textbf{\nolinkurl{DC70}}\hypertarget{lh:DC70}{}\enspace{\ttfamily DecisionQuotient.\allowbreak{}static\_\allowbreak{}sufficiency\_\allowbreak{}inP\_\allowbreak{}explicit} {\tiny\ttfamily DecisionQuotient/\allowbreak ExplicitStateMembership.lean}
\item \textbf{\nolinkurl{DC71}}\hypertarget{lh:DC71}{}\enspace{\ttfamily DecisionQuotient.\allowbreak{}static\_\allowbreak{}anchor\_\allowbreak{}inP\_\allowbreak{}explicit} {\tiny\ttfamily DecisionQuotient/\allowbreak ExplicitStateMembership.lean}
\item \textbf{\nolinkurl{DC72}}\hypertarget{lh:DC72}{}\enspace{\ttfamily StochasticSequential.\allowbreak{}stochastic\_\allowbreak{}sufficiency\_\allowbreak{}inP\_\allowbreak{}explicit} {\tiny\ttfamily DecisionQuotient/\allowbreak StochasticSequential/\allowbreak Computation.lean}
\item \textbf{\nolinkurl{DC73}}\hypertarget{lh:DC73}{}\enspace{\ttfamily StochasticSequential.\allowbreak{}stochastic\_\allowbreak{}anchor\_\allowbreak{}inP\_\allowbreak{}explicit} {\tiny\ttfamily DecisionQuotient/\allowbreak StochasticSequential/\allowbreak Computation.lean}
\item \textbf{\nolinkurl{DC74}}\hypertarget{lh:DC74}{}\enspace{\ttfamily StochasticSequential.\allowbreak{}sequential\_\allowbreak{}sufficiency\_\allowbreak{}inP\_\allowbreak{}explicit} {\tiny\ttfamily DecisionQuotient/\allowbreak StochasticSequential/\allowbreak Computation.lean}
\item \textbf{\nolinkurl{DC75}}\hypertarget{lh:DC75}{}\enspace{\ttfamily StochasticSequential.\allowbreak{}sequential\_\allowbreak{}anchor\_\allowbreak{}inP\_\allowbreak{}explicit} {\tiny\ttfamily DecisionQuotient/\allowbreak StochasticSequential/\allowbreak Computation.lean}
\item \textbf{\nolinkurl{DC76}}\hypertarget{lh:DC76}{}\enspace{\ttfamily DecisionQuotient.\allowbreak{}explicit\_\allowbreak{}state\_\allowbreak{}inP\_\allowbreak{}summary} {\tiny\ttfamily DecisionQuotient/\allowbreak ExplicitStateMembership.lean}
\item \textbf{\nolinkurl{DC77}}\hypertarget{lh:DC77}{}\enspace{\ttfamily DecisionQuotient.\allowbreak{}static\_\allowbreak{}minimum\_\allowbreak{}inP\_\allowbreak{}explicit} {\tiny\ttfamily DecisionQuotient/\allowbreak ExplicitStateMembership.lean}
\item \textbf{\nolinkurl{DC80}}\hypertarget{lh:DC80}{}\enspace{\ttfamily DecisionQuotient.\allowbreak{}staticSufficiency\_\allowbreak{}counted\_\allowbreak{}search\_\allowbreak{}witness} {\tiny\ttfamily DecisionQuotient/\allowbreak AlgorithmComplexity.lean}
\item \textbf{\nolinkurl{DC81}}\hypertarget{lh:DC81}{}\enspace{\ttfamily DecisionQuotient.\allowbreak{}static\_\allowbreak{}query\_\allowbreak{}search\_\allowbreak{}matrix} {\tiny\ttfamily DecisionQuotient/\allowbreak AlgorithmComplexity.lean}
\item \textbf{\nolinkurl{DC82}}\hypertarget{lh:DC82}{}\enspace{\ttfamily DecisionQuotient.\allowbreak{}finite\_\allowbreak{}search\_\allowbreak{}summary} {\tiny\ttfamily DecisionQuotient/\allowbreak ComplexityMain.lean}
\item \textbf{\nolinkurl{DC91}}\hypertarget{lh:DC91}{}\enspace{\ttfamily StochasticSequential.\allowbreak{}countedStochasticPreservationSearch\_\allowbreak{}spec} {\tiny\ttfamily DecisionQuotient/\allowbreak StochasticSequential/\allowbreak Computation.lean}
\item \textbf{\nolinkurl{DC92}}\hypertarget{lh:DC92}{}\enspace{\ttfamily StochasticSequential.\allowbreak{}countedStochasticPreservationSearch\_\allowbreak{}steps} {\tiny\ttfamily DecisionQuotient/\allowbreak StochasticSequential/\allowbreak Computation.lean}
\item \textbf{\nolinkurl{DC93}}\hypertarget{lh:DC93}{}\enspace{\ttfamily StochasticSequential.\allowbreak{}stochastic\_\allowbreak{}preservation\_\allowbreak{}inP\_\allowbreak{}explicit} {\tiny\ttfamily DecisionQuotient/\allowbreak StochasticSequential/\allowbreak Computation.lean}
\item \textbf{\nolinkurl{DC95}}\hypertarget{lh:DC95}{}\enspace{\ttfamily StochasticSequential.\allowbreak{}static\_\allowbreak{}sufficiency\_\allowbreak{}implies\_\allowbreak{}stochastic\_\allowbreak{}preservation\_\allowbreak{}of\_\allowbreak{}full\_\allowbreak{}support} {\tiny\ttfamily DecisionQuotient/\allowbreak StochasticSequential/\allowbreak Basic.lean}
\item \textbf{\nolinkurl{DC96}}\hypertarget{lh:DC96}{}\enspace{\ttfamily StochasticSequential.\allowbreak{}static\_\allowbreak{}sufficiency\_\allowbreak{}iff\_\allowbreak{}stochastic\_\allowbreak{}preservation\_\allowbreak{}of\_\allowbreak{}full\_\allowbreak{}support} {\tiny\ttfamily DecisionQuotient/\allowbreak StochasticSequential/\allowbreak Basic.lean}
\item \textbf{\nolinkurl{DC97}}\hypertarget{lh:DC97}{}\enspace{\ttfamily StochasticSequential.\allowbreak{}stochasticDecisionEquiv\_\allowbreak{}iff\_\allowbreak{}decisionEquiv\_\allowbreak{}of\_\allowbreak{}preservation} {\tiny\ttfamily DecisionQuotient/\allowbreak StochasticSequential/\allowbreak Quotient.lean}
\item \textbf{\nolinkurl{DC98}}\hypertarget{lh:DC98}{}\enspace{\ttfamily StochasticSequential.\allowbreak{}stochasticDecisionEquiv\_\allowbreak{}iff\_\allowbreak{}decisionEquiv\_\allowbreak{}of\_\allowbreak{}full\_\allowbreak{}support} {\tiny\ttfamily DecisionQuotient/\allowbreak StochasticSequential/\allowbreak Quotient.lean}
\item \textbf{\nolinkurl{DC99}}\hypertarget{lh:DC99}{}\enspace{\ttfamily StochasticSequential.\allowbreak{}stochasticEquivSetoid\_\allowbreak{}eq\_\allowbreak{}decisionSetoid\_\allowbreak{}of\_\allowbreak{}full\_\allowbreak{}support} {\tiny\ttfamily DecisionQuotient/\allowbreak StochasticSequential/\allowbreak Quotient.lean}
\item \textbf{\nolinkurl{DC108}}\hypertarget{lh:DC108}{}\enspace{\ttfamily StochasticSequential.\allowbreak{}zeroSupportCounterexample\_\allowbreak{}refutes\_\allowbreak{}stochastic\_\allowbreak{}relevance\_\allowbreak{}conjecture} {\tiny\ttfamily DecisionQuotient/\allowbreak StochasticSequential/\allowbreak PreservationVariants.lean}
\item \textbf{\nolinkurl{DC109}}\hypertarget{lh:DC109}{}\enspace{\ttfamily StochasticSequential.\allowbreak{}stochastic\_\allowbreak{}preservation\_\allowbreak{}iff\_\allowbreak{}contains\_\allowbreak{}stochastically\_\allowbreak{}relevant\_\allowbreak{}of\_\allowbreak{}nonneg\_\allowbreak{}support\_\allowbreak{}transport} {\tiny\ttfamily DecisionQuotient/\allowbreak StochasticSequential/\allowbreak PreservationVariants.lean}
\item \textbf{\nolinkurl{DC110}}\hypertarget{lh:DC110}{}\enspace{\ttfamily StochasticSequential.\allowbreak{}stochastic\_\allowbreak{}preservation\_\allowbreak{}iff\_\allowbreak{}contains\_\allowbreak{}stochastically\_\allowbreak{}relevant\_\allowbreak{}of\_\allowbreak{}nonneg\_\allowbreak{}primitive\_\allowbreak{}dynamics} {\tiny\ttfamily DecisionQuotient/\allowbreak StochasticSequential/\allowbreak PreservationVariants.lean}
\item \textbf{\nolinkurl{DC113}}\hypertarget{lh:DC113}{}\enspace{\ttfamily StochasticSequential.\allowbreak{}stochastic\_\allowbreak{}preservation\_\allowbreak{}iff\_\allowbreak{}support\_\allowbreak{}restricted\_\allowbreak{}and\_\allowbreak{}support\_\allowbreak{}or\_\allowbreak{}trivial} {\tiny\ttfamily DecisionQuotient/\allowbreak StochasticSequential/\allowbreak PreservationVariants.lean}
\item \textbf{\nolinkurl{DC114}}\hypertarget{lh:DC114}{}\enspace{\ttfamily StochasticSequential.\allowbreak{}QuantifierArenaPosition} {\tiny\ttfamily DecisionQuotient/\allowbreak StochasticSequential/\allowbreak AlternatingTQBF.lean}
\item \textbf{\nolinkurl{DC115}}\hypertarget{lh:DC115}{}\enspace{\ttfamily StochasticSequential.\allowbreak{}quantifierArena\_\allowbreak{}owner\_\allowbreak{}existential} {\tiny\ttfamily DecisionQuotient/\allowbreak StochasticSequential/\allowbreak AlternatingTQBF.lean}
\item \textbf{\nolinkurl{DC116}}\hypertarget{lh:DC116}{}\enspace{\ttfamily StochasticSequential.\allowbreak{}quantifierArena\_\allowbreak{}owner\_\allowbreak{}universal} {\tiny\ttfamily DecisionQuotient/\allowbreak StochasticSequential/\allowbreak AlternatingTQBF.lean}
\item \textbf{\nolinkurl{DC117}}\hypertarget{lh:DC117}{}\enspace{\ttfamily StochasticSequential.\allowbreak{}quantifierArena\_\allowbreak{}successor\_\allowbreak{}quantify} {\tiny\ttfamily DecisionQuotient/\allowbreak StochasticSequential/\allowbreak AlternatingTQBF.lean}
\item \textbf{\nolinkurl{DC118}}\hypertarget{lh:DC118}{}\enspace{\ttfamily StochasticSequential.\allowbreak{}quantifierArena\_\allowbreak{}terminal\_\allowbreak{}value} {\tiny\ttfamily DecisionQuotient/\allowbreak StochasticSequential/\allowbreak AlternatingTQBF.lean}
\item \textbf{\nolinkurl{DC119}}\hypertarget{lh:DC119}{}\enspace{\ttfamily StochasticSequential.\allowbreak{}quantifierArena\_\allowbreak{}controller\_\allowbreak{}backward\_\allowbreak{}step} {\tiny\ttfamily DecisionQuotient/\allowbreak StochasticSequential/\allowbreak AlternatingTQBF.lean}
\item \textbf{\nolinkurl{DC120}}\hypertarget{lh:DC120}{}\enspace{\ttfamily StochasticSequential.\allowbreak{}quantifierArena\_\allowbreak{}environment\_\allowbreak{}backward\_\allowbreak{}step} {\tiny\ttfamily DecisionQuotient/\allowbreak StochasticSequential/\allowbreak AlternatingTQBF.lean}
\item \textbf{\nolinkurl{DC121}}\hypertarget{lh:DC121}{}\enspace{\ttfamily StochasticSequential.\allowbreak{}quantifierArenaValue\_\allowbreak{}eq\_\allowbreak{}prenexEval} {\tiny\ttfamily DecisionQuotient/\allowbreak StochasticSequential/\allowbreak AlternatingTQBF.lean}
\item \textbf{\nolinkurl{DC122}}\hypertarget{lh:DC122}{}\enspace{\ttfamily StochasticSequential.\allowbreak{}reducePrenexTQBF\_\allowbreak{}syntaxSize} {\tiny\ttfamily DecisionQuotient/\allowbreak StochasticSequential/\allowbreak AlternatingTQBF.lean}
\item \textbf{\nolinkurl{DC123}}\hypertarget{lh:DC123}{}\enspace{\ttfamily StochasticSequential.\allowbreak{}reducePrenexTQBF\_\allowbreak{}reductionCorrectness} {\tiny\ttfamily DecisionQuotient/\allowbreak StochasticSequential/\allowbreak AlternatingTQBF.lean}
\item \textbf{\nolinkurl{DC124}}\hypertarget{lh:DC124}{}\enspace{\ttfamily StochasticSequential.\allowbreak{}reducePrenexTQBF\_\allowbreak{}anchor\_\allowbreak{}reductionCorrectness} {\tiny\ttfamily DecisionQuotient/\allowbreak StochasticSequential/\allowbreak AlternatingTQBF.lean}
\item \textbf{\nolinkurl{DC125}}\hypertarget{lh:DC125}{}\enspace{\ttfamily StochasticSequential.\allowbreak{}reducePrenexTQBF\_\allowbreak{}minimum\_\allowbreak{}reductionCorrectness} {\tiny\ttfamily DecisionQuotient/\allowbreak StochasticSequential/\allowbreak AlternatingTQBF.lean}
\item \textbf{\nolinkurl{DC126}}\hypertarget{lh:DC126}{}\enspace{\ttfamily StochasticSequential.\allowbreak{}stochastic\_\allowbreak{}preservation\_\allowbreak{}iff\_\allowbreak{}static\_\allowbreak{}sufficiency\_\allowbreak{}and\_\allowbreak{}zero\_\allowbreak{}fiber\_\allowbreak{}trivial} {\tiny\ttfamily DecisionQuotient/\allowbreak StochasticSequential/\allowbreak PreservationVariants.lean}
\item \textbf{\nolinkurl{DC127}}\hypertarget{lh:DC127}{}\enspace{\ttfamily StochasticSequential.\allowbreak{}stochastic\_\allowbreak{}anchor\_\allowbreak{}preservation\_\allowbreak{}iff\_\allowbreak{}positive\_\allowbreak{}constant\_\allowbreak{}or\_\allowbreak{}zero\_\allowbreak{}trivial} {\tiny\ttfamily DecisionQuotient/\allowbreak StochasticSequential/\allowbreak PreservationVariants.lean}
\item \textbf{\nolinkurl{DC128}}\hypertarget{lh:DC128}{}\enspace{\ttfamily StochasticSequential.\allowbreak{}stochastic\_\allowbreak{}minimum\_\allowbreak{}preservation\_\allowbreak{}iff\_\allowbreak{}relevantFinset} {\tiny\ttfamily DecisionQuotient/\allowbreak StochasticSequential/\allowbreak PreservationVariants.lean}
\item \textbf{\nolinkurl{DC129}}\hypertarget{lh:DC129}{}\enspace{\ttfamily StochasticSequential.\allowbreak{}reduceForallExistsToPreservation\_\allowbreak{}correct} {\tiny\ttfamily DecisionQuotient/\allowbreak StochasticSequential/\allowbreak PreservationPHHardness.lean}
\item \textbf{\nolinkurl{DC130}}\hypertarget{lh:DC130}{}\enspace{\ttfamily StochasticSequential.\allowbreak{}reduceForallExistsToMinimumPreservation\_\allowbreak{}correct} {\tiny\ttfamily DecisionQuotient/\allowbreak StochasticSequential/\allowbreak PreservationPHHardness.lean}
\item \textbf{\nolinkurl{DC131}}\hypertarget{lh:DC131}{}\enspace{\ttfamily StochasticSequential.\allowbreak{}mdDistribution\_\allowbreak{}sum\_\allowbreak{}eq\_\allowbreak{}one} {\tiny\ttfamily DecisionQuotient/\allowbreak StochasticSequential/\allowbreak MinimumDecisivenessHardness.lean}
\item \textbf{\nolinkurl{DC132}}\hypertarget{lh:DC132}{}\enspace{\ttfamily StochasticSequential.\allowbreak{}reduceExistsForallMajorityMinimumDecisiveness\_\allowbreak{}correct} {\tiny\ttfamily DecisionQuotient/\allowbreak StochasticSequential/\allowbreak MinimumDecisivenessHardness.lean}
\item \textbf{\nolinkurl{DC133}}\hypertarget{lh:DC133}{}\enspace{\ttfamily StochasticSequential.\allowbreak{}reduceExistsForallMajorityMinimumDecisiveness\_\allowbreak{}poly\_\allowbreak{}bound} {\tiny\ttfamily DecisionQuotient/\allowbreak StochasticSequential/\allowbreak MinimumDecisivenessHardness.lean}
\item \textbf{\nolinkurl{DC134}}\hypertarget{lh:DC134}{}\enspace{\ttfamily StochasticSequential.\allowbreak{}reduceExistsMajorityPureAnchor\_\allowbreak{}correct} {\tiny\ttfamily DecisionQuotient/\allowbreak StochasticSequential/\allowbreak NPPPHardness.lean}
\item \textbf{\nolinkurl{DC135}}\hypertarget{lh:DC135}{}\enspace{\ttfamily StochasticSequential.\allowbreak{}reduceExistsNonMajorityPureDecisiveness\_\allowbreak{}correct} {\tiny\ttfamily DecisionQuotient/\allowbreak StochasticSequential/\allowbreak NPPPHardness.lean}
\item \textbf{\nolinkurl{DC136}}\hypertarget{lh:DC136}{}\enspace{\ttfamily StochasticSequential.\allowbreak{}reduceExistsMajorityPureAnchor\_\allowbreak{}poly\_\allowbreak{}bound} {\tiny\ttfamily DecisionQuotient/\allowbreak StochasticSequential/\allowbreak NPPPHardness.lean}
\item \textbf{\nolinkurl{DC138}}\hypertarget{lh:DC138}{}\enspace{\ttfamily StochasticSequential.\allowbreak{}static\_\allowbreak{}singleton\_\allowbreak{}sufficiency\_\allowbreak{}implies\_\allowbreak{}stochastic\_\allowbreak{}decisiveness\_\allowbreak{}of\_\allowbreak{}positive\_\allowbreak{}fiber\_\allowbreak{}support} {\tiny\ttfamily DecisionQuotient/\allowbreak StochasticSequential/\allowbreak PreservationVariants.lean}
\item \textbf{\nolinkurl{DC139}}\hypertarget{lh:DC139}{}\enspace{\ttfamily StochasticSequential.\allowbreak{}ClaimClosure.\allowbreak{}claim\_\allowbreak{}stochastic\_\allowbreak{}succinct\_\allowbreak{}matrix\_\allowbreak{}closed} {\tiny\ttfamily DecisionQuotient/\allowbreak StochasticSequential/\allowbreak ClaimClosure.lean}
\item \textbf{\nolinkurl{DC140}}\hypertarget{lh:DC140}{}\enspace{\ttfamily StochasticSequential.\allowbreak{}stochasticPreservation\_\allowbreak{}inter} {\tiny\ttfamily DecisionQuotient/\allowbreak StochasticSequential/\allowbreak IntersectionSeparation.lean}
\item \textbf{\nolinkurl{DC141}}\hypertarget{lh:DC141}{}\enspace{\ttfamily StochasticSequential.\allowbreak{}DecisivenessIntersection.\allowbreak{}distribution\_\allowbreak{}positive} {\tiny\ttfamily DecisionQuotient/\allowbreak StochasticSequential/\allowbreak IntersectionSeparation.lean}
\item \textbf{\nolinkurl{DC142}}\hypertarget{lh:DC142}{}\enspace{\ttfamily StochasticSequential.\allowbreak{}DecisivenessIntersection.\allowbreak{}decisive\_\allowbreak{}not\_\allowbreak{}intersection\_\allowbreak{}closed} {\tiny\ttfamily DecisionQuotient/\allowbreak StochasticSequential/\allowbreak IntersectionSeparation.lean}
\item \textbf{\nolinkurl{DC143}}\hypertarget{lh:DC143}{}\enspace{\ttfamily StochasticSequential.\allowbreak{}DecisivenessIntersection.\allowbreak{}first\_\allowbreak{}inter\_\allowbreak{}second\_\allowbreak{}empty} {\tiny\ttfamily DecisionQuotient/\allowbreak StochasticSequential/\allowbreak IntersectionSeparation.lean}
\item \textbf{\nolinkurl{DC144}}\hypertarget{lh:DC144}{}\enspace{\ttfamily StochasticSequential.\allowbreak{}mdProductCode\_\allowbreak{}agreeOn\_\allowbreak{}iff} {\tiny\ttfamily DecisionQuotient/\allowbreak StochasticSequential/\allowbreak MinimumDecisivenessProductHardness.lean}
\item \textbf{\nolinkurl{DC145}}\hypertarget{lh:DC145}{}\enspace{\ttfamily StochasticSequential.\allowbreak{}mdProductCode\_\allowbreak{}injective} {\tiny\ttfamily DecisionQuotient/\allowbreak StochasticSequential/\allowbreak MinimumDecisivenessProductHardness.lean}
\item \textbf{\nolinkurl{DC146}}\hypertarget{lh:DC146}{}\enspace{\ttfamily StochasticSequential.\allowbreak{}mdProductBackgroundTotal\_\allowbreak{}lt\_\allowbreak{}one} {\tiny\ttfamily DecisionQuotient/\allowbreak StochasticSequential/\allowbreak MinimumDecisivenessProductHardness.lean}
\item \textbf{\nolinkurl{DC147}}\hypertarget{lh:DC147}{}\enspace{\ttfamily StochasticSequential.\allowbreak{}mdProductSignalScale\_\allowbreak{}pos} {\tiny\ttfamily DecisionQuotient/\allowbreak StochasticSequential/\allowbreak MinimumDecisivenessProductHardness.lean}
\item \textbf{\nolinkurl{DC148}}\hypertarget{lh:DC148}{}\enspace{\ttfamily StochasticSequential.\allowbreak{}mdProductDistribution\_\allowbreak{}pos} {\tiny\ttfamily DecisionQuotient/\allowbreak StochasticSequential/\allowbreak MinimumDecisivenessProductHardness.lean}
\item \textbf{\nolinkurl{DC149}}\hypertarget{lh:DC149}{}\enspace{\ttfamily StochasticSequential.\allowbreak{}mdProductDistribution\_\allowbreak{}sum\_\allowbreak{}eq\_\allowbreak{}one} {\tiny\ttfamily DecisionQuotient/\allowbreak StochasticSequential/\allowbreak MinimumDecisivenessProductHardness.lean}
\item \textbf{\nolinkurl{DC150}}\hypertarget{lh:DC150}{}\enspace{\ttfamily StochasticSequential.\allowbreak{}no\_\allowbreak{}stochastic\_\allowbreak{}preserving\_\allowbreak{}set\_\allowbreak{}iff\_\allowbreak{}relevant\_\allowbreak{}zero\_\allowbreak{}fiber} {\tiny\ttfamily DecisionQuotient/\allowbreak StochasticSequential/\allowbreak PreservationVariants.lean}
\item \textbf{\nolinkurl{DC151}}\hypertarget{lh:DC151}{}\enspace{\ttfamily StochasticSequential.\allowbreak{}mdAction\_\allowbreak{}card} {\tiny\ttfamily DecisionQuotient/\allowbreak StochasticSequential/\allowbreak MinimumDecisivenessHardness.lean}
\item \textbf{\nolinkurl{DC152}}\hypertarget{lh:DC152}{}\enspace{\ttfamily StochasticSequential.\allowbreak{}mdPairMissing\_\allowbreak{}prevents\_\allowbreak{}decisiveness} {\tiny\ttfamily DecisionQuotient/\allowbreak StochasticSequential/\allowbreak MinimumDecisivenessHardness.lean}
\item \textbf{\nolinkurl{DC153}}\hypertarget{lh:DC153}{}\enspace{\ttfamily StochasticSequential.\allowbreak{}mdProductAction\_\allowbreak{}card} {\tiny\ttfamily DecisionQuotient/\allowbreak StochasticSequential/\allowbreak MinimumDecisivenessProductHardness.lean}
\item \textbf{\nolinkurl{DC154}}\hypertarget{lh:DC154}{}\enspace{\ttfamily StochasticSequential.\allowbreak{}stochasticSufficient\_\allowbreak{}bool\_\allowbreak{}iff\_\allowbreak{}all\_\allowbreak{}fiber\_\allowbreak{}gaps\_\allowbreak{}ne\_\allowbreak{}zero} {\tiny\ttfamily DecisionQuotient/\allowbreak StochasticSequential/\allowbreak TwoActionExactCounting.lean}
\item \textbf{\nolinkurl{DC155}}\hypertarget{lh:DC155}{}\enspace{\ttfamily StochasticSequential.\allowbreak{}universalArenaFamily\_\allowbreak{}allTrue\_\allowbreak{}iff\_\allowbreak{}empty\_\allowbreak{}sufficient} {\tiny\ttfamily DecisionQuotient/\allowbreak StochasticSequential/\allowbreak BoundedAlternationCalibration.lean}
\item \textbf{\nolinkurl{DC156}}\hypertarget{lh:DC156}{}\enspace{\ttfamily StochasticSequential.\allowbreak{}universalArenaFamily\_\allowbreak{}allTrue\_\allowbreak{}iff\_\allowbreak{}zero\_\allowbreak{}minimum} {\tiny\ttfamily DecisionQuotient/\allowbreak StochasticSequential/\allowbreak BoundedAlternationCalibration.lean}
\item \textbf{\nolinkurl{DC157}}\hypertarget{lh:DC157}{}\enspace{\ttfamily StochasticSequential.\allowbreak{}anchorArenaFamily\_\allowbreak{}existsAllTrue\_\allowbreak{}iff\_\allowbreak{}anchorSufficient} {\tiny\ttfamily DecisionQuotient/\allowbreak StochasticSequential/\allowbreak BoundedAlternationCalibration.lean}
\item \textbf{\nolinkurl{DC158}}\hypertarget{lh:DC158}{}\enspace{\ttfamily StochasticSequential.\allowbreak{}omittedSelectorPair\_\allowbreak{}cancels} {\tiny\ttfamily DecisionQuotient/\allowbreak StochasticSequential/\allowbreak TwoActionExactCounting.lean}
\item \textbf{\nolinkurl{DC159}}\hypertarget{lh:DC159}{}\enspace{\ttfamily StochasticSequential.\allowbreak{}retainedSelectorSide\_\allowbreak{}isolates} {\tiny\ttfamily DecisionQuotient/\allowbreak StochasticSequential/\allowbreak TwoActionExactCounting.lean}
\item \textbf{\nolinkurl{DC160}}\hypertarget{lh:DC160}{}\enspace{\ttfamily StochasticSequential.\allowbreak{}selectorSign\_\allowbreak{}ne\_\allowbreak{}zero} {\tiny\ttfamily DecisionQuotient/\allowbreak StochasticSequential/\allowbreak TwoActionExactCounting.lean}
\item \textbf{\nolinkurl{DC161}}\hypertarget{lh:DC161}{}\enspace{\ttfamily StochasticSequential.\allowbreak{}selectorSign\_\allowbreak{}sum\_\allowbreak{}eq\_\allowbreak{}zero} {\tiny\ttfamily DecisionQuotient/\allowbreak StochasticSequential/\allowbreak TwoActionExactCounting.lean}
\item \textbf{\nolinkurl{DC162}}\hypertarget{lh:DC162}{}\enspace{\ttfamily StochasticSequential.\allowbreak{}stochastic\_\allowbreak{}preservation\_\allowbreak{}iff\_\allowbreak{}deployable\_\allowbreak{}of\_\allowbreak{}singleton\_\allowbreak{}opt} {\tiny\ttfamily DecisionQuotient/\allowbreak StochasticSequential/\allowbreak PreservationVariants.lean}
\item \textbf{\nolinkurl{DC163}}\hypertarget{lh:DC163}{}\enspace{\ttfamily StochasticSequential.\allowbreak{}threshold\_\allowbreak{}le\_\allowbreak{}iff\_\allowbreak{}forall\_\allowbreak{}lt\_\allowbreak{}ne} {\tiny\ttfamily DecisionQuotient/\allowbreak StochasticSequential/\allowbreak TwoActionExactCounting.lean}
\item \textbf{\nolinkurl{DP6}}\hypertarget{lh:DP6}{}\enspace{\ttfamily ClaimClosure.\allowbreak{}DP6} {\tiny\ttfamily DecisionQuotient/\allowbreak ClaimClosure.lean}
\item \textbf{\nolinkurl{DP7}}\hypertarget{lh:DP7}{}\enspace{\ttfamily ClaimClosure.\allowbreak{}DP7} {\tiny\ttfamily DecisionQuotient/\allowbreak ClaimClosure.lean}
\item \textbf{\nolinkurl{DQ11}}\hypertarget{lh:DQ11}{}\enspace{\ttfamily SetCoverInstance.\allowbreak{}min\_\allowbreak{}sufficient\_\allowbreak{}iff\_\allowbreak{}set\_\allowbreak{}cover} {\tiny\ttfamily DecisionQuotient/\allowbreak Hardness/\allowbreak SetCoverReduction.lean}
\item \textbf{\nolinkurl{DQ12}}\hypertarget{lh:DQ12}{}\enspace{\ttfamily DecisionQuotient.\allowbreak{}StochasticSequential.\allowbreak{}bounded\_\allowbreak{}horizon\_\allowbreak{}tractable} {\tiny\ttfamily DecisionQuotient/\allowbreak StochasticSequential/\allowbreak Tractability.lean}
\item \textbf{\nolinkurl{DQ14}}\hypertarget{lh:DQ14}{}\enspace{\ttfamily DecisionQuotient.\allowbreak{}StochasticSequential.\allowbreak{}fully\_\allowbreak{}observable\_\allowbreak{}tractable} {\tiny\ttfamily DecisionQuotient/\allowbreak StochasticSequential/\allowbreak Tractability.lean}
\item \textbf{\nolinkurl{DQ18}}\hypertarget{lh:DQ18}{}\enspace{\ttfamily DecisionQuotient.\allowbreak{}StochasticSequential.\allowbreak{}static\_\allowbreak{}simpler\_\allowbreak{}than\_\allowbreak{}stochastic} {\tiny\ttfamily DecisionQuotient/\allowbreak StochasticSequential/\allowbreak CrossRegime.lean}
\item \textbf{\nolinkurl{DQ19}}\hypertarget{lh:DQ19}{}\enspace{\ttfamily DecisionQuotient.\allowbreak{}StochasticSequential.\allowbreak{}stochastic\_\allowbreak{}simpler\_\allowbreak{}than\_\allowbreak{}sequential} {\tiny\ttfamily DecisionQuotient/\allowbreak StochasticSequential/\allowbreak CrossRegime.lean}
\item \textbf{\nolinkurl{DQ25}}\hypertarget{lh:DQ25}{}\enspace{\ttfamily DecisionQuotient.\allowbreak{}counted\_\allowbreak{}min\_\allowbreak{}sufficient\_\allowbreak{}inapproximability\_\allowbreak{}conditional} {\tiny\ttfamily DecisionQuotient/\allowbreak Hardness/\allowbreak ApproximationHardness.lean}
\item \textbf{\nolinkurl{DQ27}}\hypertarget{lh:DQ27}{}\enspace{\ttfamily DecisionQuotient.\allowbreak{}eth\_\allowbreak{}lower\_\allowbreak{}bound\_\allowbreak{}informal} {\tiny\ttfamily DecisionQuotient/\allowbreak Hardness/\allowbreak ETH.lean}
\item \textbf{\nolinkurl{DQ30}}\hypertarget{lh:DQ30}{}\enspace{\ttfamily DecisionQuotient.\allowbreak{}hard\_\allowbreak{}family\_\allowbreak{}all\_\allowbreak{}coords\_\allowbreak{}core} {\tiny\ttfamily DecisionQuotient/\allowbreak ClaimClosure.lean}
\item \textbf{\nolinkurl{DQ32}}\hypertarget{lh:DQ32}{}\enspace{\ttfamily DecisionQuotient.\allowbreak{}min\_\allowbreak{}sufficient\_\allowbreak{}factor\_\allowbreak{}inapprox\_\allowbreak{}of\_\allowbreak{}set\_\allowbreak{}cover\_\allowbreak{}factor\_\allowbreak{}inapprox} {\tiny\ttfamily DecisionQuotient/\allowbreak Hardness/\allowbreak ApproximationHardness.lean}
\item \textbf{\nolinkurl{DQ34}}\hypertarget{lh:DQ34}{}\enspace{\ttfamily DecisionQuotient.\allowbreak{}min\_\allowbreak{}sufficient\_\allowbreak{}set\_\allowbreak{}cover\_\allowbreak{}equiv} {\tiny\ttfamily DecisionQuotient/\allowbreak Hardness/\allowbreak ApproximationHardness.lean}
\item \textbf{\nolinkurl{DQ36}}\hypertarget{lh:DQ36}{}\enspace{\ttfamily ClaimClosure.\allowbreak{}minsuff\_\allowbreak{}conp\_\allowbreak{}complete\_\allowbreak{}conditional} {\tiny\ttfamily DecisionQuotient/\allowbreak ClaimClosure.lean}
\item \textbf{\nolinkurl{DQ40}}\hypertarget{lh:DQ40}{}\enspace{\ttfamily ClaimClosure.\allowbreak{}sufficiency\_\allowbreak{}conp\_\allowbreak{}complete\_\allowbreak{}conditional} {\tiny\ttfamily DecisionQuotient/\allowbreak ClaimClosure.lean}
\item \textbf{\nolinkurl{DQ45}}\hypertarget{lh:DQ45}{}\enspace{\ttfamily DecisionQuotient.\allowbreak{}tractable\_\allowbreak{}tree\_\allowbreak{}core} {\tiny\ttfamily DecisionQuotient/\allowbreak ClaimClosure.lean}
\item \textbf{\nolinkurl{DQ51}}\hypertarget{lh:DQ51}{}\enspace{\ttfamily DecisionQuotient.\allowbreak{}bounded\_\allowbreak{}actions\_\allowbreak{}tractable} {\tiny\ttfamily DecisionQuotient/\allowbreak Summary.lean}
\item \textbf{\nolinkurl{DQ57}}\hypertarget{lh:DQ57}{}\enspace{\ttfamily DecisionQuotient.\allowbreak{}separable\_\allowbreak{}utility\_\allowbreak{}tractable} {\tiny\ttfamily DecisionQuotient/\allowbreak Summary.lean}
\item \textbf{\nolinkurl{DQ62}}\hypertarget{lh:DQ62}{}\enspace{\ttfamily ClaimClosure.\allowbreak{}transition\_\allowbreak{}coupled\_\allowbreak{}bridge\_\allowbreak{}can\_\allowbreak{}fail\_\allowbreak{}on\_\allowbreak{}sufficiency} {\tiny\ttfamily DecisionQuotient/\allowbreak ClaimClosure.lean}
\item \textbf{\nolinkurl{DQ72}}\hypertarget{lh:DQ72}{}\enspace{\ttfamily DecisionQuotient.\allowbreak{}bounded\_\allowbreak{}actions\_\allowbreak{}complexity} {\tiny\ttfamily DecisionQuotient/\allowbreak Tractability/\allowbreak BoundedActions.lean}
\item \textbf{\nolinkurl{DQ74}}\hypertarget{lh:DQ74}{}\enspace{\ttfamily DecisionQuotient.\allowbreak{}sufficiency\_\allowbreak{}poly\_\allowbreak{}separable} {\tiny\ttfamily DecisionQuotient/\allowbreak Tractability/\allowbreak SeparableUtility.lean}
\item \textbf{\nolinkurl{DQ75}}\hypertarget{lh:DQ75}{}\enspace{\ttfamily TensorRankDecomposition.\allowbreak{}factorInvariantCheck\_\allowbreak{}iff} {\tiny\ttfamily DecisionQuotient/\allowbreak Tractability/\allowbreak SeparableUtility.lean}
\item \textbf{\nolinkurl{DQ76}}\hypertarget{lh:DQ76}{}\enspace{\ttfamily TensorRankDecomposition.\allowbreak{}factorInvariantComparisonCount\_\allowbreak{}le} {\tiny\ttfamily DecisionQuotient/\allowbreak Tractability/\allowbreak SeparableUtility.lean}
\item \textbf{\nolinkurl{DQ77}}\hypertarget{lh:DQ77}{}\enspace{\ttfamily TensorRankDecomposition.\allowbreak{}factorInvariant\_\allowbreak{}implies\_\allowbreak{}sufficient} {\tiny\ttfamily DecisionQuotient/\allowbreak Tractability/\allowbreak SeparableUtility.lean}
\item \textbf{\nolinkurl{DQ79}}\hypertarget{lh:DQ79}{}\enspace{\ttfamily DecisionQuotient.\allowbreak{}pairVariable\_\allowbreak{}card} {\tiny\ttfamily DecisionQuotient/\allowbreak Tractability/\allowbreak InsufficiencyCSP.lean}
\item \textbf{\nolinkurl{DQ80}}\hypertarget{lh:DQ80}{}\enspace{\ttfamily DecisionQuotient.\allowbreak{}canonicalInsufficiencyCSP\_\allowbreak{}satisfiable\_\allowbreak{}iff} {\tiny\ttfamily DecisionQuotient/\allowbreak Tractability/\allowbreak InsufficiencyCSP.lean}
\item \textbf{\nolinkurl{DQ81}}\hypertarget{lh:DQ81}{}\enspace{\ttfamily DecisionQuotient.\allowbreak{}sufficiency\_\allowbreak{}via\_\allowbreak{}bounded\_\allowbreak{}treewidth\_\allowbreak{}csp} {\tiny\ttfamily DecisionQuotient/\allowbreak Tractability/\allowbreak InsufficiencyCSP.lean}
\item \textbf{\nolinkurl{DQ82}}\hypertarget{lh:DQ82}{}\enspace{\ttfamily DecisionQuotient.\allowbreak{}orbitType\_\allowbreak{}eq\_\allowbreak{}iff} {\tiny\ttfamily DecisionQuotient/\allowbreak Tractability/\allowbreak Dimensional.lean}
\item \textbf{\nolinkurl{DQ83}}\hypertarget{lh:DQ83}{}\enspace{\ttfamily DecisionQuotient.\allowbreak{}symmetric\_\allowbreak{}optimalActions\_\allowbreak{}orbit\_\allowbreak{}invariant} {\tiny\ttfamily DecisionQuotient/\allowbreak Tractability/\allowbreak Dimensional.lean}
\item \textbf{\nolinkurl{DQ84}}\hypertarget{lh:DQ84}{}\enspace{\ttfamily DecisionQuotient.\allowbreak{}sufficiency\_\allowbreak{}reduces\_\allowbreak{}to\_\allowbreak{}cross\_\allowbreak{}orbit\_\allowbreak{}check} {\tiny\ttfamily DecisionQuotient/\allowbreak Tractability/\allowbreak Dimensional.lean}
\item \textbf{\nolinkurl{DQ85}}\hypertarget{lh:DQ85}{}\enspace{\ttfamily DecisionQuotient.\allowbreak{}symmetric\_\allowbreak{}orbit\_\allowbreak{}pair\_\allowbreak{}count\_\allowbreak{}bound} {\tiny\ttfamily DecisionQuotient/\allowbreak Tractability/\allowbreak Dimensional.lean}
\item \textbf{\nolinkurl{DQ86}}\hypertarget{lh:DQ86}{}\enspace{\ttfamily SetCoverInstance.\allowbreak{}duplicateSingleton\_\allowbreak{}not\_\allowbreak{}product\_\allowbreak{}closed} {\tiny\ttfamily DecisionQuotient/\allowbreak Hardness/\allowbreak SetCoverReduction.lean}
\item \textbf{\nolinkurl{DQ87}}\hypertarget{lh:DQ87}{}\enspace{\ttfamily DecisionQuotient.\allowbreak{}sufficiency\_\allowbreak{}via\_\allowbreak{}bounded\_\allowbreak{}treewidth\_\allowbreak{}csp\_\allowbreak{}steps} {\tiny\ttfamily DecisionQuotient/\allowbreak Tractability/\allowbreak InsufficiencyCSP.lean}
\item \textbf{\nolinkurl{DQ88}}\hypertarget{lh:DQ88}{}\enspace{\ttfamily DecisionQuotient.\allowbreak{}StochasticSequential.\allowbreak{}countedStochasticMinimumPreservationSearch\_\allowbreak{}spec} {\tiny\ttfamily DecisionQuotient/\allowbreak StochasticSequential/\allowbreak PreservationVariants.lean}
\item \textbf{\nolinkurl{DQ89}}\hypertarget{lh:DQ89}{}\enspace{\ttfamily DecisionQuotient.\allowbreak{}StochasticSequential.\allowbreak{}countedStochasticMinimumPreservationSearch\_\allowbreak{}steps} {\tiny\ttfamily DecisionQuotient/\allowbreak StochasticSequential/\allowbreak PreservationVariants.lean}
\item \textbf{\nolinkurl{DQ90}}\hypertarget{lh:DQ90}{}\enspace{\ttfamily StochasticSequential.\allowbreak{}static\_\allowbreak{}anchor\_\allowbreak{}implies\_\allowbreak{}stochastic\_\allowbreak{}anchor\_\allowbreak{}preservation\_\allowbreak{}of\_\allowbreak{}positive\_\allowbreak{}anchor\_\allowbreak{}fiber} {\tiny\ttfamily DecisionQuotient/\allowbreak StochasticSequential/\allowbreak PreservationVariants.lean}
\item \textbf{\nolinkurl{DQ91}}\hypertarget{lh:DQ91}{}\enspace{\ttfamily StochasticSequential.\allowbreak{}static\_\allowbreak{}sufficiency\_\allowbreak{}implies\_\allowbreak{}stochastic\_\allowbreak{}preservation\_\allowbreak{}of\_\allowbreak{}positive\_\allowbreak{}fiber\_\allowbreak{}support} {\tiny\ttfamily DecisionQuotient/\allowbreak StochasticSequential/\allowbreak PreservationVariants.lean}
\item \textbf{\nolinkurl{DQ92}}\hypertarget{lh:DQ92}{}\enspace{\ttfamily DecisionQuotient.\allowbreak{}StochasticSequential.\allowbreak{}stochasticAnchorPreservation\_\allowbreak{}counted\_\allowbreak{}search\_\allowbreak{}witness} {\tiny\ttfamily DecisionQuotient/\allowbreak StochasticSequential/\allowbreak PreservationVariants.lean}
\item \textbf{\nolinkurl{DQ93}}\hypertarget{lh:DQ93}{}\enspace{\ttfamily DecisionQuotient.\allowbreak{}StochasticSequential.\allowbreak{}stochasticMinimumPreservation\_\allowbreak{}counted\_\allowbreak{}search\_\allowbreak{}witness} {\tiny\ttfamily DecisionQuotient/\allowbreak StochasticSequential/\allowbreak PreservationVariants.lean}
\item \textbf{\nolinkurl{DQ94}}\hypertarget{lh:DQ94}{}\enspace{\ttfamily StochasticSequential.\allowbreak{}stochastic\_\allowbreak{}anchor\_\allowbreak{}preservation\_\allowbreak{}iff\_\allowbreak{}static\_\allowbreak{}anchor\_\allowbreak{}of\_\allowbreak{}full\_\allowbreak{}support} {\tiny\ttfamily DecisionQuotient/\allowbreak StochasticSequential/\allowbreak PreservationVariants.lean}
\item \textbf{\nolinkurl{DQ95}}\hypertarget{lh:DQ95}{}\enspace{\ttfamily StochasticSequential.\allowbreak{}stochastic\_\allowbreak{}minimum\_\allowbreak{}preservation\_\allowbreak{}iff\_\allowbreak{}static\_\allowbreak{}of\_\allowbreak{}full\_\allowbreak{}support} {\tiny\ttfamily DecisionQuotient/\allowbreak StochasticSequential/\allowbreak PreservationVariants.lean}
\item \textbf{\nolinkurl{DQ96}}\hypertarget{lh:DQ96}{}\enspace{\ttfamily StochasticSequential.\allowbreak{}stochastic\_\allowbreak{}minimum\_\allowbreak{}preservation\_\allowbreak{}static\_\allowbreak{}relevant\_\allowbreak{}card\_\allowbreak{}le} {\tiny\ttfamily DecisionQuotient/\allowbreak StochasticSequential/\allowbreak PreservationVariants.lean}
\item \textbf{\nolinkurl{DQ97}}\hypertarget{lh:DQ97}{}\enspace{\ttfamily StochasticSequential.\allowbreak{}stochastic\_\allowbreak{}preservation\_\allowbreak{}contains\_\allowbreak{}static\_\allowbreak{}relevant} {\tiny\ttfamily DecisionQuotient/\allowbreak StochasticSequential/\allowbreak PreservationVariants.lean}
\item \textbf{\nolinkurl{DQ98}}\hypertarget{lh:DQ98}{}\enspace{\ttfamily DecisionQuotient.\allowbreak{}anchor\_\allowbreak{}sufficiency\_\allowbreak{}sigma2p} {\tiny\ttfamily DecisionQuotient/\allowbreak Hardness.lean}
\item \textbf{\nolinkurl{DSC1}}\hypertarget{lh:DSC1}{}\enspace{\ttfamily DecisionQuotient.\allowbreak{}sufficient\_\allowbreak{}iff\_\allowbreak{}hits\_\allowbreak{}discernibility} {\tiny\ttfamily DecisionQuotient/\allowbreak Discernibility.lean}
\item \textbf{\nolinkurl{EF1}}\hypertarget{lh:EF1}{}\enspace{\ttfamily CoordinateSetoid.\allowbreak{}productClosure\_\allowbreak{}relevance} {\tiny\ttfamily DecisionQuotient/\allowbreak EpsilonAdmissibility.lean}
\item \textbf{\nolinkurl{EF2}}\hypertarget{lh:EF2}{}\enspace{\ttfamily EpsilonAdmissibility.\allowbreak{}exists\_\allowbreak{}breakpoint\_\allowbreak{}representation} {\tiny\ttfamily DecisionQuotient/\allowbreak EpsilonAdmissibility.lean}
\item \textbf{\nolinkurl{EF3}}\hypertarget{lh:EF3}{}\enspace{\ttfamily EpsilonAdmissibility.\allowbreak{}admissibleFamily\_\allowbreak{}card\_\allowbreak{}bound} {\tiny\ttfamily DecisionQuotient/\allowbreak EpsilonAdmissibility.lean}
\item \textbf{\nolinkurl{EF4}}\hypertarget{lh:EF4}{}\enspace{\ttfamily EpsilonAdmissibility.\allowbreak{}KernelNonmonotonicity.\allowbreak{}merged\_\allowbreak{}at\_\allowbreak{}zero} {\tiny\ttfamily DecisionQuotient/\allowbreak EpsilonAdmissibility.lean}
\item \textbf{\nolinkurl{EF5}}\hypertarget{lh:EF5}{}\enspace{\ttfamily EpsilonAdmissibility.\allowbreak{}KernelNonmonotonicity.\allowbreak{}split\_\allowbreak{}at\_\allowbreak{}three\_\allowbreak{}halves} {\tiny\ttfamily DecisionQuotient/\allowbreak EpsilonAdmissibility.lean}
\item \textbf{\nolinkurl{EF6}}\hypertarget{lh:EF6}{}\enspace{\ttfamily EpsilonAdmissibility.\allowbreak{}KernelNonmonotonicity.\allowbreak{}merged\_\allowbreak{}at\_\allowbreak{}two} {\tiny\ttfamily DecisionQuotient/\allowbreak EpsilonAdmissibility.lean}
\item \textbf{\nolinkurl{EF7}}\hypertarget{lh:EF7}{}\enspace{\ttfamily EpsilonAdmissibility.\allowbreak{}KernelNonmonotonicity.\allowbreak{}opposite\_\allowbreak{}optimizers} {\tiny\ttfamily DecisionQuotient/\allowbreak EpsilonAdmissibility.lean}
\item \textbf{\nolinkurl{EF8}}\hypertarget{lh:EF8}{}\enspace{\ttfamily EpsilonAdmissibility.\allowbreak{}KernelNonmonotonicity.\allowbreak{}opposite\_\allowbreak{}merge\_\allowbreak{}at\_\allowbreak{}one} {\tiny\ttfamily DecisionQuotient/\allowbreak EpsilonAdmissibility.lean}
\item \textbf{\nolinkurl{EF9}}\hypertarget{lh:EF9}{}\enspace{\ttfamily EpsilonAdmissibility.\allowbreak{}toleranceKernelMeet\_\allowbreak{}eq\_\allowbreak{}actionGapKernel} {\tiny\ttfamily DecisionQuotient/\allowbreak EpsilonAdmissibility.lean}
\item \textbf{\nolinkurl{EF10}}\hypertarget{lh:EF10}{}\enspace{\ttfamily EpsilonAdmissibility.\allowbreak{}admissibleMap\_\allowbreak{}eq\_\allowbreak{}univ\_\allowbreak{}of\_\allowbreak{}globalMaxRegret\_\allowbreak{}le} {\tiny\ttfamily DecisionQuotient/\allowbreak EpsilonAdmissibility.lean}
\item \textbf{\nolinkurl{EF11}}\hypertarget{lh:EF11}{}\enspace{\ttfamily EpsilonAdmissibility.\allowbreak{}stableBelow\_\allowbreak{}of\_\allowbreak{}positive\_\allowbreak{}regret\_\allowbreak{}lower\_\allowbreak{}bound} {\tiny\ttfamily DecisionQuotient/\allowbreak EpsilonAdmissibility.lean}
\item \textbf{\nolinkurl{EF12}}\hypertarget{lh:EF12}{}\enspace{\ttfamily EpsilonAdmissibility.\allowbreak{}StabilityWitness.\allowbreak{}minimum\_\allowbreak{}positive\_\allowbreak{}regret\_\allowbreak{}is\_\allowbreak{}one} {\tiny\ttfamily DecisionQuotient/\allowbreak EpsilonAdmissibility.lean}
\item \textbf{\nolinkurl{EF13}}\hypertarget{lh:EF13}{}\enspace{\ttfamily EpsilonAdmissibility.\allowbreak{}StabilityWitness.\allowbreak{}stable\_\allowbreak{}below\_\allowbreak{}five} {\tiny\ttfamily DecisionQuotient/\allowbreak EpsilonAdmissibility.lean}
\item \textbf{\nolinkurl{EF14}}\hypertarget{lh:EF14}{}\enspace{\ttfamily EpsilonAdmissibility.\allowbreak{}StabilityWitness.\allowbreak{}split\_\allowbreak{}at\_\allowbreak{}five} {\tiny\ttfamily DecisionQuotient/\allowbreak EpsilonAdmissibility.lean}
\item \textbf{\nolinkurl{EF15}}\hypertarget{lh:EF15}{}\enspace{\ttfamily EpsilonAdmissibility.\allowbreak{}StabilityWitness.\allowbreak{}not\_\allowbreak{}stable\_\allowbreak{}above\_\allowbreak{}five} {\tiny\ttfamily DecisionQuotient/\allowbreak EpsilonAdmissibility.lean}
\item \textbf{\nolinkurl{EF17}}\hypertarget{lh:EF17}{}\enspace{\ttfamily EpsilonAdmissibility.\allowbreak{}StabilityReduction.\allowbreak{}oneClassAt\_\allowbreak{}one\_\allowbreak{}iff\_\allowbreak{}forall} {\tiny\ttfamily DecisionQuotient/\allowbreak EpsilonAdmissibility.lean}
\item \textbf{\nolinkurl{EF18}}\hypertarget{lh:EF18}{}\enspace{\ttfamily EpsilonAdmissibility.\allowbreak{}FullProductToleranceSelector.\allowbreak{}sufficientAt\_\allowbreak{}iff\_\allowbreak{}forall} {\tiny\ttfamily DecisionQuotient/\allowbreak EpsilonAdmissibility.lean}
\item \textbf{\nolinkurl{EH2}}\hypertarget{lh:EH2}{}\enspace{\ttfamily StochasticSequential.\allowbreak{}existential\_\allowbreak{}anchor\_\allowbreak{}hard\_\allowbreak{}honest} {\tiny\ttfamily DecisionQuotient/\allowbreak StochasticSequential/\allowbreak ExistentialHardness.lean}
\item \textbf{\nolinkurl{EH3}}\hypertarget{lh:EH3}{}\enspace{\ttfamily StochasticSequential.\allowbreak{}existential\_\allowbreak{}anchor\_\allowbreak{}query\_\allowbreak{}family\_\allowbreak{}hard\_\allowbreak{}honest} {\tiny\ttfamily DecisionQuotient/\allowbreak StochasticSequential/\allowbreak ExistentialHardness.lean}
\item \textbf{\nolinkurl{EH4}}\hypertarget{lh:EH4}{}\enspace{\ttfamily StochasticSequential.\allowbreak{}existential\_\allowbreak{}anchor\_\allowbreak{}np\_\allowbreak{}over\_\allowbreak{}ppstyle\_\allowbreak{}hard\_\allowbreak{}honest} {\tiny\ttfamily DecisionQuotient/\allowbreak StochasticSequential/\allowbreak ExistentialHardness.lean}
\item \textbf{\nolinkurl{EH5}}\hypertarget{lh:EH5}{}\enspace{\ttfamily StochasticSequential.\allowbreak{}existential\_\allowbreak{}anchor\_\allowbreak{}query\_\allowbreak{}family\_\allowbreak{}np\_\allowbreak{}over\_\allowbreak{}ppstyle\_\allowbreak{}hard\_\allowbreak{}honest} {\tiny\ttfamily DecisionQuotient/\allowbreak StochasticSequential/\allowbreak ExistentialHardness.lean}
\item \textbf{\nolinkurl{EH6}}\hypertarget{lh:EH6}{}\enspace{\ttfamily StochasticSequential.\allowbreak{}existential\_\allowbreak{}anchor\_\allowbreak{}query\_\allowbreak{}family\_\allowbreak{}np\_\allowbreak{}over\_\allowbreak{}ppstyle\_\allowbreak{}complete\_\allowbreak{}honest} {\tiny\ttfamily DecisionQuotient/\allowbreak StochasticSequential/\allowbreak ExistentialHardness.lean}
\item \textbf{\nolinkurl{EH7}}\hypertarget{lh:EH7}{}\enspace{\ttfamily StochasticSequential.\allowbreak{}existential\_\allowbreak{}decisiveness\_\allowbreak{}complement\_\allowbreak{}np\_\allowbreak{}over\_\allowbreak{}ppstyle\_\allowbreak{}hard\_\allowbreak{}honest} {\tiny\ttfamily DecisionQuotient/\allowbreak StochasticSequential/\allowbreak ExistentialHardness.lean}
\item \textbf{\nolinkurl{EH8}}\hypertarget{lh:EH8}{}\enspace{\ttfamily StochasticSequential.\allowbreak{}existential\_\allowbreak{}decisiveness\_\allowbreak{}complement\_\allowbreak{}np\_\allowbreak{}over\_\allowbreak{}ppstyle\_\allowbreak{}complete\_\allowbreak{}honest} {\tiny\ttfamily DecisionQuotient/\allowbreak StochasticSequential/\allowbreak ExistentialHardness.lean}
\item \textbf{\nolinkurl{EH9}}\hypertarget{lh:EH9}{}\enspace{\ttfamily StochasticSequential.\allowbreak{}existential\_\allowbreak{}decisiveness\_\allowbreak{}query\_\allowbreak{}family\_\allowbreak{}np\_\allowbreak{}over\_\allowbreak{}ppstyle\_\allowbreak{}hard\_\allowbreak{}honest} {\tiny\ttfamily DecisionQuotient/\allowbreak StochasticSequential/\allowbreak ExistentialHardness.lean}
\item \textbf{\nolinkurl{EH10}}\hypertarget{lh:EH10}{}\enspace{\ttfamily StochasticSequential.\allowbreak{}existential\_\allowbreak{}decisiveness\_\allowbreak{}query\_\allowbreak{}family\_\allowbreak{}np\_\allowbreak{}over\_\allowbreak{}ppstyle\_\allowbreak{}complete\_\allowbreak{}honest} {\tiny\ttfamily DecisionQuotient/\allowbreak StochasticSequential/\allowbreak ExistentialHardness.lean}
\item \textbf{\nolinkurl{EX1}}\hypertarget{lh:EX1}{}\enspace{\ttfamily Examples.\allowbreak{}pomdp\_\allowbreak{}reduction\_\allowbreak{}to\_\allowbreak{}preservation} {\tiny\ttfamily DecisionQuotient/\allowbreak Examples/\allowbreak PreservationExamples.lean}
\item \textbf{\nolinkurl{EX4}}\hypertarget{lh:EX4}{}\enspace{\ttfamily Examples.\allowbreak{}tied\_\allowbreak{}empty\_\allowbreak{}preserves} {\tiny\ttfamily DecisionQuotient/\allowbreak Examples/\allowbreak ConcreteExamples.lean}
\item \textbf{\nolinkurl{EX5}}\hypertarget{lh:EX5}{}\enspace{\ttfamily Examples.\allowbreak{}tied\_\allowbreak{}empty\_\allowbreak{}not\_\allowbreak{}decisive} {\tiny\ttfamily DecisionQuotient/\allowbreak Examples/\allowbreak ConcreteExamples.lean}
\item \textbf{\nolinkurl{EX6}}\hypertarget{lh:EX6}{}\enspace{\ttfamily Examples.\allowbreak{}toy\_\allowbreak{}empty\_\allowbreak{}decisive} {\tiny\ttfamily DecisionQuotient/\allowbreak Examples/\allowbreak ConcreteExamples.lean}
\item \textbf{\nolinkurl{EX7}}\hypertarget{lh:EX7}{}\enspace{\ttfamily Examples.\allowbreak{}toy\_\allowbreak{}empty\_\allowbreak{}not\_\allowbreak{}preserving} {\tiny\ttfamily DecisionQuotient/\allowbreak Examples/\allowbreak ConcreteExamples.lean}
\item \textbf{\nolinkurl{IC1}}\hypertarget{lh:IC1}{}\enspace{\ttfamily IntersectionClosedFamily.\allowbreak{}core\_\allowbreak{}qualifies} {\tiny\ttfamily DecisionQuotient/\allowbreak StochasticSequential/\allowbreak IntersectionSeparation.lean}
\item \textbf{\nolinkurl{IC2}}\hypertarget{lh:IC2}{}\enspace{\ttfamily IntersectionClosedFamily.\allowbreak{}core\_\allowbreak{}subset} {\tiny\ttfamily DecisionQuotient/\allowbreak StochasticSequential/\allowbreak IntersectionSeparation.lean}
\item \textbf{\nolinkurl{IC3}}\hypertarget{lh:IC3}{}\enspace{\ttfamily IntersectionClosedFamily.\allowbreak{}existsUnique\_\allowbreak{}least} {\tiny\ttfamily DecisionQuotient/\allowbreak StochasticSequential/\allowbreak IntersectionSeparation.lean}
\item \textbf{\nolinkurl{IC4}}\hypertarget{lh:IC4}{}\enspace{\ttfamily IntersectionClosedFamily.\allowbreak{}exists\_\allowbreak{}card\_\allowbreak{}le\_\allowbreak{}iff\_\allowbreak{}core\_\allowbreak{}card\_\allowbreak{}le} {\tiny\ttfamily DecisionQuotient/\allowbreak StochasticSequential/\allowbreak IntersectionSeparation.lean}
\item \textbf{\nolinkurl{IS1}}\hypertarget{lh:IS1}{}\enspace{\ttfamily InputSettings.\allowbreak{}replacement\_\allowbreak{}preserves\_\allowbreak{}iff\_\allowbreak{}context\_\allowbreak{}agreement} {\tiny\ttfamily DecisionQuotient/\allowbreak InputSettings.lean}
\item \textbf{\nolinkurl{IS2}}\hypertarget{lh:IS2}{}\enspace{\ttfamily InputSettings.\allowbreak{}reset\_\allowbreak{}preserves\_\allowbreak{}iff\_\allowbreak{}fiber\_\allowbreak{}constant} {\tiny\ttfamily DecisionQuotient/\allowbreak InputSettings.lean}
\item \textbf{\nolinkurl{IS3}}\hypertarget{lh:IS3}{}\enspace{\ttfamily InputSettings.\allowbreak{}reset\_\allowbreak{}preservation\_\allowbreak{}baseline\_\allowbreak{}independent} {\tiny\ttfamily DecisionQuotient/\allowbreak InputSettings.lean}
\item \textbf{\nolinkurl{IS4}}\hypertarget{lh:IS4}{}\enspace{\ttfamily DecisionProblem.\allowbreak{}reset\_\allowbreak{}preserves\_\allowbreak{}iff\_\allowbreak{}projection\_\allowbreak{}factorization} {\tiny\ttfamily DecisionQuotient/\allowbreak InputSettings.lean}
\item \textbf{\nolinkurl{IS5}}\hypertarget{lh:IS5}{}\enspace{\ttfamily InputSettings.\allowbreak{}endpoint\_\allowbreak{}agreement\_\allowbreak{}does\_\allowbreak{}not\_\allowbreak{}imply\_\allowbreak{}constancy} {\tiny\ttfamily DecisionQuotient/\allowbreak InputSettings.lean}
\item \textbf{\nolinkurl{IS6}}\hypertarget{lh:IS6}{}\enspace{\ttfamily InputSettings.\allowbreak{}count\_\allowbreak{}constancy\_\allowbreak{}iff\_\allowbreak{}optional\_\allowbreak{}constancy} {\tiny\ttfamily DecisionQuotient/\allowbreak InputSettings.lean}
\item \textbf{\nolinkurl{IS7}}\hypertarget{lh:IS7}{}\enspace{\ttfamily InputSettings.\allowbreak{}optional\_\allowbreak{}endpoint\_\allowbreak{}agreement\_\allowbreak{}not\_\allowbreak{}universal\_\allowbreak{}reset} {\tiny\ttfamily DecisionQuotient/\allowbreak InputSettings.lean}
\item \textbf{\nolinkurl{IS8}}\hypertarget{lh:IS8}{}\enspace{\ttfamily InputSettings.\allowbreak{}each\_\allowbreak{}optional\_\allowbreak{}channel\_\allowbreak{}can\_\allowbreak{}change\_\allowbreak{}response} {\tiny\ttfamily DecisionQuotient/\allowbreak InputSettings.lean}
\item \textbf{\nolinkurl{KP1}}\hypertarget{lh:KP1}{}\enspace{\ttfamily DecisionProblem.\allowbreak{}supportsOptimalSelector\_\allowbreak{}iff\_\allowbreak{}common\_\allowbreak{}on\_\allowbreak{}fibers} {\tiny\ttfamily DecisionQuotient/\allowbreak SelectorRecovery.lean}
\item \textbf{\nolinkurl{KP2}}\hypertarget{lh:KP2}{}\enspace{\ttfamily DecisionProblem.\allowbreak{}preservesOptStrong\_\allowbreak{}implies\_\allowbreak{}supportsOptimalSelector} {\tiny\ttfamily DecisionQuotient/\allowbreak SelectorRecovery.lean}
\item \textbf{\nolinkurl{KP3}}\hypertarget{lh:KP3}{}\enspace{\ttfamily SelectorSeparation.\allowbreak{}collapse\_\allowbreak{}supports\_\allowbreak{}optimal\_\allowbreak{}selector} {\tiny\ttfamily DecisionQuotient/\allowbreak SelectorRecovery.lean}
\item \textbf{\nolinkurl{KP4}}\hypertarget{lh:KP4}{}\enspace{\ttfamily SelectorSeparation.\allowbreak{}collapse\_\allowbreak{}does\_\allowbreak{}not\_\allowbreak{}recover\_\allowbreak{}optimizer} {\tiny\ttfamily DecisionQuotient/\allowbreak SelectorRecovery.lean}
\item \textbf{\nolinkurl{KP5}}\hypertarget{lh:KP5}{}\enspace{\ttfamily SelectorNoncanonical.\allowbreak{}mergeLeft\_\allowbreak{}supports\_\allowbreak{}optimal\_\allowbreak{}selector} {\tiny\ttfamily DecisionQuotient/\allowbreak SelectorRecovery.lean}
\item \textbf{\nolinkurl{KP6}}\hypertarget{lh:KP6}{}\enspace{\ttfamily SelectorNoncanonical.\allowbreak{}mergeRight\_\allowbreak{}supports\_\allowbreak{}optimal\_\allowbreak{}selector} {\tiny\ttfamily DecisionQuotient/\allowbreak SelectorRecovery.lean}
\item \textbf{\nolinkurl{KP7}}\hypertarget{lh:KP7}{}\enspace{\ttfamily SelectorNoncanonical.\allowbreak{}mergeLeft\_\allowbreak{}not\_\allowbreak{}refines\_\allowbreak{}mergeRight} {\tiny\ttfamily DecisionQuotient/\allowbreak SelectorRecovery.lean}
\item \textbf{\nolinkurl{KP8}}\hypertarget{lh:KP8}{}\enspace{\ttfamily SelectorNoncanonical.\allowbreak{}mergeRight\_\allowbreak{}not\_\allowbreak{}refines\_\allowbreak{}mergeLeft} {\tiny\ttfamily DecisionQuotient/\allowbreak SelectorRecovery.lean}
\item \textbf{\nolinkurl{KP9}}\hypertarget{lh:KP9}{}\enspace{\ttfamily SelectorNoncanonical.\allowbreak{}no\_\allowbreak{}selector\_\allowbreak{}preserving\_\allowbreak{}common\_\allowbreak{}coarsening} {\tiny\ttfamily DecisionQuotient/\allowbreak SelectorRecovery.lean}
\item \textbf{\nolinkurl{KP10}}\hypertarget{lh:KP10}{}\enspace{\ttfamily SelectorSymmetry.\allowbreak{}no\_\allowbreak{}relabeling\_\allowbreak{}equivariant\_\allowbreak{}selector} {\tiny\ttfamily DecisionQuotient/\allowbreak SelectorRecovery.lean}
\item \textbf{\nolinkurl{KP11}}\hypertarget{lh:KP11}{}\enspace{\ttfamily DecisionQuotient.\allowbreak{}recoversDeclaredSelector\_\allowbreak{}iff\_\allowbreak{}kernel\_\allowbreak{}refines} {\tiny\ttfamily DecisionQuotient/\allowbreak SelectorRecovery.lean}
\item \textbf{\nolinkurl{KP12}}\hypertarget{lh:KP12}{}\enspace{\ttfamily DecisionProblem.\allowbreak{}recovers\_\allowbreak{}optimal\_\allowbreak{}declared\_\allowbreak{}selector\_\allowbreak{}implies\_\allowbreak{}supports} {\tiny\ttfamily DecisionQuotient/\allowbreak SelectorRecovery.lean}
\item \textbf{\nolinkurl{KP13}}\hypertarget{lh:KP13}{}\enspace{\ttfamily DecisionQuotient.\allowbreak{}declaredSelectorQuotient\_\allowbreak{}has\_\allowbreak{}unique\_\allowbreak{}factorization} {\tiny\ttfamily DecisionQuotient/\allowbreak SelectorRecovery.lean}
\item \textbf{\nolinkurl{LMS1}}\hypertarget{lh:LMS1}{}\enspace{\ttfamily ListedMinimumSufficiency.\allowbreak{}canonical\_\allowbreak{}sufficient\_\allowbreak{}iff} {\tiny\ttfamily DecisionQuotient/\allowbreak Hardness/\allowbreak ListedMinimumSufficiency.lean}
\item \textbf{\nolinkurl{LMS2}}\hypertarget{lh:LMS2}{}\enspace{\ttfamily ListedMinimumSufficiency.\allowbreak{}exists\_\allowbreak{}small\_\allowbreak{}sufficient\_\allowbreak{}iff} {\tiny\ttfamily DecisionQuotient/\allowbreak Hardness/\allowbreak ListedMinimumSufficiency.lean}
\item \textbf{\nolinkurl{OU1}}\hypertarget{lh:OU1}{}\enspace{\ttfamily StochasticSequential.\allowbreak{}stochastic\_\allowbreak{}anchor\_\allowbreak{}query\_\allowbreak{}fits\_\allowbreak{}np\_\allowbreak{}over\_\allowbreak{}ppstyle} {\tiny\ttfamily DecisionQuotient/\allowbreak StochasticSequential/\allowbreak OracleUpperBounds.lean}
\item \textbf{\nolinkurl{OU8}}\hypertarget{lh:OU8}{}\enspace{\ttfamily StochasticSequential.\allowbreak{}fitsCoNPOverPPStyle\_\allowbreak{}of\_\allowbreak{}no\_\allowbreak{}witness} {\tiny\ttfamily DecisionQuotient/\allowbreak StochasticSequential/\allowbreak OracleUpperBounds.lean}
\item \textbf{\nolinkurl{OU9}}\hypertarget{lh:OU9}{}\enspace{\ttfamily StochasticSequential.\allowbreak{}stochastic\_\allowbreak{}decisiveness\_\allowbreak{}query\_\allowbreak{}fits\_\allowbreak{}conp\_\allowbreak{}over\_\allowbreak{}ppstyle} {\tiny\ttfamily DecisionQuotient/\allowbreak StochasticSequential/\allowbreak OracleUpperBounds.lean}
\item \textbf{\nolinkurl{OU10}}\hypertarget{lh:OU10}{}\enspace{\ttfamily StochasticSequential.\allowbreak{}stochastic\_\allowbreak{}decisiveness\_\allowbreak{}complement\_\allowbreak{}fits\_\allowbreak{}np\_\allowbreak{}over\_\allowbreak{}ppstyle} {\tiny\ttfamily DecisionQuotient/\allowbreak StochasticSequential/\allowbreak OracleUpperBounds.lean}
\item \textbf{\nolinkurl{OU11}}\hypertarget{lh:OU11}{}\enspace{\ttfamily StochasticSequential.\allowbreak{}stochastic\_\allowbreak{}decisiveness\_\allowbreak{}scoped\_\allowbreak{}oracle\_\allowbreak{}bounds} {\tiny\ttfamily DecisionQuotient/\allowbreak StochasticSequential/\allowbreak OracleUpperBounds.lean}
\item \textbf{\nolinkurl{OU12}}\hypertarget{lh:OU12}{}\enspace{\ttfamily StochasticSequential.\allowbreak{}stochastic\_\allowbreak{}preservation\_\allowbreak{}explicit\_\allowbreak{}summary} {\tiny\ttfamily DecisionQuotient/\allowbreak StochasticSequential/\allowbreak OracleUpperBounds.lean}
\item \textbf{\nolinkurl{OU13}}\hypertarget{lh:OU13}{}\enspace{\ttfamily StochasticSequential.\allowbreak{}FitsNPOverCoNPOverPPStyle} {\tiny\ttfamily DecisionQuotient/\allowbreak StochasticSequential/\allowbreak OracleUpperBounds.lean}
\item \textbf{\nolinkurl{OU14}}\hypertarget{lh:OU14}{}\enspace{\ttfamily StochasticSequential.\allowbreak{}stochastic\_\allowbreak{}preservation\_\allowbreak{}query\_\allowbreak{}fits\_\allowbreak{}conp\_\allowbreak{}over\_\allowbreak{}ppstyle} {\tiny\ttfamily DecisionQuotient/\allowbreak StochasticSequential/\allowbreak OracleUpperBounds.lean}
\item \textbf{\nolinkurl{OU15}}\hypertarget{lh:OU15}{}\enspace{\ttfamily StochasticSequential.\allowbreak{}stochastic\_\allowbreak{}preservation\_\allowbreak{}complement\_\allowbreak{}fits\_\allowbreak{}np\_\allowbreak{}over\_\allowbreak{}ppstyle} {\tiny\ttfamily DecisionQuotient/\allowbreak StochasticSequential/\allowbreak OracleUpperBounds.lean}
\item \textbf{\nolinkurl{OU16}}\hypertarget{lh:OU16}{}\enspace{\ttfamily StochasticSequential.\allowbreak{}stochastic\_\allowbreak{}minimum\_\allowbreak{}preservation\_\allowbreak{}query\_\allowbreak{}fits\_\allowbreak{}np\_\allowbreak{}over\_\allowbreak{}conp\_\allowbreak{}over\_\allowbreak{}ppstyle} {\tiny\ttfamily DecisionQuotient/\allowbreak StochasticSequential/\allowbreak OracleUpperBounds.lean}
\item \textbf{\nolinkurl{OU17}}\hypertarget{lh:OU17}{}\enspace{\ttfamily StochasticSequential.\allowbreak{}stochastic\_\allowbreak{}anchor\_\allowbreak{}preservation\_\allowbreak{}query\_\allowbreak{}fits\_\allowbreak{}np\_\allowbreak{}over\_\allowbreak{}conp\_\allowbreak{}over\_\allowbreak{}ppstyle} {\tiny\ttfamily DecisionQuotient/\allowbreak StochasticSequential/\allowbreak OracleUpperBounds.lean}
\item \textbf{\nolinkurl{OU18}}\hypertarget{lh:OU18}{}\enspace{\ttfamily StochasticSequential.\allowbreak{}stochastic\_\allowbreak{}preservation\_\allowbreak{}general\_\allowbreak{}scoped\_\allowbreak{}oracle\_\allowbreak{}bounds} {\tiny\ttfamily DecisionQuotient/\allowbreak StochasticSequential/\allowbreak OracleUpperBounds.lean}
\item \textbf{\nolinkurl{OU36}}\hypertarget{lh:OU36}{}\enspace{\ttfamily StochasticSequential.\allowbreak{}KarpReduction4b} {\tiny\ttfamily DecisionQuotient/\allowbreak StochasticSequential/\allowbreak PolynomialReduction.lean}
\item \textbf{\nolinkurl{OU37}}\hypertarget{lh:OU37}{}\enspace{\ttfamily StochasticSequential.\allowbreak{}KarpReduction4b.\allowbreak{}to\_\allowbreak{}manyOneReducesPoly} {\tiny\ttfamily DecisionQuotient/\allowbreak StochasticSequential/\allowbreak PolynomialReduction.lean}
\item \textbf{\nolinkurl{OU38}}\hypertarget{lh:OU38}{}\enspace{\ttfamily StochasticSequential.\allowbreak{}stochastic\_\allowbreak{}minimum\_\allowbreak{}decisiveness\_\allowbreak{}query\_\allowbreak{}fits\_\allowbreak{}np\_\allowbreak{}over\_\allowbreak{}conp\_\allowbreak{}over\_\allowbreak{}ppstyle} {\tiny\ttfamily DecisionQuotient/\allowbreak StochasticSequential/\allowbreak OracleUpperBounds.lean}
\item \textbf{\nolinkurl{PA3}}\hypertarget{lh:PA3}{}\enspace{\ttfamily Physics.\allowbreak{}AnchorChecks.\allowbreak{}stochastic\_\allowbreak{}anchor\_\allowbreak{}check\_\allowbreak{}iff\_\allowbreak{}exists\_\allowbreak{}anchor\_\allowbreak{}singleton} {\tiny\ttfamily DecisionQuotient/\allowbreak Physics/\allowbreak AnchorChecks.lean}
\item \textbf{\nolinkurl{PBA1}}\hypertarget{lh:PBA1}{}\enspace{\ttfamily DecisionQuotient.\allowbreak{}pairedAccepts\_\allowbreak{}of\_\allowbreak{}tautology} {\tiny\ttfamily DecisionQuotient/\allowbreak Hardness/\allowbreak PairedBooleanApproximation.lean}
\item \textbf{\nolinkurl{PBA2}}\hypertarget{lh:PBA2}{}\enspace{\ttfamily DecisionQuotient.\allowbreak{}paired\_\allowbreak{}gate\_\allowbreak{}relevant} {\tiny\ttfamily DecisionQuotient/\allowbreak Hardness/\allowbreak PairedBooleanApproximation.lean}
\item \textbf{\nolinkurl{PBA3}}\hypertarget{lh:PBA3}{}\enspace{\ttfamily DecisionQuotient.\allowbreak{}paired\_\allowbreak{}left\_\allowbreak{}relevant\_\allowbreak{}of\_\allowbreak{}not\_\allowbreak{}tautology} {\tiny\ttfamily DecisionQuotient/\allowbreak Hardness/\allowbreak PairedBooleanApproximation.lean}
\item \textbf{\nolinkurl{PBA4}}\hypertarget{lh:PBA4}{}\enspace{\ttfamily DecisionQuotient.\allowbreak{}paired\_\allowbreak{}right\_\allowbreak{}relevant\_\allowbreak{}of\_\allowbreak{}not\_\allowbreak{}tautology} {\tiny\ttfamily DecisionQuotient/\allowbreak Hardness/\allowbreak PairedBooleanApproximation.lean}
\item \textbf{\nolinkurl{PBA5}}\hypertarget{lh:PBA5}{}\enspace{\ttfamily DecisionQuotient.\allowbreak{}paired\_\allowbreak{}gate\_\allowbreak{}sufficient\_\allowbreak{}of\_\allowbreak{}tautology} {\tiny\ttfamily DecisionQuotient/\allowbreak Hardness/\allowbreak PairedBooleanApproximation.lean}
\item \textbf{\nolinkurl{PBA6}}\hypertarget{lh:PBA6}{}\enspace{\ttfamily DecisionQuotient.\allowbreak{}paired\_\allowbreak{}sufficient\_\allowbreak{}iff\_\allowbreak{}univ\_\allowbreak{}of\_\allowbreak{}not\_\allowbreak{}tautology} {\tiny\ttfamily DecisionQuotient/\allowbreak Hardness/\allowbreak PairedBooleanApproximation.lean}
\item \textbf{\nolinkurl{PBA7}}\hypertarget{lh:PBA7}{}\enspace{\ttfamily DecisionQuotient.\allowbreak{}card\_\allowbreak{}pairedBoolCoord} {\tiny\ttfamily DecisionQuotient/\allowbreak Hardness/\allowbreak PairedBooleanApproximation.lean}
\item \textbf{\nolinkurl{PM1}}\hypertarget{lh:PM1}{}\enspace{\ttfamily StochasticSequential.\allowbreak{}FiniteValueTree.\allowbreak{}depthFirst\_\allowbreak{}value} {\tiny\ttfamily DecisionQuotient/\allowbreak StochasticSequential/\allowbreak MultistepPOMDPCores.lean}
\item \textbf{\nolinkurl{PM2}}\hypertarget{lh:PM2}{}\enspace{\ttfamily StochasticSequential.\allowbreak{}FiniteValueTree.\allowbreak{}depthFirst\_\allowbreak{}stack} {\tiny\ttfamily DecisionQuotient/\allowbreak StochasticSequential/\allowbreak MultistepPOMDPCores.lean}
\item \textbf{\nolinkurl{PM3}}\hypertarget{lh:PM3}{}\enspace{\ttfamily StochasticSequential.\allowbreak{}dyadic\_\allowbreak{}path\_\allowbreak{}denominator\_\allowbreak{}exponent} {\tiny\ttfamily DecisionQuotient/\allowbreak StochasticSequential/\allowbreak MultistepPOMDPCores.lean}
\item \textbf{\nolinkurl{PM4}}\hypertarget{lh:PM4}{}\enspace{\ttfamily StochasticSequential.\allowbreak{}geometricBonusTotal\_\allowbreak{}lt\_\allowbreak{}third} {\tiny\ttfamily DecisionQuotient/\allowbreak StochasticSequential/\allowbreak MultistepPOMDPCores.lean}
\item \textbf{\nolinkurl{PM5}}\hypertarget{lh:PM5}{}\enspace{\ttfamily StochasticSequential.\allowbreak{}sign\_\allowbreak{}gap\_\allowbreak{}survives\_\allowbreak{}perturbation} {\tiny\ttfamily DecisionQuotient/\allowbreak StochasticSequential/\allowbreak MultistepPOMDPCores.lean}
\item \textbf{\nolinkurl{PM6}}\hypertarget{lh:PM6}{}\enspace{\ttfamily StochasticSequential.\allowbreak{}preservation\_\allowbreak{}wrapper\_\allowbreak{}iff} {\tiny\ttfamily DecisionQuotient/\allowbreak StochasticSequential/\allowbreak MultistepPOMDPCores.lean}
\item \textbf{\nolinkurl{PM7}}\hypertarget{lh:PM7}{}\enspace{\ttfamily StochasticSequential.\allowbreak{}decisiveness\_\allowbreak{}wrapper\_\allowbreak{}iff} {\tiny\ttfamily DecisionQuotient/\allowbreak StochasticSequential/\allowbreak MultistepPOMDPCores.lean}
\item \textbf{\nolinkurl{PM8}}\hypertarget{lh:PM8}{}\enspace{\ttfamily StochasticSequential.\allowbreak{}forcing\_\allowbreak{}bank\_\allowbreak{}hides\_\allowbreak{}tag} {\tiny\ttfamily DecisionQuotient/\allowbreak StochasticSequential/\allowbreak MultistepPOMDPCores.lean}
\item \textbf{\nolinkurl{PM9}}\hypertarget{lh:PM9}{}\enspace{\ttfamily StochasticSequential.\allowbreak{}multistep\_\allowbreak{}wrapper\_\allowbreak{}core} {\tiny\ttfamily DecisionQuotient/\allowbreak StochasticSequential/\allowbreak MultistepPOMDPCores.lean}
\item \textbf{\nolinkurl{PM10}}\hypertarget{lh:PM10}{}\enspace{\ttfamily StochasticSequential.\allowbreak{}horizonTwoFiberScore\_\allowbreak{}eq\_\allowbreak{}fiberExpectedUtility} {\tiny\ttfamily DecisionQuotient/\allowbreak StochasticSequential/\allowbreak MultistepPOMDPCores.lean}
\item \textbf{\nolinkurl{PM11}}\hypertarget{lh:PM11}{}\enspace{\ttfamily StochasticSequential.\allowbreak{}horizonTwoTerminalOpt\_\allowbreak{}eq\_\allowbreak{}fiberOpt} {\tiny\ttfamily DecisionQuotient/\allowbreak StochasticSequential/\allowbreak MultistepPOMDPCores.lean}
\item \textbf{\nolinkurl{PM12}}\hypertarget{lh:PM12}{}\enspace{\ttfamily StochasticSequential.\allowbreak{}horizonTwo\_\allowbreak{}deployable\_\allowbreak{}iff\_\allowbreak{}oneStep} {\tiny\ttfamily DecisionQuotient/\allowbreak StochasticSequential/\allowbreak MultistepPOMDPCores.lean}
\item \textbf{\nolinkurl{PM13}}\hypertarget{lh:PM13}{}\enspace{\ttfamily StochasticSequential.\allowbreak{}horizonTwo\_\allowbreak{}decisive\_\allowbreak{}iff\_\allowbreak{}oneStep} {\tiny\ttfamily DecisionQuotient/\allowbreak StochasticSequential/\allowbreak MultistepPOMDPCores.lean}
\item \textbf{\nolinkurl{PM14}}\hypertarget{lh:PM14}{}\enspace{\ttfamily StochasticSequential.\allowbreak{}InterlockRegression.\allowbreak{}first\_\allowbreak{}coordinate\_\allowbreak{}not\_\allowbreak{}deployable} {\tiny\ttfamily DecisionQuotient/\allowbreak StochasticSequential/\allowbreak MultistepPOMDPCores.lean}
\item \textbf{\nolinkurl{PM15}}\hypertarget{lh:PM15}{}\enspace{\ttfamily StochasticSequential.\allowbreak{}InterlockRegression.\allowbreak{}first\_\allowbreak{}coordinate\_\allowbreak{}not\_\allowbreak{}horizonTwoDeployable} {\tiny\ttfamily DecisionQuotient/\allowbreak StochasticSequential/\allowbreak MultistepPOMDPCores.lean}
\item \textbf{\nolinkurl{QT5}}\hypertarget{lh:QT5}{}\enspace{\ttfamily DecisionProblem.\allowbreak{}quotientEquivOptRange} {\tiny\ttfamily DecisionQuotient/\allowbreak Quotient.lean}
\item \textbf{\nolinkurl{QT6}}\hypertarget{lh:QT6}{}\enspace{\ttfamily DecisionProblem.\allowbreak{}quotientEquivOptRange\_\allowbreak{}apply\_\allowbreak{}quotientMap} {\tiny\ttfamily DecisionQuotient/\allowbreak Quotient.lean}
\item \textbf{\nolinkurl{QT7}}\hypertarget{lh:QT7}{}\enspace{\ttfamily DecisionProblem.\allowbreak{}quotient\_\allowbreak{}has\_\allowbreak{}unique\_\allowbreak{}factorization} {\tiny\ttfamily DecisionQuotient/\allowbreak Quotient.lean}
\item \textbf{\nolinkurl{QT8}}\hypertarget{lh:QT8}{}\enspace{\ttfamily AdmissibleOutputSpec.\allowbreak{}asDecisionProblem\_\allowbreak{}opt\_\allowbreak{}eq\_\allowbreak{}singleton} {\tiny\ttfamily DecisionQuotient/\allowbreak Quotient.lean}
\item \textbf{\nolinkurl{QT9}}\hypertarget{lh:QT9}{}\enspace{\ttfamily AdmissibleOutputSpec.\allowbreak{}asDecisionProblem\_\allowbreak{}decisionEquiv\_\allowbreak{}iff} {\tiny\ttfamily DecisionQuotient/\allowbreak Quotient.lean}
\item \textbf{\nolinkurl{SK4}}\hypertarget{lh:SK4}{}\enspace{\ttfamily DecisionQuotient.\allowbreak{}tautology\_\allowbreak{}iff\_\allowbreak{}sufficient\_\allowbreak{}many} {\tiny\ttfamily DecisionQuotient/\allowbreak Reduction\_AllCoords.lean}
\item \textbf{\nolinkurl{SK5}}\hypertarget{lh:SK5}{}\enspace{\ttfamily DecisionQuotient.\allowbreak{}all\_\allowbreak{}coords\_\allowbreak{}relevant\_\allowbreak{}of\_\allowbreak{}not\_\allowbreak{}tautology} {\tiny\ttfamily DecisionQuotient/\allowbreak Reduction\_AllCoords.lean}
\item \textbf{\nolinkurl{SK6}}\hypertarget{lh:SK6}{}\enspace{\ttfamily DecisionQuotient.\allowbreak{}hard\_\allowbreak{}family\_\allowbreak{}srank\_\allowbreak{}eq\_\allowbreak{}n} {\tiny\ttfamily DecisionQuotient/\allowbreak Tractability/\allowbreak StructuralRank.lean}
\item \textbf{\nolinkurl{SK7}}\hypertarget{lh:SK7}{}\enspace{\ttfamily DecisionQuotient.\allowbreak{}structural\_\allowbreak{}rank\_\allowbreak{}tautology\_\allowbreak{}boundary} {\tiny\ttfamily DecisionQuotient/\allowbreak Tractability/\allowbreak StructuralRank.lean}
\item \textbf{\nolinkurl{SK8}}\hypertarget{lh:SK8}{}\enspace{\ttfamily ClaimClosure.\allowbreak{}hard\_\allowbreak{}family\_\allowbreak{}minimum\_\allowbreak{}zero\_\allowbreak{}iff\_\allowbreak{}tautology} {\tiny\ttfamily DecisionQuotient/\allowbreak ClaimClosure.lean}
\item \textbf{\nolinkurl{SYM1}}\hypertarget{lh:SYM1}{}\enspace{\ttfamily DimensionalStateSpace.\allowbreak{}card\_\allowbreak{}le\_\allowbreak{}commonMultiplicity\_\allowbreak{}of\_\allowbreak{}agreeOn} {\tiny\ttfamily DecisionQuotient/\allowbreak Tractability/\allowbreak Dimensional.lean}
\item \textbf{\nolinkurl{SYM2}}\hypertarget{lh:SYM2}{}\enspace{\ttfamily DimensionalStateSpace.\allowbreak{}commonMultiplicity\_\allowbreak{}eq\_\allowbreak{}of\_\allowbreak{}orbitType\_\allowbreak{}eq} {\tiny\ttfamily DecisionQuotient/\allowbreak Tractability/\allowbreak Dimensional.lean}
\item \textbf{\nolinkurl{SYM3}}\hypertarget{lh:SYM3}{}\enspace{\ttfamily DecisionQuotient.\allowbreak{}representativeCompatibilityCriterion} {\tiny\ttfamily DecisionQuotient/\allowbreak Tractability/\allowbreak Dimensional.lean}
\item \textbf{\nolinkurl{SYM4}}\hypertarget{lh:SYM4}{}\enspace{\ttfamily DecisionQuotient.\allowbreak{}representativeCompatibilityCriterion\_\allowbreak{}iff\_\allowbreak{}sufficient} {\tiny\ttfamily DecisionQuotient/\allowbreak Tractability/\allowbreak Dimensional.lean}
\item \textbf{\nolinkurl{SYM5}}\hypertarget{lh:SYM5}{}\enspace{\ttfamily DecisionQuotient.\allowbreak{}representativeHistogramCheck\_\allowbreak{}iff} {\tiny\ttfamily DecisionQuotient/\allowbreak Tractability/\allowbreak Dimensional.lean}
\item \textbf{\nolinkurl{SYM6}}\hypertarget{lh:SYM6}{}\enspace{\ttfamily DecisionQuotient.\allowbreak{}representativeHistogramCheck\_\allowbreak{}implies\_\allowbreak{}sufficient} {\tiny\ttfamily DecisionQuotient/\allowbreak Tractability/\allowbreak Dimensional.lean}
\item \textbf{\nolinkurl{SYM7}}\hypertarget{lh:SYM7}{}\enspace{\ttfamily DecisionQuotient.\allowbreak{}representativeHistogramComparisonCount\_\allowbreak{}le} {\tiny\ttfamily DecisionQuotient/\allowbreak Tractability/\allowbreak Dimensional.lean}
\item \textbf{\nolinkurl{UB1}}\hypertarget{lh:UB1}{}\enspace{\ttfamily DecisionProblem.\allowbreak{}opt\_\allowbreak{}factors\_\allowbreak{}through\_\allowbreak{}iff\_\allowbreak{}kernel\_\allowbreak{}refines} {\tiny\ttfamily DecisionQuotient/\allowbreak Quotient.lean}
\item \textbf{\nolinkurl{UB2}}\hypertarget{lh:UB2}{}\enspace{\ttfamily DecisionProblem.\allowbreak{}isSufficient\_\allowbreak{}iff\_\allowbreak{}opt\_\allowbreak{}factors\_\allowbreak{}through\_\allowbreak{}projection} {\tiny\ttfamily DecisionQuotient/\allowbreak FourInstances.lean}
\item \textbf{\nolinkurl{UB3}}\hypertarget{lh:UB3}{}\enspace{\ttfamily DecisionProblem.\allowbreak{}isMinimalSufficient\_\allowbreak{}iff\_\allowbreak{}optimizerLabelReduct} {\tiny\ttfamily DecisionQuotient/\allowbreak FourInstances.lean}
\item \textbf{\nolinkurl{UB4}}\hypertarget{lh:UB4}{}\enspace{\ttfamily DecisionProblem.\allowbreak{}isSufficient\_\allowbreak{}iff\_\allowbreak{}strongBackdoor\_\allowbreak{}constant} {\tiny\ttfamily DecisionQuotient/\allowbreak FourInstances.lean}
\item \textbf{\nolinkurl{UB5}}\hypertarget{lh:UB5}{}\enspace{\ttfamily DecisionProblem.\allowbreak{}strongBackdoor\_\allowbreak{}mono} {\tiny\ttfamily DecisionQuotient/\allowbreak FourInstances.lean}
\item \textbf{\nolinkurl{UB6}}\hypertarget{lh:UB6}{}\enspace{\ttfamily BackdoorSeparation.\allowbreak{}empty\_\allowbreak{}is\_\allowbreak{}strongBackdoor\_\allowbreak{}atMostTwo} {\tiny\ttfamily DecisionQuotient/\allowbreak FourInstances.lean}
\item \textbf{\nolinkurl{UB7}}\hypertarget{lh:UB7}{}\enspace{\ttfamily BackdoorSeparation.\allowbreak{}empty\_\allowbreak{}not\_\allowbreak{}sufficient} {\tiny\ttfamily DecisionQuotient/\allowbreak FourInstances.lean}
\item \textbf{\nolinkurl{UB8}}\hypertarget{lh:UB8}{}\enspace{\ttfamily QStarModel.\allowbreak{}qStar\_\allowbreak{}abstraction\_\allowbreak{}refines\_\allowbreak{}optimizer} {\tiny\ttfamily DecisionQuotient/\allowbreak FourInstances.lean}
\item \textbf{\nolinkurl{UB9}}\hypertarget{lh:UB9}{}\enspace{\ttfamily QStarModel.\allowbreak{}qStar\_\allowbreak{}abstraction\_\allowbreak{}factors\_\allowbreak{}optimizer} {\tiny\ttfamily DecisionQuotient/\allowbreak FourInstances.lean}
\item \textbf{\nolinkurl{UB10}}\hypertarget{lh:UB10}{}\enspace{\ttfamily MDPSeparation.\allowbreak{}same\_\allowbreak{}optimizer\_\allowbreak{}class} {\tiny\ttfamily DecisionQuotient/\allowbreak FourInstances.lean}
\item \textbf{\nolinkurl{UB11}}\hypertarget{lh:UB11}{}\enspace{\ttfamily MDPSeparation.\allowbreak{}qStar\_\allowbreak{}vectors\_\allowbreak{}differ} {\tiny\ttfamily DecisionQuotient/\allowbreak FourInstances.lean}
\item \textbf{\nolinkurl{UB12}}\hypertarget{lh:UB12}{}\enspace{\ttfamily MDPSeparation.\allowbreak{}immediate\_\allowbreak{}rewards\_\allowbreak{}differ} {\tiny\ttfamily DecisionQuotient/\allowbreak FourInstances.lean}
\item \textbf{\nolinkurl{UB13}}\hypertarget{lh:UB13}{}\enspace{\ttfamily DecisionProblem.\allowbreak{}exists\_\allowbreak{}reduct\_\allowbreak{}card\_\allowbreak{}le\_\allowbreak{}iff\_\allowbreak{}exists\_\allowbreak{}sufficient\_\allowbreak{}card\_\allowbreak{}le} {\tiny\ttfamily DecisionQuotient/\allowbreak FourInstances.lean}
\item \textbf{\nolinkurl{UB14}}\hypertarget{lh:UB14}{}\enspace{\ttfamily MDPSeparation.\allowbreak{}valueStar\_\allowbreak{}isBellmanOptimal} {\tiny\ttfamily DecisionQuotient/\allowbreak FourInstances.lean}
\item \textbf{\nolinkurl{UB15}}\hypertarget{lh:UB15}{}\enspace{\ttfamily DecisionQuotient.\allowbreak{}optimizerMarkovBlanket\_\allowbreak{}three\_\allowbreak{}way} {\tiny\ttfamily DecisionQuotient/\allowbreak OptimizerMarkovBlanket.lean}
\item \textbf{\nolinkurl{UB16}}\hypertarget{lh:UB16}{}\enspace{\ttfamily DecisionQuotient.\allowbreak{}optimizerMarkovBlanket\_\allowbreak{}iff\_\allowbreak{}staticSufficiency\_\allowbreak{}of\_\allowbreak{}fullSupport} {\tiny\ttfamily DecisionQuotient/\allowbreak OptimizerMarkovBlanket.lean}
\item \textbf{\nolinkurl{UB17}}\hypertarget{lh:UB17}{}\enspace{\ttfamily MarkovBlanketSeparation.\allowbreak{}first\_\allowbreak{}coordinate\_\allowbreak{}is\_\allowbreak{}minimal\_\allowbreak{}blanket} {\tiny\ttfamily DecisionQuotient/\allowbreak OptimizerMarkovBlanket.lean}
\item \textbf{\nolinkurl{UB18}}\hypertarget{lh:UB18}{}\enspace{\ttfamily MarkovBlanketSeparation.\allowbreak{}second\_\allowbreak{}coordinate\_\allowbreak{}is\_\allowbreak{}minimal\_\allowbreak{}blanket} {\tiny\ttfamily DecisionQuotient/\allowbreak OptimizerMarkovBlanket.lean}
\item \textbf{\nolinkurl{UB19}}\hypertarget{lh:UB19}{}\enspace{\ttfamily MarkovBlanketSeparation.\allowbreak{}empty\_\allowbreak{}not\_\allowbreak{}supported\_\allowbreak{}constancy} {\tiny\ttfamily DecisionQuotient/\allowbreak OptimizerMarkovBlanket.lean}
\end{list}
\else
\begin{longtable}{@{}>{\raggedright\arraybackslash}m{0.10\linewidth}>{\raggedright\arraybackslash}m{0.86\linewidth}@{}}
\caption{Lean handle index.}\label{tab:lean-handle-index}\\
\toprule
\textbf{ID} & \textbf{Lean Handle / Source} \\
\midrule
\endfirsthead
\toprule
\textbf{ID} & \textbf{Lean Handle / Source} \\
\midrule
\endhead
\endfoot
\bottomrule
\endlastfoot
\hypertarget{lh:ACT1}{\textbf{\nolinkurl{ACT1}}} & {\fontsize{8}{9}\selectfont\ttfamily DecisionQuotient.\allowbreak{}two\_\allowbreak{}actions\_\allowbreak{}coNP\_\allowbreak{}hard}\par\vspace{0.1ex}{\fontsize{8}{9}\selectfont\ttfamily DecisionQuotient/\allowbreak Tractability/\allowbreak Tightness.lean} \\
\midrule
\hypertarget{lh:CC4}{\textbf{\nolinkurl{CC4}}} & {\fontsize{8}{9}\selectfont\ttfamily DecisionQuotient.\allowbreak{}ClaimClosure.\allowbreak{}anchor\_\allowbreak{}sigma2p\_\allowbreak{}complete\_\allowbreak{}conditional}\par\vspace{0.1ex}{\fontsize{8}{9}\selectfont\ttfamily DecisionQuotient/\allowbreak ClaimClosure.lean} \\
\midrule
\hypertarget{lh:DC10}{\textbf{\nolinkurl{DC10}}} & {\fontsize{8}{9}\selectfont\ttfamily StochasticSequential.\allowbreak{}sequential\_\allowbreak{}to\_\allowbreak{}PSPACE}\par\vspace{0.1ex}{\fontsize{8}{9}\selectfont\ttfamily DecisionQuotient/\allowbreak StochasticSequential/\allowbreak Hierarchy.lean} \\
\midrule
\hypertarget{lh:DC19}{\textbf{\nolinkurl{DC19}}} & {\fontsize{8}{9}\selectfont\ttfamily StochasticSequential.\allowbreak{}stochastic\_\allowbreak{}anchor\_\allowbreak{}sufficient\_\allowbreak{}of\_\allowbreak{}stochastic\_\allowbreak{}sufficient}\par\vspace{0.1ex}{\fontsize{8}{9}\selectfont\ttfamily DecisionQuotient/\allowbreak StochasticSequential/\allowbreak Quotient.lean} \\
\midrule
\hypertarget{lh:DC46}{\textbf{\nolinkurl{DC46}}} & {\fontsize{8}{9}\selectfont\ttfamily StochasticSequential.\allowbreak{}prenexTQBF\_\allowbreak{}iff\_\allowbreak{}empty\_\allowbreak{}alternatingSequentialSufficient}\par\vspace{0.1ex}{\fontsize{8}{9}\selectfont\ttfamily DecisionQuotient/\allowbreak StochasticSequential/\allowbreak AlternatingTQBF.lean} \\
\end{longtable}
\fi
\makeatother
\endgroup

}{%
  \textbf{Error:} \texttt{content/lean\_handle\_ids\_auto.tex} not found.
}

\printbibliography